\documentclass[11pt]{article} 

\usepackage{natbib}
\usepackage{hyperref}
\hypersetup{hidelinks}

\usepackage{booktabs}
\usepackage{enumitem}

\newcommand\arxiv[2]{#1} 

\newcommand\suppSecProofFunctionalForms{S1}
\newcommand\suppSecProofOptimizationConvergence{S2}

\newcommand\suppSecProofTriangularArrays{S4}
\newcommand\suppSecProofConsistency{S5}
\newcommand\suppSecProofBiasVar{S6}
\newcommand\suppSecProofSharperBiasVar{S7}
\newcommand\suppSecBayesFigs{S8}
\newcommand\suppSecAlgorithmicParameters{S9}
\newcommand\suppSecDaphnia{S10}
\newcommand\suppSecMethodComparisonDhaka{S11}
\newcommand\suppSecModelsComparison{S12}

\newcommand\editAIstatement{Feedback from two anonymous referees, an associate editor, and generative artificial intelligence (AI) has led to improvements in this revision.
AI was also employed to support the coding of numerical results.
We used frontier AI models from OpenAI and Anthropic, up to and including Claude Opus 5.5 and GPT-6 Astra.}

\newcommand\articleTitle{Automatic Differentiation for Inference on Partially Observed Markov Process Models}

\newcommand\editTheoremOne{where the derivative with respect to $\theta$ on the right hand side propagates the derivative of the particle locations, via the chain rule, for fixed $\omega$.}

\newcommand\editAncestorDefinition{  
  Here,  $X_{1:N,j}^{A,F,\theta}$ is the ancestral trajectory of $X_{N,j}^{F,\theta}$ at times $t_1,\dots,t_N$.
}

  \newcommand\editPoyiadjisCorrespondence{Theorem~\ref{thm:mop-functional-forms} shows that DMOP-$1$ has a similar form to the consistent but high-variance score estimate in the first algorithm of \cite{poyiadjis11}.
Unlike  \cite{poyiadjis11}, our algorithm has the plug-and-play property; in the case $\alpha=1$, our algorithm replaces differentiation of the latent transition density in \cite{poyiadjis11} by pathwise differentation of the simulator.}

\newcommand\nPrevious{m}
  
\newcommand\editPlugAndPlay{Plug-and-play methods are necessary when transition densities are not available, but can have advantages even when they are.
The plug-and-play property avoids difficulties arising when the full latent transition densities at different parameter values are mutually singular \citep{roberts01,karppinen24}.
The plug-and-play property also facilitates the development of statistical software applicable to a broad class of POMP models.}

\newcommand\newEdit[1]{{#1}}

\usepackage{amsmath,bbm,bm}
\usepackage{amssymb}
\usepackage{amsfonts}
\usepackage{amsthm}
\usepackage{mathtools}

\newtheorem{thm}{Theorem}

\usepackage{tabularx}
\usepackage[normalem]{ulem}

\newcount\comments  
\newcommand{\genComment}[2]{\ifnum\comments=1{\textcolor{#1}{\textsf{\footnotesize #2}}}\fi}

\newcommand{\captiona}{{\em (a)}}
\newcommand{\captionb}{{\em (b)}}

\newcommand\R{\mathbb{R}}

\newcommand\lik{\mathcal{L}}
\newcommand\prob{\mathbb{P}}
\newcommand\E{\mathbb{E}}
\newcommand\loglik{\ell}
\newcommand\process{\texttt{process}}

\newcommand\param{\,;}
\newcommand\giventh\param

\def\eqref#1{equation~\ref{#1}}

\def\gN{{\mathcal{N}}}
\def\gX{{\mathcal{X}}}
\def\gY{{\mathcal{Y}}}

\usepackage{placeins}

\usepackage{tabto}

  \usepackage{fullpage}
  \usepackage{algorithm}
  \usepackage[noend]{algpseudocode}
  \usepackage{xcolor}
  
\AtEndDocument{\refstepcounter{thm}\label{finalthm}}
\AtEndDocument{\refstepcounter{lem}\label{finallem}}
\AtEndDocument{\refstepcounter{defn}\label{finaldefn}}
\AtEndDocument{\refstepcounter{prop}\label{finalprop}}
\AtEndDocument{\refstepcounter{aspt}\label{finalaspt}}
\AtEndDocument{\refstepcounter{figure}\label{finalfig}}

\begin{document}

\title{\articleTitle}

\author{K. Tan$^{1}$, A. J. Abkemeier$^{2}$, K. He$^{2}$, G. Hooker$^{1}$ and E. L. Ionides$^{2}$
  \vspace{3ex}\\
  \small{$^{1}$Department of Statistics, University of Pennsylvania}\\
  \small{$^{2}$Department of Statistics, University of Michigan}}

\date{}

\maketitle

\begin{abstract}
  Automatic differentiation (AD) has driven recent advances in machine learning, including deep neural networks and Hamiltonian Markov Chain Monte Carlo methods. Partially observed nonlinear stochastic dynamical systems have proved resistant to AD techniques because widely used particle filter algorithms yield an estimated likelihood function that is discontinuous as a function of the model parameters. To resolve this, we create a theoretical framework that embeds two existing AD particle filter methods within a new class of algorithms. This new class permits a bias-variance tradeoff and hence a mean squared error substantially lower than the existing algorithms. This allows us to develop likelihood maximization algorithms suited to the Monte Carlo properties of the AD gradient estimate. Our algorithms require only a differentiable simulator for the latent dynamic system; by contrast, most previous approaches to AD likelihood maximization for particle filters require access to the system's transition probabilities. Numerical results indicate that a hybrid algorithm, using AD to refine a coarse solution from an iterated filtering algorithm, shows substantial improvement on current state-of-the-art methods on a challenging scientific benchmark problem.
\end{abstract}


\section{Introduction}

Many scientific models involve highly nonlinear stochastic dynamic systems possessing significant random variation in both the process dynamics and the measurements.
Commonly, the latent system is modeled as a Markov process, giving rise to a partially observed Markov process (POMP) model, also known as a hidden Markov models or a state space model.
POMP models arise in fields as diverse as optimal control \citep{singh22}, epidemiology \citep{he10, stocks17}, ecology \citep{knape12} and finance \citep{kim08, breto14}.
Despite their ubiquity, estimation and inference for this broad class of models remains a challenging problem.
We develop full-information plug-and-play inference for POMP models that takes advantage of automatic differentiation (AD) computation while controlling the Monte Carlo bias and variance of the resulting estimator.
These desiderata have not previously been combined, and we discuss each component in turn.

Full-information methods are those which make statistically efficient use of data, typically via likelihood-based or Bayesian methods that avoid dependence on summary statistics.
The particle filter, also known as sequential Monte Carlo, serves as the foundation for various full-information inference algorithms for POMP models, since it provides an unbiased estimate of the likelihood function \citep{delMoral04}, enabling Bayesian inference via particle Markov chain Monte Carlo \citep{andrieu10} and likelihood-based inference via iterated filtering algorithms IF1 \citep{ionides06-pnas} and IF2 \citep{ionides15}.

An inference algorithm for a POMP model is said to be plug-and-play if it does not need to evaluate the transition density of the latent Markov process, enabling an arbitrary model simulator to be plugged into the algorithm, a convenient property that facilitates scientific discovery \citep{he10}.
Some plug-and-play methods depend on the construction of low-dimensional summary statistics \citep{wood10,toni09}, sacrificing statistical efficiency for computational convenience.
We develop computational methods that maintain statistical efficiency and the plug-and-play property.
Plug-and-play methods are sometimes called likelihood-free \citep{owen15}, and it may be counter-intuitive that likelihood-based inference is possible using likelihood-free methods. 
Basic particle filters, also known as boostrap filters, possess the plug-and-play property.
Iterated particle filtering algorithms, such as IF1 and IF2, have shown that plug-and-play likelihood maximization is practical in various scientific investigations \citep{king08,blake14,pons-salort18,subramanian21,fox22,drake23}.
However, iterated filtering algorithms have slow convergence near the MLE \citep{doucet13}

Advances in AD for gradient estimation using particle filters \citep{naesseth18, jonschkowski18, corenflos21, scibior21, singh22} have drawn attention to AD as a tool for inference in POMP models.
Widespread adoption of these approaches have been limited by their bias \citep{naesseth18, jonschkowski18}, high variance \citep{poyiadjis11, scibior21}, computational expense quadratic in the number of particles \citep{corenflos21, chen24}, or by the loss of the plug-and-play property \citep{poyiadjis11, scibior21, singh22, chen24}.
We present a new algorithm which (for reasons which will be discussed later) we call a {\it Differentiated Measurement Off-Parameter} (DMOP) particle filter.
DMOP corresponds to the derivative of a smooth variant of the particle filter that we call MOP, though DMOP is conveniently calculated using AD without the need to actually implement the full MOP filter.
MOP has the special property that, unlike most particle filters, the gradient of the filtered log-likelihood estimate provides a consistent estimate of the gradient of the log-likelihood.
Particle filters have previously been adapted to provide a smooth estimate of the likelihood function \citep{svensson18,malik11}, but MOP and DMOP differ from these by possessing the plug-and-play property.

DMOP trades bias against variance via a discount factor, $\alpha \in [0,1]$.
We find that DMOP-$\alpha$ posesses a bias-variance tradeoff that permits both low bias and low variance for suitable choices of $\alpha$ strictly in $(0,1)$, both in theory (Theorem \ref{thm:mop-biasvar}) and practice (Figure \ref{fig:biasvar}).
Theorem~\ref{thm:mop-functional-forms} shows that DMOP-$\alpha$ coincides with the biased gradient estimator of \cite{naesseth18} when $\alpha=0$.
In a limit as the number of particles increases, Theorem~\ref{thm:mop-functional-forms} also shows that DMOP-$\alpha$ with $\alpha=1$ matches the high-variance gradient estimate from \cite{poyiadjis11} and \cite{scibior21}.
In this case, the variance can become prohibitively high, as discovered by these researchers and others \citep{arya22}.
However, the variance rapidly reduces as $\alpha$ decreases, and so a choice of $\alpha$ close to one can provide a low-bias, low-variance, plug-and-play gradient estimate with computational scaling linear in the number of particles.

In practice, likelihood optimization via gradient descent with DMOP-$\alpha$ converges quickly within a neighborhood of the maximum but struggles to reach this neighborhood in the presence of local minima and saddle points.
By contrast, iterated filtering algorithms provide a relatively fast and stable way to identify this neighborhood.
We therefore build a hybrid algorithm called {\it Iterated Filtering with Automatic Differentiation} (IFAD) that warm-starts first-order or second-order gradient methods with a preliminary solution obtained from iterated filtering.
We implement our algorithm using JAX \citep{jax}, which also enables us to take advantage of graphical processing unit (GPU) hardware.
Our methodology is provided in an open-source software library, Pypomp \citep{abkemeier25}.

Numerical results indicate that IFAD beats IF2 (and by the numerical results of \cite{ionides15}, also IF1 and the Liu-West filter \citep{liuwest01}) on a challenging problem in epidemiology, the Dhaka cholera model of \cite{king08}.
IFAD also outperforms marginalized panel iterated filtering on a challenging ecological panel problem, the 38-parameter \textit{Daphnia} model of \citet{yang25daphnia}.
\newEdit{\editPlugAndPlay}

We also present an exploratory application of DMOP-$\alpha$ to Bayesian computation.
In Section~\ref{sec:bayes}, we demonstrate the DMOP-$\alpha$ gradient estimates within a No-U-Turn Sampler (NUTS) \citep{hoffman14}.
The numerical results provide preliminary evidence of faster parameter-space exploration than a previous application of particle Markov chain Monte Carlo (PMCMC) on the Dhaka cholera model of \cite{king08}, though we also report that adaptive choice of the PMCMC proposal can recover much of the difference.

\section{Problem Setup}

Consider an unobserved Markov process $\{X(t),t  \geq t_0\}$, with discrete-time observations $Y_1,...,Y_N$ realized at values $y_1^*,\dots,y_N^*$ with times $t_1,\dots, t_N$.
The process is parameterized by an unknown parameter $\theta \in \Theta \subseteq \R^p$, where the state $X(t)$ takes values in the state space $\gX \subseteq \R^d$, the observations $Y_n$ take values in $\gY,$ and we write $X_n = X(t_n)$. We suppose that the discrete-time latent process model has a density $f_{X_n|X_{n-1}}\left(x_{n} \mid x_{n-1}; \theta\right)$.
Although we assume this density exists, plug-and-play inference algorithms do not explicitly evaluate it.
Instead, they have access to a corresponding simulator, which write as $\process_n\left(x_n\mid x_{n-1}; \theta\right)$.
We set $f_{Y_n|X_n}\left(y_n \mid x_n; \theta\right)$ to be the measurement density, which can be directly evaluated. 
We call $f_{X_{1:n}|Y_{1:n}}(x_{1:n}|y_{1:n}^*; \theta)$ the posterior distribution of states, and $f_{X_{n}|Y_{1:n}}(x_n|y_{1:n}^*; \theta)$ the filtering distribution at time $t_n$.
All joint and conditional densities are defined on a probability space $(\Omega, \Sigma, \prob)$ which is also assumed to enable construction of independent replicates and all random variables defined in our algorithms.
For the algorithmic interpretation of our theory, we identify the random number seed with an element of the sample space $\omega \in \Omega$.
\newEdit{We assume that $\omega$ can be used to construct a sequence $\{\omega_k, k\in \mathbb{Z}\}$ that comprise independent uniform random variables in the mathematical formulation and a sequence of pseudo-random numbers in the equivalent computational formulation. 
We require that $\process_n$ is a function of $\omega_n$, though this dependence is not written explicitly.
Mathematically and computationally, we can permit $\omega_n$ to be vector valued.}

\section{Off-Parameter Particle Filters}

The {\it Measurement Off-Parameter with discount factor} $\alpha$ (MOP-$\alpha$), particle filter is defined by the pseudocode in Algorithm~\ref{alg:mop}.
Note that this algorithm has the plug-and-play property.
We suppose that the simulator is a differentiable function of $\theta$ for every fixed $\omega$.
This condition is equivalent to the {\it reparameterization trick} \citep{corenflos21}; by writing the generated random numbers as parameter-dependent transformation from a reference distribution---e.g., the quantile transform of a uniform random variable---we can differentiate their values, and hence the simulation, with respect to parameters.
The discrete resampling step of the bootstrap particle filter has been the primary point of concern for the use of AD to estimate derivatives for particle filters.
The discontinuity of this resampling (as a function of $\theta$, for fixed $\omega$) means that the derivative of the expectation is not consistently estimated by the expectation of the derivative. 
MOP-$\alpha$ is a mathematical construction designed to avoid this discontinuity.
We begin with a filter that resamples particles with weights calculated at a baseline parameter value, $\phi$, and then corrects these weights to target the filtering distribution at another parameter value, $\theta$, yielding the name {\it measurement off-parameter}.
The terminology is chosen by analogy with similar {\it off-policy} algorithms used for reinforcement learning in which a policy that was not implemented is evaluated by reweighting the rewards of the policy that was actually applied.
If we fix $\omega$ and $\phi$, the sample indices do not change with $\theta$, but the re-weighting scheme does, allowing MOP-$\alpha$ to be differentiated with respect to $\theta$.
When $\theta = \phi$, we have $w_{n,j}^{P,\theta} = w_{n-1,j}^{F,\theta} = 1$ and MOP-$\alpha$ reduces  to the standard particle filter; this is the setting we use in practice,  but the construction allows us to define a derivative.
One can think of MOP-$\alpha$ as reconstructing local smooth approximations of the discontinuous likelihood estimates produced by a particle filter for a fixed number of particles $J$ and fixed seed $\omega$.
These approximate the likelihood at $\theta$ in a neighborhood around $\phi$, with the intent that differentiating these smooth log-likelihood reconstructions yields a score estimate.

DMOP-$\alpha$ (Algorithm~\ref{alg:dmop}) computes the fixed-seed derivative of the likelihood estimate obtained by MOP-$\alpha$ at $\theta=\phi$.
To better understand the implementation in DMOP-$\alpha$, we review some background for AD.
 Standard AD consists of two evaluation passes: a forward pass in which the function is evaluated and a graph built of the evaluation process, followed by a backward pass where the derivative is evaluated by recursive use of the chain rule on the computation graph \citep{griewank08}.  
  The two evaluation passes in MOP-$\alpha$ at $\phi$ and $\theta$ correspond to the two AD evaluation passes in DMOP-$\alpha$.
  Differentiation of MOP-$\alpha$ with respect to $\theta$ requires a preliminary evaluation at $\phi$ and subsequent differentiation with respect to $\theta$ for the second pass which we obtain through AD.
In practice, however, we only work in the $\theta = \phi$ setting in which case MOP-$\alpha$ corresponds to the standard particle filter and we only need to make one pass.
In order to apply AD in this setting without using two passes, we must both instantiate the weighting scheme (even though the weights are always 1) and distinguish which components of the weights correspond to which pass.
This can be achieved using the {\it stop-gradient} capability that appears in Algorithm~\ref{alg:dmop} where its application in the weight correction step indicates that AD should apply the chain rule to the numerator but not the denominator of the correction; note that this construction yields weights which are always 1, but which also have non-zero derivatives.

\begin{figure}[htbp!]
  \centering
    \includegraphics[width=\textwidth/\real{2.2}]{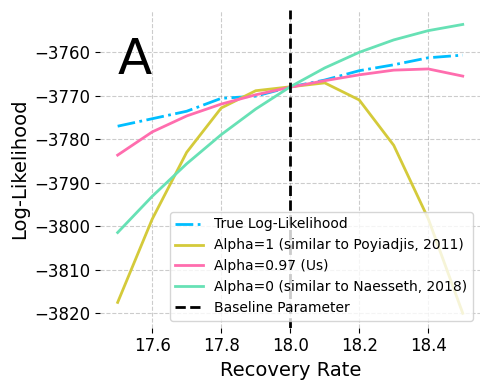}
    \includegraphics[width=\textwidth/\real{2.2}]{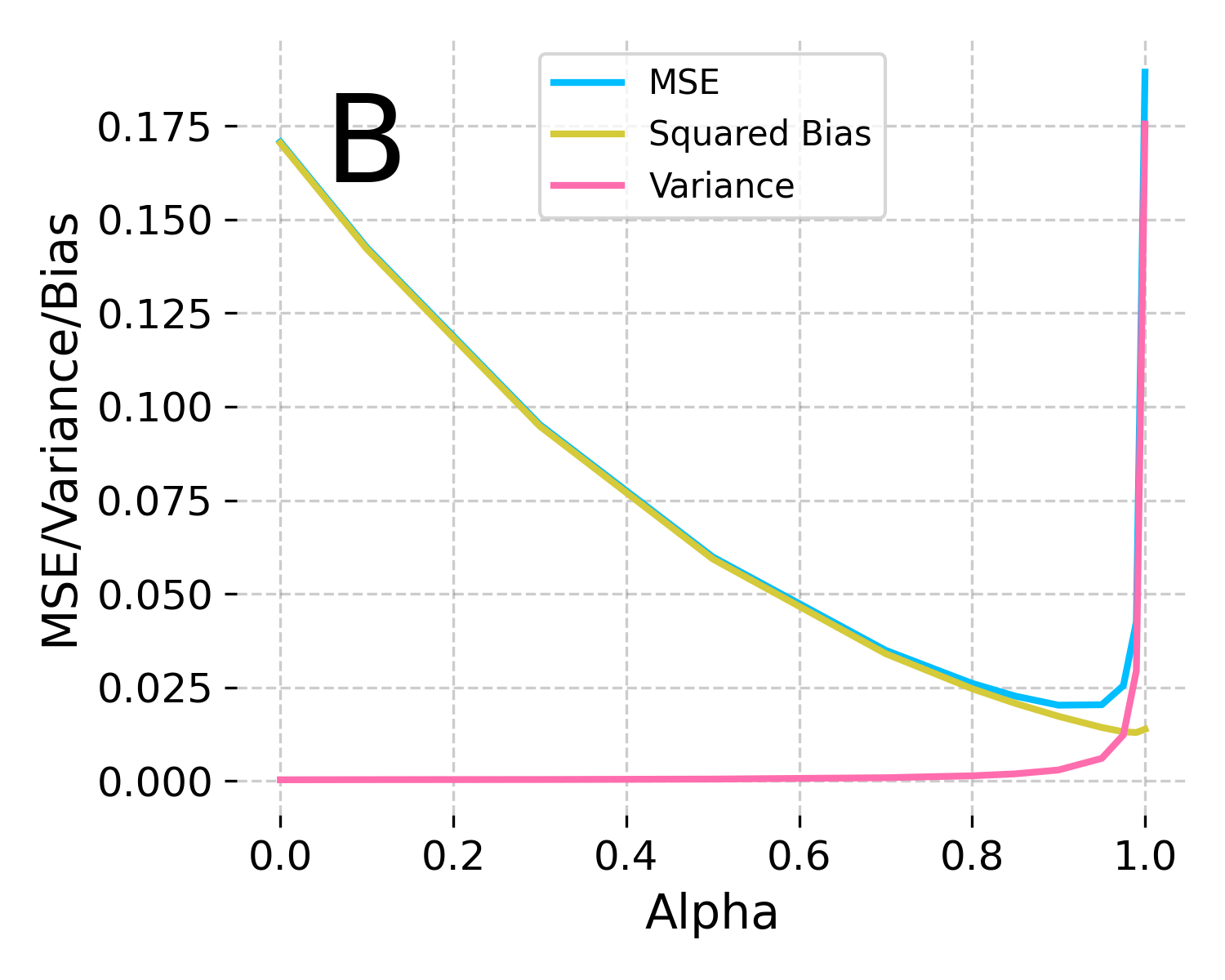}
    \caption{\textbf{A:} Non‐differentiable likelihood estimates and smooth off‐parameter approximations that hold locally around the baseline parameter $\phi$. \textbf{B:} Illustration of the bias-variance tradeoff induced by the discounting hyperparameter $\alpha$. We display the MSE of score estimates for the recovery rate at the MLE for the Dhaka cholera model of \cite{king08}. In this example, intermediate values of $\alpha$ improve the likelihood approximation and the score estimates.}
    \label{fig:biasvar}
\end{figure}

\begin{algorithm}[htbp!]
	\caption{MOP-$\alpha$}
    \label{alg:mop}
	     \textbf{Input:} Number of particles $J$, timesteps $N$, measurement model $f_{Y_n|X_n}(y_n^*|x_n, \theta)$, simulator $\process_n(\cdot|x_n; \theta)$, evaluation parameter $\theta$, baseline parameter $\phi$, seed $\omega$.
      
        \textbf{First pass:} Set $\theta=\phi$ and fix $\omega$, yielding $X_{n,j}^{P,\phi}$, $X_{n,j}^{F,\phi}$, $g^{\phi}_{n,j}$.
            
        \textbf{Second pass:}
        Fix $\omega$, and filter at $\theta\neq \phi$:
            
		\textbf{Initialize } particles ${X}_{0,j}^{F,\theta}\sim {f}_{{X}_{0}}\left(\cdot\giventh{\theta}\right)$, weights $w^{F,\theta}_{0,j}= 1$. \newline
		\textbf{For} $n=1,...,N$: \newline
            \hspace*{4mm} Accumulate discounted weights, $w_{n,j}^{P,\theta} = \big(w_{n-1,j}^{F,\theta}\big)^\alpha$.\newline
            \hspace*{4mm} Simulate process model,
            ${X}_{n,j}^{P,\theta}\sim \process_n\big(\cdot|{X}_{n-1,j}^{F, \theta};{\theta}\big)$. \newline
            \hspace*{4mm} Measurement density,
            $g^{\theta}_{n,j}={f}_{{Y}_{n}|{X}_{n}}(y_{n}^{*}|{X}_{n,j}^{P,\theta}\giventh{\theta})$. \newline
            \hspace*{4mm} Compute $L_n^{B,\theta,\alpha} ={\sum_{j=1}^Jg^\theta_{n,j} \, w^{P,\theta}_{n,j}}\, \big/\, {\sum_{j=1}^J  w^{P,\theta}_{n,j}}$. \newline
            \hspace*{4mm} Conditional likelihood under $\phi$,
            $L_n^{\phi} = \frac{1}{J}\sum_{m=1}^{J}g^{\phi}_{n,m}$.\newline
            \hspace*{4mm} Select resampling indices $k_{1:J}$ with $\prob\big(k_{j}=m\big) \propto g^{\phi}_{n,m}$. \newline
            \hspace*{4mm} Obtain resampled particles ${X}_{n,j}^{F,\theta}={X}_{n,k_{j}}^{P,\theta}$. \newline
            \hspace*{4mm} Calculate corrected weights
            $w_{n,j}^{F,\theta}= w^{P,\theta}_{n,k_j} \, g^{\theta}_{n,k_j} \, \big/ \, { g^{\phi}_{n,k_j}}$.\newline
            \hspace*{4mm} Compute $ L_n^{A,\theta,\alpha} = L_n^\phi\cdot {\sum_{j=1}^J w^{F,\theta}_{n,j}} \, \big/ \, {\sum_{j=1}^J  w^{P,\theta}_{n,j}}$.\newline
	    \textbf{Return:} likelihood estimate $\hat{\lik}^A(\theta) = \prod_{n=1}^N L_n^{A,\theta,\alpha}$ or $\hat{\lik}^B(\theta) = \prod_{n=1}^N L_n^{B,\theta,\alpha}$,
            log-likelihood $\hat{\loglik}^A(\theta)=\log\big(\hat{\lik}^A(\theta)\big)$ or  $\hat{\loglik}^B(\theta)=\log\big(\hat{\lik}^B(\theta)\big)$,
            and filtering distributions $\{(X_{n,j}^{F, \theta}, w^{F,\theta}_{n,j})\}_{n,j=1}^{N,J}.$
\end{algorithm}

\begin{algorithm}[htbp!]
	\caption{DMOP-$\alpha$}
    \label{alg:dmop}
	     \textbf{Input:} Number of particles $J$, timesteps $N$, measurement model $f_{Y_n|X_n}(y_n^*|x_n, \theta)$, simulator $\process_n(\cdot|x_n; \theta)$, parameter $\theta$.
      
             \textbf{Forward pass:} AD constructs a computation graph for the likelihood estimate, $\hat{\lik}^A(\theta) = \prod_{n=1}^N L_n^{A,\theta,\alpha}$ or $\hat{\lik}^B(\theta) = \prod_{n=1}^N L_n^{B,\theta,\alpha}$. 
            
        \textbf{Backward pass:} AD evaluates the gradient estimate, $\hat{\mathcal{D}}^A(\theta)$ or  $\hat{\mathcal{D}}^B(\theta)$.
            
		\textbf{Initialize } particles ${X}_{0,j}^{F,\theta}\sim {f}_{{X}_{0}}\left(\cdot\giventh{\theta}\right)$, weights $w^{F,\theta}_{0,j}= 1$. \newline
		\textbf{For} $n=1,...,N$: \newline
            \hspace*{4mm} Accumulate discounted weights, $w_{n,j}^{P,\theta} = \big(w_{n-1,j}^{F,\theta}\big)^\alpha$.\newline
            \hspace*{4mm} Simulate process model,
            ${X}_{n,j}^{P,\theta}\sim \process_n\big(\cdot|{X}_{n-1,j}^{F, \theta};{\theta}\big)$. \newline
            \hspace*{4mm} Measurement density,
            $g^{\theta}_{n,j}={f}_{{Y}_{n}|{X}_{n}}(y_{n}^{*}|{X}_{n,j}^{P,\theta}\giventh{\theta})$. \newline
            \hspace*{4mm} Compute $L_n^{B,\theta,\alpha} ={\sum_{j=1}^Jg^\theta_{n,j} \, w^{P,\theta}_{n,j}}\, \big/\, {\sum_{j=1}^J  w^{P,\theta}_{n,j}}$. \newline
            \hspace*{4mm} Select resampling indices $k_{1:J}$ with $\prob\big(k_{j}=m\big) \propto g^{\theta}_{n,m}$. \newline
            \hspace*{4mm} Obtain resampled particles ${X}_{n,j}^{F,\theta}={X}_{n,k_{j}}^{P,\theta}$. \newline
            \hspace*{4mm} Calculate corrected weights
            $w_{n,j}^{F,\theta}= w^{P,\theta}_{n,k_j} \, g^{\theta}_{n,k_j} \, \big/ \, { \mathtt{stop\_gradient}(g^{\theta}_{n,k_j})}$.\newline
            \hspace*{4mm} Compute $ L_n^{A,\theta,\alpha} = \mathtt{stop\_gradient}\big(L_n^{B,\theta,\alpha}\big)\, {\sum_{j=1}^J w^{F,\theta}_{n,j}} \, \big/ \, {\sum_{j=1}^J  w^{P,\theta}_{n,j}}$.\newline
	    \textbf{Return:} log-likelihood estimate $\hat{\loglik}^A(\theta)= \log\big(\hat{\lik}^A(\theta)\big) = \sum_{n=1}^N \log L_n^{A,\theta,\alpha}$ or $\hat{\loglik}^B(\theta)=\log\big(\hat{\lik}^B(\theta)\big) = \sum_{n=1}^N \log L_n^{B,\theta,\alpha}$, and gradient estimate $\hat{\mathcal{D}}^A(\theta)$ or $\hat{\mathcal{D}}^B(\theta)$.
\end{algorithm}

MOP-$\alpha$ has similarities with the algorithm of \cite{svensson18}, in which particles from a previous run at $\phi$ are reweighted to evaluate the likelihood at any sufficiently nearby parameter $\theta$ for $0$-th order optimization.
However, their method is not plug-and-play because they use transition densities to reweight particles, and the variance of their log-likelihood estimate has an exponential dependence in the horizon, $N$, when $\theta\neq\phi$.
Our approach bypasses these issues: the first since the reparameterization trick allows us to consider particle paths dependent on $\theta$, and the second as using AD allows us to only evaluate the likelihood and score estimates when $\theta=\phi$. 

\subsection{Some algorithmic details}

The resampling indices, which we write as $k_j \sim \text{Categorical}(g^{\phi}_{n,1},...,g^{\phi}_{n,J})$, are a function of $\phi$ for any value of $\theta$.
For the categorical distribution, we use systematic resampling \citep{arulampalam02,king16} which usually outperforms multinomial resampling.
The finite-sample concentration result used for Theorem~\ref{thm:linear},
however, is stated for multinomial resampling; it does not assert the same
guarantee for the systematic-resampling implementation used in our numerical
work.

In Algorithm~\ref{alg:mop}, the conditional likelihoods are reweighted by a correction factor $w^{P,\theta}_{n,j}$ that accumulates over filtering steps to account for the resampling under $\phi$. 
Writing $g^{\theta}_{n,j}={f}_{{Y}_{n}|{X}_{n}}(y_{n}^{*}|{X}_{n,j}^{P,\theta}\giventh{\theta})$ for the measurement density and $L_n^{\phi} = \frac{1}{J}\sum_{m=1}^{J}g^{\phi}_{n,m}$ for the conditional likelihood estimate under $\phi$, we obtain two estimates of the conditional likelihood under $\theta$,
\begin{equation}
     \label{eq:mop-conditional-likelihood}
     L_n^{B,\theta,\alpha} = \frac{\sum_{j=1}^Jg^\theta_{n,j}\, w^{P,\theta}_{n,j}}{\sum_{j=1}^J  w^{P,\theta}_{n,j}}, \arxiv{\hspace{15mm}}{\;\;}
     L_n^{A,\theta,\alpha} = L_n^\phi\cdot \frac{\sum_{j=1}^J w^{F,\theta}_{n,j}}{\sum_{j=1}^J  w^{P,\theta}_{n,j}},
\end{equation}
where we recursively correct the weights for each particle for the cumulative error incurred by the off-parameter resampling:
\begin{equation}
    \label{eq:weighting-scheme}
    w_{n,j}^{P,\theta} = (w_{n-1,j}^{F,\theta})^\alpha, 
    \arxiv{\hspace{5mm}}{\;}
    w^{F,\theta}_{n,j} = w^{P,\theta}_{n,k_j} \, g^{\theta}_{n,k_j} \, \big/ g^{\phi}_{n,k_j}, 
    \arxiv{\hspace{5mm}}{\;}
    w^{F,\theta}_{0,j}= 1.
\end{equation}
The before-resampling estimate $L_n^{B,\theta,\alpha}$ is preferable in practice, as it has slightly lower variance than the after-resampling estimate $L_n^{A,\theta,\alpha}$, but the latter is useful in deriving properties of the MOP-$\alpha$ gradient estimate such as that in Theorem \ref{thm:mop-functional-forms}.

\subsection{Discounted Weights and a Bias-Variance Tradeoff}

Heuristically, $\alpha$ controls the exponentially decaying memory of MOP and DMOP.
At $\theta=\phi$, the weights equal one, so $\alpha$ does not affect the function evaluation in the first pass of DMOP. It does affect the gradients evaluated in the second pass.
If $\alpha=1$, the gradient of the weights, $w^{P,\theta}_{n,j}$ and $w^{F,\theta}_{n,j}$, accumulates as $n$ increases.
This leads to numerical instability and high variance unless $N$ is small.
For $\alpha<1$, the weights from previous timesteps are discounted as in Equation~(\ref{eq:weighting-scheme}), allowing for lower variance at the expense of bias in the gradient estimate.

When $\alpha=1$, MOP-$1$ maintains complete memory of each particle's ancestral trajectory.
\newEdit{\editPoyiadjisCorrespondence}

When $\alpha=0$, MOP-$0$ considers only single-step transition dynamics, and the before-sampling variant of DMOP-$0$ recovers the low-variance but asymptotically biased score estimate of \cite{naesseth18}.
\cite{scibior21} showed that this matches the gradient of a vanilla particle filter when resampling terms are dropped. 

\begin{thm}[DMOP-$0$ and DMOP-$1$ Functional Forms]
    \label{thm:mop-functional-forms}
    Write $\nabla_\theta \hat\ell^\alpha(\theta)$ for the gradient estimate yielded by DMOP-$\alpha$ 
    and use the after-resampling conditional likelihood estimate so that $\hat\loglik(\theta)= \log \hat\lik(\theta) = \log \prod_{n=1}^N L_n^{A, \theta, \alpha}$. When $\alpha=0$,
    \begin{equation} \nonumber 
        \nabla_\theta \hat\ell^0(\theta) 
        = \frac{1}{J} \sum_{n=1}^N \sum_{j=1}^J \nabla_\theta \log f_{Y_n|X_{n}}\left(y_n^*|X_{n,j}^{F, \theta}; \theta\right),
    \end{equation}
    \newEdit{\editTheoremOne}
    When $\alpha=1$,
    \begin{equation} \nonumber 
        \nabla_\theta \hat{\ell}^1(\theta) 
        = \frac{1}{J}\sum_{j=1}^J \nabla_\theta \log f_{Y_{1:N}|X_{1:N}}\left(y_{1:N}^* | X_{1:N,j}^{A,F,\theta}; \theta \right).
    \end{equation}
    \newEdit{\editAncestorDefinition}
\end{thm}

We defer the proof to the supplement (Section~\suppSecProofFunctionalForms). 
The argument relies on a useful decomposition of the after-resampling estimate $L_n^{A,\theta,\alpha}$, repeated applications of the log-derivative identity that $\nabla_x \log(f(x)) = (\nabla_x f(x))/f(x)$, and noting $\theta=\phi$ implies that $w_{n,j}^{P,\theta}$ evaluates to $1$. 

The result in Theorem \ref{thm:mop-functional-forms} provides mathematical intuition for why \cite{scibior21}'s use of the stop-gradient procedure in the particle filter recovers the estimate of \cite{poyiadjis11}.
When $\theta=\phi$, the correction $g_{n,k_j}^\theta/g_{n,k_j}^\phi$ is equivalent to $g_{n,k_j}^\theta / \text{stop\_gradient}(g_{n,k_j}^\theta)$, and is implemented in this way in practice.
As such, the use of the stop-gradient in \cite{scibior21} can be seen as the MOP-$\alpha$ weight correction for $\alpha=1$, and $\theta=\phi$.

\section{From Gradients to Maximum Likelihood Estimation}

To design an effective procedure for likelihood maximization using DMOP-$\alpha$, we carry out iterated filtering then automatic differentiation (IFAD). 
A basic IFAD approach is presented in Algorithm~\ref{alg:ifad}.
The second stage uses a symmetric positive-definite curvature matrix $B_m$ whose
eigenvalues are kept in a fixed interval $[c,C]$.  A raw particle Hessian or a
quasi-Newton matrix may be spectrally clipped to this interval; the choice
$B_m=I_p$ gives stochastic gradient descent (SGD).
Algorithm~\ref{alg:ifad} is convenient for theoretical analysis, but two further modifications are helpful in practice.
Our best numerical results were obtained by putting the gradient estimate into the Adam optimizer \citep{kingma15} together with learning rate decay. 
Numerical results are provided in Section~\ref{sec:cholera}.

\begin{algorithm}[htbp!]
	\caption{IFAD}
    \label{alg:ifad}
	    \textbf{Input:} Number of particles $J$, timesteps $N$, IF2 cooling schedule $a_m$, AD step size $\eta$, MOP-$\alpha$ discount parameter $\alpha$, spectral bounds $0<c\leq C$, gradient tolerance $\epsilon$, stopping multiplier $\sigma$, maximum AD iterations $M$, and initial value $\theta_0$.\newline
        Run the initial IF2 search under cooling schedule $a_m$ to obtain $\{\Theta_j, j=1,...,J\}$, and set $\theta_0 := J^{-1}\sum_{j=1}^J \Theta_j.$\newline
		\textbf{For} $m=0,\ldots,M-1$: \newline
		\hspace*{4mm} Using fresh random numbers, run Algorithm \ref{alg:mop} to obtain $\hat\loglik(\theta_m)$ and set $g_m := \nabla_{\theta}(-\hat\loglik(\theta_m))$. \newline
		\hspace*{4mm} \textbf{If} $\|g_m\|_2\leq\sigma\epsilon$, \textbf{return} $\hat\theta:=\theta_m$. \newline
		\hspace*{4mm} Choose a symmetric raw curvature estimate $\widetilde B_m$ and clip its eigenvalues to $[c,C]$ to obtain $B_m$. \newline
		\hspace*{4mm} Update $\theta_{m+1} := \theta_m - \eta B_m^{-1} g_m$. \newline
		\textbf{Return} $\hat{\theta} := \theta_M.$
\end{algorithm}

The result below analyzes the specialization of Algorithm~\ref{alg:ifad} that
uses $\alpha=1$, the after-resampling likelihood estimate, and multinomial
resampling.  These restrictions are needed for the stated finite-sample
probability bound.

The convergence of IF2 (and so the first stage of IFAD) to a neighborhood of the MLE happens rapidly in practice. 
However, even when applying much larger amounts of Monte Carlo effort, IF2 fails to get much closer than 2 log-likelihood units of the maximum on this problem (Figure~\ref{fig:boxplot}).
For this example, IF2, without assistance from AD, struggles to reach the inferentially critical region within one log-likelihood unit of the maximum, even when additional computationally intensive techniques are used, such as repeated pruning and profiling,

For a fixed optimization budget, the AD stage has a high-probability linear
convergence rate up to a neighborhood of the MLE that shrinks as the particle
gradient accuracy improves, under the usual smoothness and strong convexity
assumptions, as shown below.

\begin{thm}[Linear Convergence of IFAD]
\label{thm:linear}\label{thm:mop-convergence}
Condition on the IF2 output $\theta_0$ and write $f=-\ell$.  Suppose $f$ is
$\gamma$-strongly convex and $\Gamma$-smooth on a convex region containing the
second-stage trajectory, with unique minimizer $\theta^*$, and that (A3) and
(A5) hold uniformly there.  Consider the $\alpha=1$, after-resampling,
multinomial specialization of Algorithm~\ref{alg:ifad}, using fresh particle
randomness at each iteration and curvature matrices
$cI_p\preceq B_m\preceq CI_p$.

Fix $M\geq1$, $\delta\in(0,1)$, $\epsilon>0$, and $\beta\in(0,1)$, and choose
$0<\eta\leq c(1-\beta)/\Gamma$ and
$\sigma\geq2C/[c(1-\beta)]$.  For $J$ sufficiently large, as quantified in
Section~\suppSecProofOptimizationConvergence, with conditional probability at
least $1-\delta$ every nonstopping iteration up to $M$ satisfies
\[
\ell(\theta^*)-\ell(\theta_{m+1})
\leq q\{\ell(\theta^*)-\ell(\theta_m)\},
\qquad
q=1-\frac{2\eta\beta\gamma}{C}
\left(\frac{\sigma}{\sigma+1}\right)^2\in(1/2,1).
\]
Thus, if no stop occurs, the final likelihood gap is at most
$q^M\{\ell(\theta^*)-\ell(\theta_0)\}$; if the algorithm stops at
$\theta_\tau$, then
$\|\nabla\ell(\theta_\tau)\|_2\leq(\sigma+1)\epsilon$ and
$\ell(\theta^*)-\ell(\theta_\tau)
\leq(\sigma+1)^2\epsilon^2/(2\gamma)$.
\end{thm}
The proof of Theorem~\ref{thm:linear} is a finite-horizon inexact-gradient
argument.  It retains the two-sided spectral control and the per-iteration
probability accounting that are essential in the related subsampled-Newton
analysis of \cite{mahoney16}.
Details are deferred to the supplement (Section~\suppSecProofOptimizationConvergence). 
This result shows that, over any prespecified finite optimization budget,
particle gradient estimates can retain a linear rate until their Monte Carlo
accuracy becomes the limiting factor.  The theorem is stated for DMOP-$1$:
an analogous result for $\alpha<1$ additionally requires a high-probability
bound on the discount bias relative to the exact score and is not implied merely
by taking $\alpha$ close to one.

We therefore see that the second stage of IFAD converges linearly toward a
Monte Carlo neighborhood of the MLE if (1) the log-likelihood surface is
well-conditioned in a neighborhood of the MLE and (2) the first stage of IFAD
successfully reaches a basin in which the theorem's assumptions hold.
Within a neighborhood of the MLE that grows as more data become available \citep{ning21}, local quadratic approximations such as local asymptotic normality apply under general conditions \citep{lecam00}.
This suggests that the linear convergence may occur often in practice.

\section{Application to a Cholera Transmission Model}
\label{sec:cholera}

The cholera transmission model that \cite{king08} developed for Dhaka, Bangladesh, has been used to benchmark the performance of various POMP inference methods \citep{ionides15, fasiolo16, wycoff26}, and we employ it here for the same purpose.
This model categorizes individuals in a population as susceptible, $S(t)$, infected, $I(t)$, and recovered, $R(t)$ and so is called an SIRS compartmental model.
In this case, the compartment $R(t)$ is further subdivided into three recovered compartments $R^1(t)$, $R^2(t)$, $R^3(t)$ denoting varying degrees of cholera immunity.
We write $P(t)$ for the total population, and $M_n$ for the cholera deaths in each month.
As in \cite{king08}, the transition dynamics follow a stochastic differential equation (SDE):
\begin{align*}
    dS&=\big(k \epsilon R^k+\delta(P-S)-\lambda(t) S\big)\, dt+d P-({\sigma S I}/{P})\, dB, \\
    dI&=\big(\lambda(t) S-(m+\delta+\gamma) I\big)\, dt+({\sigma S I}/{P})\, dB, \\
    dR^1&=\big(\gamma I-(k \epsilon+\delta) R^1\big)\, dt, \hspace{3ex} \dots \\
    dR^k&=\big(k \epsilon R^{k-1} -(k \epsilon+\delta) R^k\big)\, dt,
\end{align*}
with Brownian motion $B(t)$, cholera death rate $m$, recovery rate $\gamma$, mean immunity duration $1/\epsilon$, standard deviation of the force of infection $\sigma$, and population death rate $\delta=0.02$. The force of infection, $\lambda_t$, is modeled by splines $(s_j)_{j=1}^6$
\begin{equation*}
  \lambda_t=\exp\hspace{-1mm}\bigg(\hspace{-.5mm}\beta_{\text{trend}}(t-t_0)+\hspace{-1mm}\sum_{j=1}^{6} \beta_j s_j(t)\hspace{-1mm}\bigg)\hspace{-1mm}\frac{I}{P} + \exp\hspace{-1mm} \bigg(\sum_{j=1}^{6} \omega_j s_j(t)\hspace{-1mm}\bigg)\hspace{-1mm},
    \arxiv{}{\vspace*{-3ex}}
\end{equation*}
where the coefficients $(\beta_j)_{j=1}^6$ model seasonality in the force of infection, $\beta_{\text{trend}}$ models the trend in the force of infection, and the $\omega_j$ represent seasonality of a non-human environmental reservoir of disease.
The measurement model for observed monthly cholera deaths is
    $Y_n \sim \gN(M_n, \tau^2M_n^2)$,
where $M_n= m\int_{t_{n-1}}^{t_n}I(s)\, ds$ is the modeled cholera deaths in that month.

We numerically solve this SDE using Euler's method with a time step of $1/20$ months, matching \citep{king08}.
This example is a partially observed hypo-elliptical SDE, and methodology specific to this model class has been previously studied by \citet{ditlevsen19} (henceforth, DS19).
Supplementary Section~{\suppSecMethodComparisonDhaka} compares IFAD with an adaptation of DS19 and the differentiable particle filter of \citet{corenflos21} (CTDD21), including the use of a warm start using IF2.
Section~{\suppSecModelsComparison} expands the comparison to five models of varying complexity.
IFAD is found to be advantageous on the more complex situations such as this Dhaka model. 
The plug-and-play property also ensures that IFAD is readily applicable to a wide class of scientific models.

\begin{figure}[htbp!]
    \centering
    \includegraphics[width=\arxiv{14cm}{\textwidth}]{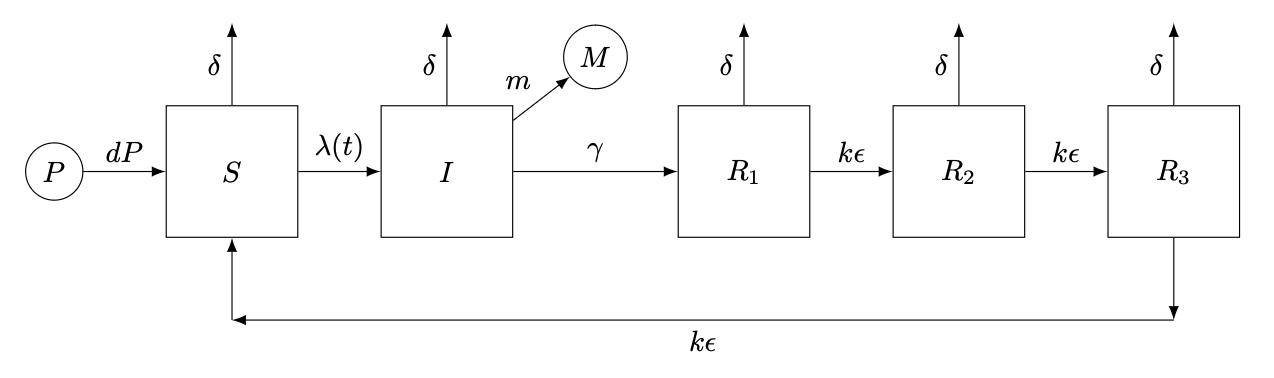}
    \vspace*{-7mm}
    \caption{A compartment flow diagram for the \cite{king08} Dhaka cholera model.}
    \label{fig:tikz-cholera}
    \arxiv{}{\vspace*{-3mm}}
\end{figure}


We tested IFAD with $\alpha = 0$, $\alpha=0.97$ and $\alpha=1$ against IF2 on a global search problem for the Dhaka cholera model using our implementation in Pypomp 1.0.5 \citep{abkemeier25}.
This example does not satisfy the conditions for our theoretical analysis of IFAD, so the successful empirical performance suggests these conditions are stronger than necessary.
For each method, we performed 100 searches, initialized with 100 starting parameter vectors drawn uniformly from the same wide bounding box used by \cite{ionides15}. 
To ensure a fair comparison under comparable computational effort, we ran each method for a similar amount of time.
For each method, we chose algorithmic parameters that maximized the log-likelihood well for that method while finishing within the allotted time. 
We ran IF2 for 1950 iterations, while IFAD-0.97 is run using a short IF2 warm-start of 175 iterations followed by the Adam optimizer for 400 iterations.
IFAD-0 and IFAD-1 were run for 60 IF2 iterations followed by 410 and 425 Adam iterations, respectively.
The cooling rate used for IF2 was slower than that used for the warm start in IFAD. 
To assess whether other methods can catch up to IFAD-0.97 given more time, we separately ran the IF2-only approach for 4 times as many iterations and the IFAD variants with 4 times as many Adam iterations, slowing the cooling rate for each as well.
A list of optimization settings is provided in the supplement (Section~\suppSecAlgorithmicParameters).

\begin{table}[htbp!]
\centering
\begin{tabular}{lrrrrr}
\toprule
 & Maximum & Re-evaluated & Median & IF2 time (s) & DMOP time (s) \\
\midrule
\multicolumn{6}{l}{\textit{Comparable computational effort ( $\approx$ 15 minutes)}} \\
\midrule
IFAD-0.97 & -3743.68 & -3744.26 & -3744.40 & 92 & 851 \\
IFAD-1 & -3744.36 & -3744.40 & -3746.84 & 35 & 902 \\
IFAD-0 & -3748.80 & -3749.02 & -3751.62 & 36 & 912 \\
IF2 & -3749.00 & -3748.97 & -3754.24 & 946 & - \\
\midrule
\multicolumn{6}{l}{\textit{Extended computational effort ( $\approx$ 60 minutes)}} \\
\midrule
IFAD-0.97 & -3743.53 & -3744.30 & -3744.03 & 96 & 3557 \\
IFAD-1 & -3743.96 & -3744.33 & -3745.39 & 35 & 3583 \\
IF2 & -3745.97 & -3746.11 & -3752.63 & 3782 & - \\
IFAD-0 & -3748.32 & -3748.38 & -3750.37 & 35 & 3465 \\
\bottomrule
\end{tabular}

\caption{Maximum log-likelihood found by IF2 and IFAD under both comparable and extended computational effort. The maximum is the largest of the 100 searches, and is biased upward by the winner's curse; the re-evaluated column gives a fresh particle filter evaluation of that same search, which removes the selection bias at the cost of being a single noisy evaluation. Where the searches cluster tightly, as they do for IFAD-0.97 under extended effort, the re-evaluation regresses toward the median. IFAD-0.97 performs the best among all methods.}
\arxiv{}{\vspace*{-7mm}}
\label{table:mle}
\end{table}

Our findings are summarized in Table \ref{table:mle} and Figures \ref{fig:boxplot}, \ref{fig:density}, and \ref{fig:optim}. 
Because estimating the log-likelihoods at every iteration of IF2 with particle filtering for Figure \ref{fig:optim} lengthens the computation time and thus compromises the fair comparison between methods, we obtained the times in Table \ref{table:mle} and Figure \ref{fig:optim} by running the methods again without the extra calls to the particle filter. 
Every other value displayed for the Dhaka results comes from the runs with the extra particle filter calls. 

Examining Table \ref{table:mle}, IFAD-0.97 attained a log-likelihood of -3743.68, which surpassed IF2 even when IF2 ran for 4 times as long.
IFAD-1 sometimes came close to IFAD-0.97, but its additional noise gave it worse median performance.
IFAD-1 needed to be run for an extended time for the median log-likelihood to start catching up with IFAD-0.97.

\begin{figure}[ht]
    \centering
    \arxiv{}{\vspace*{0.1ex}}
    \includegraphics[width=\arxiv{0.48\textwidth}{\textwidth/\real{2.2}}]{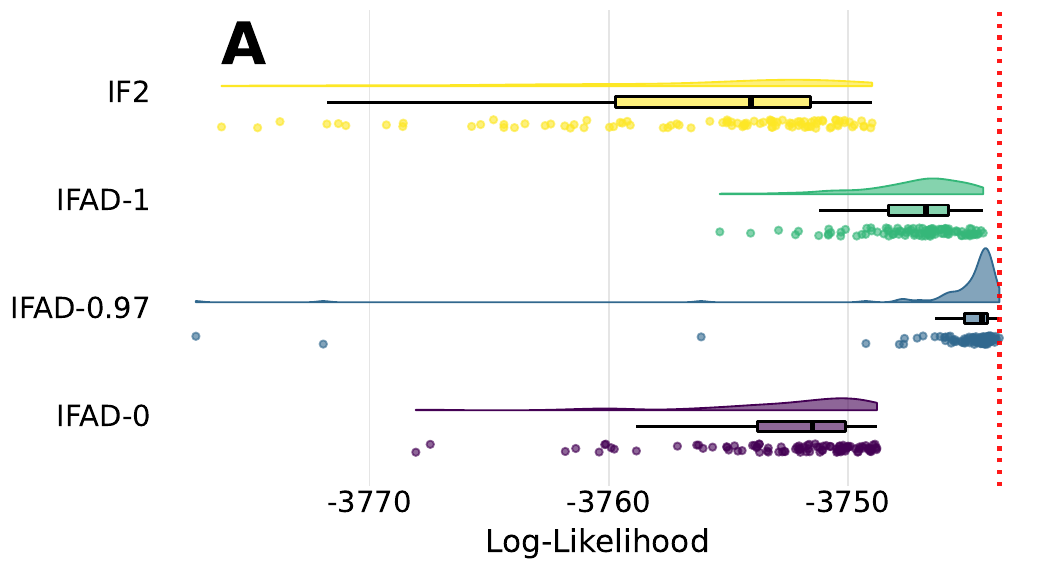}\hfill
    \includegraphics[width=\arxiv{0.48\textwidth}{\textwidth/\real{2.2}}]{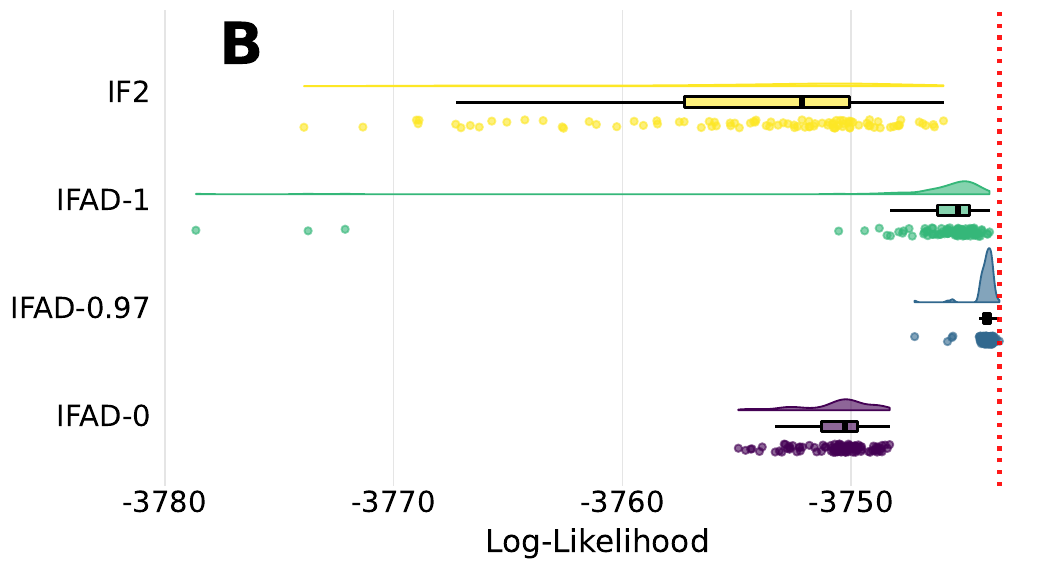}
    \caption{Raincloud plot of the searches. 
    \captiona\ Comparable computational effort. 
    \captionb\ Extended computational effort. 
    Log-likelihoods lower than -3800 are omitted. 
    IFAD-0.97 outperforms all other methods.
    The dotted red line shows the best-known maximized log-likelihood.}
    \label{fig:boxplot}
    \arxiv{}{\vspace*{-5mm}}
\end{figure}

\begin{figure}[htbp!]
    \centering
    \includegraphics[width=\arxiv{14cm}{\textwidth}]{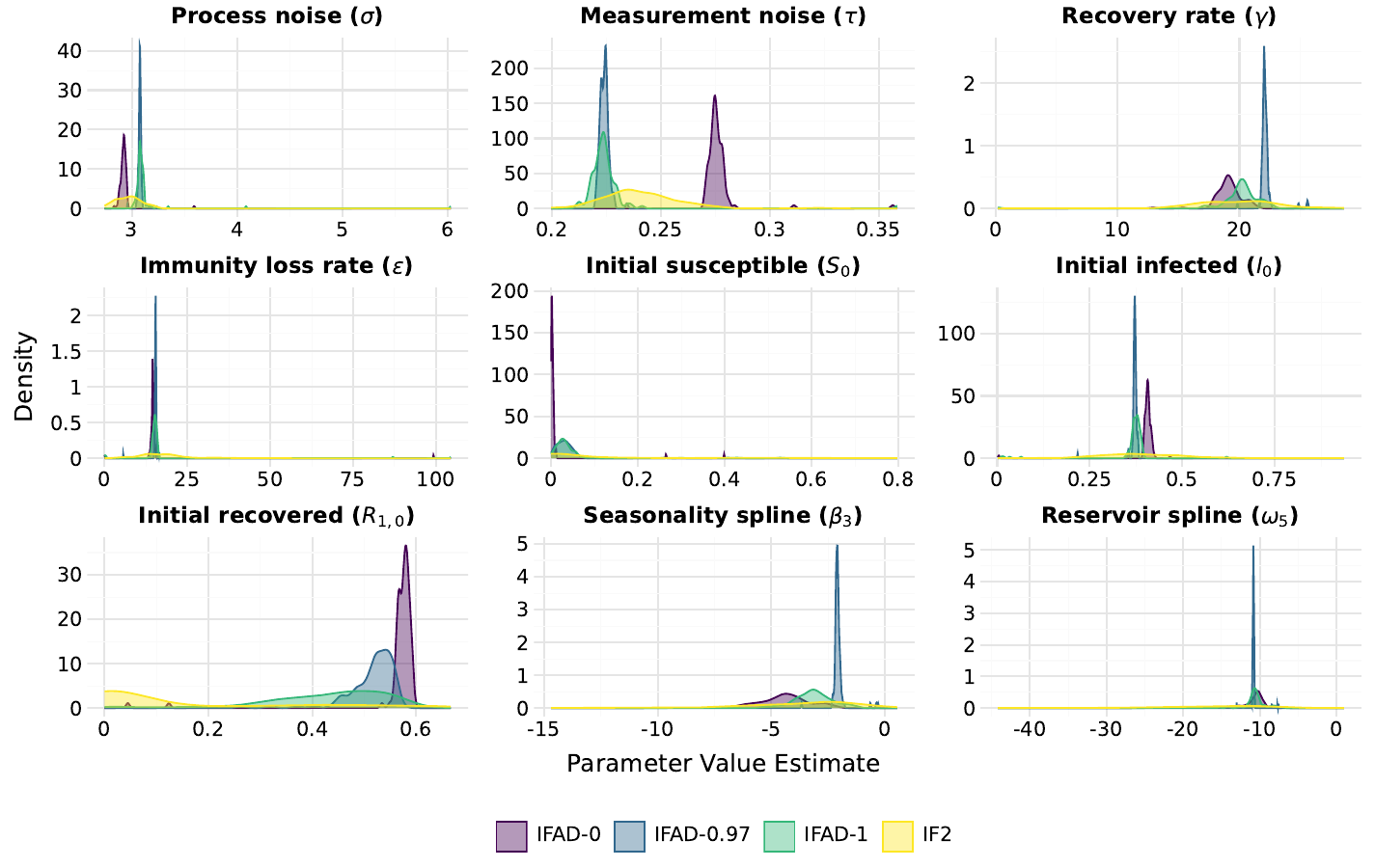}
    \caption{Density plots of the final parameter estimates across the 100 replicates for nine key parameters of the cholera model under extended computational effort.}
    \label{fig:density}
\end{figure}

\begin{figure}[ht]
    \centering
    \vspace{1mm}
    \includegraphics[width=\arxiv{14cm}{\textwidth/2}]{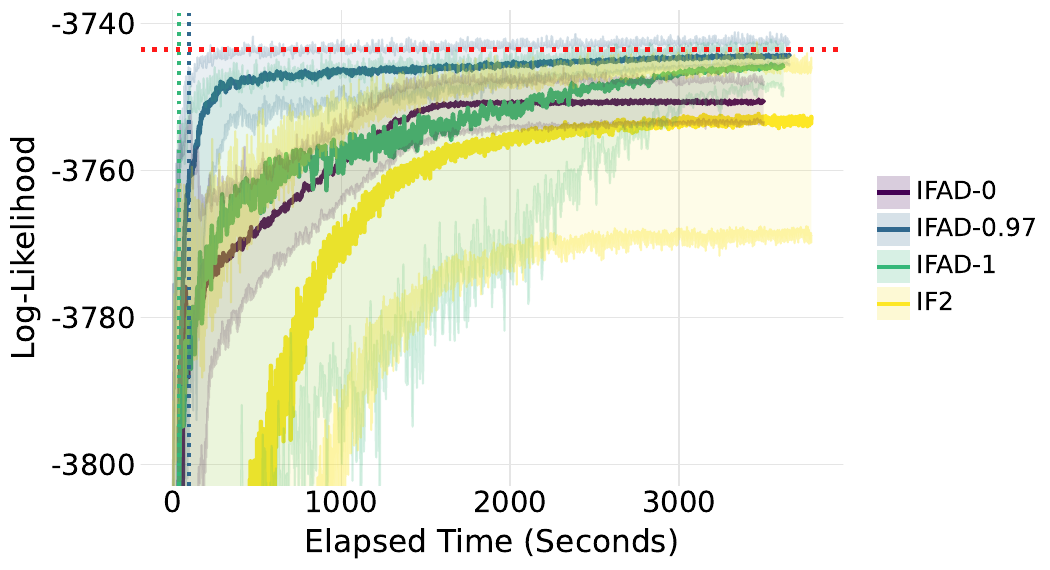}
    \arxiv{}{\vspace*{-1mm}}
    \caption{Optimization progress of IFAD and IF2 under extended computational effort.
    The shaded regions cover the area between the 100th and 10th percentiles of log-likelihood estimates at each iteration.
    We use a dotted red line to display the MLE.
    The vertical dotted lines indicate where the IFAD runs switch from IF2 to Adam optimization.
    Some log-likelihood estimates appear to be slightly above the MLE; this is due to Monte~Carlo evaluation error.}
    \label{fig:optim}
    \arxiv{}{\vspace*{-3ex}}
\end{figure}

IFAD-0 fell nearly 5 log-units short of IFAD-0.97 even with extended time, presumably due to its high bias. 
Based on Figure \ref{fig:optim}, IFAD-0 and IF2-only plateau heavily after being run for long enough and could not catch up to IFAD-0.97 within a reasonable timeframe.

Previously, a maximized log-likelihood of $-3748.6$ for this model and data was reported by \cite{king16}.
This MLE was obtained with much computational effort, using many global and local IF2 searches and with the assistance of likelihood profiling, 8 years after the model was first proposed by \cite{king08}.
Meanwhile, \cite{ionides15} only achieve a maximum log-likelihood of $-3768.6$, and \cite{king08} only $-3793.4$.
Using the improved computational efficiency of Pypomp (elaborated upon in Section \ref{sec:efficiency}), IF2 alone surpasses all of these previous results within a wall clock time of about an hour, while IFAD-0.97 does so in under 16 minutes and holds the record at $-3743.53$.

Notably, the log-likelihood and parameter estimates for IFAD-0.97 frequently exhibit much less Monte~Carlo variance than that of the other methods, as can be seen in Figures~\ref{fig:boxplot} and~\ref{fig:density}. 
IFAD in general exhibits less Monte~Carlo variance than IF2, especially for the parameter estimates. 
IFAD-0 has the second-lowest Monte~Carlo variance for most parameter estimates, although it has the lowest for the initial state proportions $S_0$, $I_0$ and $R_{1,0}$.
That said, the choice $\alpha=0$ is also known to have the greatest bias, and IFAD-1 agrees more closely with IFAD-0.97. 
Because multiple searches are employed to mitigate the effects of Monte~Carlo variance, the lower Monte~Carlo variance of IFAD-0.97 suggests that fewer searches are needed when running IFAD-0.97 to obtain a log-likelihood close to the maximum. 
Lowering the Monte~Carlo variance also narrows the width of Monte-Carlo-adjusted profile confidence intervals \citep{ionides17}.


While IF2 effectively approaches a neighborhood of the MLE (Figure \ref{fig:optim}), IF2 alone ultimately fails to achieve the last few log-likelihood units, as IF2 struggles to get much closer than about 2 log-likelihood units of the MLE even when run for an extended time (Table \ref{table:mle}).
Conversely, we encountered many failed searches with gradient descent alone. 
This is a difficult, nonconvex, and noisy problem with $23$ estimated parameters, and gradient descent without a warm-start gets stuck in local minima and saddle points.

IFAD, by comparison, approaches the MLE quickly due to the IF2 warm-start (Figure \ref{fig:optim}) and also succeeds at refining the coarse solution found by the warm-start with DMOP gradient steps to find the MLE.
In this experiment, combining IF2 initialization with DMOP refinement outperforms either method alone.
Monte~Carlo replication is appropriate for IFAD, or any Monte Carlo algorithm used to solve a challenging numerical problem, but Figure~\ref{fig:optim} shows that IFAD can find higher likelihood values more quickly and more reliably than the previous state-of-the-art. 
Since replicated searches are inevitably needed for Monte~Carlo maximization of this challenging likelihood maximization task, it is important to give extra consideration to the right tail of the distributions in Figure~\ref{fig:boxplot} when comparing the capabilities of the methods.

\section{Application to Bayesian Inference}
\label{sec:bayes}

Particle MCMC \citep{andrieu10} is arguably the most popular method for full-information plug-and-play Bayesian inference for POMP models.
However, particle Metropolis-Hastings often mixes slowly, requiring long burn-in periods.
For instance, \cite{fasiolo16} ran $1.4$ million iterations and discarded the first half as burn-in for a modified version of the \cite{king08} model, using a different transition scheme, process noise and observation distribution.
To alleviate this, we demonstrate a No-U-Turn Sampler (NUTS) \citep{hoffman14} using DMOP-$\alpha$ gradients, with a Gaussian prior common to NUTS and particle Metropolis--Hastings.
The prior has independent Gaussian distributions on the estimation scale, including logarithms for the five positive parameters. Each mean is the midpoint of the corresponding parameter search interval, and each standard deviation is half its width, where we note that these intervals do not truncate the prior. Particle Metropolis--Hastings uses a Gaussian random walk with a full covariance matrix estimated from preliminary runs \citep[p.~29]{dahlin19}, proposing fresh particle-filter random numbers at each iteration. This adapted proposal was crucial for convergence; we provide additional results in Supplementary Section~{\suppSecBayesFigs} that depict slow mixing otherwise. NUTS uses DMOP-$\alpha$ gradients to construct trajectories conditional on the filter random numbers, which are updated separately by a Metropolis step. Its preconditioning matrix is also estimated from preliminary PMMH runs. Both matrices are held fixed, with scalar step sizes adapted during the 500 warmup iterations.  Supplementary Section~{\suppSecBayesFigs} gives the prior parameters and proposal settings.

Using NUTS with DMOP-$\alpha$ gives promising preliminary results for parameter-space exploration after discarding $500$ iterations (Figure \ref{fig:nuts-eb}).
NUTS has been implemented with previous differentiable particle filters \citep{rosato22} but not in our plug-and-play context with control over the bias/variance tradeoff.

\begin{figure}[tb]
    \centering
    \vspace{1mm}
    \includegraphics[width=\textwidth]{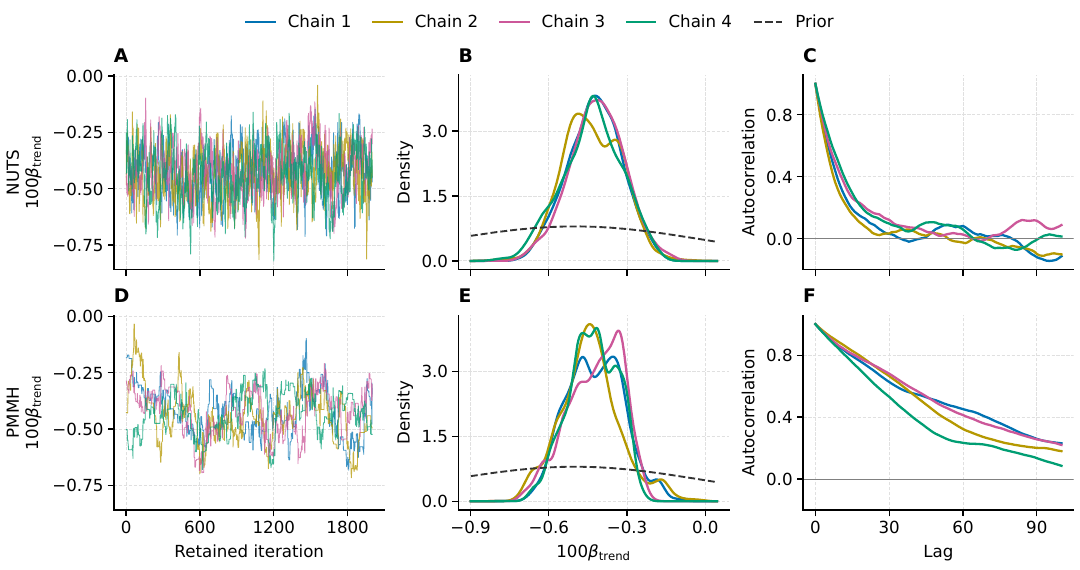}
    \caption{Sampling diagnostics for the transformed trend coordinate $100\beta_{\mathrm{trend}}$ in the Dhaka cholera model of \cite{king08}. Four chains of each method have 2,000 retained iterations after 500 warmup iterations. Panels show the traces, marginal density estimates, and autocorrelations of NUTS (A--C) NUTS and particle Metropolis--Hastings (D--F) respectively. The common prior is depicted in panels B and E with dashed lines.}
    \label{fig:nuts-eb}
    \arxiv{}{\vspace*{-4mm}}
\end{figure}

With the same prior and starting points, the NUTS sampler explores the trend parameter more effectively than particle Metropolis--Hastings (Figure~\ref{fig:nuts-eb}). We include quantitative comparisons for all estimated parameters in the SI. These preliminary results suggest that plug-and-play gradient estimates could assist particle Bayesian computation for complex scientific models.

\section{Theoretical Analysis of MOP-$\alpha$}
\label{sec:thms}

Here, we show that MOP-$\alpha$ targets the filtering distribution and likelihood under $\theta$, is consistent when $\alpha \uparrow 1$, and characterize rates for its bias and variance under different choices of $\alpha$. To do so, we require the following assumptions:

\arxiv{}{\vspace*{-3ex}}
\begin{enumerate}[label=(A\arabic*),itemsep=-1ex, , font=\small] 
    \item \textbf{Continuity of the Likelihood.} $\ell(\theta)$ has more than two continuous derivatives in a neighborhood $\left\{\theta: \ell(\theta)>\lambda_1\right\}$ for some $\lambda_1<\sup _{\varphi} \ell(\varphi)$.\label{assump:conti-lik}
    \item \textbf{Bounded Process Model.} There exist $\underbar{M}, \bar{M}$ such that $0 < \underbar{M} \leq f_{X_n|X_{n-1}}(x_n | x_{n-1};\theta) \leq \bar{M} < \infty$.\label{assump:bounded-process}
    \item \textbf{Bounded Measurement Model.} Uniformly in $n,j$, states, simulator draws, reference-filter resampling indices, and $\theta$ near $\phi$,
    $0<\underbar G\leq f_{Y_n\mid X_n}(y_n^*\mid x_n;\theta)\leq\bar G<\infty$,
    $|\partial_{\theta_a}\log g_{n,j}^{A,F,\theta}|\leq G'(\theta)<\infty$, and
    $|\partial_{\theta_a\theta_b}^2\log g_{n,j}^{A,F,\theta}|\leq G''(\theta)<\infty$
    for $1\leq a,b\leq p$. The derivatives include the dependence of the simulated state on $\theta$.
    The symbols $G'$ and $G''$ denote separate derivative bounds,
    not derivatives of the function $G$ in (A4).\label{assump:bounded-measurement}
    \item \textbf{Bounded Gradient Estimates.} There are functions $G(\theta), H(\theta): \Theta \to [0,\infty)$ uniformly bounded by $G^*, H^*<\infty$, such that the MOP-$\alpha$ derivatives at $\theta=\phi$ satisfy $\|\nabla_\theta\hat\ell_J^\alpha(\theta)\|_2\leq G(\theta)$ and $\|\nabla_\theta^2\hat\ell_J^\alpha(\theta)\|_{\mathrm{op}}\leq H(\theta)$ almost surely, uniformly in $J$ and $\alpha$. For $\alpha=1$, these bounds hold on one probability-one event simultaneously for all $J$ and all $\theta$ in a fixed closed ball $K$ of positive radius centered at $\phi$ and contained in the interior of $\Theta$. All estimates are coupled on one probability space, with the simulator draws and reference-filter resampling indices generated at $\phi$ held fixed as $\theta$ varies.\label{assump:local-bounded-derivative}
    \item \textbf{Differentiable Density Ratios and Simulator.} For fixed $\omega$, and any fixed collection of times, the simulated path and its log measurement densities have more than two continuous derivatives in $\theta$.\label{assump:diff-meas-and-sim}
\item{\textbf{Measurement-sensitivity decay.}}
\newEdit{
The simulator is $C^1$ in its incoming state. 
Starting the simulator at $X^\theta_\nPrevious=x_\nPrevious$, with $\theta$ and subsequent simulator draws held fixed,
\[
\left\|\nabla_{x_\nPrevious}\log f_{Y_n\mid X_n}
\bigl(y_n^*\mid X_n^\theta;\theta\bigr)\right\|_2
\leq C\kappa^{n-\nPrevious},\qquad C<\infty,\quad 0<\kappa<1.
\]
The bound is uniform in $\nPrevious\leq n$, $x_\nPrevious$, $X^\theta_n$, $\omega$, and $\theta$ near $\phi$.
The parameter derivatives of the initial state, simulator, and log measurement
density are bounded, holding incoming states fixed.
}
\item[(A6-J)]{\textbf{Alternative assumption to (A6).}} 
\newEdit{
The log measurement density is $C^1$ in
$(x,\theta)$, and (A2) holds for $\theta$ near $\phi$ for the joint Markov process
$(X_n^\theta,\partial_\theta X_n^\theta)$ relative to a common finite
measure on its invariant state space.
}
\end{enumerate}

Assumptions (A6) and (A6-J) provide two ways to control the simulated derivatives.
Under (A6), the effect of an earlier state on a later measurement decays with the elapsed time, holding the simulator draws fixed.
Under (A6-J), the state and its parameter derivative instead satisfy a joint mixing condition.
Mixing of the physical state alone does not ensure either property.
We require one of these additional assumptions for the particle error bounds; neither requires the algorithm to evaluate transition densities.
Section~\suppSecProofSharperBiasVar\ gives a sufficient condition for (A6) and an example satisfying (A6-J) but not (A6).

In particular, assumptions \ref{assump:bounded-process} and \ref{assump:bounded-measurement} imply that the POMP model satisfies a strong mixing condition, and as such are fairly strong but commonplace in the literature \citep{karjalainen23, delMoral04}. The fact that the likelihood estimate yielded by MOP-$\alpha$ has more than two continuous derivatives in $\theta$ follows from the construction of the likelihood estimate in equations ()\ref{eq:mop-conditional-likelihood}) and~(\ref{eq:weighting-scheme}), as well as Assumptions \ref{assump:bounded-measurement} and \ref{assump:diff-meas-and-sim}. 

We show that MOP-$\alpha$ targets the filtering distribution under $\theta$ and is strongly consistent for the likelihood under $\theta$ when $\alpha=1$ or $\theta=\phi$.
Employing $\alpha<1$ lead to inconsistency when $\theta$ deviates from $\phi$, but this bias disappears as $\theta$ approaches $\phi$.

While the result is presented here as specific to MOP-$\alpha$, we prove a more general result in the supplement (Section~\suppSecProofTriangularArrays).
This is a strong law of large numbers for triangular arrays of particles resampled arbitrarily, with resampling probabilities not proportional to the targeted distribution of interest, but reweighted to correct for the cumulative discrepancy between the resampling and the target distribution.

\begin{thm}[MOP-$\alpha$ Targets the Filtering Distribution and Likelihood]
    \label{thm:mop-targeting}
    When $\alpha=1$ or $\theta=\phi$, MOP-$\alpha$ targets the filtering distribution under $\theta$ and is strongly consistent for the likelihood under $\theta$. That is, for $\pi_n(\theta)=f_{X_{n}|Y_{1:n}}(x_n|y_{1:n}^* ; \theta)$ and any measurable and bounded functional $h$ and for $\hat\lik(\theta) = \prod_{n=1}^N L_n^{A,\theta,\alpha}$ or $\prod_{n=1}^N L_n^{B,\theta,\alpha}$, it holds that
    \arxiv{}{\vspace*{-1.5mm}}
    \begin{equation} \nonumber
        \frac{\sum_{j=1}^J h(x_{n,j}^{F, \theta}) \, w_{n,j}^{F,\theta}}{\sum_{j=1}^J w_{n,j}^{F,\theta}} \stackrel{a.s.}{\to} E_{\pi_n(\theta)} \big[h(X_n)\big], \hspace{5mm} \hat\lik(\theta)  \stackrel{a.s.}{\to} \lik(\theta).
    \arxiv{}{\vspace*{-1.5mm}}
    \end{equation}
\end{thm}

\begin{proof}[Proof.]
    \setlength{\emergencystretch}{2em}
    We provide a proof sketch here, deferring the details to the supplement (Section~\suppSecProofTriangularArrays). 
    When $\theta=\phi$, the ratio ${g_{n,j}^\theta}/{g_{n,j}^\phi}=1$ regardless of $\alpha$, reducing to the vanilla particle filter.
    When $\alpha=1$ and $\theta\neq\phi,$ suppose inductively that $\{(X^{F,\theta}_{n-1,j},w^{F,\theta}_{n-1,j})\}_{j=1}^J$ targets $f_{X_{n-1}|Y_{1:n-1}}(x_{n-1}|y^*_{1:n-1};\theta)$.
    It can then be shown that $\{(X^{P,\theta}_{n,j},w^{P,\theta}_{n,j})\}_{j=1}^J$ targets $f_{X_{n}|Y_{1:n-1}}(x_{n}|y^*_{1:n-1};\theta)$, that $\{(X^{P,\theta}_{n,j},w^{P,\theta}_{n,j} \, g^\theta_{n,j} )\}_{j=1}^J$ targets  $f_{X_{n}|Y_{1:n}}(x_{n}|y^*_{1:n};\theta)$, and that weighting the particles by $(X^{F,\theta}_{n,j},w^{F,\theta}_{n,j}) = (X^{P,\theta}_{n,k_j}, w^{P,\theta}_{n,k_j} \, g^\theta_{n,k_j} \big/ g^\phi_{n,k_j})$,
    resampling the $k_j$ with probabilities proportional to $g^\phi_{n,j}$, also targets $f_{X_{n}|Y_{1:n}}(x_{n}|y^*_{1:n};\theta)$.
    Assumption~\ref{assump:bounded-measurement} guarantees the boundedness and positivity conditions required by Lemma~S7 for the measurement update.
    Estimating the likelihood with the before-resampling estimator $\hat\lik(\theta) = \prod_{n=1}^N L_n^{B,\theta,\alpha}$, we see that the strong consistency is a direct consequence of our result that $\{ \big(X^{P,\theta}_{n,j},w^{P,\theta}_{n,j}\big) \}$ targets $f_{X_{n}|Y_{1:n-1}}(x_{n}|y^*_{1:n-1};\theta)$. 
\end{proof}

We take advantage of the equivalence between DMOP-$\alpha$ and the derivative of MOP-$\alpha$ to investigate the former by studying the latter.
Specifically, we show the consistency of the MOP-$1$ gradient estimate below.
Note that \cite{poyiadjis11} and \cite{scibior21} estimate the Fisher score by 
$\frac{1}{J}\sum_{j=1}^J \nabla_\theta \log f_{X_{0:N}, Y_{1:N}}(x_{0:n,j}^{A, F,\theta}, y_{1:N}^* ; \theta)$ and not
$\frac{1}{J}\sum_{j=1}^J \nabla_\theta \log f_{Y_{1:N}| X_{1:N}}(y_{1:N}^* | x_{1:n,j}^{A, F,\theta}; \theta).$
The estimators are evidently somewhat similar, but it is not immediately apparent that they converge to the same thing.

\begin{thm}[Consistency of MOP-$1$ Gradient Estimate]
    The gradient estimate of MOP-$\alpha$ when $\alpha=1$, $\theta=\phi$ is strongly consistent for the score: $\nabla_\theta \hat\ell_J^1(\theta) \stackrel{a.s.}{\to} \nabla_\theta \ell(\theta)$ as $J \to \infty$.
    \label{thm:mop-grad-consistency}
\end{thm}

We present a proof outline here, postponing the full argument to the supplement (Section~\suppSecProofConsistency).
On the compact neighborhood $K$ in Assumption \ref{assump:local-bounded-derivative}, Theorem \ref{thm:mop-targeting} with $\alpha=1$ gives pointwise almost-sure convergence of $\hat\ell_J^1(\theta)$ to $\ell(\theta)$, including when $\theta\ne\phi$. Intersecting these events over a countable dense subset of $K$, the gradient bound upgrades this to uniform convergence on $K$. The Hessian bound makes the gradients uniformly equicontinuous, so Arzel\`a--Ascoli gives uniformly convergent gradient subsequences. Passing to the limit in the fundamental theorem of calculus identifies each such limit as $\nabla_\theta\ell(\theta)$ in the interior of $K$. Thus the full gradient sequence converges locally uniformly there almost surely, and in particular $\nabla_\theta\hat\ell_J^1(\phi)\to\nabla_\theta\ell(\phi)$.

We now bound the particle approximation error and population discount bias of DMOP-$\alpha$.
The result describes the bias-variance tradeoff as $\alpha$ varies between zero and one.

\begin{thm}[Error bounds for DMOP-$\alpha$]
\label{thm:mop-biasvar}
Suppose (A2), (A3), (A5), and either (A6) or (A6-J) hold.
Fix the observations, use multinomial resampling, and let $J\geq2$ and $\theta=\phi$.
Write $H_N(\alpha)=\sum_{r=0}^{N-1}\alpha^r$,
using the after-resampling likelihood estimate.
For $0\leq\alpha<1$,
\begin{equation}
\label{eq:mop-mse}
\E\|\nabla_\theta\hat\ell_J^\alpha(\theta)-\nabla_\theta\ell(\theta)\|_2^2
\lesssim pG'(\theta)^2
\left[
\left\{\frac{N(1+\log J)^2}{J}+\frac{N^2}{J^2}\right\}
H_N(\alpha)^2+N^2(1-\alpha)^2
\right].
\end{equation}
At $\alpha=1$,
\begin{equation}
\label{eq:mop-one-mse}
\E\|\nabla_\theta\hat\ell_J^1(\theta)-\nabla_\theta\ell(\theta)\|_2^2
\lesssim pG'(\theta)^2\min\{N^2,N^3/J\}.
\end{equation}
The constants depend on the model and the stated stability bounds, independently of $N,J,\alpha$.
Expectations are over the particle-filter randomness.
\end{thm}

The bound separates particle approximation error from population discount bias.
The first term in Equation~\ref{eq:mop-mse} controls the error relative to the population DMOP-$\alpha$ score, while the second controls its squared difference from the full score.
Increasing the number of particles reduces the particle error, but note that it does not remove the population bias introduced by discounting.

Table~\ref{tab:dmop-error-bounds} summarizes the bounds.
We suppress the common factor $pG'(\theta)^2$, write $L_J=1+\log J$, and take $J\geq N$.
The population bias is the bias remaining as $J\to\infty$.

\begin{table}[htbp!]
\centering
\begin{tabular}{@{}lccc@{}}
\toprule
& \multicolumn{2}{c}{Particle MSE} & Squared population \\
\cmidrule(lr){2-3}
Choice of $\alpha$ & Under (A6) & Under (A6-J) & bias \\
\midrule
$\alpha=0$ & $N/J$ & $N/J$ & $N^2$ \\
Fixed $0<\alpha<1$ & $\dfrac{NL_J^2}{J(1-\alpha)^2}$
& $\dfrac{NL_J}{J(1-\alpha)^2}$ & $N^2(1-\alpha)^2$ \\[8pt]
$\alpha=1$ & $N^3/J$ & $N^3/J$ & $0$ \\
\bottomrule
\end{tabular}
\caption{Upper bounds for DMOP-$\alpha$ with $J\geq N$, omitting the common factor $pG'(\theta)^2$, model-dependent constants, and logarithmic factors.
The particle MSE bounds also bound the variance; adding the population term bounds the full-score MSE up to a constant.}
\label{tab:dmop-error-bounds}
\end{table}

For DMOP-$0$, the particle error grows at most linearly with the number of observations, apart from logarithmic factors.
The joint mixing assumption (A6-J) gives one fewer factor of $1+\log J$ in this case.
However, DMOP-$0$ can retain population bias even with arbitrarily many particles.

For any fixed $0<\alpha<1$, the particle bound has the same dependence on $N$ and $J$, with a constant that increases as $\alpha$ approaches one.
The population bias bound decreases at the same time.
The exact memory factor is $H_N(\alpha)\leq\min\{N,(1-\alpha)^{-1}\}$, so the finite-horizon particle bound remains finite as $\alpha$ approaches one.

For DMOP-$1$, any bias introduced by population discounting is zero, and so the MSE tends to zero as $J$ increases for fixed $N$.
Its $N^3/J$ upper bound allows substantially greater particle error for long series.

We analyze geometrically discounted scores through an exact weighted average of scores truncated at finite lags.
Posterior forgetting controls the population discount bias, while particle approximation and covariance bounds control the Monte Carlo error.

The main difficulties are that measurement gradients depend on simulated histories and that resampling creates dependence between particles and across time.
Under (A6), we expand each measurement gradient according to the age of its simulator contributions.
Their geometric decay permits us to sum bounds for finite histories using mixing of the physical particle cloud.
Under (A6-J), the measurement gradient is already a bounded function of the current state--derivative pair, so we apply the particle bounds directly to that joint process.
Section~\suppSecProofBiasVar\ develops the covariance argument for DMOP-$0$; Section~\suppSecProofSharperBiasVar\ handles general discounting and proves the theorem.

\section{Computational Efficiency} \label{sec:efficiency}

The cheap gradient principle of \cite{kakade2019provably} suggests that a gradient estimate from DMOP-$\alpha$ should take no more than 6 times that of the runtime of the particle filter.
We found the empirical value of this factor to be 3.75, using the implementation of our methods in the Pypomp library \citep{abkemeier25}.
DMOP-$\alpha$ and IFAD therefore share the same $O(NJ)$ time complexity as the particle filter, unlike the $O(NJ^2)$ complexity of \cite{corenflos21} and Algorithm 2 in \citep{poyiadjis11, scibior21}.

Implementation of the particle filter, IF2, and DMOP-$\alpha$ in JAX \citep{jax} takes advantage of just-in-time compilation and GPU acceleration.
While CPU-based implementations of the particle filtering methods use multiple CPU cores by parallelizing over independent replications of a given method, Pypomp uses tens of thousands of GPU cores by vectorizing over both particles and independent method replications.
To examine the performance improvement on the Dhaka model, we ran 1 instance of IF2 with 60 iterations and 5000 particles using 1 CPU core and compared this with running 100 replications on a GPU.
The results in Table \ref{table:speedup} show that, despite running 100 times as many replications, the GPU finishes 26.0 times faster, giving it a throughput 2597.2 times greater than the single CPU core.

\begin{table}[htbp!]
\centering
\begin{tabular}{lll}
\toprule
Metric & \texttt{pomp} \citep{king16} & Pypomp \citep{abkemeier25} \\
\midrule
Language & R (C simulator) & Python (JAX) \\
Hardware& 3.0 GHz Intel Xeon Gold 6154 (1 core) & NVIDIA RTX6000 Pro Blackwell \\
IF2 Replications & 1 & 100 \\
Latency (s) & 935 & 36 \\
\bottomrule
\end{tabular}
\caption{Comparison of execution latency for running IF2 (60 iterations, 5000 particles, 1 versus 100 replications) between the \texttt{pomp} (CPU-only, R and compiled C) and Pypomp (GPU-accelerated, compiled Python).}
\label{table:speedup}
\end{table}

\section{Discussion}

Practical likelihood-based data analysis involves many likelihood optimizations \citep{king08,blake14,pons-salort18,subramanian21,fox22,drake23}.
This occurs with profile likelihood confidence intervals and model selection. A careful scientist should consider many model variations, to see whether the conclusions of the study are sensitive to alternative model specifications.
Coding model variations is relatively simple when using plug-and-play inference methodology.
However, assessing the scientific potential of these variations requires likelihood optimization.
Monte~Carlo adjusted profile methods can help to acommodate maximization and evaluation error \citep{ionides17,ning21} but these methods are most effective when the maximized log-likelihood can be obtained within a few units.
Improving the process of model fitting has practical consequences for the ability of scientists to evaluate hypotheses.
We have demonstrated, for the first time, a plug-and-play maximum likelihood approach which provides a high-quality log-likelihood maximization in the complex benchmark model of \cite{king08}.
The \textit{Daphnia} application in the supplement (Section~\suppSecDaphnia) demonstrates that these gains also extend to complex panel data.

If the simulator is discontinuous in $\theta$, DMOP-$\alpha$ no longer applies.
We are working on a variation of DMOP-$\alpha$ applicable to this case.
Specifically, differentiable transition density ratios can be used instead of a differentiable simulator.
The off-parameter treatment of the measurement model is extended to an off-parameter simulation for the dynamic process.
In that case, one requires access to the transition densities, or at least their ratios. 
This algorithm may be better suited to large discrete state spaces than either DMOP-$\alpha$ or differentiation of the Baum-Welch algorithm.

Discounting the weights by some $\alpha \in [0,1]$ is not the only way to interpolate between the estimators of \cite{naesseth18} and \cite{poyiadjis11}. 
As the proof of Theorem \ref{thm:mop-biasvar} implies, we can also truncate the weights at a fixed lag, corresponding to the MOP-$(1,k)$ algorithm mentioned. 
The analysis is similar, with comparable rates but with a slightly less convenient implementation. 

\section*{Reproducibility}

Source code for this article is archived on Zenodo with DOI \href{https://doi.org/10.5281/zenodo.23538801}{10.5281/zenodo.23538801}.
It is also available on GitHub at \href{https://github.com/pypomp/dmop}{https://github.com/pypomp/dmop}.

\section*{Acknowledgments}
This research was supported by National Science Foundation grants DMS-1761603, DEB-1933497 and  OAC-2607869. 
We thank Nicolas Chopin for communications regarding the strong law of large numbers for off-parameter resampled particle filters, and other feedback. We thank Jun Chen and Arnaud Doucet for helpful discussions.
{\editAIstatement}


\end{document}


\title{Supporting Materials for ``\articleTitle''}

\author{Kevin Tan$^1$, Aaron J. Abkemeier$^2$, Kunyang He$^2$,\\
Giles Hooker$^1$ and Edward L. Ionides$^2$\\\small $^1$University of Pennsylvania and $^2$University of Michigan}
\date{}

\maketitle

\tableofcontents

\pagebreak

\section{Proof of Theorem~{\thmMopFunctionalForms} (DMOP-$\alpha$ Functional Forms)}

\label{appendix:functional}

Theorem~{\thmMopFunctionalForms} follows immediately as a consequence of the following results, Lemmas \ref{lem:mop-1-formula} and \ref{lem:mop-0-formula}.
Here, an $A$ superscript is used to denote ancestral trajectories, so that $x_{1:n,j}^{A, P,\theta}$ contains the history of the prediction particle $x_{n,j}^{P,\theta}$ at observation times $\seq{1}{n}$.
Similarly, $g_{i,j}^{A,P,\theta}$ is the ancestral weight of $x_{n,j}^{P,\theta}$ at observation time $i$.

\begin{lem}
\label{lem:mop-1-formula}
Write $\nabla_\theta \hat\ell^\alpha(\theta)$ for the DMOP-$\alpha$ gradient estimate (or, equivalently, the gradient of MOP-$\alpha$ when $\theta=\phi$).
Consider the case where we use the after-resampling conditional likelihood estimate so that $\hat\lik(\theta) = \prod_{n=1}^N L_n^{A, \theta, \alpha}$.
When $\alpha=1$,
    \begin{equation}
        \nabla_\theta \hat{\ell}^1(\theta) 
        = \frac{1}{J}\sum_{j=1}^J \nabla_\theta \log f_{Y_{1:N}|X_{1:N}}\big(y_{1:N}^* | x_{1:N,j}^{A, F,\theta};\theta\big),
    \end{equation}
    with the same undiscounted path-space score target as the bootstrap filter estimators of \cite{poyiadjis11} and \cite{scibior21}.
\end{lem}

\begin{proof}[Proof.]
Consider MOP-$\alpha$ with $\alpha=1$. Hold $\phi$, the simulator draws,
and the reference-filter resampling indices fixed when differentiating,
then evaluate at $\theta=\phi$. Write $k_{n,j}$ for $k_j$ at time $n$.
Algorithm~1 gives
\begin{equation}
w_{n,j}^{P,\theta}=w_{n-1,j}^{F,\theta},\qquad
w_{n,j}^{F,\theta}
=w_{n,k_{n,j}}^{P,\theta}\frac{g_{n,k_{n,j}}^\theta}{g_{n,k_{n,j}}^\phi}
=w_{n-1,k_{n,j}}^{F,\theta}\frac{g_{n,k_{n,j}}^\theta}{g_{n,k_{n,j}}^\phi}.
\label{eq:lemma-s1-filtered-recursion}
\end{equation}
The after-resampling likelihood estimate therefore telescopes:
\begin{align}
\hat\lik^1(\theta)
&=\prod_{n=1}^N L_n^{A,\theta,1}
 =\prod_{n=1}^N\left(L_n^\phi
 \frac{\sum_{j=1}^Jw_{n,j}^{F,\theta}}{\sum_{j=1}^Jw_{n,j}^{P,\theta}}\right)
 \notag\\*
&=\left(\prod_{n=1}^N L_n^\phi\right)
 \prod_{n=1}^N\frac{\sum_{j=1}^Jw_{n,j}^{F,\theta}}
                       {\sum_{j=1}^Jw_{n-1,j}^{F,\theta}}
 =\left(\frac1J\sum_{j=1}^Jw_{N,j}^{F,\theta}\right)
  \prod_{n=1}^N L_n^\phi.
\label{eq:lemma-s1-telescoping}
\end{align}
Here each intermediate weight sum cancels, and the initial sum is $J$
because $w_{0,j}^{F,\theta}=1$. The log-derivative identity gives
\begin{equation}
\nabla_\theta\hat\ell^1(\theta)
=\frac{\nabla_\theta\hat\lik^1(\theta)}{\hat\lik^1(\theta)}
=\frac{J^{-1}\sum_{j=1}^J\nabla_\theta w_{N,j}^{F,\theta}
       \prod_{n=1}^N L_n^\phi}
       {J^{-1}\sum_{j=1}^Jw_{N,j}^{F,\theta}
       \prod_{n=1}^N L_n^\phi}
=\frac{\sum_{j=1}^J\nabla_\theta w_{N,j}^{F,\theta}}
       {\sum_{j=1}^Jw_{N,j}^{F,\theta}}.
\label{eq:log-derivative-identity}
\end{equation}

We next expand the filtered weight along its ancestry. Using the existing
$A$ convention, let $g_{i,j}^{A,F,\theta}$ denote the measurement density
at time $i$ along the filtered trajectory $x_{1:n,j}^{A,F,\theta}$;
the endpoint $n$ is implicit. At time $n$, this density is
$g_{n,k_{n,j}}^\theta$. At earlier times, the trajectory agrees with that
of the selected parent $x_{n-1,k_{n,j}}^{F,\theta}$.
Starting from $w_{0,j}^{F,\theta}=1$ and repeatedly substituting the
filtered-weight recursion therefore gives
\begin{equation}
w_{n,j}^{F,\theta}
=w_{n-1,k_{n,j}}^{F,\theta}
 \frac{g_{n,k_{n,j}}^\theta}{g_{n,k_{n,j}}^\phi}
=\left(\prod_{i=1}^{n-1}
 \frac{g_{i,j}^{A,F,\theta}}{g_{i,j}^{A,F,\phi}}\right)
 \frac{g_{n,j}^{A,F,\theta}}{g_{n,j}^{A,F,\phi}}
=\prod_{i=1}^{n}\frac{g_{i,j}^{A,F,\theta}}{g_{i,j}^{A,F,\phi}}.
\label{eq:lemma-s1-ratio}
\end{equation}
The product through $n-1$ is empty when $n=1$; for $n>1$, it is the
product already accumulated by the selected parent. The reference
factors and ancestor indices are fixed under differentiation, so
\begin{align}
\frac{\nabla_\theta w_{n,j}^{F,\theta}}{w_{n,j}^{F,\theta}}
&=\nabla_\theta\log w_{n,j}^{F,\theta}
 =\nabla_\theta\log\left(
 \prod_{i=1}^{n}\frac{g_{i,j}^{A,F,\theta}}{g_{i,j}^{A,F,\phi}}\right)
 \notag\\*
&=\nabla_\theta\sum_{i=1}^{n}
 \left(\log g_{i,j}^{A,F,\theta}-\log g_{i,j}^{A,F,\phi}\right)
 =\sum_{i=1}^{n}\nabla_\theta\log g_{i,j}^{A,F,\theta}.
\end{align}
For the terminal filtered trajectory at time $N$, conditional independence
of the observations gives
\begin{align}
\sum_{n=1}^{N}\nabla_\theta\log g_{n,j}^{A,F,\theta}
&=\nabla_\theta\log\left(\prod_{n=1}^{N}g_{n,j}^{A,F,\theta}\right)
 \notag\\*
&=\nabla_\theta\log\left(\prod_{n=1}^{N}f_{Y_n|X_n}
 \big(y_n^*\mid x_{n,j}^{A,F,\theta};\theta\big)\right)
 \notag\\*
&=\nabla_\theta\log f_{Y_{1:N}|X_{1:N}}
 \big(y_{1:N}^*\mid x_{1:N,j}^{A,F,\theta};\theta\big).
\end{align}
The derivative includes the dependence of the simulated particle locations
on $\theta$ through the chain rule. Multiplying the log-weight derivative
by $w_{N,j}^{F,\theta}$ yields
\begin{equation}
\nabla_\theta w_{N,j}^{F,\theta}
=w_{N,j}^{F,\theta}\sum_{n=1}^{N}\nabla_\theta\log g_{n,j}^{A,F,\theta}
=w_{N,j}^{F,\theta}\nabla_\theta\log f_{Y_{1:N}|X_{1:N}}
 \big(y_{1:N}^*\mid x_{1:N,j}^{A,F,\theta};\theta\big).
\label{eq:lemma-s1-filtered-score}
\end{equation}
Finally, at $\theta=\phi$, every ratio in
Equation~(\ref{eq:lemma-s1-ratio}) equals one, so
$w_{N,j}^{F,\theta}=1$ and $\sum_jw_{N,j}^{F,\theta}=J$.
Substituting into Equation~(\ref{eq:log-derivative-identity}) after
differentiating gives
\begin{align}
\left.\nabla_\theta\hat\ell^1(\theta)\right|_{\theta=\phi}
&=\frac1J\sum_{j=1}^{J}
 \left.\nabla_\theta w_{N,j}^{F,\theta}\right|_{\theta=\phi}
 \notag\\*
&=\frac1J\sum_{j=1}^{J}
 \left.\nabla_\theta\log f_{Y_{1:N}|X_{1:N}}
 \big(y_{1:N}^*\mid x_{1:N,j}^{A,F,\theta};\theta\big)
 \right|_{\theta=\phi},
\end{align}
as claimed.
This has the same undiscounted path-space score target as the estimators of
\cite{poyiadjis11, scibior21} for the bootstrap filter. The
reparameterization score contribution need not equal their complete-data
likelihood-ratio score contribution particle by particle.
\end{proof}

For $\theta\neq\phi$, the same calculation instead gives the normalized
weighted path score
\begin{equation}
\nabla_\theta\hat\ell^1(\theta)
=\frac{\displaystyle\sum_{j=1}^J w_{N,j}^{F,\theta}
\nabla_\theta\log f_{Y_{1:N}|X_{1:N}}
\big(y_{1:N}^*\mid x_{1:N,j}^{A,F,\theta};\theta\big)}
{\displaystyle\sum_{j=1}^J w_{N,j}^{F,\theta}}.
\end{equation}
Growth of the unnormalized weight products alone does not establish a
variance rate for this normalized expression.
When $\theta=\phi$, a generic normalized path-functional inequality gives
the coarse upper bound $O(N^4/J)$ after rescaling to the full
unnormalized score.  Section~\ref{appendix:sharper-biasvar} gives an exact
constant-weight calculation with order $N^2/J$ at $\alpha=1$ and explains
why this example does not by itself provide a general finite-$J$ upper
bound.

\begin{lem}
\label{lem:mop-0-formula}
Write $\nabla_\theta \hat\ell^\alpha(\theta)$ for the DMOP-$\alpha$ gradient estimate (or, equivalently, the gradient of MOP-$\alpha$ when $\theta=\phi$).
Consider the case where we use the after-resampling conditional likelihood estimate so that $\hat\lik(\theta) = \prod_{n=1}^N L_n^{A, \theta, \alpha}$.
When $\alpha=0$,
    \begin{equation}
        \nabla_\theta \hat\ell^0(\theta) 
        = \frac{1}{J} \sum_{n=1}^N \sum_{j=1}^J \nabla_\theta \log\left(f_{Y_n|X_{n}}\big(y_n^*|x_{n,j}^{F, \theta}; \theta\big)\right),
    \end{equation}
    yielding the estimate of \cite{naesseth18} when applied to the bootstrap filter. 
\end{lem}
\begin{proof}[Proof.]
First, write
\begin{equation}
s_{n,j} = \frac{g_{n,k_{n,j}}^\theta}{g_{n,k_{n,j}}^\phi}
\end{equation}
as shorthand for the selected measurement density ratios, with the
reference parameter and resampling indices held fixed.
Observe that, when $\alpha=0$, the likelihood estimate becomes
\begin{eqnarray}
    \hat{\lik}^0(\theta) := \prod_{n=1}^N L_n^{A, \theta, \alpha} &=& \prod_{n=1}^N\left(L_n^\phi \cdot \frac{\sum_{j=1}^J w_{n,j}^{F,\theta}}{\sum_{j=1}^J w_{n,j}^{P,\theta}}\right)
    \\
    &=& \prod_{n=1}^N\left(L_n^\phi \cdot \frac{1}{J}\sum_{j=1}^J s_{n,j}\right)
    \\
    &=& \prod_{n=1}^N\left(L_n^\phi \cdot \frac{1}{J}\sum_{j=1}^J \frac{f_{Y_n|X_n}\big(y_n^*|x_{n,k_{n,j}}^{P, \theta};\theta\big)}{f_{Y_n|X_n}\big(y_n^*|x_{n,k_{n,j}}^{P, \phi};\phi\big)}\right).
\end{eqnarray}
We lose the nice telescoping property observed in the MOP-$1$ case, but this expression still yields something useful. 
This is because its gradient when $\theta=\phi$ is therefore 
\begin{align}
    \nabla_\theta \hat{\ell}^0(\theta) &:= \sum_{n=1}^N \nabla_\theta \log\left(L_n^\phi \frac{1}{J} \sum_{j=1}^J s_{n,j}\right) \\
    &= \sum_{n=1}^N \frac{\nabla_\theta \left(L_n^\phi \frac{1}{J} \sum_{j=1}^J s_{n,j}\right)}{\left(L_n^\phi \frac{1}{J} \sum_{j=1}^J s_{n,j}\right)} \\
    &= \sum_{n=1}^N \frac{\sum_{j=1}^J \nabla_\theta s_{n,j}}{\sum_{j=1}^J s_{n,j}} \\
    &= \sum_{n=1}^N \frac{1}{J} \sum_{j=1}^J \frac{\nabla_\theta f_{Y_n|X_{n}}\big(y_n^*|x_{n,j}^{F, \theta}; \theta\big)}{f_{Y_n|X_{n}}(y_n^*|x_{n,j}^{F, \phi}; \phi)} \\
    &= \frac{1}{J} \sum_{n=1}^N \sum_{j=1}^J \nabla_\theta \log\left(f_{Y_n|X_{n}}\big(y_n^*|x_{n,j}^{F, \theta}; \theta\big)\right),
\end{align}
where we use the log-derivative trick in the second equality, observe that $\sum_{j=1}^J s_{n,j} = J$ when $\theta=\phi$ in the fourth equality, and use the log-derivative trick again while noting that $\theta=\phi$ in the fifth equality. This yields the desired result.
\end{proof}

\section{Proof of Theorem~{\thmMopConvergence} (Optimization Convergence Analysis)}
\label{appendix:convergence}

All derivatives in this section are evaluated at the common filtering and
evaluation parameter $\theta=\phi$.  Let $Z_n$ contain the complete ancestral
path through time $n$, the parameter-independent simulator variables, and the
first- and second-order gradients propagated by automatic differentiation.  Set
\begin{align}
G_n^\theta(Z_n)
  &=f_{Y_n\mid X_n}(y_n^*\mid X_n^\theta;\theta),
&W_\theta(Z_N)&=\prod_{n=1}^N G_n^\theta(Z_n),                 \label{eq:s2-weight}\\
\psi_\theta(Z_N)
  &=\nabla_\theta\log W_\theta(Z_N),
&\Xi_\theta(Z_N)&=\nabla_\theta^2
                         \log W_\theta(Z_N).                  \label{eq:s2-score-curvature}
\end{align}
Differentiate through $X_n^\theta$ in $G_n^\theta(Z_n)$, holding the
simulator draws and resampling indices fixed.

Write $\rho_N^J=J^{-1}\sum_{j=1}^J\delta_{Z_{N,j}}$ for the empirical measure
of the after-resampling genealogical particles, and define the normalized
Feynman--Kac path measure by
\begin{equation}
\rho_N(F)=\frac{\E\{W_\theta(Z_N)F(Z_N)\}}
                 {\E\{W_\theta(Z_N)\}}.
\label{eq:s2-rho-definition}
\end{equation}
Assumptions~{\refAssumpBoundedMeasurement} and
{\refAssumpDifferentiableSimulator} give
$|\partial_{\theta_i}W_\theta|\leq N\bar G^NG'(\theta)$, so dominated
convergence and the log-derivative identity yield
\begin{equation}
\nabla_\theta\ell(\theta)=\rho_N(\psi_\theta),
\qquad
\nabla_\theta\hat\ell_J^{A,1}(\theta)=\rho_N^J(\psi_\theta).
\label{eq:s2-gradient-identities}
\end{equation}
The second equality is Lemma~\ref{lem:mop-1-formula}; importantly, it is one
terminal empirical average of the complete path score, not a sum of empirical
filtering expectations at the individual observation times.

For completeness, multinomial resampling gives the required genetic
Feynman--Kac particle system explicitly.  For semigroup indexing, let
$\mathsf Z_0$ be the augmented prediction path at observation $1$; for
$q=1,\ldots,N-1$, let $\mathsf Z_q$ be the prediction path at observation
$q+1$; and let $\mathsf Z_N$ be the terminal path after the final resampling.
For $0\leq q\leq N$, write
$\rho_q^J=J^{-1}\sum_{j=1}^J\delta_{\mathsf Z_{q,j}}$, where the
$\mathsf Z_{0,j}$ are independent draws from the unweighted prediction law of
$\mathsf Z_0$ and $\mathsf Z_{N,j}=Z_{N,j}$.
Write $V_q^\theta(\mathsf Z_q)=G_{q+1}^\theta(\mathsf Z_q)$ for
$q=0,\ldots,N-1$.  Let the locally indexed augmented kernel
$\mathcal M_{q+1}^\theta$ append the next independent simulator innovation and
its gradients for $0\leq q\leq N-2$, and put
$\mathcal M_N^\theta=\operatorname{Id}$ for the final selection.  The one-step
population map is
\begin{equation}
\Phi_{q+1}^\theta(\mu)(dz')
=\frac{\int\mu(dz)V_q^\theta(z)\mathcal M_{q+1}^\theta(z,dz')}
       {\mu(V_q^\theta)},\qquad 0\leq q<N.                    \label{eq:s2-fk-update}
\end{equation}
Conditionally on the current cloud, the $J$ children are independent with
common law $\Phi_{q+1}^\theta(\rho_q^J)$.  The final observation is handled by
the $q=N-1$ transition, namely multinomial selection with potential
$V_{N-1}^\theta=G_N^\theta$ followed by the identity mutation.  Thus
$\rho_N^J$ is exactly the terminal empirical measure of this augmented
historical system.  Systematic resampling does not have this
conditional-independence property and is not covered by the calculation.

We next write out the particle concentration constant used below.  For
$0\leq q\leq N$, let
\begin{align}
Q_{q,N}^\theta F(z_q)
 &=\E\left[
   F(\mathsf Z_N)\prod_{r=q}^{N-1}V_r^\theta(\mathsf Z_r)
   \,\middle|\,\mathsf Z_q=z_q
   \right],                                                     \label{eq:s2-q-semigroup}\\
P_{q,N}^\theta F
 &=\frac{Q_{q,N}^\theta F}{Q_{q,N}^\theta1},\qquad
g_{q,N}=\sup_{z,z'}
 \frac{Q_{q,N}^\theta1(z)}{Q_{q,N}^\theta1(z')},\qquad
b_{q,N}=\sup_{\operatorname{osc}(F)\leq1}
 \operatorname{osc}(P_{q,N}^\theta F),                         \label{eq:s2-semigroup-constants}
\end{align}
where $Q_{N,N}^\theta F=F$ and
$\operatorname{osc}(F)=\sup_{z,z'}|F(z)-F(z')|$.  Thus
$0\leq b_{q,N}\leq1$.  With
\begin{equation}
R_G=\frac{\bar G}{\underbar G},\qquad
S_N=\sum_{q=0}^Ng_{q,N}^3b_{q,N},\qquad K_N=32S_N,\qquad
K_N^{\rm pf}=32\sum_{r=0}^NR_G^{3r},                           \label{eq:s2-explicit-constant}
\end{equation}
the potential bounds give, term by term,
\begin{equation}
\underbar G^{N-q}\leq Q_{q,N}^\theta1(z_q)
 \leq\bar G^{N-q},\qquad
g_{q,N}\leq R_G^{N-q},\qquad K_N\leq K_N^{\rm pf}.             \label{eq:s2-explicit-constant-bound}
\end{equation}
In particular, the constant used in the theorem is the closed-form quantity
\begin{equation}
K_N^{\rm pf}=
\begin{cases}
32(N+1),&R_G=1,\\[1mm]
\displaystyle 32\frac{R_G^{3(N+1)}-1}{R_G^3-1},&R_G>1.
\end{cases}                                                     \label{eq:s2-explicit-constant-closed}
\end{equation}
Here is the constant calculation behind the tail bound used below.  Let
$r_N^{\rm DM}$ and $\beta_N^{\rm DM}$ denote the second- and first-order
coefficients in Theorem~1.2 of \citet{delMoral11}.  The Feynman--Kac bounds in
their Appendix~A.3 give
\begin{equation}
r_N^{\rm DM}\leq4S_N,
\qquad
(\beta_N^{\rm DM})^2
 \leq4\sum_{q=0}^Ng_{q,N}^2b_{q,N}^2
 \leq4S_N,                                                       \label{eq:s2-dm-coefficient-bounds}
\end{equation}
where the last inequality uses $g_{q,N}\geq1$ and $0\leq b_{q,N}\leq1$.
If $S_N=0$, the particle error is identically zero.  Hence suppose $S_N>0$.
For $u\in(0,1/2]$, set $x=Ju^2/(32S_N)$.  If $x\leq\log6$, then
$6e^{-x}\geq1$ and the desired probability bound is trivial.  If
$x>\log6$, use
$h^{-1}(x)\leq2x+2\sqrt x$ for
$h(t)=\{t-\log(1+t)\}/2$ and
$1+2x+2\sqrt x\leq8x$ to obtain
\begin{align*}
\frac{r_N^{\rm DM}}J\{1+h^{-1}(x)\}
&\leq\frac{u^2}{8x}\{1+2x+2\sqrt x\}
 \leq u^2\leq\frac u2,\\
\beta_N^{\rm DM}\sqrt{\frac{2x}{J}}
&\leq2\sqrt{S_N}\sqrt{\frac{2x}{J}}=\frac u2.
\end{align*}
Theorem~1.2 applied to both $F$ and $-F$, after centering and scaling $F$ to
oscillation one, therefore gives the explicit Feynman--Kac corollary
\begin{equation}
\Pr\left(
 |(\rho_N^J-\rho_N)(F)|>\operatorname{osc}(F)u
 \right)
\leq6\exp\left(-\frac{Ju^2}{K_N}\right)
\leq6\exp\left(-\frac{Ju^2}{K_N^{\rm pf}}\right).             \label{eq:s2-scalar-concentration}
\end{equation}
for every $F$ with $\operatorname{osc}(F)>0$ and every $u\in(0,1/2]$.
If $F$ is constant, $(\rho_N^J-\rho_N)(F)=0$ identically.
We use the constants in Equations~\ref{eq:s2-explicit-constant}--
\ref{eq:s2-scalar-concentration} in the concentration bounds below.

\begin{lem}[Explicit concentration of the DMOP-$1$ gradient]
\label{lemma:grad_bound}
Assume multinomial resampling and let
\begin{equation}
A_N(\theta)=\max_{1\leq i\leq p}
 \operatorname{osc}(\psi_{\theta,i}),
\qquad
u_g(J,\delta)=
 \sqrt{\frac{K_N^{\rm pf}}{J}\log\frac{6p}{\delta}}.
\label{eq:s2-gradient-constants}
\end{equation}
If $J\geq4K_N^{\rm pf}\log(6p/\delta)$, then, with probability at least
$1-\delta$,
\begin{equation}
\|\nabla_\theta\hat\ell_J^{A,1}(\theta)
      -\nabla_\theta\ell(\theta)\|_2
\leq A_N(\theta)\sqrt p\,u_g(J,\delta).                        \label{eq:s2-gradient-concentration}
\end{equation}
Moreover, Assumption~{\refAssumpBoundedMeasurement} gives
$A_N(\theta)\leq2NG'(\theta)$.  Consequently, the explicit condition
\begin{equation}
J\geq K_N^{\rm pf}\log\frac{6p}{\delta}
 \max\left\{4,\frac{4pN^2G'(\theta)^2}{\epsilon^2}\right\}     \label{eq:s2-gradient-particle-count}
\end{equation}
implies that the left side of \eqref{eq:s2-gradient-concentration} is at most
$\epsilon$ with probability at least $1-\delta$.
\end{lem}

\begin{proof}
The lower bound on $J$ makes $u_g(J,\delta)\leq1/2$.  Apply
\eqref{eq:s2-scalar-concentration} to each normalized coordinate
$\psi_{\theta,i}/A_N(\theta)$; a constant coordinate has zero error and needs
no normalization.  A union bound over the $p$ coordinates gives
\begin{equation*}
\max_{1\leq i\leq p}
 |(\rho_N^J-\rho_N)(\psi_{\theta,i})|
 \leq A_N(\theta)u_g(J,\delta)
\end{equation*}
with probability at least $1-\delta$.  Now use
$\|v\|_2\leq\sqrt p\|v\|_\infty$ and
\eqref{eq:s2-gradient-identities}.  Finally,
$|\psi_{\theta,i}|\leq NG'(\theta)$, so
$A_N(\theta)\leq2NG'(\theta)$; substituting this inequality and solving
$2NG'(\theta)\sqrt p\,u_g(J,\delta)\leq\epsilon$ gives
\eqref{eq:s2-gradient-particle-count}.
\end{proof}

The next result bounds the particle Hessian, including the outer-product
and empirical-mean terms produced by differentiating the logarithm of the
terminal empirical average.

\begin{lem}[Explicit concentration of the DMOP-$1$ particle Hessian]
\label{lemma:hess_bound}
Suppose in addition that Assumption~{\refAssumpBoundedMeasurement} gives the
entrywise bound
\(
|\Xi_{\theta,ab}(Z_N)|\leq NG''(\theta)
\)
for $1\leq a,b\leq p$.  Define
\begin{equation}
d_p=\frac{p(p+3)}2,\qquad
H_N(\theta)=2p\{NG''(\theta)+3N^2G'(\theta)^2\},\qquad
u_H(J,\delta)=
 \sqrt{\frac{K_N^{\rm pf}}{J}\log\frac{6d_p}{\delta}}.         \label{eq:s2-hessian-constants}
\end{equation}
If $J\geq4K_N^{\rm pf}\log(6d_p/\delta)$, then, with probability at least
$1-\delta$,
\begin{equation}
\|\nabla_\theta^2\hat\ell_J^{A,1}(\theta)
      -\nabla_\theta^2\ell(\theta)\|_{\rm op}
\leq H_N(\theta)u_H(J,\delta).                                 \label{eq:s2-hessian-concentration}
\end{equation}
Fix $0<\kappa<\gamma$.  If
\begin{equation}
\gamma I_p\preceq-\nabla_\theta^2\ell(\theta)\preceq\Gamma I_p
\quad\hbox{and}\quad
J\geq K_N^{\rm pf}\log\frac{6d_p}{\delta}
 \max\left\{4,\frac{H_N(\theta)^2}{\kappa^2}\right\},       \label{eq:s2-hessian-particle-count}
\end{equation}
then, with probability at least $1-\delta$,
\begin{equation}
(\gamma-\kappa)I_p
 \preceq-\nabla_\theta^2\hat\ell_J^{A,1}(\theta)
 \preceq(\Gamma+\kappa)I_p.                                  \label{eq:s2-hessian-eigenvalues}
\end{equation}
In particular, the raw particle curvature is then positive definite and
invertible.
\end{lem}

\begin{proof}
For fixed reference-filter random numbers and resampling indices, the
telescoping identity in the proof of Lemma~\ref{lem:mop-1-formula} gives
\begin{equation*}
\hat\ell_J^{A,1}(\theta)
=C_\phi+\log\left\{\frac1J\sum_{j=1}^J\exp(a_j(\theta))\right\},
\qquad a_j(\theta)=\log w_{N,j}^{F,\theta},\qquad a_j(\phi)=0,
\end{equation*}
where $C_\phi$ is constant in $\theta$.  Twice differentiating this log-sum-exp at
$\theta=\phi$ gives the exact finite-particle identity
\begin{equation}
\nabla_\theta^2\hat\ell_J^{A,1}
=\rho_N^J(\Xi_\theta+\psi_\theta\psi_\theta^T)
 -\rho_N^J(\psi_\theta)\rho_N^J(\psi_\theta)^T.                \label{eq:s2-empirical-hessian-identity}
\end{equation}
Assumption~{\refAssumpBoundedMeasurement} also gives the entrywise dominating
bound
\begin{equation*}
|\partial_{\theta_a\theta_b}^2W_\theta|
=W_\theta|\Xi_{\theta,ab}+\psi_{\theta,a}\psi_{\theta,b}|
\leq\bar G^N\{NG''(\theta)+N^2G'(\theta)^2\}.
\end{equation*}
Thus differentiation through second order may pass under the expectation, and
differentiating $\log\E W_\theta$ gives the population identity
\begin{equation}
\nabla_\theta^2\ell
=\rho_N(\Xi_\theta+\psi_\theta\psi_\theta^T)
 -\rho_N(\psi_\theta)\rho_N(\psi_\theta)^T.                   \label{eq:s2-population-hessian-identity}
\end{equation}

Let $A_1=2NG'(\theta)$ and
$A_2=2\{NG''(\theta)+N^2G'(\theta)^2\}$.  The $p$ scalar functions
$\psi_a$ have oscillation at most $A_1$, and the symmetric entries of
$\Xi+\psi\psi^T$ have oscillation at most $A_2$.  Apply
\eqref{eq:s2-scalar-concentration} jointly to these functions, using
$d_p=p+p(p+1)/2$ scalar events.  On the resulting event,
\begin{align*}
\|(\rho_N^J-\rho_N)(\psi)\|_2
 &\leq A_1\sqrt p\,u_H(J,\delta),\\
\|(\rho_N^J-\rho_N)(\Xi+\psi\psi^T)\|_{\rm op}
 &\leq pA_2u_H(J,\delta).
\end{align*}
Both $\|\rho_N^J(\psi)\|_2$ and $\|\rho_N(\psi)\|_2$ are at most
$\sqrt p\,N G'(\theta)$.  Therefore
\begin{align*}
&\|\rho_N^J(\psi)\rho_N^J(\psi)^T
       -\rho_N(\psi)\rho_N(\psi)^T\|_{\rm op}\\
&\hspace{18mm}\leq
2\sqrt p\,N G'(\theta)\,
 \|(\rho_N^J-\rho_N)(\psi)\|_2
\leq4pN^2G'(\theta)^2u_H(J,\delta).
\end{align*}
Combining the last three displays with
\eqref{eq:s2-empirical-hessian-identity} and
\eqref{eq:s2-population-hessian-identity} proves
\eqref{eq:s2-hessian-concentration}.  Condition
\eqref{eq:s2-hessian-particle-count} makes its right side at most $\kappa$;
Weyl's eigenvalue inequality then gives
\eqref{eq:s2-hessian-eigenvalues}.
\end{proof}

For the optimization theorem, no stochastic Hessian claim is needed if the
curvature is regularized deterministically.  If
$\widetilde B=Q\operatorname{diag}(\lambda_1,\ldots,\lambda_p)Q^T$ is an
orthogonal eigendecomposition of any symmetric raw estimate, define
\begin{equation}
B=Q\operatorname{diag}([\lambda_1]_c^C,\ldots,[\lambda_p]_c^C)Q^T,
\qquad [x]_c^C=\min\{C,\max\{c,x\}\}.                          \label{eq:s2-spectral-clipping}
\end{equation}
For every $v\in\mathbb R^p$, direct expansion in the eigenbasis gives
\begin{equation}
c\|v\|_2^2\leq v^TBv\leq C\|v\|_2^2,
\qquad
C^{-1}\|v\|_2^2\leq v^TB^{-1}v\leq c^{-1}\|v\|_2^2.         \label{eq:s2-preconditioner-bounds}
\end{equation}
Thus $cI_p\preceq B\preceq CI_p$ and
$C^{-1}I_p\preceq B^{-1}\preceq c^{-1}I_p$.  Lemma~
\ref{lemma:hess_bound} shows that the unclipped particle Hessian has these
properties with $c=\gamma-\kappa$ and $C=\Gamma+\kappa$ under
\eqref{eq:s2-hessian-particle-count}; \eqref{eq:s2-spectral-clipping}
guarantees the chosen $c,C$ bounds for every realization.

\begin{proof}[Proof of Theorem~{\thmMopConvergence}.]
Let $\mathcal D$ be the convex region in Theorem~{\thmMopConvergence}, and set
\begin{equation}
G'_{\mathcal D}=\sup_{\theta\in\mathcal D}G'(\theta),\qquad
J_0=\left\lceil
K_N^{\rm pf}\log\!\left(\frac{6pM}{\delta}\right)
\max\left\{4,
\frac{4pN^2(G'_{\mathcal D})^2}{\epsilon^2}
\right\}\right\rceil.                                      \label{eq:s2-theorem-particle-count}
\end{equation}
Thus it is enough to take $J\geq J_0$.  With
$g_m=-\nabla_\theta\hat\ell_J^{A,1}(\theta_m)$, define
\begin{equation}
\tau=\min\big\{m\in\{0,\ldots,M-1\}:\|g_m\|_2\leq\sigma\epsilon\big\},
\qquad \min\varnothing=M,                                  \label{eq:s2-stopping-time}
\end{equation}
and
\begin{equation}
\mathcal E_M=
\bigcap_{m=0}^{\min\{\tau,M-1\}}
\big\{\|g_m-\nabla f(\theta_m)\|_2\leq\epsilon\big\}.    \label{eq:s2-gradient-event}
\end{equation}
If $\mathcal F_m$ is the history before the $m$th particle-filter call, then
$\theta_m$ and
$\{m\leq\min(\tau,M-1)\}$ are $\mathcal F_m$-measurable.  Lemma~
\ref{lemma:grad_bound}, with failure probability $\delta/M$, gives
\begin{equation}
\begin{aligned}
\Pr(\mathcal E_M^c\mid\theta_0)
&\leq\sum_{m=0}^{M-1}
 \E\left[\mathbf1_{\{m\leq\min(\tau,M-1)\}}
 \Pr\left(
 \|g_m-\nabla f(\theta_m)\|_2>\epsilon\mid\mathcal F_m
 \right)\,\middle|\,\theta_0\right]\\
&\leq\sum_{m=0}^{M-1}\frac{\delta}{M}=\delta.
\end{aligned}                                                   \label{eq:s2-all-gradient-events}
\end{equation}
This calculation is valid at the adaptive iterates because each call uses a
fresh particle-randomness block.

Work on $\mathcal E_M$ and write
$p_m=-B_m^{-1}g_m$.  If $\|g_m\|_2\leq\sigma\epsilon$, then
\begin{equation}
\|\nabla f(\theta_m)\|_2
 \leq\|g_m\|_2+\epsilon\leq(\sigma+1)\epsilon.                 \label{eq:s2-stopping-gradient}
\end{equation}
Strong convexity implies
$f(\theta)-f(\theta^*)\leq\|\nabla f(\theta)\|_2^2/(2\gamma)$;
therefore
\begin{equation}
f(\theta_m)-f(\theta^*)
\leq\frac{(\sigma+1)^2\epsilon^2}{2\gamma}.                   \label{eq:s2-stopping-neighborhood}
\end{equation}

Suppose instead that $\|g_m\|_2>\sigma\epsilon$.  Since
$g_m=-B_mp_m$ and $cI_p\preceq B_m\preceq CI_p$,
\begin{equation}
c\|p_m\|_2\leq\|g_m\|_2\leq C\|p_m\|_2,
\qquad
\epsilon\|p_m\|_2<\frac{C}{\sigma}\|p_m\|_2^2.              \label{eq:s2-noise-direction-bound}
\end{equation}
The smoothness inequality and the gradient-error event give
\begin{align}
f(\theta_m+\eta p_m)-f(\theta_m)
&\leq\eta\nabla f(\theta_m)^Tp_m
      +\frac{\Gamma\eta^2}{2}\|p_m\|_2^2\notag\\
&\leq-\eta p_m^TB_mp_m
      +\eta\epsilon\|p_m\|_2
      +\frac{\Gamma\eta^2}{2}\|p_m\|_2^2.                    \label{eq:s2-descent-start}
\end{align}
The choices of $\sigma$ and $\eta$ imply
\begin{equation*}
\frac{C}{\sigma}\|p_m\|_2^2
 \leq\frac{c(1-\beta)}2\|p_m\|_2^2,
\qquad
\frac{\Gamma\eta}{2}\|p_m\|_2^2
 \leq\frac{c(1-\beta)}2\|p_m\|_2^2.
\end{equation*}
Using $p_m^TB_mp_m\geq c\|p_m\|_2^2$ in
\eqref{eq:s2-descent-start} now gives
\begin{equation}
f(\theta_{m+1})-f(\theta_m)
\leq-\eta\beta p_m^TB_mp_m.                                  \label{eq:s2-descent}
\end{equation}

On a nonstopping iteration,
\begin{equation}
\|\nabla f(\theta_m)\|_2
 \leq\|g_m\|_2+\epsilon
 <\frac{\sigma+1}{\sigma}\|g_m\|_2.                          \label{eq:s2-gradient-comparison}
\end{equation}
Consequently, \eqref{eq:s2-preconditioner-bounds} and strong convexity give
\begin{samepage}
\begin{align}
p_m^TB_mp_m
 &=g_m^TB_m^{-1}g_m\notag\\
 &\geq\frac1C\left(\frac{\sigma}{\sigma+1}\right)^2
       \|\nabla f(\theta_m)\|_2^2\notag\\
 &\geq\frac{2\gamma}{C}
       \left(\frac{\sigma}{\sigma+1}\right)^2
       \{f(\theta_m)-f(\theta^*)\}.                           \label{eq:s2-gap-direction}
\end{align}
\end{samepage}
Combining \eqref{eq:s2-descent} and \eqref{eq:s2-gap-direction} proves
\begin{equation}
f(\theta_{m+1})-f(\theta^*)
\leq q\{f(\theta_m)-f(\theta^*)\},
\qquad
q=1-\frac{2\eta\beta\gamma}{C}
       \left(\frac{\sigma}{\sigma+1}\right)^2.               \label{eq:s2-linear-rate}
\end{equation}
The step-size bound also verifies the claimed range of the contraction:
\begin{equation*}
0<1-q
<2\beta(1-\beta)\frac{c\gamma}{C\Gamma}
\leq\frac12,
\qquad\text{so}\qquad \frac12<q<1.
\end{equation*}
Iteration gives the $q^m$ conclusion in the theorem, while
\eqref{eq:s2-stopping-neighborhood} gives its stopping conclusion after
substituting $f=-\ell$.
\end{proof}

The result covers DMOP-$1$. For $\alpha<1$, the population limit is a
discounted gradient, so extending the theorem requires a finite-sample bound
on the discount bias relative to the exact score.

\section{Feynman-Kac Models and Monte Carlo Approximations}
\label{appendix:feynman}

In this section, we introduce the Feynman-Kac convention of \cite{delMoral04} that has since become commonplace \citep{karjalainen23} for the analysis of the particle filter.
The mathematical formalization and notation introduced here will be adopted in the remainder of the analysis, in order to prove Theorems~{\thmMopTargeting}, \thmMopGradConsistency, and \thmMopBiasvar.
Let $(\eta_n)_{n=1}^N$ and $(\pi_n)_{n=1}^N$ be the prediction and filtering
probability measures on the current state space $\gX$. The posterior measure
$\rho_n$ is on simulator histories through time $n$, including the base
variables, as defined in Equation~(\ref{eq:canonical-extended-posterior-path-law}).
Its state-path marginal is $f_{X_{0:n}|Y_{1:n}}$. These are the measures
we seek to approximate with the particle filter.
For any measurable bounded functional $h$, we adopt the following functional-analytic notation, borrowed from \cite{delMoral04, chopin20, karjalainen23}.
We choose our specific choice of notation and definitions to be in line with that of \cite{karjalainen23}.

\paragraph{\bf Markov kernels and the process model:} A Markov kernel $M$ with source $\gX_1$ and target $\gX_2$ is a map $M: \gX_1 \times \mathcal{B}(\gX_2) \to [0,1]$ such that for every set $A \in \mathcal{B}(\gX_2)$ and every point $x \in \gX_1$, the map $x \mapsto M(x, A)$ is a measurable function of $x$, and the map $A \mapsto M(x,A)$ is a probability measure on $\gX_2$. The quantity $M(x,A)$ can be thought of as the probability of transitioning to the set $A$ given that we are at the point $x$. If this yields a density, this then corresponds to the process density $f_{X_{2}|X_1}$ conditional on $x$ and integrated over $A$.  

\paragraph{\bf Markov kernels and measures:} For any measure $\eta$, any Markov kernel $M$ on $\gX$, any point $x \in \gX$ and any measurable subset $A \subseteq \gX$, let 
\begin{align}
    \eta(h) &= \int h \, d\eta = \int h(x) \eta(dx), \\(\eta M)(A) &= \int \eta(dx)M(x,A), \\
    (Mh)(x) &= \int M(x, dy) h(y).
\end{align}

\paragraph{\bf Compositions of Markov kernels:} The composition of a Markov kernel $M_1$ with another Markov kernel $M_2$ is another Markov kernel, given by 
\begin{equation}
 (M_1M_2)(x, A) = \int M_1(x, dy) M_2(y, A).
\end{equation}

\paragraph{\bf Total variation distance:} The total variation distance between two measures $\mu$ and $\nu$ on $\gX$ is
\begin{equation}
\|\mu-\nu\|_{\mathrm{TV}}=\sup _{\|h\|_{\infty} \leq 1 / 2}|\mu(h)-\nu(h)|=\sup _{\operatorname{osc}(h) \leq 1}|\mu(h)-\nu(h)|.    
\end{equation}

\paragraph{\bf Dobrushin contraction:} The Dobrushin contraction coefficient $\beta_{\text{TV}}$ of a Markov kernel $M$ is given by
\begin{equation}
\beta_{\mathrm{TV}}(M)=\sup _{x, y \in \gX}\big\|M(x, \cdot)-M(y, \cdot)\big\|_{\mathrm{TV}}=\sup _{\mu, \nu \in \mathcal{P}, \mu \neq \nu} \frac{\|\mu M-\nu M\|_{\mathrm{TV}}}{\|\mu-\nu\|_{\mathrm{TV}}}.    
\end{equation}

\paragraph{\bf Potential functions and the measurement model:} A potential function $G : \gX \to [0,\infty)$ is a non-negative function of an element of the state space $x \in \gX$. 
In our case, this corresponds to the measurement model, and in our previous notation is written as $g_{n,j} = f_{Y_n|X_n}\big(y_n^*|x_{n,j}^P;\theta\big) = G_n(x_{n,j}^P)$, where we suppress the dependence on $\theta$ for notational simplicity.
Note that $G_n(\cdot) = f_{Y_n|X_n}(y_n^*|\;\cdot\;)$ is the conditional density of the observed measurement at time $n$, evaluated at a state in $\gX$. At the filtered particle,
$G_n(x_{n,j}^F)=g_{n,k_{n,j}}$.

\paragraph{\bf Feynman-Kac models:} A Feynman-Kac model on $\gX$ is a tuple $(\pi_0, (M_n)_{n=1}^N, (G_n)_{n=1}^N)$ of an initial probability measure on the state space $\pi_0$, a sequence of transition kernels $(M_n)_{n=1}^N$, and a sequence of potential functions $(G_n)_{n=1}^N$. In the notation used in the main text, this corresponds to the starting distribution $f_{X_0}$, the sequence of transition densities $f_{X_{n}|X_{n-1}}$, and the measurement densities $f_{Y_n|X_n}$. This induces a set of mappings from the set of probability measures on $\gX$ to itself, $\mathcal{P}(\gX) \to \mathcal{P}(\gX)$, as follows:
\begin{itemize}
    \item The update from the prediction to the filtering distributions is given by 
    \begin{equation}
    \pi_n(dx) = \Psi_n(\eta_n)(dx) = \frac{G_n(x)\cdot\eta_n(dx)}{\eta_n(G_n)}.
    \end{equation}
    \item The map from the prediction distribution at timestep $n$ to timestep $n+1$ is given by 
    \begin{equation}
 \Phi_{n+1}(\eta_n) = \Psi_n(\eta_n) M_{n+1}.       
    \end{equation}
    \item The composition of maps between prediction distributions yields the map from the prediction distribution at time $k$ to the prediction distribution at time $n$ where $k \leq n$,
    \begin{equation}
 \Phi_{k,n} = \Phi_n \circ ... \circ \Phi_{k+1}.       
    \end{equation}
\end{itemize}
\paragraph{\bf The particle filter:} The particle filter then yields a Monte Carlo approximation to the above Feynman-Kac model, via a sequence of mixture Dirac measures. When one resamples at every timestep, the prediction measure at timestep $n$ is then given by 
\begin{equation}
\eta_n^J = \frac{1}{J}\sum_{j=1}^J \delta_{x_{n,j}^P},
\end{equation}
and the filtering measure at timestep $n$ is given by
\begin{equation}
\pi_n^J = \frac{\sum_{j=1}^J g_{n,j} \delta_{x_{n,j}^P}}{\sum_{j=1}^J g_{n,j}} \approx \frac{1}{J} \sum_{j=1}^J \delta_{x_{n,j}^F}.
\end{equation}
In a slight abuse of notation, we will identify $x_{n, 1:J}^P \equiv \eta_n^J$, and $x_{n, 1:J}^F \equiv \pi_n^J$. 
As in \cite{karjalainen23}, one can view this as an inhomogenous Markov process evolving on $\gX^{J}$. For $\theta=\phi$ and multinomial resampling, the corresponding
Markov transition kernel is
\begin{equation}
\textbf{M}_n(x_{n-1, 1:J}^P, \cdot) = \left(\Phi_{n}\left(\eta_{n-1}^J\right)\right)^{\otimes J} = \left(\Phi_{n}\left(\frac{1}{J}\sum_{j=1}^J \delta_{x_{n-1,j}^P}\right)\right)^{\otimes J},
\end{equation}
and the composition of Markov kernels on particles from timestep $n$ to timestep $n+k$ is written 
\begin{equation}
\textbf{M}_{n, n+k} = \textbf{M}_{n+1}\cdots \textbf{M}_{n+k}.
\end{equation}
The kernel acts on the full prediction cloud. Conditional on the preceding
prediction cloud, its law is
\begin{equation}
\operatorname{Law}\big(x_{n,1:J}^P\mid x_{n-1,1:J}^P\big)
=\textbf{M}_n(x_{n-1,1:J}^P,\cdot).
\end{equation}

\paragraph{\bf Forgetting of the particle filter:} 
The above formalization yields a result from \cite{karjalainen23} on the forgetting of the particle filter that we require for our analysis of the bias, variance, and error of MOP-$\alpha$.  
That is, \cite{karjalainen23} show that
\begin{equation}
\beta_{\mathrm{TV}}\left(\mathbf{M}_{n, n+k}\right) \leq(1-\epsilon)^{\lfloor k /(O(\log J))\rfloor},
\end{equation}
for some $\epsilon$ dependent on $\bar{G}, \underbar{G}, \bar{M}, \underbar{M}$ in Assumptions~{\refAssumpBoundedMeasurement} and~{\refAssumpBoundedProcess}. As a result, the mixing time of the particle filter is only on the order of $O(\log(J))$ timesteps.

\vspace{3mm}

\section{Proof of Theorem~{\thmMopTargeting} via a Strong Law of Large Numbers for Triangular Arrays of Particles With Off-Parameter Resampling}
\label{appendix:targeting}

Here, we prove Theorem~{\thmMopTargeting} by deriving more general results, from which it will be clear that MOP-$\alpha$ targets the filtering distribution and is strongly consistent for the likelihood.
We will prove a strong law of large numbers for triangular arrays of particles with off-parameter resampling, meaning that we resample the particles according to an arbitrary resampling rule that is not necessarily in proportion to the target distribution of interest. 
Using weights that encode the cumulative discrepancy between the resampling distribution and the target distribution (instead of resampling to equal weights, as in the basic particle filter) provides a sufficient correction to ensure almost sure convergence.
We now introduce the precise definition of an off-parameter resampled particle filter. 

\begin{defn}[Off-Parameter Resampled Particle Filters]
    \label{defn:off-parameter-filter}
    An off-parameter resampled particle filter is a Monte Carlo approximation to a Feynman-Kac model $\left(\pi_0,\left(M_n\right)_{n=1}^N,\left(G_n\right)_{n=1}^N\right)$, where we inductively define given some $x_{n-1,j}^F$, $w_{n-1,j}^F$ comprising the filtering measure approximation at time $n-1$:
    \begin{enumerate}
        \item Conditional on the filtered particles and weights at time $n-1$,
        draw independently $x_{n,j}^P\sim M_n(x_{n-1,j}^F,\cdot)$ using fresh
        simulator draws. Retain the corresponding parent weight:
        $w_{n,j}^P=w_{n-1,j}^F$.
        \item The prediction measure at timestep $n$ is given by 
        \begin{equation}
        \eta_n^J = \frac{\sum_{j=1}^J w_{n,j}^P \, \delta_{x_{n,j}^P}}{\sum_{j=1}^J w_{n,j}^P}.
        \end{equation}
        \item The particles are resampled at every timestep $n$, yielding indices $k_j$ according to some arbitrary probabilities $(p_{n,j})_{j=1}^J$.
        \item The filtering measure at timestep $n$ is given by
        \begin{equation}
        \pi_n^J = \frac{\sum_{j=1}^J \delta_{x_{n,j}^P } w_{n,j}^P \, g_{n,j}}{\sum_{j=1}^J w_{n,j}^P \, g_{n,j}} \approx \frac{\sum_{j=1}^J \delta_{x_{n,k_j}^P }w_{n,k_j}^P \, g_{n,k_j}\big/p_{n,k_j}}{\sum_{j=1}^J w_{n,k_j}^P\, g_{n,k_j}\big/p_{n,k_j}} = \frac{\sum_{j=1}^J w_{n,j}^F \, \delta_{x_{n,j}^F}}{\sum_{j=1}^J w_{n,j}^F}.
        \end{equation}
        \item The prediction weights at the next timestep are given by $w_{n+1,j}^P = w_{n,j}^F =  w_{n,k_j}^P \, g_{n,k_j}\big/p_{n,k_j}$.
    \end{enumerate}
\end{defn}

To prove that the weight correction is sufficient for almost sure convergence, we first introduce what it means for a triangular array of particles to \textit{target} a given target distribution. 

\begin{defn}[Targeting]
    A pair of random vectors $(X, W)$ drawn from some measure $g$ \textbf{targets} another measure $\pi$ if, for any measurable and bounded functional $h$,
\begin{equation}
    E_g\big[h(X) \cdot W\big]=E_\pi\big[h(X)\big].
\end{equation}  
    A particle approximation is a triangular array of pairs of random vectors $(X^J_j, W^J_j), j=1,2, \ldots,J$.
    The particle approximation \textbf{targets} $\pi$ if, for any measurable and bounded functional $h$,
\begin{equation}
    \frac{\sum_{j=1}^J h(X^J_j) \, W^J_j}{\sum_{j=1}^J W^J_j} \stackrel{a.s.}{\to} E_\pi\big[h(X)\big]
\end{equation}
as $J \to \infty$.
\end{defn}
\cite{chopin04} asserted without proof that common particle filter algorithms target the filtering distribution in this sense, while \cite{chopin20} proved a related result assuming bounded densities.
We follow a similar approach to \cite{chopin20}, based on showing strong laws of large numbers for triangular arrays, noting that triangular array strong laws do not generally hold without an additional regularity condition such as boundedness.
 
In order to prove the consistency of our variation on the particle filter, we now present three helper lemmas.
The first follows from standard importance sampling arguments, the second from integrating out the marginal, and the third from Bayes' theorem. 
We state Lemma~\ref{lem:change-measure-proper-weights} assuming multinomial resampling, which is convenient for the proof though other resampling strategies may be preferable in practice.

\begin{lem}[Change of Weight Measure]
    \label{lem:change-measure-proper-weights}
    Suppose that $\{(\tilde X_j^J,U_j^J),j=1,\dots,J\}$ targets $f_X$. Now, let $\{(Y_j^J,V_j^J),j=1,\dots,J\}$ be a multinomial sample with indices $k_j$ drawn from $\{(\tilde X_j^J,U_j^J)\}$ where $(\tilde X_j^J,U_j^J)$ is represented, on average, proportional to $\pi^J_j J$ times, with $\pi_j^J>0$. Write
    \begin{equation}
    (Y_j^J,V_j^J) = \big(\tilde X^J_{k_j},U^J_{k_j}\big/\pi^J_{k_j}\big).
    \end{equation} 
    If the importance sampling weights $U_j/\pi_j$ are bounded uniformly in $j,J$ after normalizing $\sum_m U_m^J=1$, then $\{(Y^J_j,V^J_j),j=1,\dots,J\}$ targets $f_X$.
\end{lem}

\begin{proof}[Proof.]
    Normalize $\sum_j U_j^J=1$, which leaves both weighted empirical measures unchanged.
    Note that as the $Y_j^J$ are a subsample from $X_j^J$, $h$ can be a function of $Y$ as well as it is one for $X$. We then expand
    \begin{equation}\frac{\sum_j h(Y_j^J) \, V_j^J}{\sum_j V_j^J} = \frac{\sum_j h(\tilde X_{k_j}^J)\, \frac{U_{k_j}^J}{\pi_{k_j}^J}}{\sum_j \frac{U_{k_j}^J}{\pi_{k_j}^J}}.\end{equation}
    By hypothesis,
    \begin{equation}
    \frac{\sum_j h(\tilde X_j^J)\, U_j^J}{\sum_j U_j^J} \stackrel{a.s.}{\to} \E_{f_X}\big[h(X)\big].
    \end{equation}
    We want to show
    \begin{equation}\label{eq:lemma1:h}
    \frac{\sum_j h(\tilde X_{k_j}^J)\, \frac{U_{k_j}^J}{\pi_{k_j}^J}}{\sum_j \frac{U_{k_j}^J}{\pi_{k_j}^J}} - \frac{\sum_j h(\tilde X_j^J)\, U_j^J}{\sum_j U_j^J} \stackrel{a.s.}{\to} 0.
    \end{equation}
    For this, it is sufficient to show that
   \begin{equation} 
   \frac1J\sum_j h(\tilde X_{k_j}^J)\frac{U_{k_j}^J}{\pi_{k_j}^J}
    -  \sum_j h(\tilde X_j^J)\frac{U_j^J}{\pi_j^J}\pi_j^J \stackrel{a.s.}{\to} 0 
    \end{equation}
    since an application of this result with $h(x)=1$ provides almost sure convergence to $1$ of $J^{-1}$ times the denominator in (\ref{eq:lemma1:h}).
    Write $g_j^J = h(\tilde X_j^J)\frac{U_{j}^J}{\pi_{j}^J}$. We therefore need to show that
    \begin{equation}
    \frac1J\sum_j Z_j^J := \frac1J\sum_j \left(g_{k_j}^J - \sum_{m=1}^J g_m^J\pi_m^J \right) \stackrel{a.s.}{\to} 0.
\end{equation}
    Because the functional $h$ and importance sampling weights $U_{k_j}^J/\pi_{k_j}^J$ are bounded, we have that $\sup_{J,j}\E\big[(Z_j^J)^4\big] < \infty$. We can then follow the argument of \cite{chopin20} from this point on, where noting that the $Z_j^J$ are conditionally independent, identically distributed, and centered given the $\tilde{X}_j^J$, $U_j^J$, and $\pi_j^J$,
    \begin{eqnarray} 
    \E\left[\left(\sum_j Z_j^J\right)^4 \Bigg| (\tilde X_j^J,U_j^J,\pi_j^J)_{j=1}^J \right]
    &=& J\, \E\left[(Z_1^J)^4|(\tilde X_j^J,U_j^J,\pi_j^J)_{j=1}^J\right]
    \nonumber
    \\
    && \hspace{1mm}+ \hspace{2mm} 3J(J-1)\left(\E\big[(Z_{1}^J)^2\big|(\tilde X_j^J,U_j^J,\pi_j^J)_{j=1}^J\big]\right)^{2}
    \\ 
    &\leq& CJ^2,
    \end{eqnarray}
    for some $C>0$ independent of $J$.
    Taking expectations on both sides yields
    \begin{equation}\E\left[\left(\sum_j Z_j^J\right)^4 \right]  \leq CJ^2\end{equation}
    by the tower property. Now by Markov's inequality, 
    \begin{equation}\mathbb{P}\left(\left|\frac{1}{J} \sum_{j=1}^J Z_j^J\right|>\epsilon \right) 
    \leq \frac{1}{\epsilon^4J^4 } 
    \mathbb{E}\left[\left(\sum_{j=1}^J Z_j^J\right)^4\right] \leq \frac{C}{\epsilon^4J^2},\end{equation}
    and as these terms are summable we can apply Borel-Cantelli to conclude that these deviations happen only finitely often for every $\epsilon>0,$ giving us the almost-sure convergence for   
    \begin{equation} 
    \frac1J\sum_j h(\tilde X_{k_j}^J)\, \frac{U_{k_j}^J}{\pi_{k_j}^J}
    -  \sum_j h(\tilde X_j^J)\, \frac{U_j^J}{\pi_j^J}\, \pi_j^J \stackrel{a.s.}{\to} 0.
    \end{equation} 
    Similarly, we also have that
    \begin{equation} 
    \sum_j \frac{U_j^J}{\pi_j^J}\pi_j^J
    - \frac1J\sum_j \frac{U_{k_j}^J}{\pi_{k_j}^J} \stackrel{a.s.}{\to} 0,
    \end{equation}
    and the result follows by taking the ratio, since $\sum_j U_j^J=1$.

\end{proof}

\noindent \textbf{Remark:} Note that Lemma \ref{lem:change-measure-proper-weights} permits $\pi_{1:J}$ to depend on $\{(X_j,u_j)\}$ as long as the resampling is carried out independently of $\{(X_j,u_j)\}$, conditional on $\pi_{1:J}$.

\begin{lem}[Particle Marginals]
\label{lem:marginal-proper-weights}
Suppose that $\{(\tilde X_j^J,U_j^J),j=1,\dots,J\}$ targets $f_X$.
Also suppose that conditional on $(\tilde X_m^J,U_m^J)_{m=1}^J$,
the variables $\tilde Z_1^J,\ldots,\tilde Z_J^J$ are independent, with
$\tilde Z_j^J\sim f_{Z|X}(\cdot\mid\tilde X_j^J)$,
where $f_{Z|X}$ is a conditional probability density function corresponding
to a joint density $f_{X,Z}$ with marginal densities $f_X$ and $f_Z$.
Then, if the $U_j^J$ are uniformly bounded above and away from zero,
$\{(\tilde Z_j^J,U_j^J)\}$ targets $f_Z$.
\end{lem}

\begin{proof}[Proof.]
We want to show that, for any measurable bounded $h$,
\begin{equation}
 \frac{\sum_j h(\tilde Z_j^J)\,U_j^J}{\sum_j U_j^J}
 \stackrel{a.s.}{\to}\E_{f_Z}[h(Z)]
 =\E_{f_X}\big[\E_{f_{Z|X}}[h(Z)\mid X]\big].
\end{equation}
By assumption, for any measurable and bounded functional $g$ with domain $\gX$,
\begin{equation}\label{eq:lemma2:g}
 \frac{\sum_j g(\tilde X_j^J)\,U_j^J}{\sum_j U_j^J}
 \stackrel{a.s.}{\to}\E_{f_X}[g(X)].
\end{equation}
Let $\bar U_j^J=JU_j^J/\sum_j U_j^J$.
For constants $0<c\le U_j^J\le C<\infty$, uniformly in $j,J$,
we have $\bar U_j^J\le C/c$.
Examine the numerator and denominator of the quantity
\begin{equation}
 \frac{J^{-1}\sum_j h(\tilde Z_j^J)\,U_j^J}{J^{-1}\sum_j U_j^J}
 =\frac{J^{-1}\sum_j h(\tilde Z_j^J)\,\bar U_j^J}
 {J^{-1}\sum_j\bar U_j^J}.
\end{equation}
The denominator converges to $1$ almost surely. The numerator, on the other hand, is
\begin{equation}\label{eq:lemma6:numerator}
 \frac1J\sum_j h(\tilde Z_j^J)\,\bar U_j^J.
\end{equation}
Set
\[
 D_J=\frac1J\sum_j\bar U_j^J
 \bigl\{h(\tilde Z_j^J)-\E[h(Z)\mid X=\tilde X_j^J]\bigr\}.
\]
Conditioning on $(\tilde X_j^J,U_j^J)_{j=1}^J$, we apply the fourth-moment
argument of Lemma~\ref{lem:change-measure-proper-weights} to the centered, uniformly bounded summands, obtaining
\[
 \E[D_J^4]\le\frac{C_h}{J^2},\qquad
 \mathbb P(|D_J|>\epsilon)\le\frac{C_h}{\epsilon^4J^2},
\]
for some $C_h<\infty$ independent of $J$ and every $\epsilon>0$.
Thus $D_J\stackrel{a.s.}{\to}0$ by Borel--Cantelli.
For the expectation of (\ref{eq:lemma6:numerator}), the tower property gives
\begin{align}
 \lim_{J\to\infty}\E\left[\frac1J\sum_j h(\tilde Z_j^J)\,\bar U_j^J\right]
 &=\lim_{J\to\infty}\E\left[\frac1J\sum_j
 \E\big[h(\tilde Z_j^J)\,\bar U_j^J\mid\tilde X_j^J,\bar U_j^J\big]\right]
 \\
 &=\lim_{J\to\infty}\E\left[\frac1J\sum_j
 \E\big[h(Z)\mid X=\tilde X_j^J\big]\,\bar U_j^J\right].
\end{align}
Applying (\ref{eq:lemma2:g}) with $g(x)=\E[h(Z)\mid X=x]$, the average on the
right hand side converges almost surely to
$\E\{\E[h(Z)\mid X]\}=\E[h(Z)]$.
It remains to swap the limit and expectations. We can do so with the bounded
convergence theorem. The almost-sure convergence of this average,
together with $D_J\stackrel{a.s.}{\to}0$, gives
\begin{equation}
 \frac1J\sum_j h(\tilde Z_j^J)\,\bar U_j^J
 \stackrel{a.s.}{\to}\E_{f_Z}[h(Z)].\qedhere
\end{equation}
\end{proof}

\noindent\textbf{Remark.} The conditional independence assumption in Lemma~\ref{lem:marginal-proper-weights} holds for
standard particle filters, and for our MOP filter it is guaranteed by step 1
of Definition~\ref{defn:off-parameter-filter}. 

\vspace{1mm}

\noindent \textbf{Remark:}  The uniform boundedness condition in Lemma~\ref{lem:marginal-proper-weights} can be guaranteed by construction in the context of Definition~\ref{defn:off-parameter-filter} for suitable choices of the resampling probabilities $\pi_{1:J}$ at an earlier iteration.
In the context of DMOP-$\alpha$, this is guaranteed if the measurement density is similarly bounded, which holds under Assumption (A2) for our theoretical analysis. 
It is trivially true for MOP-$\alpha$ at $\theta=\phi$, since in this case the weights at each propagation step are equal.
The smoothness assumption (A5) therefore also guarantees this condition in a neighborhood of $\theta=\phi$.

\vspace{1mm}

\noindent\textbf{Remark.} For nonnegative weights with positive sum, the weight condition in Lemma~\ref{lem:marginal-proper-weights} can be replaced by the weaker requirement
\[
 \max_{1\le j\le J}\frac{JU_j^J}{\sum_mU_m^J}\le B<\infty,
\]
almost surely, for a constant $B$ independent of $J$. The same proof applies.

\begin{lem}[Particle Posteriors]
    \label{lem:posterior-proper-weights}
    Suppose that $\{(X_j^J,U_j^J),j=1,\dots,J\}$ targets $f_X$.
    Also suppose that $(X^{\prime J}_j,U^{\prime J}_j) = \big(X_j^J,U_j^J \, f_{Z|X}(z^*|X_j^J)\big)$. 
    Then, if $x\mapsto f_{Z|X}(z^*|x)$ is bounded and $0<f_Z(z^*)<\infty$, $\{(X^{\prime J}_j,U^{\prime J}_j)\}$ targets $f_{X|Z}(\cdot | z^*)$.
\end{lem}

\begin{proof}[Proof.]
    Again, we want to show that, for any measurable bounded $h$,
     \begin{equation}
     \frac{\sum_j h(X_j^J) \cdot U_j^J \cdot f_{Z|X}(z^*|X_j^J)}{\sum_j U_j^J \cdot f_{Z|X}(z^*|X_j^J)} \stackrel{a.s.}{\to} \E_{f_{X|Z}}\big[h(X)\big|z^*\big].
     \end{equation}
    We already have that for any measurable bounded $g$,
    \begin{equation}
        \frac{\sum_j g(X_j^J)\, U_j^J}{\sum_j U_j^J} \stackrel{a.s.}{\to} \E_{f_X}\big[g(X)\big].
        \label{eq:lemma-posterior-hypothesis}
    \end{equation}
    Consider the following:
    \begin{equation}
    \frac{J^{-1} \sum_j h(X_j^J) \, f_{Z|X}(z^*|X_j^J)\,  {U}_j^{J}}{J^{-1} \sum_j {U}_j^{J}}
    \times \left( \frac{J^{-1} \sum_j f_{Z|X}(z^*|X_j^J) \, {U}_j^{J}}{J^{-1} \sum_j {U}_j^{J}} \right)^{-1}.
    \end{equation}
    The functions $x\mapsto f_{Z|X}(z^*|x)$ and $x\mapsto h(x)f_{Z|X}(z^*|x)$ are bounded.
    We will apply Equation (\ref{eq:lemma-posterior-hypothesis}) to the numerator and the denominator in the ratio above individually. The numerator converges to 
    \begin{equation}
    \frac{J^{-1} \sum_j h(X_j^J) \, f_{Z|X}(z^*|X_j^J) \, {U}_j^{J}}{J^{-1} \sum_j {U}_j^{J}} \stackrel{a.s.}{\to} \E_{f_X}\big[h(X)\, f_{Z|X}(z^*|X)\big],
    \end{equation}
    while the denominator converges to
    \begin{equation} 
    \frac{J^{-1} \sum_j f_{Z|X}(z^*|X_j^J) \, {U}_j^{J}}{J^{-1} \sum_j {U}_j^{J}}  \stackrel{a.s.}{\to} \E_{f_X}\big[f_{Z|X}(z^*|X)\big] = f_Z(z^*).
    \end{equation}
    Now we take advantage of the identities
    \begin{eqnarray}
    \frac{\E_{f_X}\big[h(X)\, f_{Z|X}(z^*|X)\big]}{f_Z(z^*)} &=& \E_{f_X}\left[\frac{h(X)\, f_{Z|X}(z^*|X)}{f_Z(z^*)}\right] 
    \\
    &=& \E_{f_X}\left[h(X)\frac{f_{X|Z}(X|z^*)}{f_X(X)}\right] =  \E_{f_{X|Z}}\big[h(X)|z^*\big],
    \end{eqnarray}
    to give the desired result.


\end{proof}

\noindent\textbf{Remark.} In the MOP application, the boundedness and positivity conditions in Lemma~\ref{lem:posterior-proper-weights} follow from Assumption~\refAssumpBoundedMeasurement.

\begin{thm}[Off-Parameter Resampled Particle Filters Target the Filtering Distribution]
    \label{thm:off-parameter-targeting}
    The off-parameter resampled particle filter of Definition \ref{defn:off-parameter-filter} targets the filtering distribution when the resampling probabilities $(p_{n,j})_{j=1}^J$ are chosen so that $g_{n,j}/(J\, p_{n,j})$ is uniformly bounded above and away from zero whenever $g_{n,j}>0$.
\end{thm}
\begin{proof}[Proof.]
    Recursively applying Lemmas \ref{lem:change-measure-proper-weights}, \ref{lem:marginal-proper-weights}, and \ref{lem:posterior-proper-weights}, we obtain that the off-parameter resampled particle filter targets the posterior.
    Specifically, suppose inductively that $\big\{\big(X^{F}_{n-1,j},w^{F}_{n-1,j}\big)\big\}$ targets $\pi_{n-1}$.
    Then, Lemma \ref{lem:marginal-proper-weights} tells us that $\big\{\big(X^{P}_{n,j},w^{P}_{n,j}\big)\big\}$ targets $\eta_n$.
    Lemma \ref{lem:posterior-proper-weights} tells us that $\big\{\big(X^{P}_{n,j},w^{P}_{n,j} g^\theta_{n,j} \big)\big\}$ therefore targets  $\pi_n$.
    Lemma \ref{lem:change-measure-proper-weights} guarantees that the resampling rule, given by 
    \begin{equation}
    \big(X^{F}_{n,j},w^{F}_{n,j}\big) = \big(X^{P}_{n,k_j}, w^{P}_{n,k_j} \, g_{n,k_j}\big/ p_{n,k_j}\big),
    \end{equation}
    with resampling probabilities proportional to $p_{n,j}$, therefore also targets $\pi_n$.
\end{proof}

\begin{prop}[MOP-1 Targets the Filtering Distribution]
    When $\alpha=1$ or $\phi=\theta$, MOP-$\alpha$ targets the filtering distribution. 
\end{prop}
\begin{proof}[Proof.]
    When $\theta=\phi$, regardless of the value of $\alpha$, the ratio $\frac{g_{n,j}^\theta}{g_{n,j}^\phi}=1,$ and this reduces to the vanilla particle filter estimate.

    When $\alpha=1$, and $\theta\neq\phi,$ the proof is identical to that of \ref{thm:off-parameter-targeting}. Recursively applying Lemmas \ref{lem:change-measure-proper-weights}, \ref{lem:marginal-proper-weights}, and \ref{lem:posterior-proper-weights}, we obtain that 
    the MOP-1 filter targets the posterior.
    Specifically, suppose inductively that $\big\{\big(X^{F,\theta}_{n-1,j},w^{F,\theta}_{n-1,j}\big)\big\}$ is properly weighted for $f_{X_{n-1}|Y_{1:n-1}}(x_{n-1}|y^*_{1:n-1};\theta)$.
    Then, Lemma \ref{lem:marginal-proper-weights} tells us that $\big\{\big(X^{P,\theta}_{n,j},w^{P,\theta}_{n,j}\big)\big\}$ targets $f_{X_{n}|Y_{1:n-1}}(x_{n}|y^*_{1:n-1};\theta)$.
    Lemma \ref{lem:posterior-proper-weights} tells us that $\big\{\big(X^{P,\theta}_{n,j},w^{P,\theta}_{n,j} \, g^\theta_{n,j} \big)\big\}$ therefore targets  $f_{X_{n}|Y_{1:n}}(x_{n}|y^*_{1:n};\theta)$.
    Lemma \ref{lem:change-measure-proper-weights} guarantees that the resampling rule, given by 
    \begin{equation}
    \big(X^{F,\theta}_{n,j},w^{F,\theta}_{n,j}\big) = \big(X^{P,\theta}_{n,k_j}, w^{P,\theta}_{n,k_j} \, g^\theta_{n,k_j}\big/ g^\phi_{n,k_j}\big),
    \end{equation}
    with resampling weights proportional to $g^\phi_{n,j}$, therefore also targets $f_{X_{n}|Y_{1:n}}(x_{n}|y^*_{1:n};\theta)$.
\end{proof}

This has addressed filtering, but not quite yet the likelihood evaluation. For this we use the following lemma.

\begin{lem}[Likelihood Proper Weighting]
    \label{lem:lik-proper-weight}
  $f_{Y_n|Y_{1:n-1}}(y_n^*|y_{1_n-1}^*;\theta)$ is strongly consistently estimated by either the before-resampling estimate,
\begin{equation}\label{L1}
L_n^{B,\theta} =  \frac{\sum_{j=1}^Jg^\theta_{n,j} \, w^{P,\theta}_{n,j}}{\sum_{j=1}^J  w^{P,\theta}_{n,j}},
\end{equation}
or by the after-resampling estimate,
\begin{equation}\label{L2}
L_n^{A,\theta} = L^\phi_n \frac{\sum_{j=1}^Jw^{F,\theta}_{n,j}}{\sum_{j=1}^J  w^{P,\theta}_{n,j}}.
\end{equation}
where $L^\phi_n$ is as defined in the various algorithms.
\end{lem}

Here, (\ref{L1}) is a direct consequence of our earlier result that $\{ \big(X^{P,\theta}_{n,j},w^{P,\theta}_{n,j}\big) \}$ targets\\
$f_{X_{n}|Y_{1:n-1}}(x_{n}|y^*_{1:n-1};\theta)$.
To see (\ref{L2}), define the corrected weights before resampling and the
reference resampling probabilities by
\begin{equation}
w_{n,j}^{FC,\theta}
=w_{n,j}^{P,\theta}\frac{g_{n,j}^\theta}{g_{n,j}^\phi},
\qquad
p_{n,j}=\frac{g_{n,j}^\phi}{JL_n^\phi},
\qquad \sum_{j=1}^Jp_{n,j}=1.
\end{equation}
The numerator of (\ref{L1}) can then be written as
\begin{equation}
\sum_{j=1}^Jw_{n,j}^{P,\theta}g_{n,j}^\theta
=JL_n^\phi\sum_{j=1}^Jp_{n,j}w_{n,j}^{FC,\theta}.
\end{equation}
Resampling selects $w_{n,j}^{F,\theta}=w_{n,k_j}^{FC,\theta}$.
With the pre-resampling particles, weights, and reference densities fixed,
the expectation over multinomial resampling satisfies
\begin{align}
\E_{\mathrm{res}}\left[L_n^\phi\sum_{j=1}^Jw_{n,j}^{F,\theta}\right]
&=JL_n^\phi\sum_{m=1}^Jp_{n,m}w_{n,m}^{FC,\theta} \notag\\*
&=\sum_{m=1}^Jw_{n,m}^{P,\theta}g_{n,m}^\theta.
\end{align}
The same conditional triangular-array argument as in
Lemma~\ref{lem:change-measure-proper-weights} gives consistency of this
resampling approximation under its boundedness conditions. Dividing by
the unchanged prediction-weight sum gives (\ref{L2}).
Using Lemma \ref{lem:lik-proper-weight}, we obtain a likelihood estimate,
\begin{equation}
L^{A,\theta} = \prod_{n=1}^N \left( L^\phi_n \, \frac{\sum_{j=1}^J w^{F,\theta}_{n,j}}{\sum_{j=1}^J w^{P,\theta}_{n,j}}\right).
\end{equation}
Since $w^{F,\theta}_{n,j}=w^{P,\theta}_{n+1,j}$, this is a telescoping product. The remaining terms are
$\sum_{j=1}^J w^{F,\theta}_{0,j} = J$ on the denominator and $\sum_{j=1}^J w^{F,\theta}_{N,j}$ on the numerator.
This derives the MOP-1 likelihood estimates.

$L^{B,\theta}$ should generally be preferred in practice, since there is no reason to include the extra variability from resampling when calculating the conditional log likelihood, but it lacks the nice telescoping product that lets us derive exact expressions for the gradient in Theorem~\thmMopFunctionalForms.

\section{Proof of Theorem~{\thmMopGradConsistency} via Consistency of Off-Parameter Resampled Gradient Estimates}
\label{appendix:consistency}

We now show that, under sufficient local regularity conditions, one can interchange differentiation and the particle limit. We apply Arzel\`a--Ascoli, working on a compact neighborhood of the evaluation parameter. Consistency is used for the estimates themselves: differentiating a weighted particle estimate also differentiates its particles, weights, and normalizer, so boundedness alone does not make its derivative an empirical average to which Theorem \ref{thm:off-parameter-targeting} applies.

\begin{thm}[Off-Parameter Particle Filters Yield Strongly Consistent Estimates of Derivatives of Functionals]
\label{thm:consistent-derivatives}
Let $K$ be a closed ball of positive radius centered at $\phi$ and contained in the interior of $\Theta\subset\mathbb R^p$. Let $F_J(\theta)$ be scalar particle estimates, such as $\eta_n^J(h_\theta)$, defined on a common probability space. On one probability-one event, suppose every $F_J$ has two continuous derivatives on a neighborhood of $K$, with
\[
\sup_{J\geq1}\sup_{\theta\in K}\|\nabla_\theta F_J(\theta)\|_2\leq G^*,
\qquad
\sup_{J\geq1}\sup_{\theta\in K}\|\nabla_\theta^2 F_J(\theta)\|_{\mathrm{op}}\leq H^*
\]
for finite constants $G^*,H^*$. Suppose $F_J(\theta)\to h_\theta^*$ almost surely for each fixed $\theta\in K$, where $h^*$ is deterministic. Then $h^*$ is continuously differentiable in the interior of $K$, and almost surely
\[
\nabla_\theta F_J(\theta)\longrightarrow\nabla_\theta h_\theta^*
\]
uniformly on every compact subset of that interior, in particular at $\theta=\phi$.
\end{thm}
\begin{proof}[Proof.]
For each fixed pair $x,y\in K$, pointwise almost-sure convergence and the gradient bound give $|h_y^*-h_x^*|\leq G^*\|y-x\|_2$. Thus the deterministic function $h^*$ is Lipschitz, and in particular continuous, on $K$.

Choose a countable dense subset $D$ of $K$. Intersect the probability-one event of the derivative bounds with the events of convergence at each $\theta\in D$, and fix $\omega$ in this intersection. Suppress $\omega$ below. By the gradient bound, the functions $F_J$ are uniformly Lipschitz on the convex set $K$. Convergence on $D$, compactness of $K$, and continuity of $h^*$ therefore imply $F_J\to h^*$ uniformly on $K$: for any $\epsilon>0$, choose a sufficiently fine finite net in $D$, use convergence at its finitely many points, and use the common Lipschitz bound and uniform continuity of $h^*$ between those points.

The gradients are uniformly bounded and, by the Hessian bound, uniformly equicontinuous. Explicitly, for $\theta,\theta'\in K$,
\begin{equation}
\|\nabla_\theta F_J(\theta)-\nabla_\theta F_J(\theta')\|_2
\leq H^*\|\theta-\theta'\|_2.
\end{equation}
Thus, by Arzel\`a--Ascoli, every subsequence of gradients has a further subsequence converging uniformly on $K$ to a continuous function $g$. We claim that every such limit is $\nabla_\theta h^*$ in the interior of $K$. For $x,y\in K$, the fundamental theorem of calculus gives
\begin{equation}
F_J(y)-F_J(x)
=\int_0^1\nabla_\theta F_J\bigl(x+t(y-x)\bigr)^{\mathsf T}(y-x)\,dt.
\end{equation}
Passing to the limit along the convergent subsequence yields
\begin{equation}
h_y^*-h_x^*
=\int_0^1 g\bigl(x+t(y-x)\bigr)^{\mathsf T}(y-x)\,dt.
\end{equation}
Taking $y=x+te_i$ for each coordinate direction and letting $t\to0$ shows that $\partial_{\theta_i}h_x^*=g_i(x)$ for every interior $x$. Since $g$ is continuous, $h^*$ is continuously differentiable there. Every subsequential gradient limit is therefore the same on the interior, and hence on $K$ by continuity. Relative compactness and uniqueness of this limit imply convergence of the full gradient sequence. In particular, almost surely, for interior $\theta$,
\begin{equation}
\lim_{J\to\infty}\nabla_\theta F_J(\theta)
=\nabla_\theta\lim_{J\to\infty}F_J(\theta)
=\nabla_\theta h_\theta^*,
\end{equation}
with convergence uniform on every compact subset of the interior, as required.
\end{proof}
For $F_J(\theta)=\eta_n^J(h_\theta)$, the conclusion is convergence to $\nabla_\theta\eta_n(h_\theta)$. No identity with $\eta_n(\nabla_\theta h_\theta)$ is asserted, since the target distribution can itself depend on $\theta$.

Theorem~{\thmMopGradConsistency} follows by taking $F_J(\theta)=\hat\ell_J^1(\theta)$ and $h_\theta^*=\ell(\theta)$. Keep the reference parameter $\phi$ fixed throughout $K$. For $\alpha=1$, the off-parameter consistency result and Lemma~\ref{lem:lik-proper-weight} give $\hat\lik_J^1(\theta)\to\lik(\theta)>0$ almost surely for each fixed $\theta\in K$, not only at $\theta=\phi$. The continuous mapping theorem gives the corresponding convergence of the log likelihoods. The neighborhood form of Assumption~\refAssumpLocalBoundedDerivative\ supplies their uniform gradient and Hessian bounds, and Assumption~\refAssumpDifferentiableSimulator\ supplies their pathwise smoothness with the simulator draws and reference-filter resampling indices held fixed. Applying the theorem to these log likelihoods establishes score consistency at $\phi$ directly.

\section{Variance Analysis for DMOP-$0$}
\label{appendix:biasvar}

We first study DMOP-$0$ to introduce the particle variance and covariance estimates used in Section~\ref{appendix:sharper-biasvar}.
Under (A6-J), the measurement gradient is a bounded function of the current joint state, allowing a direct particle-cloud mixing argument.
Under (A6), its dependence on simulator history requires the additional decomposition developed in Section~\ref{appendix:sharper-biasvar}.
We distinguish these two arguments below; the general discounting and MSE bounds are established in Section~\ref{appendix:sharper-biasvar}.

First, we relate cloud contraction to the strong-mixing coefficient used
in the covariance bounds below.
\begin{lem}
\label{lem:dobrushin-implies-alpha-mixing}
Let $\mathcal F_m=\sigma(x_{0,1:J}^P,\ldots,x_{m,1:J}^P)$ and
$\mathcal G_n=\sigma(x_{n,1:J}^P,\ldots,x_{N,1:J}^P)$ be the past and
future sigma-fields of the Markov particle cloud. For $1\leq\ell\leq N$, define
\begin{equation}
a(\ell)=\sup_{0\leq m\leq N-\ell}
\sup_{A\in\mathcal F_m,\,B\in\mathcal G_{m+\ell}}
\big|\Pr(A\cap B)-\Pr(A)\Pr(B)\big|.
\end{equation}
Then
\begin{equation}
a(\ell)\leq\sup_{0\leq m\leq N-\ell}
\beta_{\mathrm{TV}}(\textbf M_{m,m+\ell}).
\end{equation}
The same statement applies to the joint state--derivative cloud under
(A6-J), using its transition kernel.
\end{lem}
\begin{proof}[Proof.]
Fix $A\in\mathcal F_m$ and $B\in\mathcal G_n$, where $n=m+\ell$.
By the Markov property, with
$h(x)=\Pr(B\mid x_{n,1:J}^P=x)\in[0,1]$,
\begin{equation}
\Pr(B\mid\mathcal F_m)
=(\textbf M_{m,n}h)(x_{m,1:J}^P).
\end{equation}
The expectation of this conditional probability is $\Pr(B)$, so
\begin{align}
\big|\Pr(A\cap B)-\Pr(A)\Pr(B)\big|
&=\left|\E\left[\mathbf1_A
\left\{(\textbf M_{m,n}h)(x_{m,1:J}^P)-\Pr(B)\right\}\right]\right|
\notag\\*
&\leq\operatorname{osc}(\textbf M_{m,n}h)
\leq\beta_{\mathrm{TV}}(\textbf M_{m,n})\operatorname{osc}(h)
\notag\\*
&\leq\beta_{\mathrm{TV}}(\textbf M_{m,n}).
\end{align}
Taking the suprema proves the claim. The coefficient is defined through
sigma-fields, so the argument does not require joint Lebesgue densities.
\end{proof}
The same argument bounds dependence between past and future windows of
simulation and resampling updates. A past window ending after resampling
at time $b$ is absorbed into the next prediction cloud at $b+1$; condition
there and at the prediction cloud starting the future window, and use
the intervening cloud kernel. Fresh update variables are conditionally independent of earlier
updates given the incoming cloud. Under (A6), apply this to the physical
cloud and the derivative-age windows in
Equation~(\ref{eq:lag-truncation-covariance-mixing-s13}); under (A6-J),
use the joint cloud. The number of intervening transitions determines
the lag in $a(\ell)$.

Assumptions~{\refAssumpBoundedProcess} and~{\refAssumpBoundedMeasurement} allow us to use the particle-filter forgetting results of \citet{karjalainen23}.
To control the measurement gradients, we additionally require either (A6) or (A6-J), stated in Section~\ref{appendix:sharper-biasvar}.

\begin{thm}[Variance of the DMOP-$0$ gradient estimate]
\label{thm:dmop-zero-variance}
Suppose the observations are fixed, $J\geq2$, $\theta=\phi$, Assumptions
{\refAssumpBoundedProcess}, {\refAssumpBoundedMeasurement}, and
{\refAssumpDifferentiableSimulator} hold, and multinomial resampling is used.
Then
\begin{equation}
\Var\{\nabla_\theta\hat\ell^0(\theta)\}
\lesssim \frac{NpG'(\theta)^2}{J}
\begin{cases}
(1+\log J)^2, & \text{under (A6)},\\
1+\log J, & \text{under (A6-J)}.
\end{cases}
\end{equation}
\end{thm}
\begin{proof}[Proof.]
First assume (A6-J). In this calculation, a current particle means the
pair $(X_n,\partial_\theta X_n)$, and the particle-cloud transition is
the one induced by its joint kernel.
When $\alpha=0$, we have
\begin{align}
    \hat{\lik}^0(\theta) := \prod_{n=1}^N L_n^{A, \theta, 0} &= \prod_{n=1}^N\left(L_n^\phi \cdot \frac{\sum_{j=1}^J w_{n,j}^{F,\theta,0}}{\sum_{j=1}^J w_{n,j}^{P,\theta,0}}\right)
    \\&
    = \prod_{n=1}^N\left(L_n^\phi \cdot \frac{1}{J}\sum_{j=1}^J s_{n,j}\right) = \prod_{n=1}^N\left(L_n^\phi \cdot \frac{1}{J}\sum_{j=1}^J \frac{g_{n,k_{n,j}}^\theta}{g_{n,k_{n,j}}^\phi}\right).
\end{align}
Similarly to the proof of Lemma \ref{lem:mop-0-formula}, we define 
\begin{equation}
s_{n,j}=\frac{g_{n,k_{n,j}}^\theta}{g_{n,k_{n,j}}^\phi},
\end{equation}
and from the exact same arguments conclude that its gradient when $\theta=\phi$ is therefore 
\begin{align}
    \nabla_\theta \hat{\ell}^0(\theta) &:= \sum_{n=1}^N \nabla_\theta \log\left(L_n^\phi \frac{1}{J} \sum_{j=1}^J s_{n,j}\right) \\
    &= \sum_{n=1}^N \frac{\nabla_\theta \left(L_n^\phi \frac{1}{J} \sum_{j=1}^J s_{n,j}\right)}{\left(L_n^\phi \frac{1}{J} \sum_{j=1}^J s_{n,j}\right)} \\
    &= \sum_{n=1}^N \frac{\sum_{j=1}^J \nabla_\theta s_{n,j}}{\sum_{j=1}^J s_{n,j}} \\
    &= \sum_{n=1}^N \frac{1}{J} \sum_{j=1}^J \frac{\nabla_\theta f_{Y_n|X_{n}}\big(y_n^*|x_{n,j}^{F, \theta}; \theta\big)}{f_{Y_n|X_{n}}\big(y_n^*|x_{n,j}^{F, \phi}; \phi\big)} \\
    &= \frac{1}{J} \sum_{n=1}^N \sum_{j=1}^J \nabla_\theta \log\Big(f_{Y_n|X_{n}}\big(y_n^*|x_{n,j}^{F, \theta}; \theta\big)\Big),
\end{align}
where we use the derivative of the logarithm, observe that $\sum_{j=1}^J s_{n,j} = J$ when $\theta=\phi$, and use the derivative of the logarithm where $\theta=\phi$ again. We do this to establish that the gradient of the log-likelihood estimate is given by the sum of terms over all $N$ and $J$:
\begin{equation}
\nabla_\theta \hat{\ell}^0(\theta) = \frac{1}{J} \sum_{n=1}^N \sum_{j=1}^J \nabla_\theta \log\Big(f_{Y_n|X_{n}}\big(y_n^*|x_{n,j}^{F, \theta}; \theta\big)\Big).
\end{equation}
Therefore, for a given $\theta_i$ in the parameter vector $\theta$,
\begin{align}
    \Var\left(\frac{\partial}{\partial_{\theta_i}} \hat\ell^0(\theta)\right) &= \frac{1}{J^2}\Var\left(\sum_{n=1}^N\sum_{j=1}^{J}\frac{\partial}{\partial_{\theta_i}} \log\left(f_{Y_n|X_{n}}\big(y_n^*|x_{n,j}^{F, \theta}; \theta\big)\right)\right) \\
    & \hspace{-20mm}= \frac{1}{J^2}\sum_{n=1}^N\Var\left(\sum_{j=1}^{J}\frac{\partial}{\partial_{\theta_i}} \log\left(f_{Y_n|X_{n}}(y_n^*|x_{n,j}^{F, \theta}; \theta)\right)\right)
\nonumber \\
    &\hspace{-15mm}+ 2\sum_{m<n}\Cov\left(\frac{1}{J}\sum_{j=1}^{J}\frac{\partial}{\partial_{\theta_i}} \log\left(f_{Y_m|X_{m}}(y_m^*|x_{m,j}^{F, \theta}; \theta)\right), \frac{1}{J}\sum_{j=1}^{J}\frac{\partial}{\partial_{\theta_i}} \log\left(f_{Y_n|X_{n}}(y_n^*|x_{n,j}^{F, \theta}; \theta)\right)\right)
\nonumber \\
    &\hspace{-20mm}= \sum_{n=1}^N \Var\left(\frac{\partial}{\partial_{\theta_i}}\hat\ell_n^0(\theta)\right) + 2\sum_{m<n} \Cov\left(\frac{\partial}{\partial_{\theta_i}}\hat\ell_m^0(\theta),\frac{\partial}{\partial_{\theta_i}}\hat\ell_n^0(\theta)\right).
\end{align}
Here (A6-J) and {\refAssumpBoundedMeasurement} ensure strong mixing
of the joint model. We know from Theorem 3.1 of \cite{karjalainen23} that when $\textbf{M}_{n,n+k}$ is the $k$-step Markov operator from timestep $n$ and $\beta_{\text{TV}}(M) = \sup _{x, y \in E}\|M(x, \cdot)-M(y, \cdot)\|_{\mathrm{TV}}=\sup _{\mu, \nu \in \mathcal{P}, \mu \neq \nu} \frac{\|\mu M-\nu M\|_{\mathrm{TV}}}{\|\mu-\nu\|_{\mathrm{TV}}}$ is the Dobrushin contraction coefficient of a Markov operator,
\begin{equation}
\beta_{\text{TV}}(\textbf{M}_{n,n+k}) \leq (1-\epsilon)^{\floor{k/(c\log(J))}},
\end{equation}
i.e. the mixing time of this joint particle filter is $O(\log(J))$,
where $\epsilon$ and $c$ depend on $\bar M_D,\underbar M_D,\bar G,\underbar G$.
The physical particle cloud satisfies a bound of the same form under
(A2), with its own constants; that version is used under (A6) below.

Under (A6-J), apply Lemma~2.3 of \citet{karjalainen23} to the joint
prediction cloud. Conditional resampling and the Bayes-ratio calculation
give the centered $L^4$ bound below. Theorem~3.1 and Lemma~\ref{lem:dobrushin-implies-alpha-mixing} bound the
strong-mixing coefficient $a(k)$ of the score sequence. The shift by one
accounts for the endpoint resampling. At $\alpha=0$ these results give
\begin{align}
\left\|
\frac{\partial}{\partial_{\theta_i}}\hat\ell_n^0(\theta)
-\E\left[
\frac{\partial}{\partial_{\theta_i}}\hat\ell_n^0(\theta)
\right]
\right\|_{L^4}
&\leq
\frac{CG'(\theta)}{\sqrt J},                                \notag\\
a(k)
&\leq
(1-\epsilon)^{\floor{(k-1)/[c\{1+\log J\}]}}.
\label{eq:dmop-zero-centered-l4-and-mixing}
\end{align}
Davydov's inequality must be applied to the centered variables. Since
covariance is unchanged by centering, Equation~
(\ref{eq:dmop-zero-centered-l4-and-mixing}) gives
\begin{equation}
\Cov\left(\frac{\partial}{\partial_{\theta_i}}\hat\ell_m^0(\theta), \frac{\partial}{\partial_{\theta_i}}\hat\ell_n^0(\theta)\right) \leq (1-\epsilon)^{\frac{1}{2}\floor{\frac{|n-m|-1}{c\{1+\log(J)\}}}}\frac{G'(\theta)^2}{J}.
\end{equation}
The geometric sum is bounded by $C(1+\log J)$. Putting it all together, we see that
\begin{align}
    \Var\left(\frac{\partial}{\partial_{\theta_i}} \hat\ell^0(\theta)\right) &= \sum_{n=1}^N \Var\left(\frac{\partial}{\partial_{\theta_i}}\hat\ell_n^0(\theta)\right) + 2\sum_{m<n} \Cov\left(\frac{\partial}{\partial_{\theta_i}}\hat\ell_m^0(\theta),\frac{\partial}{\partial_{\theta_i}}\hat\ell_n^0(\theta)\right)\\
    &\leq \frac{NG'(\theta)^2}{J} + 2\sum_{n=1}^{N-1} (N-n)(1-\epsilon)^{\frac{1}{2}\floor{\frac{n-1}{c\{1+\log(J)\}}}}\frac{G'(\theta)^2}{J} \\
    &\leq \frac{G'(\theta)^2}{J} \left(N + 2\sum_{n=1}^{N-1} (N-n) (1-\epsilon)^{\frac{1}{2}\floor{\frac{n-1}{c\{1+\log(J)\}}}}\right) \\
    &\lesssim \frac{G'(\theta)^2N(1+\log J)}{J}.
\end{align}
It then follows that as $\theta \in \Theta \subseteq \R^p$, 
\begin{equation}
    \Var\big(\nabla_{\theta} \hat\ell^0(\theta)\big) \lesssim \frac{Np \, G'(\theta)^2(1+\log J)}{J}.
\end{equation}
Under (A6), apply the age decomposition and covariance calculation in
Steps~1a--1d of Lemma~\ref{lem:dmop-alpha-variance} with $k=0$.
In particular, Equation~(\ref{eq:fixed-lag-variance-s13}) gives
\[
\Var\{\nabla_\theta\hat\ell^0(\theta)\}
=\Var\{\nabla_\theta\hat s_0^1(\theta)\}
\lesssim NpG'(\theta)^2(1+\log J)^2/J.
\]
That argument uses the physical-cloud mixing bound above and bounded
finite-history inputs, not mixing of the full gradient cloud under (A6).
\end{proof}
The factor of $N$ comes from using the unnormalized gradient.
Dividing by $\sqrt N$ gives a variance bound independent of the horizon.
Dividing by $N$ gives a centered $L^2$ error bound of order
$(1+\log J)/\sqrt{NJ}$.

The preceding calculation gives the order
$NpG'(\theta)^2(1+\log J)^2/J$ for the DMOP-$0$ baseline.  It is not asserted
to be a lower bound for every model or every value of $\alpha$.

\section{Proof of Theorem~{\thmMopBiasvar} (Bias-Variance Analysis)}
\label{appendix:sharper-biasvar}

Write $\nabla_\theta\ell_n^\alpha(\theta)$ for the $n$th conditional score
of the population recursion. At $\theta=\phi$ and $\alpha=1$, it equals
$\nabla_\theta\ell_n(\theta)$. Along each selected lineage,
\begin{equation}
\left.\nabla_\theta\log w_{n,j}^{F,\theta}\right|_{\theta=\phi}
=
\alpha
\left.\nabla_\theta\log w_{n-1,k_j}^{F,\theta}\right|_{\theta=\phi}
+
\left.\nabla_\theta\log g_{n,k_j}^{\theta}
\right|_{\theta=\phi}.
\label{eq:location-weight-gradient-pair-recursion}
\end{equation}
The last term is the measurement gradient $h_n$ defined in
Equation~(\ref{eq:actual-path-score}). This follows by differentiating
\begin{align}
\log w_{n,j}^{P,\theta}
&=\alpha\log w_{n-1,j}^{F,\theta},                         \notag\\
\log w_{n,j}^{F,\theta}
&=\alpha\log w_{n-1,k_j}^{F,\theta}
+\log g_{n,k_j}^{\theta}-\log g_{n,k_j}^{\phi}.
\label{eq:log-weight-recursion}
\end{align}
at $\theta=\phi$. The observations are fixed throughout this section;
expectations and variances are over the particle-filter randomness.

We next investigate conditions (A6) and (A6-J), repeating them here to assist the discussion.

\paragraph*{(A6) Measurement-sensitivity decay.}
\label{assump:measurement-sensitivity}
The simulator is $C^1$ in its incoming state. For a continuation
$X_n^\theta(x_t,\omega_{t+1:n})$ from $x_t$, with $X_t^\theta(x_t)=x_t$,
\begin{equation}
\left\|
D_{x_t}\log f_{Y_n\mid X_n}
\bigl(y_n^*\mid X_n^\theta(x_t,\omega_{t+1:n});\theta\bigr)
\right\|_2
\leq C\kappa^{n-t},\qquad C<\infty,\quad 0<\kappa<1.
\label{eq:measurement-sensitivity-assumption}
\end{equation}
The bound is uniform in $t\leq n$, states, draws, and $\theta$ near $\phi$.
The parameter derivatives of the initial state, simulator, and log measurement
density are bounded, holding incoming states fixed.

\paragraph*{(A6-J) Joint-transition alternative.}
\label{assump:joint-gradient-density}
The log measurement density is $C^1$ in $(x,\theta)$.
The joint process $(X_n^\theta,D_n^\theta)$,
$D_n^\theta=\partial_\theta X_n^\theta$, has
transition densities satisfying
\begin{equation}
0<\underbar M_D
\leq f_{(X_n,D_n)\mid(X_{n-1},D_{n-1})}
(x',v'\mid x,v;\theta)
\leq\bar M_D<\infty
\label{eq:joint-gradient-density-assumption}
\end{equation}
relative to a common finite measure on the invariant joint state space,
uniformly in time, admissible pairs, and $\theta$ near $\phi$.

The particle bounds assume either (A6) or (A6-J); the consistency results
and population Lemmas~\ref{lem:derivative-measure-forgetting}
and~\ref{lem:dmop-population-discount-bias} retain (A1)--(A5).
Neither additional condition requires density evaluation by the algorithm.
Their constants enter the particle bounds, while $\varrho$ retains its
physical-model definition. 

\paragraph*{Sufficient condition for (A6).}
Suppose the parameter derivatives in (A6) are bounded, the spatial log
measurement derivatives are bounded by $L_g$, and the one-step simulator
Jacobians are bounded by $L$. If, for fixed $b\geq1$ and $q\in(0,1)$,
$\|\partial X_{t+b}^\theta/\partial X_t^\theta\|_{\rm op}\leq q$
uniformly with the draws fixed, then writing $n-t=mb+v$, $0\leq v<b$,
gives
\begin{equation}
\left\|D_{x_t}\log f_{Y_n\mid X_n}(y_n^*\mid X_n^\theta;\theta)\right\|_2
\leq L_g L^v q^m
\leq C\kappa^{n-t},\qquad \kappa=q^{1/b}\in(0,1).
\label{eq:measurement-sensitivity-block-criterion}
\end{equation}
Here $C=L_g\max_{0\leq v<b}(L/\kappa)^v$; a zero block Jacobian permits
any $q\in(0,1)$.

\paragraph*{Example: physical mixing without (A6).}
Consider the annulus $\mathcal X=\{x\in\mathbb R^2:1\leq\|x\|_2\leq2\}$.
Let $R_u$ denote planar rotation by $u$, and let independent simulator
draws satisfy $U_n\sim\operatorname{Unif}[0,2\pi]$ and
$V_n\sim\operatorname{Unif}[1,2]$. Set
\begin{equation}
X_0^\theta=R_\theta x_0,\qquad
X_n^\theta=V_n R_{U_n}\frac{X_{n-1}^\theta}{\|X_{n-1}^\theta\|_2},
\qquad
\Pr(Y_n=1\mid X_n=x)=\frac12+\frac{x_1}{4\|x\|_2}.
\label{eq:gradient-stability-rotation-example}
\end{equation}
The physical transition forgets the incoming state in one step: its
Lebesgue density on $\mathcal X$ is
\begin{equation}
f_{X_n\mid X_{n-1}}(x'\mid x;\theta)=\frac{1}{2\pi\|x'\|_2},
\qquad \frac1{4\pi}\leq f_{X_n\mid X_{n-1}}\leq\frac1{2\pi}.
\end{equation}
Both measurement probabilities lie in $[1/4,3/4]$.
For fixed draws, $X_n^\theta=R_\theta X_n^0$, so all parameter
derivatives exist and are bounded; in particular the first two log
measurement derivatives have magnitudes at most $1$ and $2$.
The likelihood is smooth, and the MOP derivatives are bounded for each
fixed horizon. Thus (A1)--(A5) hold.

Nevertheless, fix the observations $y_n^*=1$. For $n>t$, the composition is
\begin{equation}
X_n^\theta(x,\omega_{t+1:n})
=V_nR_{U_{t+1}+\cdots+U_n}\frac{x}{\|x\|_2}.
\end{equation}
Its final angle is uniform, and a fixed-random-number change in the
incoming angle is preserved. Differentiating the displayed measurement
probability therefore gives
\begin{equation}
\E\left\|D_x\log f_{Y_n\mid X_n}
\bigl(1\mid X_n^\theta(x,\omega_{t+1:n});\theta\bigr)\right\|_2^2
=\frac{1}{2\pi\|x\|_2^2}
\int_0^{2\pi}\frac{\sin^2 u}{(2+\cos u)^2}\,du>0.
\label{eq:gradient-stability-rotation-nondecay}
\end{equation}
The right side is independent of $n-t$, so (A6) fails despite (A1)--(A5).
Here $D_n^\theta=R_{\pi/2}X_n^\theta$, and the joint transition has the
physical density relative to $dx'\,\delta_{R_{\pi/2}x'}(dv')$.
Thus (A6-J) holds and the particle bounds still apply. In particular,
(A6-J) does not imply (A6).

\paragraph*{Differentiation and notation.}
By (A5), differentiation in $\theta$ holds the simulator draws and
reference-filter resampling indices fixed and includes the dependence of
$X_n^\theta$ on $\theta$. Assumption~{\refAssumpBoundedMeasurement} gives,
uniformly in $n,j$ and $\theta$ near $\phi$,
\begin{equation}
\left\|
\nabla_\theta\log g_{n,j}^{A,F,\theta}
\right\|_\infty
\leq G'(\theta).
\label{eq:bounded-realized-measurement-score}
\end{equation}
\begin{samepage}
The product rule and Equation~(\ref{eq:bounded-realized-measurement-score})
give
\begin{align}
&\left\|
\nabla_\theta\prod_{s=1}^t G_s(X_s)
\right\|_\infty                                                   \notag\\*
&\quad=
\left\|
\prod_{s=1}^t G_s(X_s)
\sum_{s=1}^t\nabla_\theta\log G_s(X_s)
\right\|_\infty                                                   \notag\\*
&\leq
\prod_{s=1}^t G_s(X_s)
\sum_{s=1}^t
\left\|\nabla_\theta\log G_s(X_s)\right\|_\infty                  \notag\\*
&\leq
t\bar G^tG'(\theta).
\label{eq:finite-horizon-product-derivative-bound}
\end{align}
\end{samepage}
For fixed $t$, the mean-value theorem bounds the coordinatewise difference
quotients by the integrable constant $t\bar G^tG'(\theta)$.
Dominated convergence therefore gives
\begin{equation}
\nabla_\theta
\E_U\left[\prod_{s=1}^tG_s(X_s)\right]
=
\E_U\left[
\nabla_\theta\prod_{s=1}^tG_s(X_s)
\right].
\label{eq:finite-horizon-differentiation-under-expectation}
\end{equation}
At $\theta=\phi$, let $\rho_t$ be the posterior law of the simulator
history through time $t$, including its base variables. Its
$X_{0:t}$-marginal is the posterior path law, and
\begin{equation}
\rho_t(h)
=
\frac{
\E_U\left[
h
\prod_{s=1}^t G_s(X_s)
\right]
}
{
\E_U\left[\prod_{s=1}^t G_s(X_s)\right]
}.
\label{eq:canonical-extended-posterior-path-law}
\end{equation}
Here $\E_U$ integrates over the base simulator variables before
measurement weighting. Expectations in particle-error calculations are
instead over the particle-filter law, including resampling.
Future simulator draws are independent of the past.

Iterating the weight recursion at \(\theta=\phi\) gives
\begin{equation}
\left.\nabla_\theta\log w_{n,j}^{F,\theta}\right|_{\theta=\phi}
=
\sum_{r=0}^{n-1}\alpha^r
\left.\nabla_\theta\log g_{n-r,j}^{A,F,\theta}\right|_{\theta=\phi}.
\label{eq:unrolled-filtered-weight}
\end{equation}
Write
\begin{equation}
h_k
=
\left.\nabla_\theta\log g_k^{A,F,\theta}\right|_{\theta=\phi}
\label{eq:actual-path-score}
\end{equation}
for the measurement gradient on an ancestral path. Following
\citet{delmoral18}, let
\begin{equation}
m(\xi_t)
=
\frac1J\sum_{j=1}^{J}\delta_{\xi_{t,j}},
\qquad
m(\xi_t)(h_k)
=
\frac1J\sum_{j=1}^{J}h_k(\xi_{t,j}),
\quad 1\leq k\leq t,
\label{eq:explicit-path-empirical-measure}
\end{equation}
where \(\xi_{t,j}\) is the simulator history along particle \(j\)'s ancestry
through time \(t\), including the base variables needed to evaluate
\(h_k\) and its pathwise derivative contributions.
Thus \(m(\xi_t)(h_k)\) averages the time-$k$ gradients over these paths.

At \(\theta=\phi\), all weights equal one. Differentiating the two sample
averages in the definition of \(L_n^{A,\theta,\alpha}\) gives
\begin{align}
\nabla_\theta\hat\ell_n^\alpha(\theta)
&=
\frac1J\sum_{j=1}^J
\left.\nabla_\theta\log w_{n,j}^{F,\theta}\right|_{\theta=\phi}
-
\frac1J\sum_{j=1}^J
\left.\nabla_\theta\log w_{n,j}^{P,\theta}\right|_{\theta=\phi} \notag\\
&=
\frac1J\sum_{j=1}^J
\left.\nabla_\theta\log w_{n,j}^{F,\theta}\right|_{\theta=\phi}
-
\frac{\alpha}{J}\sum_{j=1}^J
\left.\nabla_\theta\log w_{n-1,j}^{F,\theta}\right|_{\theta=\phi},
\label{eq:conditional-score-weight-difference}
\end{align}
where the second equality uses
\(w_{n,j}^{P,\theta}=(w_{n-1,j}^{F,\theta})^\alpha\). In the
MOP-$\alpha$ recursion, the discounted-weight accumulation and process
propagation preserve the index
\(j\) and occur before the time-\(n\) resampling step. The resampling index
\(k_j\) enters only afterward, through
\begin{equation}
w_{n,j}^{F,\theta}
=
w_{n,k_j}^{P,\theta}
\frac{g_{n,k_j}^{\theta}}{g_{n,k_j}^{\phi}}.
\label{eq:filtered-weight-current-resampling-index}
\end{equation}
Consequently, the second average in
Equation~(\ref{eq:conditional-score-weight-difference}) is pathwise
\begin{equation}
\frac1J\sum_{j=1}^J
\left.\nabla_\theta\log w_{n,j}^{P,\theta}\right|_{\theta=\phi}
=
\frac{\alpha}{J}\sum_{j=1}^J
\left.\nabla_\theta\log w_{n-1,j}^{F,\theta}\right|_{\theta=\phi},
\label{eq:prediction-weight-index-identity}
\end{equation}
which is the empirical average over the time-\((n-1)\) filtered genealogy.
Substituting
Equation~(\ref{eq:unrolled-filtered-weight}) into
Equation~(\ref{eq:conditional-score-weight-difference}) and changing the
index in the second sum give
\begin{align}
\nabla_\theta\hat\ell_n^\alpha(\theta)
&=
\sum_{r=0}^{n-1}\alpha^r m(\xi_n)(h_{n-r})
-\sum_{r=1}^{n-1}\alpha^r m(\xi_{n-1})(h_{n-r}) \notag\\
&=
m(\xi_n)(h_n)
+\sum_{r=1}^{n-1}\alpha^r
\left\{m(\xi_n)(h_{n-r})-m(\xi_{n-1})(h_{n-r})\right\}.
\label{eq:conditional-score-path-decomposition}
\end{align}
Consequently, define
\begin{align}
\Delta_{n,0}^J
&=m(\xi_n)(h_n),
&
\Delta_{n,0}
&=\rho_n(h_n), \notag\\
\Delta_{n,r}^J
&=m(\xi_n)(h_{n-r})-m(\xi_{n-1})(h_{n-r}),
&
\Delta_{n,r}
&=\rho_n(h_{n-r})-\rho_{n-1}(h_{n-r}),
\quad 1\leq r<n.
\label{eq:lag-coefficient-definitions}
\end{align}
For the population coefficients, Equations~
(\ref{eq:finite-horizon-differentiation-under-expectation}) and
(\ref{eq:actual-path-score}) give, at $\theta=\phi$,
\begin{align}
&\nabla_\theta
\log\E_U\left[\prod_{s=1}^nG_s(X_s)\right]                       \notag\\
&\quad=
\frac{
\E_U\left[
\prod_{s=1}^nG_s(X_s)\sum_{s=1}^nh_s
\right]
}
{
\E_U\left[\prod_{s=1}^nG_s(X_s)\right]
}                                                               \notag\\
&\quad=
\sum_{s=1}^n\rho_n(h_s).
\label{eq:full-population-score-path-average}
\end{align}
Subtracting the same identity at time $n-1$ gives
\begin{align}
\nabla_\theta\ell_n(\theta)
&=
\sum_{s=1}^n\rho_n(h_s)-\sum_{s=1}^{n-1}\rho_{n-1}(h_s)          \notag\\
&=
\rho_n(h_n)
+\sum_{r=1}^{n-1}
\left\{\rho_n(h_{n-r})-\rho_{n-1}(h_{n-r})\right\}               \notag\\
&=
\sum_{r=0}^{n-1}\Delta_{n,r}.
\label{eq:exact-population-conditional-score-path-decomposition}
\end{align}
Finally, differentiating the two population weighted averages in the
deterministic MOP-$\alpha$ recursion and using
Equation~(\ref{eq:finite-horizon-differentiation-under-expectation}) gives
\begin{align}
\nabla_\theta\ell_n^\alpha(\theta)
&=
\sum_{r=0}^{n-1}\alpha^r\rho_n(h_{n-r})
-\sum_{r=1}^{n-1}\alpha^r\rho_{n-1}(h_{n-r})                   \notag\\
&=
\rho_n(h_n)
+\sum_{r=1}^{n-1}\alpha^r
\left\{\rho_n(h_{n-r})-\rho_{n-1}(h_{n-r})\right\}             \notag\\
&=
\sum_{r=0}^{n-1}\alpha^r\Delta_{n,r}.
\label{eq:population-conditional-score-path-decomposition}
\end{align}

\begin{lem}[Lagwise first-derivative forgetting]
\label{lem:derivative-measure-forgetting}
Suppose $\theta=\phi$. Under Assumptions {\refAssumpBoundedProcess},
{\refAssumpBoundedMeasurement}, and {\refAssumpDifferentiableSimulator},
there are
constants $C<\infty$ and $\varrho\in(0,1)$ such that, for
$0\leq r<n$,
\begin{equation}
\|\Delta_{n,r}\|_2
\leq C\sqrt p\,G'(\theta)\varrho^r.
\label{eq:derivative-measure-forgetting}
\end{equation}
\end{lem}

\begin{proof}[Proof.]
We use the posterior update and backward-kernel contraction.

\medskip\noindent\emph{Step 1: one-step posterior update.}\quad
First suppose that \(1\leq r<n\), and put \(k=n-r\). Fix a coordinate
\(i\in\{1,\ldots,p\}\), and let \(h_{k,i}\) be the corresponding
coordinate of $h_k$ in Equation~(\ref{eq:actual-path-score}). By
Equation~(\ref{eq:bounded-realized-measurement-score}),
\begin{equation}
|h_{k,i}|
\leq
G'(\theta).
\label{eq:population-forgetting-score-bound}
\end{equation}
Because \(h_{k,i}\) is measurable with respect to the canonical extended
path through time \(k\leq n-1\), Equation~
(\ref{eq:canonical-extended-posterior-path-law}) and conditional expectation
with respect to the new base simulator variable give
\begin{align}
\rho_n(h_{k,i})
&=
\frac{
\E_U\left[
h_{k,i}\prod_{s=1}^{n-1}G_s(X_s)
\E_U\{G_n(X_n)\mid X_{n-1}\}
\right]
}
{
\E_U\left[
\prod_{s=1}^{n-1}G_s(X_s)
\E_U\{G_n(X_n)\mid X_{n-1}\}
\right]
}                                                               \notag\\
&=
\frac{
\rho_{n-1}\{h_{k,i}M_n(G_n)\}
}
{\rho_{n-1}M_n(G_n)},
\label{eq:posterior-path-update-functional}
\end{align}
where the first equality uses the independence of the new base simulator
variable from the path through time \(n-1\). The second equality first uses
the identity below and then divides its numerator and denominator by
\(\E_U\{\prod_{s=1}^{n-1}G_s(X_s)\}\). Using the manuscript's kernel
notation, the required identity is
\begin{equation}
M_n(G_n)(x_{n-1})
=
\int M_n(x_{n-1},dx_n)G_n(x_n).
\label{eq:predictive-potential-functional}
\end{equation}
The bounds on \(G_n\) and the fact that \(M_n\) is a Markov kernel imply
for every \(x_{n-1}\) that
\begin{align}
M_n(G_n)(x_{n-1})
&\geq
\underbar G\int M_n(x_{n-1},dx_n)
=
\underbar G,
\label{eq:predictive-potential-lower}
\\
M_n(G_n)(x_{n-1})
&\leq
\bar G\int M_n(x_{n-1},dx_n)
=
\bar G.
\label{eq:predictive-potential-upper}
\end{align}

\medskip\noindent\emph{Step 2: backward-kernel contraction.}\quad
For \(s\geq1\), the posterior backward kernel from \(X_s\) to \(X_{s-1}\) is
\begin{equation}
B_s(x_s,dx_{s-1})
=
\frac{
\pi_{s-1}(dx_{s-1})
f_{X_s\mid X_{s-1}}(x_s\mid x_{s-1};\theta)
}
{\displaystyle
\int
\pi_{s-1}(du)
f_{X_s\mid X_{s-1}}(x_s\mid u;\theta)}.
\label{eq:posterior-backward-kernel}
\end{equation}
Assumption {\refAssumpBoundedProcess} gives the following minorization for
every measurable set \(A\):
\begin{align}
B_s(x_s,A)
&=
\frac{
\displaystyle
\int_A
\pi_{s-1}(dx_{s-1})
f_{X_s\mid X_{s-1}}(x_s\mid x_{s-1};\theta)
}
{\displaystyle
\int
\pi_{s-1}(du)
f_{X_s\mid X_{s-1}}(x_s\mid u;\theta)}                       \notag\\
&\geq
\frac{
\underbar M\,\pi_{s-1}(A)
}
{\bar M\int\pi_{s-1}(du)}                                   \notag\\
&=
\frac{\underbar M}{\bar M}\pi_{s-1}(A).
\label{eq:posterior-backward-minorization}
\end{align}
Put \(\epsilon=\underbar M/\bar M\). If \(\epsilon<1\), Equation~
(\ref{eq:posterior-backward-minorization}) permits the decomposition
\begin{equation}
B_s(x_s,\cdot)
=
\epsilon\pi_{s-1}(\cdot)
+(1-\epsilon)R_s(x_s,\cdot),
\label{eq:posterior-backward-splitting}
\end{equation}
where
\begin{equation}
R_s(x_s,A)
=
\frac{B_s(x_s,A)-\epsilon\pi_{s-1}(A)}{1-\epsilon}
\end{equation}
is a Markov kernel. Hence, for every bounded real function \(f\) and any
\(x_s,x'_s\),
\begin{align}
B_s(f)(x_s)-B_s(f)(x'_s)
&=
(1-\epsilon)
\{R_s(f)(x_s)-R_s(f)(x'_s)\},                                \notag\\
\left|B_s(f)(x_s)-B_s(f)(x'_s)\right|
&\leq
(1-\epsilon)\operatorname{osc}(f).
\label{eq:posterior-backward-one-step-contraction}
\end{align}
If \(\epsilon=1\), Equation~
(\ref{eq:posterior-backward-minorization}) implies
\(B_s(x_s,\cdot)=\pi_{s-1}(\cdot)\), so the left-hand side is zero.
Taking the supremum over \(x_s,x'_s\) gives
\begin{equation}
\operatorname{osc}\{B_s(f)\}
\leq
(1-\epsilon)\operatorname{osc}(f).
\label{eq:posterior-backward-oscillation-contraction}
\end{equation}

\medskip\noindent\emph{Step 3: project and propagate the time-$k$
contribution.}\quad
For the historical gradient $h_{k,i}$, define
\begin{equation}
\bar h_{k,i}(X_k)
=
\E_{\rho_k}(h_{k,i}\mid X_k).
\label{eq:population-score-current-projection}
\end{equation}
To identify this conditional expectation under later posterior laws,
for $t\geq k$ let
\begin{equation}
L_{k,t}(X_k)
=
\E_U\left\{
\left.\prod_{s=k+1}^{t}G_s(X_s)\right|X_k
\right\},
\qquad L_{k,k}=1.
\label{eq:population-score-future-likelihood}
\end{equation}
For every bounded function $q$ of $X_k$, the Markov property and
Equation~(\ref{eq:canonical-extended-posterior-path-law}) give
\begin{align}
\rho_t\{h_{k,i}q(X_k)\}
&=
\frac{
\E_U\left[
h_{k,i}q(X_k)
\prod_{s=1}^{k}G_s(X_s)L_{k,t}(X_k)
\right]
}{
\E_U\left[
\prod_{s=1}^{k}G_s(X_s)L_{k,t}(X_k)
\right]
}                                                           \notag\\
&=
\frac{
\E_U\left[
\bar h_{k,i}(X_k)q(X_k)
\prod_{s=1}^{k}G_s(X_s)L_{k,t}(X_k)
\right]
}{
\E_U\left[
\prod_{s=1}^{k}G_s(X_s)L_{k,t}(X_k)
\right]
}                                                           \notag\\
&=
\rho_t\{\bar h_{k,i}(X_k)q(X_k)\}.
\label{eq:population-score-projection-invariance}
\end{align}
The second equality is the defining conditional-expectation identity in
Equation~(\ref{eq:population-score-current-projection}), multiplied by the
$X_k$-measurable function $q(X_k)L_{k,t}(X_k)$.  Hence
\begin{equation}
\E_{\rho_t}(h_{k,i}\mid X_k)
=\bar h_{k,i}(X_k),
\qquad t\geq k.
\label{eq:population-score-projection-conditional}
\end{equation}
Repeated conditioning through the physical backward kernels therefore gives
\begin{align}
\E_{\rho_{n-1}}\{h_{k,i}\mid X_{n-1}=x_{n-1}\}
&=
(B_{n-1}B_{n-2}\cdots B_{k+1}\bar h_{k,i})(x_{n-1}).
\label{eq:iterated-backward-conditioning}
\end{align}
There are \(n-1-k=r-1\) backward kernels in this display. Moreover,
Equation~(\ref{eq:population-forgetting-score-bound}) gives
\begin{align}
\|\bar h_{k,i}\|_\infty
&\leq
\|h_{k,i}\|_\infty
\leq G'(\theta),                                            \notag\\
\operatorname{osc}(\bar h_{k,i})
&\leq
2\|\bar h_{k,i}\|_\infty                                  \notag\\
&\leq
2G'(\theta).
\label{eq:conditional-past-score-oscillation}
\end{align}
If \(r=1\), Equation~(\ref{eq:iterated-backward-conditioning}) contains no
backward kernel. If \(r\geq2\), applying Equation~
(\ref{eq:posterior-backward-oscillation-contraction}) successively gives
\begin{align}
&\operatorname{osc}\!\left[
\E_{\rho_{n-1}}\{h_{k,i}\mid X_{n-1}=\mathord\cdot\}
\right]                                                       \notag\\
&\qquad\leq
2G'(\theta)
\begin{cases}
1, & r=1,\\
(1-\epsilon)^{r-1}, & r\geq2.
\end{cases}
\label{eq:conditional-past-score-forgetting}
\end{align}

\medskip\noindent\emph{Step 4: covariance identity and geometric bound.}\quad
To finish, abbreviate the conditional expectation on the left of
Equation~(\ref{eq:conditional-past-score-forgetting}) by \(H\), only for
the next calculation. Since the \(X_{n-1}\)-marginal of \(\rho_{n-1}\)
is \(\pi_{n-1}\), the tower property gives
\begin{align}
\rho_{n-1}(h_{k,i})
&=
\pi_{n-1}(H),
\label{eq:tower-past-score}
\\
\rho_{n-1}\{h_{k,i}M_n(G_n)\}
&=
\pi_{n-1}\{H M_n(G_n)\}.
\label{eq:tower-past-score-potential}
\end{align}
Substituting Equations~(\ref{eq:tower-past-score})--
(\ref{eq:tower-past-score-potential}) into
Equation~(\ref{eq:posterior-path-update-functional}) and then subtracting
\(\rho_{n-1}(h_{k,i})\) gives the exact covariance identity
\begin{align}
\Delta_{n,r,i}
&=
\rho_n(h_{k,i})-\rho_{n-1}(h_{k,i})                           \notag\\
&=
\frac{
\pi_{n-1}\{H M_n(G_n)\}
-\pi_{n-1}(H)\pi_{n-1}M_n(G_n)
}
{\pi_{n-1}M_n(G_n)}                                          \notag\\
&=
\frac{
\pi_{n-1}\!\left[
\{H-\pi_{n-1}(H)\}
\{M_n(G_n)-\pi_{n-1}M_n(G_n)\}
\right]
}
{\pi_{n-1}M_n(G_n)}.
\label{eq:lag-coefficient-covariance-identity}
\end{align}
For any probability measure \(\mu\) and bounded real \(f\),
\begin{equation}
|f(x)-\mu(f)|
\leq
\operatorname{osc}(f).
\label{eq:centered-function-oscillation-bound}
\end{equation}
The first positive lag also inherits one contraction from the forward
transition.  To see this directly, fix \(x^\star\).  Assumption
{\refAssumpBoundedProcess} gives, for every measurable set \(A\),
\begin{align}
M_n(x,A)
&=
\int_A
f_{X_n\mid X_{n-1}}(x_n\mid x;\theta)\,dx_n                  \notag\\
&\geq
\underbar M\int_A dx_n                                      \notag\\
&\geq
\frac{\underbar M}{\bar M}
\int_A
f_{X_n\mid X_{n-1}}(x_n\mid x^\star;\theta)\,dx_n           \notag\\
&=
\epsilon M_n(x^\star,A).
\label{eq:forward-transition-minorization}
\end{align}
If \(\epsilon<1\), Equation~
(\ref{eq:forward-transition-minorization}) permits the decomposition
\begin{equation}
M_n(x,\cdot)
=
\epsilon M_n(x^\star,\cdot)
+(1-\epsilon)R_n(x,\cdot),
\label{eq:forward-transition-splitting}
\end{equation}
where \(R_n\) is a Markov kernel.  Therefore, for any \(x,x'\),
\begin{align}
M_n(G_n)(x)-M_n(G_n)(x')
&=
(1-\epsilon)
\{R_n(G_n)(x)-R_n(G_n)(x')\},                              \notag\\
\left|M_n(G_n)(x)-M_n(G_n)(x')\right|
&\leq
(1-\epsilon)\operatorname{osc}(G_n)                        \notag\\
&\leq
(1-\epsilon)(\bar G-\underbar G).
\end{align}
Taking the supremum over \(x,x'\) gives
\begin{equation}
\operatorname{osc}\{M_n(G_n)\}
\leq
(1-\epsilon)(\bar G-\underbar G).
\label{eq:predictive-potential-oscillation}
\end{equation}
If \(\epsilon=1\), Equation~
(\ref{eq:forward-transition-minorization}) implies that \(M_n(x,\cdot)\)
does not depend on \(x\), so the left side of Equation~
(\ref{eq:predictive-potential-oscillation}) is zero and the same bound
holds.
Taking absolute values in Equation~
(\ref{eq:lag-coefficient-covariance-identity}), using
Equations~(\ref{eq:centered-function-oscillation-bound})--
(\ref{eq:predictive-potential-oscillation}), and then using
Equation~(\ref{eq:conditional-past-score-forgetting}), yields
\begin{align}
|\Delta_{n,r,i}|
&\leq
\frac{
\operatorname{osc}(H)
\operatorname{osc}\{M_n(G_n)\}
}
{\underbar G}                                                 \notag\\
&\leq
\frac{2(\bar G-\underbar G)}{\underbar G}
G'(\theta)
\begin{cases}
(1-\epsilon), & r=1,\\
(1-\epsilon)^r, & r\geq2.
\end{cases}
\label{eq:lag-coordinate-geometric-bound}
\end{align}
For both cases, take
\begin{equation}
\varrho=1-\frac{\epsilon}{2}
=1-\frac{\underbar M}{2\bar M}\in(0,1).
\label{eq:population-forgetting-varrho}
\end{equation}
Since $1-\epsilon\leq\varrho$, the two cases in Equation~
(\ref{eq:lag-coordinate-geometric-bound}) are both bounded by
$\varrho^r$.
Summing the squared coordinate bounds and taking square roots give
\begin{align}
\|\Delta_{n,r}\|_2
&=
\left\{\sum_{i=1}^p|\Delta_{n,r,i}|^2\right\}^{1/2}          \notag\\
&\leq
C\sqrt p\,G'(\theta)\varrho^r,
\qquad 1\leq r<n.
\label{eq:lag-vector-geometric-bound}
\end{align}

\medskip\noindent\emph{Step 5: the zero-lag case.}\quad
Finally, when \(r=0\), Equation~
(\ref{eq:lag-coefficient-definitions}) and
Equation~(\ref{eq:bounded-realized-measurement-score}) give
\begin{align}
\|\Delta_{n,0}\|_2
&=
\|\rho_n(h_n)\|_2                                            \notag\\
&\leq
\rho_n(\|h_n\|_2)                                            \notag\\
&\leq
\sqrt p\,G'(\theta)                                         \notag\\
&\leq
C\sqrt p\,G'(\theta)\varrho^0.
\label{eq:current-lag-vector-bound}
\end{align}
Equations~(\ref{eq:lag-vector-geometric-bound}) and
(\ref{eq:current-lag-vector-bound}) prove
Equation~(\ref{eq:derivative-measure-forgetting}).
If \(\epsilon=1\), Equation~
(\ref{eq:lag-coefficient-covariance-identity}) and the zero oscillation in
Equation~(\ref{eq:predictive-potential-oscillation}) show that
\(\Delta_{n,r}=0\) for every \(r\geq1\), so any fixed
\(\varrho\in(0,1)\) proves the same result. The chosen \(\varrho\) depends
only on the model minorization bounds, independently of \(N,J,\alpha,k\).
The factor \(\varrho^r\) includes the first positive lag.
This argument bounds linear posterior expectations; particle covariances
are treated in Lemma~\ref{lem:dmop-alpha-variance}.

\end{proof}

\begin{lem}[Uniform weak bias of current-particle averages]
\label{lem:lagwise-weak-error}
Suppose $\theta=\phi$. Under Assumptions {\refAssumpBoundedProcess},
{\refAssumpBoundedMeasurement}, and
{\refAssumpDifferentiableSimulator}, with multinomial resampling, there is
$C<\infty$ such that, for every bounded real function $f$ of the current
particle location and every $t\geq1$,
\begin{equation}
\left|
\E m(\xi_t)(f)-\rho_t(f)
\right|
\leq
\frac{C}{J}\operatorname{osc}(f).
\label{eq:uniform-current-particle-weak-bias}
\end{equation}
The constant is independent of $t$ and $J$.
Under (A6-J), the same conclusion holds for a bounded function of the
current pair $(X_t,\partial_\theta X_t)$, with constants depending on the
joint-transition bounds. Under (A6), Lemma~\ref{lem:dmop-alpha-variance}
handles $h_t$ by decomposing its simulator history.
\end{lem}

\begin{proof}[Proof.]
We telescope the bias over particle updates, using semigroup contraction.

\medskip\noindent\emph{Step 1: prediction-to-filtering reduction.}\quad
First take $f$ to be a function of the physical location, in the
Feynman--Kac model defined in Section~\ref{appendix:feynman}.
The current-location marginal of
$\rho_t$ is the filtering distribution, and hence
\begin{equation}
\rho_t(f)=\pi_t(f)=\Psi_t(\eta_t)(f).
\label{eq:current-state-posterior-filter-identity}
\end{equation}
Conditionally on the prediction cloud $\eta_t^J$, multinomial resampling
draws the filtering particles independently from $\Psi_t(\eta_t^J)$.
Therefore
\begin{align}
\E\{m(\xi_t)(f)\mid\eta_t^J\}
&=
\frac1J\sum_{j=1}^J
\E\{f(x_{t,j}^F)\mid\eta_t^J\}                              \notag\\
&=
\frac1J\sum_{j=1}^J\Psi_t(\eta_t^J)(f)                     \notag\\
&=
\Psi_t(\eta_t^J)(f).
\label{eq:terminal-resampling-unbiased}
\end{align}

\medskip\noindent\emph{Step 2: normalized-semigroup contraction.}\quad
Take $\Phi_{t,t}$ to be the identity. For $1\leq p\leq t$, define
\begin{align}
Q_{t,t}(f)(x)
&=G_t(x)f(x),                                                \notag\\
Q_{p,t}(f)(x)
&=G_p(x)M_{p+1}\{Q_{p+1,t}(f)\}(x),
\qquad p<t.
\label{eq:weak-bias-semigroup-definition}
\end{align}
The definitions of $\Phi_{p,t}$ and $\Psi_t$ give, for every probability
measure $\mu$,
\begin{equation}
\Psi_t\{\Phi_{p,t}(\mu)\}(f)
=
\frac{\mu Q_{p,t}(f)}{\mu Q_{p,t}(1)}.
\label{eq:normalized-feynman-kac-semigroup}
\end{equation}

For $p<t$,
\begin{equation}
Q_{p,t}(1)(x)
=
G_p(x)\int
f_{X_{p+1}\mid X_p}(u\mid x;\theta)
Q_{p+1,t}(1)(u)\,du.
\label{eq:weak-bias-future-normalizer}
\end{equation}
Thus, for every $x,x'$,
\begin{align}
\frac{Q_{p,t}(1)(x)}{Q_{p,t}(1)(x')}
&\leq
\frac{
\bar G\bar M\int Q_{p+1,t}(1)(u)\,du
}{
\underbar G\underbar M\int Q_{p+1,t}(1)(u)\,du
}                                                        \notag\\
&=
\frac{\bar G\bar M}{\underbar G\underbar M}.
\label{eq:weak-bias-normalizer-ratio}
\end{align}
For $p=t$, the same ratio is at most $\bar G/\underbar G$. Put
\begin{equation}
R=\frac{\bar G\bar M}{\underbar G\underbar M}\geq1.
\label{eq:weak-bias-normalizer-constant}
\end{equation}
Then
\begin{equation}
\sup_{x,x'}
\frac{Q_{p,t}(1)(x)}{Q_{p,t}(1)(x')}
\leq R,
\qquad 1\leq p\leq t.
\label{eq:weak-bias-uniform-normalizer-ratio}
\end{equation}

For $p\leq s<t$, the normalized forward kernel is
\begin{equation}
P_{s,s+1}^{(t)}(x,A)
=
\frac{
\displaystyle\int_A
f_{X_{s+1}\mid X_s}(u\mid x;\theta)
Q_{s+1,t}(1)(u)\,du
}{
\displaystyle\int
f_{X_{s+1}\mid X_s}(u\mid x;\theta)
Q_{s+1,t}(1)(u)\,du
}.
\label{eq:weak-bias-forward-normalized-kernel}
\end{equation}
Define, only within this proof,
\begin{equation}
\nu_{s,t}(A)
=
\frac{
\displaystyle\int_AQ_{s+1,t}(1)(u)\,du
}{
\displaystyle\int Q_{s+1,t}(1)(u)\,du
}.
\label{eq:weak-bias-forward-minorizing-measure}
\end{equation}
Assumption~{\refAssumpBoundedProcess} gives
\begin{align}
P_{s,s+1}^{(t)}(x,A)
&\geq
\frac{
\underbar M\displaystyle\int_AQ_{s+1,t}(1)(u)\,du
}{
\bar M\displaystyle\int Q_{s+1,t}(1)(u)\,du
}                                                        \notag\\
&=
\frac{\underbar M}{\bar M}\nu_{s,t}(A).
\label{eq:weak-bias-forward-minorization}
\end{align}
Consequently, by Equation~
(\ref{eq:population-forgetting-varrho}), the Dobrushin coefficient of
every kernel in Equation~
(\ref{eq:weak-bias-forward-normalized-kernel}) is at most $\varrho$.
The factor $G_s(x)$ cancels from the normalized conditional
expectation, and hence
\begin{equation}
\frac{Q_{p,t}(f)}{Q_{p,t}(1)}
=
P_{p,p+1}^{(t)}P_{p+1,p+2}^{(t)}\cdots
P_{t-1,t}^{(t)}f.
\label{eq:weak-bias-forward-kernel-factorization}
\end{equation}
It follows that
\begin{equation}
\operatorname{osc}
\left\{
\frac{Q_{p,t}(f)}{Q_{p,t}(1)}
\right\}
\leq
\varrho^{t-p}\operatorname{osc}(f).
\label{eq:weak-bias-normalized-semigroup-contraction}
\end{equation}

\medskip\noindent\emph{Step 3: one complete particle update.}\quad
For one selection--mutation update, if $p=1$, put
$\mu=\eta_1=\pi_0M_1$. For $p\geq2$, condition on $\eta_{p-1}^J$ and put
\begin{equation}
\mu=\Phi_p(\eta_{p-1}^J).
\label{eq:weak-bias-local-conditional-law}
\end{equation}
In either case, the particles forming $\eta_p^J$ are conditionally
independent with common law $\mu$. Define the normalizer $q$ and
centered fluctuation $d$ by
\begin{align}
q
&=
\frac{Q_{p,t}(1)}{\mu Q_{p,t}(1)},                            &
a
&=
\frac{\mu Q_{p,t}(f)}{\mu Q_{p,t}(1)},                        &
d
&=
\frac{Q_{p,t}(f-a)}{\mu Q_{p,t}(1)}.
\label{eq:weak-bias-local-functions}
\end{align}
Then
\begin{equation}
\mu(q)=1,
\qquad
\mu(d)=0,
\label{eq:weak-bias-local-centering}
\end{equation}
and Equation~(\ref{eq:normalized-feynman-kac-semigroup}) gives
\begin{align}
&\Psi_t\{\Phi_{p,t}(\eta_p^J)\}(f)
-\Psi_t\{\Phi_{p,t}(\mu)\}(f)                              \notag\\
&\quad=
\frac{\eta_p^JQ_{p,t}(f)}{\eta_p^JQ_{p,t}(1)}-a             \notag\\
&\quad=
\frac{\eta_p^JQ_{p,t}(f-a)}{\eta_p^JQ_{p,t}(1)}             \notag\\
&\quad=
\frac{\eta_p^J(d)}{\eta_p^J(q)}.
\label{eq:weak-bias-local-ratio}
\end{align}

Equation~(\ref{eq:weak-bias-uniform-normalizer-ratio}) implies
\begin{equation}
R^{-1}\leq q(x)\leq R,
\qquad
\eta_p^J(q)\geq R^{-1},
\qquad
\|q-1\|_\infty\leq R.
\label{eq:weak-bias-q-bounds}
\end{equation}
Writing
\begin{equation}
F=\frac{Q_{p,t}(f)}{Q_{p,t}(1)},
\end{equation}
we have $a=\mu(qF)$ and $\mu(q)=1$. Thus $a$ is an average of
$F$ under the probability measure $q\,d\mu$, so
\begin{equation}
\inf F\leq a\leq\sup F.
\label{eq:weak-bias-a-range}
\end{equation}
Since $d=q(F-a)$, Equations~(\ref{eq:weak-bias-q-bounds}),
(\ref{eq:weak-bias-a-range}), and
(\ref{eq:weak-bias-normalized-semigroup-contraction}) give
\begin{align}
\|d\|_\infty
&\leq
R\operatorname{osc}(F)                                    \notag\\
&\leq
R\varrho^{t-p}\operatorname{osc}(f).
\label{eq:weak-bias-d-bound}
\end{align}

Conditional independence and $\mu(d)=0$ give
\begin{align}
\E\{\eta_p^J(d)\mid\mu\}
&=0,                                                       \notag\\
\E\{\eta_p^J(d)^2\mid\mu\}
&=\frac1J\mu(d^2),                                        \notag\\
\E\{(\eta_p^J(q)-1)^2\mid\mu\}
&=\frac1J\mu\{(q-1)^2\}.
\label{eq:weak-bias-empirical-second-moments}
\end{align}
Consequently,
\begin{align}
&\E\left\{
\left.
\frac{\eta_p^J(d)}{\eta_p^J(q)}
\right|\mu
\right\}                                                    \notag\\
&\quad=
-\E\left\{
\left.
\frac{\eta_p^J(d)\{\eta_p^J(q)-1\}}
{\eta_p^J(q)}
\right|\mu
\right\}.
\label{eq:weak-bias-local-cancellation}
\end{align}
Conditional Cauchy--Schwarz and Equations~
(\ref{eq:weak-bias-q-bounds})--
(\ref{eq:weak-bias-empirical-second-moments}) yield
\begin{align}
&\left|
\E\left\{
\left.
\frac{\eta_p^J(d)}{\eta_p^J(q)}
\right|\mu
\right\}
\right|                                                       \notag\\
&\quad\leq
R
\left[
\E\{\eta_p^J(d)^2\mid\mu\}
\E\{(\eta_p^J(q)-1)^2\mid\mu\}
\right]^{1/2}                                                \notag\\
&\quad=
\frac{R}{J}
\left[
\mu(d^2)\mu\{(q-1)^2\}
\right]^{1/2}                                                \notag\\
&\quad\leq
\frac{C}{J}\varrho^{t-p}\operatorname{osc}(f).
\label{eq:weak-bias-one-complete-update}
\end{align}

\medskip\noindent\emph{Step 4: telescope over update times.}\quad
Since
$\eta_t=\Phi_{1,t}(\eta_1)$,
\begin{align}
&\E\Psi_t(\eta_t^J)(f)-\Psi_t(\eta_t)(f)                    \notag\\
&\quad=
\E\left[
\Psi_t\{\Phi_{1,t}(\eta_1^J)\}(f)
-\Psi_t\{\Phi_{1,t}(\eta_1)\}(f)
\right]                                                     \notag\\
&\qquad+
\sum_{p=2}^t
\E\left[
\Psi_t\{\Phi_{p,t}(\eta_p^J)\}(f)
-\Psi_t\{\Phi_{p,t}(\Phi_p(\eta_{p-1}^J))\}(f)
\right].
\label{eq:weak-bias-full-update-telescope}
\end{align}
Indeed, for $2\leq p\leq t$,
\begin{equation}
\Phi_{p,t}\circ\Phi_p=\Phi_{p-1,t},
\label{eq:weak-bias-composition-identity}
\end{equation}
so consecutive terms cancel. The only terms left are
$\Psi_t(\eta_t^J)(f)$ and $-\Psi_t(\eta_t)(f)$.

Apply Equation~(\ref{eq:weak-bias-one-complete-update}) to every
summand in Equation~(\ref{eq:weak-bias-full-update-telescope}). Then
\begin{align}
\left|
\E\Psi_t(\eta_t^J)(f)-\Psi_t(\eta_t)(f)
\right|
&\leq
\frac{C\operatorname{osc}(f)}{J}
\sum_{p=1}^t\varrho^{t-p}                                  \notag\\
&\leq
\frac{C}{J(1-\varrho)}\operatorname{osc}(f).
\end{align}
Absorb $(1-\varrho)^{-1}$ into $C$ and combine this display with
Equations~(\ref{eq:current-state-posterior-filter-identity}) and
(\ref{eq:terminal-resampling-unbiased}).

For (A6-J), apply the same calculation to the pair
$(X_t,\partial_\theta X_t)$. In every display replace the physical
transition density by that in
Equation~(\ref{eq:joint-gradient-density-assumption}), $du$ by the common
reference measure, $(\underbar M,\bar M)$ by
$(\underbar M_D,\bar M_D)$, and the contraction factor by
$1-\underbar M_D/(2\bar M_D)$. The potential remains $G_t(X_t)$, so its
bounds and all ratio and conditional-resampling identities are unchanged.
For each coordinate, the current measurement gradient is
\begin{equation}
h_{t,i}(x,v)
=\partial_{\theta_i}\log f_{Y_t\mid X_t}(y_t^*\mid x;\theta)
+\sum_{a=1}^d\partial_{x_a}\log f_{Y_t\mid X_t}(y_t^*\mid x;\theta)
v_{ai},\qquad \theta=\phi.
\label{eq:joint-current-measurement-gradient}
\end{equation}
It is bounded by (A3) on the admissible joint state space.
\end{proof}

\begin{lem}[Full-score summation-by-parts identities]
\label{lem:dmop-full-score-summation-by-parts}
Suppose $\theta=\phi$ and $\alpha\in[0,1]$.  Then the particle and
population scores satisfy Equations~
(\ref{eq:full-score-summation-by-parts}) and
(\ref{eq:population-full-score-summation-by-parts}) below.  Moreover,
for every coordinate $i$ and $t\leq N$,
\begin{equation}
\left\|
\sum_{k=1}^{t}\alpha^{t-k}h_{k,i}
\right\|_\infty
\leq
G'(\theta)\sum_{k=1}^{t}\alpha^{t-k}.
\label{eq:discounted-score-supremum-statement}
\end{equation}
\end{lem}

\begin{proof}[Proof.]
\medskip\noindent\emph{Step 1: particle summation by parts.}\quad
Sum Equation~(\ref{eq:conditional-score-path-decomposition}) over
$n=1,\ldots,N$ and put $k=n-r$. Since the sums are finite,
\begin{align}
\nabla_\theta\hat\ell^\alpha(\theta)
&=
\sum_{k=1}^{N}
\left[
m(\xi_k)(h_k)
+\sum_{t=k+1}^{N}\alpha^{t-k}
\{m(\xi_t)(h_k)-m(\xi_{t-1})(h_k)\}
\right].
\label{eq:full-score-reindexed-by-contribution}
\end{align}
For each fixed $k$, collecting the coefficient of every
$m(\xi_t)(h_k)$ gives
\begin{align}
&m(\xi_k)(h_k)
+\sum_{t=k+1}^{N}\alpha^{t-k}
\{m(\xi_t)(h_k)-m(\xi_{t-1})(h_k)\}                         \notag\\
&\quad=
(1-\alpha)\sum_{t=k}^{N-1}\alpha^{t-k}m(\xi_t)(h_k)
+\alpha^{N-k}m(\xi_N)(h_k).
\label{eq:one-score-summation-by-parts}
\end{align}
Indeed, the coefficients are $1-\alpha$ at $t=k$,
$\alpha^{t-k}-\alpha^{t+1-k}=(1-\alpha)\alpha^{t-k}$ for
$k<t<N$, and $\alpha^{N-k}$ at $t=N$. Interchanging the finite sums once
more yields
\begin{align}
\nabla_\theta\hat\ell^\alpha(\theta)
&=
(1-\alpha)\sum_{t=1}^{N-1}
m(\xi_t)\left(\sum_{k=1}^{t}\alpha^{t-k}h_k\right)
+m(\xi_N)\left(\sum_{k=1}^{N}\alpha^{N-k}h_k\right).
\label{eq:full-score-summation-by-parts}
\end{align}
\medskip\noindent\emph{Step 2: population identity and uniform bound.}\quad
The same algebra applied to
Equation~(\ref{eq:population-conditional-score-path-decomposition}) gives
\begin{align}
\nabla_\theta\ell^\alpha(\theta)
&=
(1-\alpha)\sum_{t=1}^{N-1}
\rho_t\left(\sum_{k=1}^{t}\alpha^{t-k}h_k\right)
+\rho_N\left(\sum_{k=1}^{N}\alpha^{N-k}h_k\right).
\label{eq:population-full-score-summation-by-parts}
\end{align}
For every coordinate $i$, Equation~
(\ref{eq:bounded-realized-measurement-score}) gives
\begin{equation}
\left\|
\sum_{k=1}^{t}\alpha^{t-k}h_{k,i}
\right\|_\infty
\leq
G'(\theta)\sum_{k=1}^{t}\alpha^{t-k}.
\label{eq:discounted-score-supremum}
\end{equation}

Equation~(\ref{eq:discounted-score-supremum}) is identical to
Equation~(\ref{eq:discounted-score-supremum-statement}), completing the
proof.
\end{proof}

\begin{claim}[Conditional empirical moments]
\label{claim:conditional-empirical-moments}
Let $\mathcal F$ be a sigma-field and let $\mu$ be an
$\mathcal F$-measurable probability measure.  Suppose that,
conditionally on $\mathcal F$, $X_1,\ldots,X_J$ are independent with
common law $\mu$.  The following bounds hold for bounded
$\mathcal F$-measurable real functions $f,f_1,f_2$:
\begin{align}
\E\left\{
\left.
\frac1J\sum_{j=1}^J\{f(X_j)-\mu(f)\}
\right|\mathcal F
\right\}
&=0,
\label{eq:conditional-empirical-mean-zero}
\\
\E\left\{
\left.
\left[
\frac1J\sum_{j=1}^J\{f(X_j)-\mu(f)\}
\right]^2
\right|\mathcal F
\right\}
&=
\frac{\mu\{(f-\mu f)^2\}}J
\leq
\frac{\operatorname{osc}(f)^2}{4J},
\label{eq:conditional-empirical-second-moment}
\\
\E\left\{
\left.
\left[
\frac1J\sum_{j=1}^J\{f(X_j)-\mu(f)\}
\right]^4
\right|\mathcal F
\right\}
&=
\frac{\mu\{(f-\mu f)^4\}}{J^3}
+\frac{3(J-1)}{J^3}\mu\{(f-\mu f)^2\}^2                 \notag\\
&\leq
\frac{\mu\{(f-\mu f)^4\}}{J^3}
+\frac{3}{J^2}\mu\{(f-\mu f)^2\}^2                    \notag\\
&\leq
\frac{3\operatorname{osc}(f)^4}{J^2}.
\label{eq:conditional-empirical-fourth-moment}
\end{align}
Consequently, for bounded $f_1$ and $f_2$,
\begin{align}
&\left\|
\left[
\frac1J\sum_{j=1}^J\{f_1(X_j)-\mu(f_1)\}
\right]
\left[
\frac1J\sum_{j=1}^J\{f_2(X_j)-\mu(f_2)\}
\right]
\right\|_{L^2\mid\mathcal F}                               \notag\\
&\qquad\leq
\frac{\sqrt3}{J}
\operatorname{osc}(f_1)\operatorname{osc}(f_2).
\label{eq:conditional-empirical-product-moment}
\end{align}
\end{claim}

\begin{proof}[Proof.]
Put $Y_j=f(X_j)-\mu(f)$.  Conditional independence and
$\E(Y_j\mid\mathcal F)=0$ give Equations~
(\ref{eq:conditional-empirical-mean-zero}) and
(\ref{eq:conditional-empirical-second-moment}).  They also give
\begin{align}
&\E\left\{
\left.
\left(\frac1J\sum_{j=1}^JY_j\right)^4
\right|\mathcal F
\right\}                                                     \notag\\
&\quad=
\frac1{J^4}
\left\{
J\mu\{(f-\mu f)^4\}
+3J(J-1)\mu\{(f-\mu f)^2\}^2
\right\}                                                    \notag\\
&\quad=
\frac{\mu\{(f-\mu f)^4\}}{J^3}
+\frac{3(J-1)}{J^3}\mu\{(f-\mu f)^2\}^2                 \notag\\
&\quad\leq
\frac{3}{J^2}\operatorname{osc}(f)^4,
\end{align}
which is Equation~(\ref{eq:conditional-empirical-fourth-moment}).
Conditional Cauchy--Schwarz applied to the two fourth moments proves
Equation~(\ref{eq:conditional-empirical-product-moment}).
\end{proof}

\begin{claim}[Normalized-ratio perturbations]
\label{claim:normalized-ratio-perturbations}
Let $\mu$ and $\nu$ be probability measures.  For $l\in\{0,1\}$,
let $a_l>0$ and $b_l$ be bounded functions such that $\mu(a_l)=1$, and
put
\begin{equation}
\mathcal R_l(\lambda)=\frac{\lambda(b_l)}{\lambda(a_l)},
\qquad
D_\mu\mathcal R_l=b_l-\mathcal R_l(\mu)a_l.
\label{eq:auxiliary-normalized-ratio-definition}
\end{equation}
Then
\begin{align}
\mathcal R_l(\nu)-\mathcal R_l(\mu)
&=(\nu-\mu)D_\mu\mathcal R_l
-\frac{(\nu-\mu)(a_l)}{\nu(a_l)}
(\nu-\mu)D_\mu\mathcal R_l,
\label{eq:auxiliary-exact-ratio-expansion}
\end{align}
and hence
\begin{align}
&\{\mathcal R_1(\nu)-\mathcal R_0(\nu)\}
-\{\mathcal R_1(\mu)-\mathcal R_0(\mu)\}                  \notag\\
&\quad-(\nu-\mu)
\{D_\mu\mathcal R_1-D_\mu\mathcal R_0\}                  \notag\\
&=-\frac{(\nu-\mu)(a_1-a_0)}{\nu(a_0)\nu(a_1)}
(\nu-\mu)D_\mu\mathcal R_1                                \notag\\
&\quad-
\frac{(\nu-\mu)(a_0)}{\nu(a_0)}
(\nu-\mu)
\{D_\mu\mathcal R_1-D_\mu\mathcal R_0\}.
\label{eq:auxiliary-paired-ratio-remainder}
\end{align}
If $\inf_{l\in\{0,1\}}\nu(a_l)\geq C^{-1}$, then the absolute
value of the left side of
Equation~(\ref{eq:auxiliary-paired-ratio-remainder}) is at most
\begin{align}
&C^2\left|(\nu-\mu)(a_1-a_0)\right|
\left|(\nu-\mu)D_\mu\mathcal R_1\right|                    \notag\\
&\quad+
C\left|(\nu-\mu)(a_0)\right|
\left|(\nu-\mu)
\{D_\mu\mathcal R_1-D_\mu\mathcal R_0\}\right|.
\label{eq:auxiliary-paired-ratio-bound}
\end{align}
\end{claim}

\begin{proof}[Proof.]
Since $\mu(a_l)=1$,
\begin{align}
\mathcal R_l(\nu)-\mathcal R_l(\mu)
&=
\frac{\nu(b_l)-\mathcal R_l(\mu)\nu(a_l)}{\nu(a_l)}       \notag\\
&=
\frac{(\nu-\mu)D_\mu\mathcal R_l}{\nu(a_l)}              \notag\\
&=(\nu-\mu)D_\mu\mathcal R_l
-\frac{(\nu-\mu)(a_l)}{\nu(a_l)}
(\nu-\mu)D_\mu\mathcal R_l,
\end{align}
which proves Equation~(\ref{eq:auxiliary-exact-ratio-expansion}).
Moreover,
\begin{align}
&\frac{(\nu-\mu)(a_1)}{\nu(a_1)}
-\frac{(\nu-\mu)(a_0)}{\nu(a_0)}                         \notag\\
&\quad=
\left(1-\frac1{\nu(a_1)}\right)
-\left(1-\frac1{\nu(a_0)}\right)                         \notag\\
&\quad=
\frac{\nu(a_1)-\nu(a_0)}{\nu(a_0)\nu(a_1)}              \notag\\
&\quad=
\frac{(\nu-\mu)(a_1-a_0)}{\nu(a_0)\nu(a_1)}.
\end{align}
Subtracting the cases $l=0$ and $l=1$ in
Equation~(\ref{eq:auxiliary-exact-ratio-expansion}) now proves
Equation~(\ref{eq:auxiliary-paired-ratio-remainder}).  Taking absolute
values there and using
$\nu(a_0)^{-1}\leq C$ and $\nu(a_1)^{-1}\leq C$ proves
Equation~(\ref{eq:auxiliary-paired-ratio-bound}).
This exact identity is the paired specialization of the normalized
Feynman--Kac perturbation calculation in Equations~(3.1)--(3.4) and
Appendix~A.3 of \citet{delMoral11}.
\end{proof}

\begin{lem}[Particle error and variance of the DMOP-$\alpha$ score]
\label{lem:dmop-alpha-variance}
Suppose $J\geq2$, $\theta=\phi$, $\alpha\in[0,1]$, Assumptions
{\refAssumpBoundedProcess}, {\refAssumpBoundedMeasurement}, and
{\refAssumpDifferentiableSimulator} hold, and multinomial resampling is
used. Suppose also that either (A6) or (A6-J) holds. Then
\begin{align}
\Var\left\{\nabla_\theta\hat\ell^\alpha(\theta)\right\}
&\lesssim
\frac{NpG'(\theta)^2(1+\log J)^2}{J}
\left(\sum_{r=0}^{N-1}\alpha^r\right)^2,
\label{eq:refined-dmop-alpha-variance}
\\
\left\|
\E\nabla_\theta\hat\ell^\alpha(\theta)
-\nabla_\theta\ell^\alpha(\theta)
\right\|_2
&\lesssim
\frac{N\sqrt p\,G'(\theta)}{J}
\sum_{r=0}^{N-1}\alpha^r,
\label{eq:refined-dmop-alpha-particle-bias}
\\
\E\left\|
\nabla_\theta\hat\ell^\alpha(\theta)
-\nabla_\theta\ell^\alpha(\theta)
\right\|_2^2
&\lesssim
pG'(\theta)^2
\left(\frac{N(1+\log J)^2}{J}+\frac{N^2}{J^2}\right)
\left(\sum_{r=0}^{N-1}\alpha^r\right)^2.
\label{eq:refined-dmop-alpha-direct-particle-error}
\end{align}
The finite geometric sum is
\begin{equation}
\sum_{r=0}^{N-1}\alpha^r
=
\begin{cases}
\dfrac{1-\alpha^N}{1-\alpha},&0\leq\alpha<1,\\[6pt]
N,&\alpha=1.
\end{cases}
\label{eq:dmop-alpha-finite-memory}
\end{equation}
The constants suppressed by $\lesssim$ depend only on the fixed model
and stability bounds and are independent of $N,J$, and $\alpha$.
As elsewhere,
$\Var(Z)=\E\|Z-\E Z\|_2^2$ for a random vector $Z$.
\end{lem}

\begin{proof}[Proof.]
We bound lag-truncated scores uniformly in their starting time, then use
their geometric mixture to recover DMOP-$\alpha$. Under (A6), the
chain rule first decomposes $h_s$ by the age of its simulator contributions.

\medskip\noindent\emph{Step 1: auxiliary lag truncations.}\quad
For $0\leq k<N$, define
\begin{align}
\nabla_\theta\hat s_k^1(\theta)
&=
\sum_{n=1}^{N}\sum_{r=0}^{\min(k,n-1)}\Delta_{n,r}^J,
&
\nabla_\theta s_k^1(\theta)
&=
\sum_{n=1}^{N}\sum_{r=0}^{\min(k,n-1)}\Delta_{n,r}.
\label{eq:fixed-lag-score-definition-s13}
\end{align}
For a fixed creation time $s$, the terms in the first sum telescope:
\begin{align}
&m(\xi_s)(h_s)
+\sum_{n=s+1}^{\min(N,s+k)}
\{m(\xi_n)(h_s)-m(\xi_{n-1})(h_s)\}                     \notag\\
&\qquad=
m(\xi_{\min(N,s+k)})(h_s).
\label{eq:fixed-lag-one-input-telescope-s13}
\end{align}
Summing Equation~(\ref{eq:fixed-lag-one-input-telescope-s13}) over
$s=1,\ldots,N$ gives
\begin{align}
\nabla_\theta\hat s_k^1(\theta)
&=
\sum_{s=1}^{N-k-1}m(\xi_{s+k})(h_s)
+m(\xi_N)\left(\sum_{s=N-k}^{N}h_s\right),
\label{eq:fixed-lag-score-telescope-s13}
\\
\nabla_\theta s_k^1(\theta)
&=
\sum_{s=1}^{N-k-1}\rho_{s+k}(h_s)
+\rho_N\left(\sum_{s=N-k}^{N}h_s\right).
\label{eq:fixed-lag-population-telescope-s13}
\end{align}
An empty sum in these displays is zero.

\medskip\noindent\emph{Step 1a: chain-rule decomposition.}\quad
All derivatives in this step are evaluated at $\theta=\phi$, with
reference resampling indices fixed. Write
$X_u^\theta=\process_u(X_{u-1}^\theta;\theta,\omega_u)$, and let
$\partial_{\theta_i}\process_u$ denote its partial derivative holding
$X_{u-1}^\theta$ fixed. The chain rule, iterated from time zero, gives
\begin{equation}
\partial_{\theta_i}X_s^\theta
=
\frac{\partial X_s^\theta}{\partial X_0^\theta}
\partial_{\theta_i}X_0^\theta
+\sum_{u=1}^s
\frac{\partial X_s^\theta}{\partial X_u^\theta}
\partial_{\theta_i}\process_u(X_{u-1}^\theta;\theta,\omega_u).
\label{eq:simulator-gradient-chain-expansion}
\end{equation}
Under (A6), decompose the coordinate-$i$ measurement gradient at time $s$
by the age of its chain-rule contributions. The superscript $[r]$ is an
age index, not a power or derivative order: for $1\leq r<s$,
$h_{s,i}^{[r]}$ is the contribution from the simulator's partial parameter
derivative at time $s-r$, propagated to the measurement at time $s$.
This age differs from the discount lag $n-s$, which gives the coefficient
$\alpha^{n-s}$ when the time-$s$ measurement gradient contributes at time $n$.
Define
\begin{align}
h_{s,i}^{[0]}
&=
\left.\partial_{\theta_i}\log f_{Y_s\mid X_s}(y_s^*\mid x;\theta)
\right|_{x=X_s^\theta}
+\left.D_x\log f_{Y_s\mid X_s}(y_s^*\mid x;\theta)\right|_{x=X_s^\theta}
\partial_{\theta_i}\process_s(X_{s-1}^\theta;\theta,\omega_s),
\notag\\
h_{s,i}^{[r]}
&=
\left.
D_x\log f_{Y_s\mid X_s}
\bigl(y_s^*\mid X_s^\theta(x,\omega_{s-r+1:s});\theta\bigr)
\right|_{x=X_{s-r}^\theta}
\partial_{\theta_i}\process_{s-r}(X_{s-r-1}^\theta;\theta,\omega_{s-r}),
\quad 1\leq r<s,
\notag\\
h_{s,i}^{[s]}
&=
\left.
D_x\log f_{Y_s\mid X_s}
\bigl(y_s^*\mid X_s^\theta(x,\omega_{1:s});\theta\bigr)
\right|_{x=X_0^\theta}
\partial_{\theta_i}X_0^\theta.
\label{eq:measurement-gradient-age-components}
\end{align}
Here $h_{s,i}^{[0]}$ combines the direct parameter derivative of the
log measurement density, holding $x$ fixed, with the contribution from the
time-$s$ simulator's partial parameter derivative. The endpoint
$h_{s,i}^{[s]}$ is the contribution from the initial-state derivative.
Multiplying Equation~(\ref{eq:simulator-gradient-chain-expansion}) by
the spatial log measurement derivative gives
\begin{equation}
h_{s,i}=\sum_{r=0}^s h_{s,i}^{[r]},
\qquad
\|h_{s,i}^{[r]}\|_\infty\leq CG'(\theta)\kappa^r.
\label{eq:measurement-gradient-age-sum}
\end{equation}
For $r=0$ use (A6) with $t=s$ and the bounded partial derivatives;
for $1\leq r<s$ use it with $t=s-r$; for $r=s$ use $t=0$ and the
initial-gradient bound. For $r<s$, the displayed function depends only
on $X_{s-r-1}$ and the simulator variables from $s-r$ through $s$.
For $r=s$ it also uses the initial simulator draw, whose particle
replicates are independent. Set $h_{s,i}^{[r]}=0$ for $r>s$.

Under (A6-J), take instead $h_{s,i}^{[0]}=h_{s,i}$, the bounded current
function in Equation~(\ref{eq:joint-current-measurement-gradient}), and
$h_{s,i}^{[r]}=0$ for $r\geq1$. Use the joint pair as the state throughout
the particle calculation. The sums over $r$ below then contain only
$r=0$; $\kappa$ is needed only under (A6).

\medskip\noindent\emph{Step 1b: finite-window particle error.}\quad
Let $f$ be a bounded real function of the state at time $a$ and the
subsequent simulator variables through time $b\geq a$, with no earlier
history. Write $Q_{a,b}$ for the path lift of the semigroup in
Equation~(\ref{eq:weak-bias-semigroup-definition}).
First suppose $a\geq1$. Under (A6), the starting cloud is
the physical cloud; under (A6-J), it is the joint cloud.
Adding and subtracting the conditional population continuation gives
\begin{align}
m(\xi_b)(f)-\rho_b(f)
&=
\left\{m(\xi_b)(f)-\frac{\eta_a^JQ_{a,b}(f)}
{\eta_a^JQ_{a,b}(1)}\right\}
+\left\{\frac{\eta_a^JQ_{a,b}(f)}
{\eta_a^JQ_{a,b}(1)}-\frac{\eta_aQ_{a,b}(f)}
{\eta_aQ_{a,b}(1)}\right\}.
\label{eq:lag-truncation-conditional-decomposition-s13}
\end{align}

For the first brace, condition on $\eta_a^J$. Let
$\nu=m(\xi_b^P)$ and $\mu=\Phi_{a,b}(\eta_a^J)$ on the prediction
window. The first prediction following selection from $\eta_a^J$
has conditionally independent particles with common law
$\Phi_{a,a+1}(\eta_a^J)$. To apply Theorems~3.1 and~3.3 of
\citet{delmoral18}, first note that future normalizers depend on the current
physical state and satisfy
Equation~(\ref{eq:weak-bias-uniform-normalizer-ratio}). For a bounded
function $\psi(X_b,\omega_b)$ of the last state and its innovation,
its normalized continuation, for $b>a$, is
\begin{align}
\frac{Q_{a,b}(\psi)}{Q_{a,b}(1)}
&=P_{a,a+1}^{(b)}\cdots P_{b-2,b-1}^{(b)}
\left[
\frac{\E\{G_b(X_b)\psi(X_b,\omega_b)\mid X_{b-1}=\mathord\cdot\}}
{M_b(G_b)}
\right],\notag\\
\operatorname{osc}\left\{\frac{Q_{a,b}(\psi)}{Q_{a,b}(1)}\right\}
&\leq\varrho^{b-a-1}\operatorname{osc}(\psi).
\label{eq:innovation-marginal-forgetting}
\end{align}
The product is empty when $b=a+1$. The bracket is a conditional average
of $\psi$, and the other kernels contract as in
Equation~(\ref{eq:weak-bias-normalized-semigroup-contraction}). Thus the
innovation-augmented marginal has a summable forgetting bound, with a
one-step delay. For prediction histories, omit the terminal $G_b$:
the final bracket becomes $\E[\psi(X_b,\omega_b)\mid X_{b-1}=\mathord\cdot]$,
and the same contraction argument gives the corresponding summable bound.
Apply the path construction of \citet{delmoral18} to this marginal.
Under (A6-J), use the joint marginal and its transition bounds instead,
as in Lemma~\ref{lem:lagwise-weak-error}. The cited theorems now apply
over at most $b-a$ updates, uniformly in $a$.
For the second moment, put $\bar f=f-\mu(f)$ and
choose $c$ so that $0\leq\bar f+c\leq\operatorname{osc}(f)$.
Then $\mu(\bar f+c)=c$, and
\begin{align}
\E\{|\nu(f)-\mu(f)|^2\mid\eta_a^J\}
&=\frac1{J^2}\sum_j\E\{\bar f(\xi_{b,j}^P)^2\mid\eta_a^J\}
+\frac1{J^2}\sum_{j\ne l}
\E\{\bar f(\xi_{b,j}^P)\bar f(\xi_{b,l}^P)\mid\eta_a^J\},
\notag\\
\E\{\bar f(\xi_{b,j}^P)\bar f(\xi_{b,l}^P)\mid\eta_a^J\}
&=\E\{(\bar f+c)(\xi_{b,j}^P)(\bar f+c)(\xi_{b,l}^P)
\mid\eta_a^J\}-c^2
\notag\\
&\quad-c\{\E[(\bar f+c)(\xi_{b,j}^P)\mid\eta_a^J]-c\}
\notag\\
&\quad-c\{\E[(\bar f+c)(\xi_{b,l}^P)\mid\eta_a^J]-c\}.
\label{eq:lag-truncation-prediction-diagonal-expansion-s13}
\end{align}
The diagonal terms are at most $\operatorname{osc}(f)^2$.
Theorem~3.3 bounds the first difference in the second identity by
$C(b-a)\operatorname{osc}(f)^2/J$; Theorem~3.1 bounds each remaining
difference by $C(b-a)\operatorname{osc}(f)/J$, and
$|c|\leq\operatorname{osc}(f)$. Thus, initially for
$J\geq c_0(b-a+1)$,
\begin{align}
\E\{|\nu(f)-\mu(f)|^2\mid\eta_a^J\}
&\leq \frac{C(b-a+1)}J\operatorname{osc}(f)^2,
\label{eq:lag-truncation-prediction-l2-s13}\\
|\E\{\nu(f)-\mu(f)\mid\eta_a^J\}|
&\leq \frac{C(b-a+1)}J\operatorname{osc}(f).
\label{eq:lag-truncation-prediction-bias-s13}
\end{align}
When $a=b$, $\nu=\mu=\eta_a^J$, so both prediction errors are zero.

To pass to filtering, put $d=G_b\{f-\Psi_b(\mu)(f)\}$. Then $\mu(d)=0$
and the exact ratio expansion is
\begin{align}
\Psi_b(\nu)(f)-\Psi_b(\mu)(f)
&=\frac{(\nu-\mu)(d)}{\nu(G_b)}
\notag\\
&=\frac{(\nu-\mu)(d)}{\mu(G_b)}
-\frac{(\nu-\mu)(d)(\nu-\mu)(G_b)}
{\nu(G_b)\mu(G_b)}.
\label{eq:lag-truncation-window-bayes-bias-expansion-s13}
\end{align}
Both denominators are at least $\underbar G$, and
$\operatorname{osc}(d)\leq2\bar G\operatorname{osc}(f)$.
Equations~(\ref{eq:lag-truncation-prediction-l2-s13}) and
(\ref{eq:lag-truncation-prediction-bias-s13}) consequently give
\begin{align}
\E\{|\Psi_b(\nu)(f)-\Psi_b(\mu)(f)|^2\mid\eta_a^J\}
&\leq\frac{C(b-a+1)}J\operatorname{osc}(f)^2,\notag\\
|\E\{\Psi_b(\nu)(f)-\Psi_b(\mu)(f)\mid\eta_a^J\}|
&\leq
\frac{|\E\{(\nu-\mu)(d)\mid\eta_a^J\}|}{\underbar G}
+\frac{\|(\nu-\mu)(d)\|_{L^2(\cdot\mid\eta_a^J)}
\|(\nu-\mu)(G_b)\|_{L^2(\cdot\mid\eta_a^J)}}{\underbar G^2}
\notag\\
&\leq\frac{C(b-a+1)}J\operatorname{osc}(f).
\label{eq:lag-truncation-window-bayes-bias-s13}
\end{align}
Finally, conditional multinomial resampling gives
\begin{align}
\E\{m(\xi_b)(f)-\Psi_b(\nu)(f)\mid\xi_b^P\}&=0,\notag\\
\E\{|m(\xi_b)(f)-\Psi_b(\nu)(f)|^2\mid\xi_b^P\}
&\leq\operatorname{osc}(f)^2/J.
\label{eq:lag-truncation-terminal-resampling-s13}
\end{align}
Adding this centered term proves the same orders for the first brace
of Equation~(\ref{eq:lag-truncation-conditional-decomposition-s13}).

For the second brace, put
\begin{equation}
q=\frac{Q_{a,b}(1)}{\eta_aQ_{a,b}(1)},\qquad
d=\frac{Q_{a,b}\{f-\rho_b(f)\}}{\eta_aQ_{a,b}(1)}.
\end{equation}
These functions of the starting state satisfy
\begin{equation}
\eta_a(q)=1,\quad \eta_a(d)=0,\quad
R^{-1}\leq q\leq R,\quad \|d\|_\infty\leq R\operatorname{osc}(f).
\label{eq:lag-truncation-continuation-ratio-bounds-s13}
\end{equation}
Karjalainen et al.'s Lemma~2.3 applies to these bounded current
prediction functions. The prediction weak-bias bound follows from
Lemma~\ref{lem:lagwise-weak-error}, since for $a\geq2$
\begin{align}
\E\eta_a^J(u)-\eta_a(u)
&=\E m(\xi_{a-1})(M_a u)-\rho_{a-1}(M_a u),\notag\\
|\E\eta_a^J(u)-\eta_a(u)|&\leq C\operatorname{osc}(u)/J.
\end{align}
At initialization the corresponding empirical averages are independent
and unbiased. Applying these bounds to $d$ and $q$ yields
\begin{align}
\E\left|\frac{\eta_a^J(d)}{\eta_a^J(q)}\right|^2
&\leq R^2\E|\eta_a^J(d)|^2
\leq C\operatorname{osc}(f)^2/J,
\label{eq:lag-truncation-initial-cloud-l2-s13}\\
\left|\E\frac{\eta_a^J(d)}{\eta_a^J(q)}\right|
&=\left|\E\eta_a^J(d)
-\E\frac{\eta_a^J(d)\{\eta_a^J(q)-1\}}{\eta_a^J(q)}\right|
\notag\\
&\leq |\E\eta_a^J(d)|
+R\{\E|\eta_a^J(d)|^2\,\E|\eta_a^J(q)-1|^2\}^{1/2}
\leq C\operatorname{osc}(f)/J.
\label{eq:lag-truncation-initial-cloud-bias-s13}
\end{align}
Combine the two braces using $|u+v|^2\leq2|u|^2+2|v|^2$.
For $a=0$ and $b\geq1$,
the time-one prediction histories, including the independent initial
draw and its gradient when present, are unconditionally independent
with their common population prediction law. Apply the same
diagonal--off-diagonal calculation to these histories over the remaining
$b-1$ updates, followed by the displayed Bayes-ratio and final resampling
steps. The normalizer and marginal-forgetting bounds above still apply,
because future transitions depend on these histories only through $X_1$
under (A6), or the joint pair under (A6-J).
This gives bounds $Cb\operatorname{osc}(f)^2/J$ and
$Cb\operatorname{osc}(f)/J$. For $a=b=0$, the empirical average is
unbiased and its variance is at most $\operatorname{osc}(f)^2/J$.
For $J<c_0(b-a+1)$ the pathwise bound $2\|f\|_\infty$ gives the
same orders after a constant shift of $f$ and enlargement of $C$.
For every $J\geq2$ we have therefore proved
\begin{align}
\|m(\xi_b)(f)-\rho_b(f)\|_{L^2}
&\leq C\operatorname{osc}(f)\sqrt{(b-a+1)/J},\notag\\
|\E m(\xi_b)(f)-\rho_b(f)|
&\leq C\operatorname{osc}(f)(b-a+1)/J.
\label{eq:finite-window-gradient-input-estimate}
\end{align}

\medskip\noindent\emph{Step 1c: sum over gradient ages.}\quad
Apply Equation~(\ref{eq:finite-window-gradient-input-estimate}) to
$f=h_{s,i}^{[r]}$ at $b=s+q$. Under (A6) start at
$a=\max(0,s-r-1)$; when $r=s$ include the independent initial draw.
The window contains at most $q+r+2$ times. Under (A6-J), use only
$r=0$ and start at $s$. Thus
\begin{align}
\|m(\xi_{s+q})(h_{s,i}^{[r]})-\rho_{s+q}(h_{s,i}^{[r]})\|_{L^2}
&\leq \frac{CG'(\theta)\kappa^r}{\sqrt J}\sqrt{q+r+2},\notag\\
|\E m(\xi_{s+q})(h_{s,i}^{[r]})-\rho_{s+q}(h_{s,i}^{[r]})|
&\leq \frac{CG'(\theta)\kappa^r}{J}(q+r+2).
\label{eq:gradient-age-particle-errors}
\end{align}
Under (A6-J), only the $r=0$ inequalities are used.
Since $q+r+2\leq(q+1)(r+2)$,
\begin{align}
\sum_{r=0}^s\kappa^r\sqrt{q+r+2}
&\leq\sqrt{q+1}\sum_{r=0}^\infty\kappa^r\sqrt{r+2}
\leq C_\kappa\sqrt{q+1},\notag\\
\sum_{r=0}^s\kappa^r(q+r+2)
&\leq \frac{q+2}{1-\kappa}+\frac{\kappa}{(1-\kappa)^2}
\leq C_\kappa(q+1).
\label{eq:gradient-age-summability}
\end{align}
Minkowski's inequality, the triangle inequality for expectations, and
Equation~(\ref{eq:measurement-gradient-age-sum}) give
\begin{align}
\E|m(\xi_{s+q})(h_{s,i})-\rho_{s+q}(h_{s,i})|^2
&\leq C(q+1)G'(\theta)^2/J,
\label{eq:fixed-lag-one-input-l2-s13}\\
|\E m(\xi_{s+q})(h_{s,i})-\rho_{s+q}(h_{s,i})|
&\leq C(q+1)G'(\theta)/J.
\label{eq:fixed-lag-one-input-bias-s13}
\end{align}
Taking $q=0$ proves the uniform $J^{-1}$ weak-bias bound for $h_s$.

\begin{samepage}
For the terminal block, put $s=N-q$ and apply Minkowski:
\begin{align}
&\left\|m(\xi_N)\left(\sum_{s=N-k}^N h_{s,i}\right)
-\rho_N\left(\sum_{s=N-k}^N h_{s,i}\right)\right\|_{L^2}
\notag\\*
&\quad\leq
\sum_{q=0}^k\|m(\xi_N)(h_{N-q,i})-\rho_N(h_{N-q,i})\|_{L^2}
\leq\frac{CG'(\theta)}{\sqrt J}\sum_{q=0}^k\sqrt{q+1}
\leq\frac{CG'(\theta)(k+1)^{3/2}}{\sqrt J}.
\end{align}
\end{samepage}
Consequently,
\begin{align}
\E\left|m(\xi_N)\left(\sum_{s=N-k}^N h_{s,i}\right)
-\rho_N\left(\sum_{s=N-k}^N h_{s,i}\right)\right|^2
&\leq C(k+1)^3G'(\theta)^2/J,
\label{eq:fixed-lag-terminal-l2-s13}\\
\left|\E m(\xi_N)\left(\sum_{s=N-k}^N h_{s,i}\right)
-\rho_N\left(\sum_{s=N-k}^N h_{s,i}\right)\right|
&\leq\frac{CG'(\theta)}J\sum_{q=0}^k(q+1)
\leq C(k+1)^2G'(\theta)/J.
\label{eq:fixed-lag-terminal-bias-s13}
\end{align}

\medskip\noindent\emph{Step 1d: covariance summation.}\quad
Fix $r$ and consider the interior terms
$m(\xi_{s+k})(h_{s,i}^{[r]})$, $1\leq s<N-k$.
Under (A6) these depend on the physical cloud at
$\max(0,s-r-1)$ and subsequent variables through $s+k$.
Under (A6-J) the $r=0$ terms depend on the joint cloud at $s$ and
the next $k$ updates. Write $a(u)$ for the mixing coefficient bound
in Equation~(\ref{eq:dmop-zero-centered-l4-and-mixing}), using the
physical cloud in the first case and the joint cloud in the second.
Equation~(\ref{eq:gradient-age-particle-errors}) and Cauchy--Schwarz give
\begin{align}
&|\Cov\{m(\xi_{s+k})(h_{s,i}^{[r]}),
m(\xi_{t+k})(h_{t,i}^{[r]})\}|
\leq \frac{CG'(\theta)^2\kappa^{2r}(k+r+2)}J.
\label{eq:lag-truncation-covariance-cauchy-s13}
\end{align}
For $t-s>k+r+2$, the prediction clouds separating the windows are
at least $t-s-k-r-2$ transitions apart, allowing one transition to
absorb the past endpoint resampling. The centered variables have $L^4$ norms at most
$2CG'(\theta)\kappa^r$ by their pathwise bounds, so Davydov's inequality gives
\begin{align}
&|\Cov\{m(\xi_{s+k})(h_{s,i}^{[r]}),
m(\xi_{t+k})(h_{t,i}^{[r]})\}|
\leq CG'(\theta)^2\kappa^{2r}a(t-s-k-r-2)^{1/2}.
\label{eq:lag-truncation-covariance-mixing-s13}
\end{align}
Under (A6-J), the larger gap $t-s-k-1$ is available; the weaker displayed
bound with $r=0$ suffices. Combining these inequalities gives
\begin{align}
&\sum_{t=s+1}^{N-k-1}
|\Cov\{m(\xi_{s+k})(h_{s,i}^{[r]}),
m(\xi_{t+k})(h_{t,i}^{[r]})\}|
\notag\\
&\quad\leq CG'(\theta)^2\kappa^{2r}
\left\{\frac{(k+r+2)^2}J+
\sum_{u=1}^\infty\min\left(\frac{k+r+2}J,a(u)^{1/2}\right)\right\}.
\label{eq:fixed-lag-covariance-s13}
\end{align}
Choose $u_0\leq C(1+\log J)^2$ so that the
geometric bound for $a(u_0)^{1/2}$ is at most $J^{-1}$.
The geometric tail has sum at most $C(1+\log J)/J$, and hence
\begin{equation}
\sum_{u=1}^\infty\min\left(\frac{k+r+2}J,a(u)^{1/2}\right)
\leq\frac{u_0(k+r+2)}J+\frac{C(1+\log J)}J
\leq\frac{C(k+r+2)(1+\log J)^2}J.
\label{eq:lag-truncation-mixing-sum-s13}
\end{equation}
Retaining this logarithmic factor, expanding the variance as the diagonal
sum plus twice the covariance sum gives
\begin{equation}
\Var\left\{\sum_{s=1}^{N-k-1}m(\xi_{s+k})(h_{s,i}^{[r]})\right\}
\lesssim \frac{NG'(\theta)^2\kappa^{2r}(k+r+2)^2(1+\log J)^2}J.
\end{equation}
Center each age sum and apply Minkowski:
\begin{align}
&\left\|\sum_{s=1}^{N-k-1}
\{m(\xi_{s+k})(h_{s,i})-\E m(\xi_{s+k})(h_{s,i})\}\right\|_{L^2}
\notag\\
&\quad\leq
\sum_{r=0}^{N}\left\|\sum_{s=1}^{N-k-1}
\{m(\xi_{s+k})(h_{s,i}^{[r]})-\E m(\xi_{s+k})(h_{s,i}^{[r]})\}
\right\|_{L^2}
\notag\\
&\quad\lesssim
G'(\theta)(1+\log J)\sqrt{N/J}\sum_{r=0}^{N}\kappa^r(k+r+2)
\leq C_\kappa G'(\theta)(1+\log J)\sqrt{N/J}(k+1).
\label{eq:fixed-lag-interior-variance-s13}
\end{align}
Again the joint alternative uses only $r=0$.
The terminal variance is bounded by its second moment in
Equation~(\ref{eq:fixed-lag-terminal-l2-s13}); since $k+1\leq N$,
\begin{equation}
\Var\left\{m(\xi_N)\left(\sum_{s=N-k}^N h_{s,i}\right)\right\}
\leq C(k+1)^3G'(\theta)^2/J
\leq CN(k+1)^2G'(\theta)^2/J.
\label{eq:fixed-lag-terminal-variance-s13}
\end{equation}
Use $\Var(A+B)\leq2\Var(A)+2\Var(B)$ in
Equation~(\ref{eq:fixed-lag-score-telescope-s13}) and sum coordinates:
\begin{equation}
\Var\{\nabla_\theta\hat s_k^1(\theta)\}
\lesssim NpG'(\theta)^2(k+1)^2(1+\log J)^2/J.
\label{eq:fixed-lag-variance-s13}
\end{equation}
Finally, the one-input and terminal bias bounds yield
\begin{align}
\|\E\nabla_\theta\hat s_k^1(\theta)-\nabla_\theta s_k^1(\theta)\|_2
&\leq\frac{C\sqrt p\,G'(\theta)}J
\{(N-k-1)(k+1)+(k+1)^2\}
\notag\\
&=\frac{CN\sqrt p\,G'(\theta)(k+1)}J.
\label{eq:fixed-lag-particle-bias-s13}
\end{align}

\medskip\noindent\emph{Step 2: geometric mixture.}\quad
For $0\leq k<N-1$, put
\begin{equation}
\beta_k=(1-\alpha)\alpha^k,
\qquad
\beta_{N-1}=\alpha^{N-1}.
\label{eq:geometric-mixture-weights-s13}
\end{equation}
These nonnegative numbers sum to one.  Moreover, for every
$0\leq r<N$,
\begin{align}
\sum_{k=r}^{N-1}\beta_k
&=
\sum_{k=r}^{N-2}(1-\alpha)\alpha^k+\alpha^{N-1}             \notag\\
&=
\alpha^r(1-\alpha^{N-1-r})+\alpha^{N-1}                    \notag\\
&=\alpha^r.
\label{eq:geometric-mixture-tail-s13}
\end{align}
Equations~(\ref{eq:conditional-score-path-decomposition}),
(\ref{eq:population-conditional-score-path-decomposition}), and
(\ref{eq:fixed-lag-score-definition-s13}) therefore give the pathwise
identities
\begin{align}
\nabla_\theta\hat\ell^\alpha(\theta)
&=
\sum_{k=0}^{N-1}\beta_k\nabla_\theta\hat s_k^1(\theta),
&
\nabla_\theta\ell^\alpha(\theta)
&=
\sum_{k=0}^{N-1}\beta_k\nabla_\theta s_k^1(\theta).
\label{eq:dmop-alpha-fixed-lag-mixture-s13}
\end{align}
The coefficient of each $\Delta_{n,r}^J$ or $\Delta_{n,r}$ is the tail
sum in Equation~(\ref{eq:geometric-mixture-tail-s13}).
For $N\geq2$, the mixture weight on $\nabla_\theta\hat s_0^1(\theta)$
is $\beta_0=1-\alpha$. Every truncated score contains $\Delta_{n,0}^J$,
so its coefficient is
\begin{equation}
\sum_{k=0}^{N-1}\beta_k=1.
\label{eq:geometric-mixture-zero-lag-coefficient-s13}
\end{equation}
In particular, at $\alpha=0$,
\begin{equation}
\nabla_\theta\hat\ell_n^0(\theta)=\Delta_{n,0}^J,
\qquad
\nabla_\theta\ell_n^0(\theta)=\Delta_{n,0}.
\label{eq:dmop-zero-exact-lag-coefficient-s13}
\end{equation}
For $N=1$, the only mixture weight is $\beta_0=1$ for every $\alpha$.

The mean lag in this censored geometric mixture satisfies
\begin{align}
\sum_{k=0}^{N-1}\beta_k(k+1)
&=
\sum_{r=0}^{N-1}\sum_{k=r}^{N-1}\beta_k                  \notag\\
&=
\sum_{r=0}^{N-1}\alpha^r.
\label{eq:geometric-mixture-mean-lag-s13}
\end{align}
Center the first identity in Equation~
(\ref{eq:dmop-alpha-fixed-lag-mixture-s13}).  Minkowski's inequality,
Equation~(\ref{eq:fixed-lag-variance-s13}), and Equation~
(\ref{eq:geometric-mixture-mean-lag-s13}) give
\begin{align}
\left\|
\nabla_\theta\hat\ell^\alpha(\theta)
-\E\nabla_\theta\hat\ell^\alpha(\theta)
\right\|_{L^2}
&\leq
\sum_{k=0}^{N-1}\beta_k
\left\|
\nabla_\theta\hat s_k^1(\theta)
-\E\nabla_\theta\hat s_k^1(\theta)
\right\|_{L^2}                                             \notag\\
&\lesssim
C(1+\log J)\sqrt{\frac{NpG'(\theta)^2}{J}}
\sum_{k=0}^{N-1}\beta_k(k+1)                              \notag\\
&\lesssim
C(1+\log J)\sqrt{\frac{NpG'(\theta)^2}{J}}
\sum_{r=0}^{N-1}\alpha^r.
\label{eq:geometric-mixture-variance-root-s13}
\end{align}
Squaring proves Equation~(\ref{eq:refined-dmop-alpha-variance}).
Taking expectations in Equation~
(\ref{eq:dmop-alpha-fixed-lag-mixture-s13}), applying the triangle
inequality and Equation~(\ref{eq:fixed-lag-particle-bias-s13}), and
then using Equation~(\ref{eq:geometric-mixture-mean-lag-s13}) give
\begin{align}
\left\|
\E\nabla_\theta\hat\ell^\alpha(\theta)
-\nabla_\theta\ell^\alpha(\theta)
\right\|_2
&\leq
\sum_{k=0}^{N-1}\beta_k
\left\|
\E\nabla_\theta\hat s_k^1(\theta)
-\nabla_\theta s_k^1(\theta)
\right\|_2                                                \notag\\
&\leq
\frac{CN\sqrt p\,G'(\theta)}{J}
\sum_{k=0}^{N-1}\beta_k(k+1)                              \notag\\
&=
\frac{CN\sqrt p\,G'(\theta)}{J}
\sum_{r=0}^{N-1}\alpha^r,
\end{align}
which proves Equation~(\ref{eq:refined-dmop-alpha-particle-bias}).
Finally, for a random vector $Z$ and deterministic $z$,
\begin{equation}
\E\|Z-z\|_2^2
=\Var(Z)+\|\E Z-z\|_2^2.
\end{equation}
Apply this identity with
$Z=\nabla_\theta\hat\ell^\alpha(\theta)$ and
$z=\nabla_\theta\ell^\alpha(\theta)$, and substitute Equations~
(\ref{eq:refined-dmop-alpha-variance}) and
(\ref{eq:refined-dmop-alpha-particle-bias}).  This proves Equation~
(\ref{eq:refined-dmop-alpha-direct-particle-error}).
\end{proof}

\begin{cor}[Particle bias of the DMOP-$\alpha$ score]
\label{lem:dmop-full-particle-bias}
Under the conditions of Lemma~\ref{lem:dmop-alpha-variance},
Equation~(\ref{eq:refined-dmop-alpha-particle-bias}) holds.
\end{cor}

\begin{proof}[Proof.]
This is the second conclusion of Lemma~
\ref{lem:dmop-alpha-variance}.
\end{proof}

\begin{lem}[Population discounting bias]
\label{lem:dmop-population-discount-bias}
Under the conditions of Lemma~\ref{lem:derivative-measure-forgetting},
suppose $\theta=\phi$ and $\alpha\in[0,1]$. Then
\begin{align}
\left\|
\nabla_\theta\ell_n^\alpha(\theta)
-\nabla_\theta\ell_n(\theta)
\right\|_2
&\leq
C\sqrt p\,G'(\theta)
\sum_{r=1}^{n-1}(1-\alpha^r)\varrho^r                       \notag\\
&\leq
C\sqrt p\,G'(\theta)
\frac{\varrho(1-\alpha)}
{(1-\varrho)(1-\alpha\varrho)}.
\label{eq:dmop-population-discount-bias}
\end{align}
For the full score,
\begin{align}
&\left\|
\nabla_\theta\ell(\theta)
-\nabla_\theta\ell^\alpha(\theta)
\right\|_2                                                     \notag\\
&\quad\leq
C\sqrt p\,G'(\theta)\varrho(1-\alpha)
\sum_{q=0}^{N-2}(N-1-q)(\alpha\varrho)^q                    \notag\\
&\quad=
C\sqrt p\,G'(\theta)\varrho(1-\alpha)
\frac{(N-1)-N\alpha\varrho+(\alpha\varrho)^N}
{(1-\alpha\varrho)^2}                                       \notag\\
&\quad\leq
C(N-1)\sqrt p\,G'(\theta)
\frac{\varrho(1-\alpha)}{1-\alpha\varrho}.
\label{eq:dmop-full-population-discount-bias}
\end{align}
For $N=1$, the finite sum is empty, the displayed fraction is zero, and
the full-score difference is zero.
The dependence on the transition bounds can be sharpened further.
With the existing $\varrho=1-\underbar M/(2\bar M)$ from Equation~
(\ref{eq:population-forgetting-varrho}), for $\tfrac12<\varrho<1$,
\begin{align}
&\left\|\nabla_\theta\ell(\theta)
-\nabla_\theta\ell^\alpha(\theta)\right\|_2                    \notag\\
&\quad\leq
2\sqrt p\,G'(\theta)
\frac{1-\sqrt{2(1-\varrho)}}{1+\sqrt{2(1-\varrho)}}(1-\alpha)
\sum_{q=0}^{N-2}(N-1-q)
\left\{\alpha\frac{2\varrho-1}{3-2\varrho}\right\}^{q}       \notag\\
&\quad\leq
2(N-1)\sqrt p\,G'(\theta)
\frac{1-\sqrt{2(1-\varrho)}}{1+\sqrt{2(1-\varrho)}}
\frac{(1-\alpha)(3-2\varrho)}
{3-2\varrho-\alpha(2\varrho-1)}.
\label{eq:population-discount-sharp-transition-bound}
\end{align}
If every transition is independent of its starting location, in
particular if $\varrho=\tfrac12$, the population discounting bias is zero
for every $N$ and $\alpha$.
\end{lem}

\begin{proof}[Proof.]
\medskip\noindent\emph{Step 1: exact lagwise difference.}\quad
Equation~(\ref{eq:population-conditional-score-path-decomposition}) at
$\alpha$ and at $\alpha=1$ gives
\begin{align}
\nabla_\theta\ell_n(\theta)
-\nabla_\theta\ell_n^\alpha(\theta)
&=
\sum_{r=0}^{n-1}(1-\alpha^r)\Delta_{n,r}                     \notag\\
&=
\sum_{r=1}^{n-1}(1-\alpha^r)\Delta_{n,r},
\label{eq:population-discount-lag-difference}
\end{align}
The coefficient of $\Delta_{n,0}$ is one in both scores, so its
difference is zero.  In particular, at the left endpoint Equation~
(\ref{eq:dmop-zero-exact-lag-coefficient-s13}) gives
\begin{equation}
\nabla_\theta\ell_n^0(\theta)=\Delta_{n,0},
\qquad
\nabla_\theta\ell_n(\theta)-\nabla_\theta\ell_n^0(\theta)
=\sum_{r=1}^{n-1}\Delta_{n,r}.
\label{eq:population-discount-alpha-zero-endpoint}
\end{equation}
Thus DMOP-$0$ retains the zero-lag contribution with coefficient one and
retains no positive-lag contribution.

\medskip\noindent\emph{Step 2: apply forgetting and sum.}\quad
Taking norms, applying the triangle
inequality and Lemma~\ref{lem:derivative-measure-forgetting}, and then
extending the nonnegative scalar sum to infinity give
\begin{align}
\left\|
\nabla_\theta\ell_n(\theta)
-\nabla_\theta\ell_n^\alpha(\theta)
\right\|_2
&\leq
\sum_{r=1}^{n-1}(1-\alpha^r)\|\Delta_{n,r}\|_2              \notag\\
&\leq
C\sqrt p\,G'(\theta)
\sum_{r=1}^{n-1}(1-\alpha^r)\varrho^r                       \notag\\
&\leq
C\sqrt p\,G'(\theta)
\sum_{r=1}^{\infty}(1-\alpha^r)\varrho^r                    \notag\\
&=
C\sqrt p\,G'(\theta)
\left\{
\sum_{r=1}^{\infty}\varrho^r
-\sum_{r=1}^{\infty}(\alpha\varrho)^r
\right\}                                                     \notag\\
&=
C\sqrt p\,G'(\theta)
\left\{
\frac{\varrho}{1-\varrho}
-\frac{\alpha\varrho}{1-\alpha\varrho}
\right\}                                                     \notag\\
&=
C\sqrt p\,G'(\theta)
\frac{\varrho(1-\alpha)}
{(1-\varrho)(1-\alpha\varrho)}.
\end{align}

\medskip\noindent\emph{Step 3: exact summation by parts for the full
score.}\quad
Reindexing Equation~
(\ref{eq:population-conditional-score-path-decomposition}) by
$s=n-r$ shows that the contribution of $h_s$ to the full
DMOP-$\alpha$ population score is
\begin{align}
&\rho_s(h_s)
+\sum_{q=1}^{N-s}\alpha^q
\{\rho_{s+q}(h_s)-\rho_{s+q-1}(h_s)\}                       \notag\\
&\quad=
\rho_s(h_s)
+\sum_{q=1}^{N-s}\alpha^q\rho_{s+q}(h_s)
-\sum_{q=0}^{N-s-1}\alpha^{q+1}\rho_{s+q}(h_s)             \notag\\
&\quad=
(1-\alpha)\sum_{q=0}^{N-s-1}\alpha^q\rho_{s+q}(h_s)
+\alpha^{N-s}\rho_N(h_s).
\label{eq:population-discount-summation-by-parts}
\end{align}
The contribution of $h_s$ to the exact score in Equation~
(\ref{eq:full-population-score-path-average}) is $\rho_N(h_s)$.
Subtracting Equation~
(\ref{eq:population-discount-summation-by-parts}) from this contribution
gives
\begin{align}
&\rho_N(h_s)
-\rho_s(h_s)
-\sum_{q=1}^{N-s}\alpha^q
\{\rho_{s+q}(h_s)-\rho_{s+q-1}(h_s)\}                       \notag\\
&\quad=
(1-\alpha)\sum_{q=0}^{N-s-1}\alpha^q
\{\rho_N(h_s)-\rho_{s+q}(h_s)\}.
\label{eq:population-discount-one-time-boundary-comparison}
\end{align}
Summing Equation~
(\ref{eq:population-discount-one-time-boundary-comparison}) over
$s=1,\ldots,N-1$ yields the exact identity
\begin{align}
\nabla_\theta\ell(\theta)-\nabla_\theta\ell^\alpha(\theta)
&=
(1-\alpha)\sum_{s=1}^{N-1}\sum_{q=0}^{N-s-1}\alpha^q
\{\rho_N(h_s)-\rho_{s+q}(h_s)\}.
\label{eq:population-discount-full-boundary-comparison}
\end{align}

\medskip\noindent\emph{Step 4: contract each partial-smoother
comparison.}\quad
Fix $s\leq N-1$, $0\leq q\leq N-s-1$, and a coordinate
$i\in\{1,\ldots,p\}$.  The Markov property implies that observations
after time $s+q$ are conditionally independent of the path through time
$s$ given $X_{s+q}$.  Hence the following two conditional expectations
are the same function of $x$:
\begin{align}
&\E_{\rho_N}\{h_{s,i}\mid X_{s+q}=x\}
=\E_{\rho_{s+q}}\{h_{s,i}\mid X_{s+q}=x\}                  \notag\\
&\qquad=
\begin{cases}
\bar h_{s,i}(x),&q=0,\\
(B_{s+q}B_{s+q-1}\cdots B_{s+1}\bar h_{s,i})(x),&q\geq1.
\end{cases}
\label{eq:population-discount-common-conditional-past}
\end{align}
Here $\bar h_{s,i}$ and $B_t$ are already defined in Equations~
(\ref{eq:population-score-current-projection}) and
(\ref{eq:posterior-backward-kernel}).  We first sharpen the bound on the
change of a probability expectation under a positive likelihood.

For a probability measure $\mu$, a positive function $L$ with
$0<a\leq L\leq b$, and a bounded real function $f$, suppose first that
$a<b$.  Convexity of the absolute value gives, pointwise,
\begin{equation}
|L-\mu(L)|
\leq
\frac{b-L}{b-a}\{\mu(L)-a\}
+\frac{L-a}{b-a}\{b-\mu(L)\}.
\label{eq:population-discount-likelihood-chord}
\end{equation}
Subtract the midpoint of the range of $f$ inside
$\mu[f\{L-\mu(L)\}]$, and then integrate Equation~
(\ref{eq:population-discount-likelihood-chord}).  Since
$\mu\{L-\mu(L)\}=0$, this gives
\begin{align}
\left|\frac{\mu(fL)}{\mu(L)}-\mu(f)\right|
&\leq
\frac{\operatorname{osc}(f)}{2\mu(L)}
\mu|L-\mu(L)|                                               \notag\\*
&\leq
\operatorname{osc}(f)
\frac{\{b-\mu(L)\}\{\mu(L)-a\}}{\mu(L)(b-a)}                 \notag\\*
&=
\operatorname{osc}(f)
\frac{a+b-\mu(L)-ab/\mu(L)}{b-a}                            \notag\\*
&\leq
\operatorname{osc}(f)\frac{\sqrt b-\sqrt a}{\sqrt b+\sqrt a}.
\label{eq:population-discount-bayes-range-bound}
\end{align}
The last inequality is $\mu(L)+ab/\mu(L)\geq2\sqrt{ab}$.
When $a=b$, $L$ is constant and the difference is zero.

Equation~(\ref{eq:population-forgetting-varrho}) gives
$\underbar M/\bar M=2(1-\varrho)$ and $\tfrac12\leq\varrho<1$.
For $\tfrac12<\varrho<1$, apply Equation~
(\ref{eq:population-discount-bayes-range-bound}) to the two backward
kernels at $x$ and $x'$.  In this application set
\begin{equation}
L(u)=
\frac{f_{X_t\mid X_{t-1}}(x\mid u;\theta)}
{f_{X_t\mid X_{t-1}}(x'\mid u;\theta)}.
\end{equation}
Equation~(\ref{eq:posterior-backward-kernel}) and (A2) give
\begin{equation}
2(1-\varrho)\leq L(u)\leq\{2(1-\varrho)\}^{-1},
\qquad
B_t(f)(x)=\frac{B_t(fL)(x')}{B_t(L)(x')}.
\end{equation}
Taking $\mu=B_t(x',\cdot)$, $a=2(1-\varrho)$, and $b=\{2(1-\varrho)\}^{-1}$
in Equation~(\ref{eq:population-discount-bayes-range-bound}) therefore
gives
\begin{equation}
\operatorname{osc}\{B_t(f)\}
\leq\frac{2\varrho-1}{3-2\varrho}\operatorname{osc}(f).
\label{eq:population-discount-sharp-backward-contraction}
\end{equation}
Applying this inequality to the $q$ kernels in Equation~
(\ref{eq:population-discount-common-conditional-past}), starting from
$\operatorname{osc}(\bar h_{s,i})\leq2G'(\theta)$, yields
\begin{align}
\operatorname{osc}\!\left[
x\longmapsto\E_{\rho_N}\{h_{s,i}\mid X_{s+q}=x\}
\right]
&\leq
2G'(\theta)\left(\frac{2\varrho-1}{3-2\varrho}\right)^q
\leq2G'(\theta)\varrho^q.
\label{eq:population-discount-partial-smoother-oscillation}
\end{align}
This also covers $q=0$, when the product contains no backward kernel.

Next, the future likelihood already defined in Equation~
(\ref{eq:population-score-future-likelihood}) satisfies, since $s+q<N$,
\begin{align}
L_{s+q,N}(x)
&=\int M_{s+q+1}(x,du)G_{s+q+1}(u)L_{s+q+1,N}(u),          \notag\\
L_{s+q,N}(x)&\geq2(1-\varrho) L_{s+q,N}(x')
\qquad\text{for every }x,x'.
\label{eq:population-discount-future-likelihood-ratio}
\end{align}
The second line follows by applying
$M_{s+q+1}(x,\cdot)\geq2(1-\varrho) M_{s+q+1}(x',\cdot)$ to the
nonnegative integrand.  Thus its supremum divided by its infimum is
at most $\{2(1-\varrho)\}^{-1}$.  Apply Equation~
(\ref{eq:population-discount-bayes-range-bound}) with $\mu=\pi_{s+q}$
and $L=L_{s+q,N}$ to obtain
\begin{align}
\left|
\frac{\pi_{s+q}(fL_{s+q,N})}{\pi_{s+q}(L_{s+q,N})}
-\pi_{s+q}(f)
\right|
&\leq
\frac{1-\sqrt{2(1-\varrho)}}{1+\sqrt{2(1-\varrho)}}\operatorname{osc}(f).
\label{eq:population-discount-future-tilt-contraction}
\end{align}
The tower property and the common conditional expectation in Equation~
(\ref{eq:population-discount-common-conditional-past}) identify the
two quantities being compared:
\begin{align}
\rho_N(h_{s,i})
&=
\frac{
\rho_{s+q}\!\left[
\E_{\rho_{s+q}}\{h_{s,i}\mid X_{s+q}\}L_{s+q,N}(X_{s+q})
\right]}
{\rho_{s+q}\{L_{s+q,N}(X_{s+q})\}},                         \notag\\*
\rho_{s+q}(h_{s,i})
&=\rho_{s+q}\!\left[\E_{\rho_{s+q}}\{h_{s,i}\mid X_{s+q}\}\right].
\label{eq:population-discount-two-boundary-expectations}
\end{align}
Consequently, Equations~
(\ref{eq:population-discount-partial-smoother-oscillation})--
(\ref{eq:population-discount-two-boundary-expectations}) give
\begin{align}
|\rho_N(h_{s,i})-\rho_{s+q}(h_{s,i})|
&\leq
2G'(\theta)\frac{1-\sqrt{2(1-\varrho)}}{1+\sqrt{2(1-\varrho)}}
\left(\frac{2\varrho-1}{3-2\varrho}\right)^q                \notag\\
&\leq2G'(\theta)\varrho^{q+1}.
\label{eq:population-discount-coordinate-boundary-comparison}
\end{align}
For the second inequality, each of
$(1-\sqrt{2(1-\varrho)})/(1+\sqrt{2(1-\varrho)})$ and
$(2\varrho-1)/(3-2\varrho)$ is at most $2\varrho-1\leq\varrho$.
Summing the squared coordinate bounds and taking square roots gives
\begin{align}
\|\rho_N(h_s)-\rho_{s+q}(h_s)\|_2
&\leq
2\sqrt p\,G'(\theta)\frac{1-\sqrt{2(1-\varrho)}}{1+\sqrt{2(1-\varrho)}}
\left(\frac{2\varrho-1}{3-2\varrho}\right)^q                \notag\\
&\leq2\sqrt p\,G'(\theta)\varrho^{q+1}.
\label{eq:population-discount-vector-boundary-comparison}
\end{align}
If every transition is independent of its starting location, as (A2)
implies when $\varrho=\tfrac12$, then $L_{s+q,N}$ is constant, Equation~
(\ref{eq:population-discount-two-boundary-expectations}) has equal
right-hand sides, and the full bias is zero.

\medskip\noindent\emph{Step 5: sum the boundary comparisons.}\quad
Insert Equation~
(\ref{eq:population-discount-vector-boundary-comparison}) into Equation~
(\ref{eq:population-discount-full-boundary-comparison}), apply the
triangle inequality, and reverse the two finite sums.  This gives
\begin{align}
&\left\|
\nabla_\theta\ell(\theta)-\nabla_\theta\ell^\alpha(\theta)
\right\|_2                                                     \notag\\
&\quad\leq
2\sqrt p\,G'(\theta)\varrho(1-\alpha)
\sum_{s=1}^{N-1}\sum_{q=0}^{N-s-1}(\alpha\varrho)^q         \notag\\
&\quad=
2\sqrt p\,G'(\theta)\varrho(1-\alpha)
\sum_{q=0}^{N-2}(N-1-q)(\alpha\varrho)^q                   \notag\\
&\quad\leq
2(N-1)\sqrt p\,G'(\theta)\varrho(1-\alpha)
\sum_{q=0}^{\infty}(\alpha\varrho)^q                       \notag\\
&\quad=
2(N-1)\sqrt p\,G'(\theta)
\frac{\varrho(1-\alpha)}{1-\alpha\varrho}.
\label{eq:population-discount-boundary-comparison-bound}
\end{align}
Finally, multiplying the finite sum by $1-\alpha\varrho$ gives
\begin{align}
&(1-\alpha\varrho)
\sum_{q=0}^{N-2}(N-1-q)(\alpha\varrho)^q
\notag\\
&\qquad=
(N-1)-\sum_{q=1}^{N-1}(\alpha\varrho)^q                    \notag\\
&\qquad=
(N-1)-
\frac{\alpha\varrho\{1-(\alpha\varrho)^{N-1}\}}
{1-\alpha\varrho}.
\label{eq:population-discount-finite-boundary-sum-step}
\end{align}
Multiplying Equation~
(\ref{eq:population-discount-finite-boundary-sum-step}) once more by
$1-\alpha\varrho$ and collecting terms therefore gives
\begin{align}
\sum_{q=0}^{N-2}(N-1-q)(\alpha\varrho)^q
&=
\frac{(N-1)-N\alpha\varrho+(\alpha\varrho)^N}
{(1-\alpha\varrho)^2}.
\label{eq:population-discount-finite-boundary-sum}
\end{align}
Equations~(\ref{eq:population-discount-boundary-comparison-bound}) and
(\ref{eq:population-discount-finite-boundary-sum}) prove Equation~
(\ref{eq:dmop-full-population-discount-bias}).
To obtain the sharper bound, retain the first line of Equation~
(\ref{eq:population-discount-vector-boundary-comparison}) in the same
finite sum.  For $\tfrac12<\varrho<1$,
\begin{align}
&\left\|\nabla_\theta\ell(\theta)
-\nabla_\theta\ell^\alpha(\theta)\right\|_2                    \notag\\
&\quad\leq
2\sqrt p\,G'(\theta)
\frac{1-\sqrt{2(1-\varrho)}}{1+\sqrt{2(1-\varrho)}}(1-\alpha)
\sum_{q=0}^{N-2}(N-1-q)
\left\{\alpha\frac{2\varrho-1}{3-2\varrho}\right\}^{q}       \notag\\
&\quad\leq
2(N-1)\sqrt p\,G'(\theta)
\frac{1-\sqrt{2(1-\varrho)}}{1+\sqrt{2(1-\varrho)}}(1-\alpha)
\sum_{q=0}^{\infty}
\left\{\alpha\frac{2\varrho-1}{3-2\varrho}\right\}^{q}       \notag\\
&\quad=
2(N-1)\sqrt p\,G'(\theta)
\frac{1-\sqrt{2(1-\varrho)}}{1+\sqrt{2(1-\varrho)}}
\frac{(1-\alpha)(3-2\varrho)}
{3-2\varrho-\alpha(2\varrho-1)}.
\end{align}
The ratio in the geometric series is strictly less than one even at
$\alpha=1$.  This proves Equation~
(\ref{eq:population-discount-sharp-transition-bound}).
\end{proof}

\begin{remark}[Signed cancellation in the population bias]
\label{rem:population-bias-cancellation}
Retaining signs gives a first-order expansion of the population bias.
The score is a finite polynomial in $\alpha$, so
\begin{align}
\left.\frac{\partial}{\partial\alpha}
\nabla_\theta\ell^\alpha(\theta)\right|_{\alpha=1}
&=\sum_{n=1}^{N}\sum_{r=1}^{n-1}r\Delta_{n,r}                \notag\\
&=\sum_{s=1}^{N-1}\sum_{q=0}^{N-s-1}
\{\rho_N(h_s)-\rho_{s+q}(h_s)\}.
\label{eq:population-bias-signed-derivative}
\end{align}
The first equality follows by differentiating Equation~
(\ref{eq:population-conditional-score-path-decomposition}) and summing
over $n$.  For the second, divide Equation~
(\ref{eq:population-discount-full-boundary-comparison}) by
$1-\alpha$ and let $\alpha\uparrow1$; all sums are finite.
Subtracting $(1-\alpha)$ times Equation~
(\ref{eq:population-bias-signed-derivative}) from Equation~
(\ref{eq:population-discount-full-boundary-comparison}) gives exactly
\begin{align}
&\nabla_\theta\ell(\theta)-\nabla_\theta\ell^\alpha(\theta)
-(1-\alpha)
\left.\frac{\partial}{\partial\alpha}
\nabla_\theta\ell^\alpha(\theta)\right|_{\alpha=1}             \notag\\
&\quad=-(1-\alpha)\sum_{s=1}^{N-1}\sum_{q=0}^{N-s-1}
(1-\alpha^q)\{\rho_N(h_s)-\rho_{s+q}(h_s)\}.
\label{eq:population-bias-signed-remainder}
\end{align}
Thus Equation~(\ref{eq:population-discount-vector-boundary-comparison}),
the triangle inequality, and reversal of the finite sums give, for
$\tfrac12<\varrho<1$,
\begin{align}
&\left\|
\nabla_\theta\ell(\theta)-\nabla_\theta\ell^\alpha(\theta)
-(1-\alpha)
\left.\frac{\partial}{\partial\alpha}
\nabla_\theta\ell^\alpha(\theta)\right|_{\alpha=1}
\right\|_2                                                   \notag\\*
&\quad\leq
2\sqrt p\,G'(\theta)\frac{1-\sqrt{2(1-\varrho)}}{1+\sqrt{2(1-\varrho)}}
(1-\alpha)\sum_{q=0}^{N-2}(N-1-q)(1-\alpha^q)
\left(\frac{2\varrho-1}{3-2\varrho}\right)^q                \notag\\
&\quad\leq
2(N-1)\sqrt p\,G'(\theta)\frac{1-\sqrt{2(1-\varrho)}}{1+\sqrt{2(1-\varrho)}}
(1-\alpha)
\left\{\frac{3-2\varrho}{4(1-\varrho)}
-\frac{3-2\varrho}{3-2\varrho-\alpha(2\varrho-1)}\right\}     \notag\\
&\quad=
(N-1)\sqrt p\,G'(\theta)
\frac{(1-\sqrt{2(1-\varrho)})^2(3-2\varrho)(1-\alpha)^2}
{2(1-\varrho)\{3-2\varrho-\alpha(2\varrho-1)\}}.
\label{eq:population-bias-signed-remainder-bound}
\end{align}
The last equality uses
$2\varrho-1=(1-\sqrt{2(1-\varrho)})(1+\sqrt{2(1-\varrho)})$.
For $\varrho=\tfrac12$, the score difference and its signed derivative are
both zero, so the same bound holds.
In particular, the population bias is bounded by
\begin{align}
&\|\nabla_\theta\ell(\theta)
-\nabla_\theta\ell^\alpha(\theta)\|_2                         \notag\\
&\quad\leq
(1-\alpha)
\left\|\sum_{n=1}^{N}\sum_{r=1}^{n-1}r\Delta_{n,r}\right\|_2
+
(N-1)\sqrt p\,G'(\theta)
\frac{(1-\sqrt{2(1-\varrho)})^2(3-2\varrho)(1-\alpha)^2}
{2(1-\varrho)\{3-2\varrho-\alpha(2\varrho-1)\}}.
\label{eq:population-bias-cancellation-bound}
\end{align}
For every fixed observation sequence, if the signed sum vanishes, this
gives $O\{N(1-\alpha)^2\}$ bias and $O\{N^2(1-\alpha)^4\}$ squared bias,
with fixed model constants. The next example shows that the assumptions
do not ensure this cancellation.
\end{remark}

\begin{remark}[Sharpness of the fixed-observation population order]
The powers of $N$ and $1-\alpha$ in Equation~
(\ref{eq:dmop-full-population-discount-bias}) cannot be improved uniformly
under the conditions of Lemma~\ref{lem:dmop-population-discount-bias}.
To see this, take state space $\{0,1\}$ with counting measure,
\begin{align}
X_0&\sim {\rm Bernoulli}(3/4),
&\Pr(X_n=X_{n-1}\mid X_{n-1})&=\frac34,
\label{eq:population-bias-sharpness-transition}
\end{align}
fix $y_n^*=1$ for every $n$, and, for $|\theta|\leq1/4$, set
\begin{align}
g_\theta(1\mid0)&=\frac14,
&g_\theta(1\mid1)&=\frac9{20}(1+\theta),
\label{eq:population-bias-sharpness-measurement}
\end{align}
with the probabilities at observation zero defined by complementation.
At $\theta=\phi=0$,
\begin{equation}
h_n(x)=x.
\label{eq:population-bias-sharpness-score}
\end{equation}
If the filtering probability of state one at time $n-1$ is $3/4$, its
one-step prediction and update are
\begin{align}
\Pr(X_n=1\mid y_{1:n-1}^*)&=
\frac34\frac34+\frac14\frac14=\frac58,                    \notag\\
\Pr(X_n=1\mid y_{1:n}^*)&=
\frac{(5/8)(9/20)}{(5/8)(9/20)+(3/8)(1/4)}=\frac34.
\label{eq:population-bias-sharpness-filter}
\end{align}
The law of $X_0$ initializes this induction, so the filtered probability
is $3/4$ at every time.  The resulting time-homogeneous posterior
backward kernel satisfies
\begin{align}
B(1\mid1)&=
\frac{(3/4)(3/4)}{5/8}=\frac9{10},
&B(1\mid0)&=
\frac{(3/4)(1/4)}{3/8}=\frac12.
\label{eq:population-bias-sharpness-backward-kernel}
\end{align}
Moreover, the next-observation likelihoods conditional on $X_{n-1}$ are
\begin{align}
M_n(G_n)(1)&=
\frac34\frac9{20}+\frac14\frac14=\frac25,
&M_n(G_n)(0)&=
\frac14\frac9{20}+\frac34\frac14=\frac3{10}.
\label{eq:population-bias-sharpness-next-likelihood}
\end{align}
Consequently,
\begin{align}
\Pr_{\rho_n}(X_{n-1}=1)
&=
\frac{(3/4)(2/5)}{(3/4)(2/5)+(1/4)(3/10)}
=\frac45,                                                   \notag\\
\Pr_{\rho_{n-1}}(X_{n-1}=1)&=\frac34.
\label{eq:population-bias-sharpness-one-step-shift}
\end{align}
Equation~(\ref{eq:population-bias-sharpness-backward-kernel}) maps every
affine function of the state to an affine function whose slope is
multiplied by
\begin{equation}
B(1\mid1)-B(1\mid0)=\frac25.
\end{equation}
Applying the $r-1$ backward kernels between $X_{n-r}$ and $X_{n-1}$ and
then using Equation~
(\ref{eq:population-bias-sharpness-one-step-shift}) therefore gives, for
$1\leq r<n$,
\begin{align}
\Delta_{n,r}
&=
\rho_n(h_{n-r})-\rho_{n-1}(h_{n-r})                         \notag\\
&=
\left(\frac45-\frac34\right)
\left(\frac25\right)^{r-1}
=\frac1{20}\left(\frac25\right)^{r-1}.
\label{eq:population-bias-sharpness-lag}
\end{align}
All these quantities are nonnegative.  Equation~
(\ref{eq:population-discount-lag-difference}), summed over $n$, now
implies
\begin{align}
\nabla_\theta\ell(0)-\nabla_\theta\ell^\alpha(0)
&=
\frac1{20}\sum_{r=1}^{N-1}
(N-r)(1-\alpha^r)\left(\frac25\right)^{r-1}                \notag\\
&\geq
\frac{(N-1)(1-\alpha)}{20}.
\label{eq:population-bias-sharpness-full-score}
\end{align}
Thus the squared population bias is at least
$(N-1)^2(1-\alpha)^2/400$.  This model satisfies the uniform transition
and measurement bounds on $|\theta|\leq1/4$, its simulator is independent
of $\theta$, and $h_n$ is a bounded current function. A uniform improvement
therefore requires additional assumptions, such as signed cancellation.
\end{remark}

\begin{thm}[MSE bounds for DMOP-$\alpha$]
\label{thm:refined-dmop-alpha-mse}
Under the conditions of Lemmas~\ref{lem:dmop-alpha-variance} and
\ref{lem:dmop-population-discount-bias}, take the same
$\varrho=1-\underbar M/(2\bar M)\in[\tfrac12,1)$.
For $\alpha\in[0,1]$,
\begin{align}
&\E\left\|
\nabla_\theta\hat\ell^\alpha(\theta)
-\nabla_\theta\ell(\theta)
\right\|_2^2                                                \notag\\*
&\quad\lesssim pG'(\theta)^2
\begin{aligned}[t]
\Bigg[&
\left(\frac{N(1+\log J)^2}{J}+\frac{N^2}{J^2}\right)
\left(\sum_{r=0}^{N-1}\alpha^r\right)^2\\
&+\frac{(1-\sqrt{2(1-\varrho)})^2(1-\alpha)^2}
{(1+\sqrt{2(1-\varrho)})^2}
\left\{
N-1+\sum_{q=1}^{N-2}(N-1-q)
\left(\alpha\frac{2\varrho-1}{3-2\varrho}\right)^q
\right\}^2\Bigg].
\end{aligned}
\label{eq:verified-dmop-alpha-mse-finite-sum}
\end{align}
For $0\leq\alpha<1$, this implies
\begin{align}
&\E\left\|
\nabla_\theta\hat\ell^\alpha(\theta)
-\nabla_\theta\ell(\theta)
\right\|_2^2                                                \notag\\*
&\quad\lesssim pG'(\theta)^2
\begin{aligned}[t]
\Bigg[&
\frac{N(1+\log J)^2}{J(1-\alpha)^2}
+\frac{N^2}{J^2(1-\alpha)^2}\\
&+N^2\left\{
\frac{1-\sqrt{2(1-\varrho)}}{1+\sqrt{2(1-\varrho)}}
\frac{(1-\alpha)(3-2\varrho)}
{3-2\varrho-\alpha(2\varrho-1)}
\right\}^2\Bigg].
\end{aligned}
\label{eq:verified-dmop-alpha-mse-bound}
\end{align}
The finite sum beginning at $q=1$ is empty for $N\leq2$.
At $\varrho=\tfrac12$, the displayed population terms are zero.
In particular, the two endpoints satisfy
\begin{align}
&\E\left\|
\nabla_\theta\hat\ell^0(\theta)-\nabla_\theta\ell(\theta)
\right\|_2^2                                                \notag\\*
&\quad\lesssim pG'(\theta)^2
\left[
\frac{N(1+\log J)^2}{J}+\frac{N^2}{J^2}
+(N-1)^2\left\{
\frac{1-\sqrt{2(1-\varrho)}}{1+\sqrt{2(1-\varrho)}}
\right\}^2
\right],
\label{eq:verified-dmop-zero-mse-bound}
\\
&\E\left\|
\nabla_\theta\hat\ell^1(\theta)-\nabla_\theta\ell(\theta)
\right\|_2^2
\lesssim pG'(\theta)^2
\min\left\{N^2,\frac{N^3}{J}\right\}.
\label{eq:verified-dmop-one-mse-bound}
\end{align}
At $\alpha=0$, the variance is at most order $NpG'(\theta)^2(1+\log J)^2/J$;
at $\alpha=1$, the population discounting bias is exactly zero.
Remark~\ref{rem:population-bias-cancellation} gives a complementary
bound retaining signed cancellation.
\end{thm}

\begin{proof}[Proof.]
Add and subtract $\nabla_\theta\ell^\alpha(\theta)$ and use
$\|a+b\|_2^2\leq2\|a\|_2^2+2\|b\|_2^2$.  This gives
\begin{align}
&\E\left\|
\nabla_\theta\hat\ell^\alpha(\theta)
-\nabla_\theta\ell(\theta)
\right\|_2^2                                                \notag\\
&\quad\leq
2\E\left\|
\nabla_\theta\hat\ell^\alpha(\theta)
-\nabla_\theta\ell^\alpha(\theta)
\right\|_2^2
+2\left\|
\nabla_\theta\ell^\alpha(\theta)
-\nabla_\theta\ell(\theta)
\right\|_2^2.
\label{eq:generic-mse-particle-population-split}
\end{align}
Lemma~\ref{lem:dmop-alpha-variance} bounds the first term by
\begin{equation}
CpG'(\theta)^2
\left(\frac{N(1+\log J)^2}{J}+\frac{N^2}{J^2}\right)
\left(\sum_{r=0}^{N-1}\alpha^r\right)^2.
\end{equation}
The sharper bound in Equation~
(\ref{eq:population-discount-sharp-transition-bound}) of Lemma~
\ref{lem:dmop-population-discount-bias} gives
\begin{align}
&\left\|
\nabla_\theta\ell^\alpha(\theta)-\nabla_\theta\ell(\theta)
\right\|_2^2                                                \notag\\
&\quad\leq 4pG'(\theta)^2\,\frac{(1-\sqrt{2(1-\varrho)})^2(1-\alpha)^2}
{(1+\sqrt{2(1-\varrho)})^2}
\left\{
N-1+\sum_{q=1}^{N-2}(N-1-q)
\left(\alpha\frac{2\varrho-1}{3-2\varrho}\right)^q
\right\}^2.
\end{align}
Here the $q=0$ summand has been written separately as $N-1$.
For $\varrho=\tfrac12$, both sides of this last inequality are zero
by Lemma~\ref{lem:dmop-population-discount-bias}.
Substitution into Equation~
(\ref{eq:generic-mse-particle-population-split}), absorbing fixed
numerical factors into $\lesssim$, proves the first bound.
For $0\leq\alpha<1$,
\begin{equation}
\sum_{r=0}^{N-1}\alpha^r\leq\frac1{1-\alpha}.
\end{equation}
Also,
\begin{align}
&N-1+\sum_{q=1}^{N-2}(N-1-q)
\left(\alpha\frac{2\varrho-1}{3-2\varrho}\right)^q             \notag\\
&\quad\leq
(N-1)\left\{1+\sum_{q=1}^{\infty}
\left(\alpha\frac{2\varrho-1}{3-2\varrho}\right)^q\right\}     \notag\\
&\quad=
\frac{(N-1)(3-2\varrho)}{3-2\varrho-\alpha(2\varrho-1)}.
\end{align}
The geometric ratio is in $[0,1)$, including when
$\varrho=\tfrac12$.  These inequalities and $N-1\leq N$ prove
Equation~(\ref{eq:verified-dmop-alpha-mse-bound}).
Alternatively, Equation~(\ref{eq:population-bias-cancellation-bound})
may be substituted in Equation~
(\ref{eq:generic-mse-particle-population-split}), taking the smaller
population bound without changing the particle contribution.

For $\alpha=0$, the geometric sum is one and every summand with
$q\geq1$ is zero.  Substitution in Equation~
(\ref{eq:verified-dmop-alpha-mse-finite-sum}) gives Equation~
(\ref{eq:verified-dmop-zero-mse-bound}); Equation~
(\ref{eq:refined-dmop-alpha-variance}) gives the stated variance.
For $\alpha=1$, use the terminal-block estimate directly, without
separately bounding and squaring the particle bias.  Equation~
(\ref{eq:dmop-alpha-fixed-lag-mixture-s13}) and the two telescopes in
Equations~(\ref{eq:fixed-lag-score-telescope-s13})--
(\ref{eq:fixed-lag-population-telescope-s13}), with $k=N-1$, give
\begin{align}
\nabla_\theta\hat\ell^1(\theta)
&=m(\xi_N)\left(\sum_{s=1}^{N}h_s\right),
&\nabla_\theta\ell^1(\theta)
&=\rho_N\left(\sum_{s=1}^{N}h_s\right)
=\nabla_\theta\ell(\theta).
\end{align}
Consequently, Equation~(\ref{eq:fixed-lag-terminal-l2-s13}) implies
\begin{align}
&\E\left\|
\nabla_\theta\hat\ell^1(\theta)-\nabla_\theta\ell(\theta)
\right\|_2^2                                                \notag\\
&\quad=
\sum_{i=1}^{p}\E\left|
m(\xi_N)\left(\sum_{s=1}^{N}h_{s,i}\right)
-\rho_N\left(\sum_{s=1}^{N}h_{s,i}\right)
\right|^2
\leq\frac{CN^3pG'(\theta)^2}{J}.
\label{eq:dmop-one-terminal-direct-mse}
\end{align}
Since $\|h_s\|_2\leq\sqrt p\,G'(\theta)$, the same difference
satisfies, pathwise,
\begin{align}
\left\|
\nabla_\theta\hat\ell^1(\theta)-\nabla_\theta\ell(\theta)
\right\|_2
&\leq
m(\xi_N)\left(\left\|\sum_{s=1}^{N}h_s\right\|_2\right)
+\rho_N\left(\left\|\sum_{s=1}^{N}h_s\right\|_2\right)       \notag\\
&\leq 2N\sqrt p\,G'(\theta).
\end{align}
Squaring and combining with Equation~
(\ref{eq:dmop-one-terminal-direct-mse}) proves Equation~
(\ref{eq:verified-dmop-one-mse-bound}).
\end{proof}

\begin{remark}[Neutral-model check]
To check the rates in Lemma~\ref{lem:dmop-alpha-variance}, take
iid ${\rm Uniform}[0,1]$ states,
constant weights at $\theta=\phi=0$, and
\begin{equation}
\Pr_\theta(Y_n=1\mid X_n=x)
=\frac12+\frac\theta4(2x-1),
\qquad y_n^*=1.
\end{equation}
Then $h_n=(2X_n-1)/2$ has variance $\sigma^2=1/12$, and direct
conditional multinomial calculation gives
\begin{equation}
\Var\left\{\nabla_\theta\hat\ell^\alpha(\theta)\right\}
=
\frac{\sigma^2}{J}
\left[
N(2-J^{-1})
+\sum_{r=1}^{N-1}(N-r)\alpha^{2r}(1-J^{-1})^{r+1}
\right].
\label{eq:neutral-dmop-alpha-exact-variance}
\end{equation}
For fixed $\alpha<1$, the variance is of order $N/J$. Its genealogical
contribution is of order $N/\{J(1-\alpha^2)\}$ before horizon saturation,
so the general factor $(1-\alpha)^{-2}$ is loose in this model.
At $\alpha=1$ and $J/N\to\infty$, the same exact formula gives
\begin{equation}
\Var\left\{\nabla_\theta\hat\ell^1(\theta)\right\}
\sim\frac{\sigma^2}{J}
\left\{2N+\sum_{r=1}^{N-1}(N-r)\right\}
=\frac{\sigma^2N(N+3)}{2J}.
\end{equation}
Thus the general $N^3/J$ endpoint upper bound is also conservative
in this model.
\end{remark}

\begin{remark}[Rates and normalization]
The finite particle term in Lemma~\ref{lem:dmop-alpha-variance} is
\begin{equation}
pG'(\theta)^2
\left(\frac{N(1+\log J)^2}{J}+\frac{N^2}{J^2}\right)
\left(\sum_{r=0}^{N-1}\alpha^r\right)^2.
\label{eq:generic-alpha-particle-bound-summary}
\end{equation}
At $\alpha=0$, the particle-MSE upper bound is
$pG'(\theta)^2\{N(1+\log J)^2/J+N^2/J^2\}$. Its variance bound is
$NpG'(\theta)^2(1+\log J)^2/J$, whereas its MSE relative to the full score also
contains the population term displayed in Equation~
(\ref{eq:verified-dmop-zero-mse-bound}).  Every fixed $\alpha<1$
has an $N(1+\log J)^2/J$ particle-variance upper bound and, when $J\geq N$, the
same leading particle-MSE order, with constants depending on $\alpha$.
The particle bound increases with the effective memory
$\sum_{r=0}^{N-1}\alpha^r$.

The population bound in Theorem~\ref{thm:refined-dmop-alpha-mse}
decreases to zero as $\alpha\uparrow1$.
At $\alpha=1$, the terminal $L^2$ estimate gives
Equation~(\ref{eq:verified-dmop-one-mse-bound}) directly.
Substitution into Equation~(\ref{eq:generic-alpha-particle-bound-summary})
would instead give $N^3(1+\log J)^2/J+N^4/J^2$. The neutral-model example shows that
the $(1-\alpha)^{-2}$ dependence and the $N^3/J$ endpoint bound can be loose.

A particle MSE bound for $\operatorname{osc}(f)\leq1$ is multiplied by
$\operatorname{osc}(f)^2$ when it is applied to an unnormalized
functional.  Thus a generic path-space bound of order $N^2/J$ for an
oscillation-one functional becomes $N^4/J$ if it is applied once to a
full score whose oscillation is of order $N$.
Lemma~\ref{lem:dmop-alpha-variance} bounds each $h_s$ before summing
lags, so Theorem~\ref{thm:refined-dmop-alpha-mse} needs no further rescaling.
The resulting $N^4/J$ generic upper bound is separate from the quadratic
asymptotic-variance lower bound of \citet{poyiadjis11}.
\end{remark}

\begin{remark}[Choosing $\alpha$ in the upper bound]
\label{rem:alpha-tuning-upper-bound}
We tune the upper bound. For $0\leq\alpha<1$,
\begin{equation}
\sum_{r=0}^{N-1}\alpha^r
\leq\min\{N,(1-\alpha)^{-1}\}
\end{equation}
and Theorem~\ref{thm:refined-dmop-alpha-mse} give
\begin{align}
&\E\left\|
\nabla_\theta\hat\ell^\alpha(\theta)
-\nabla_\theta\ell(\theta)
\right\|_2^2                                                \notag\\*
&\quad\lesssim pG'(\theta)^2
\begin{aligned}[t]
\Bigg[&
\left(\frac{N(1+\log J)^2}{J}+\frac{N^2}{J^2}\right)
\min\{N^2,(1-\alpha)^{-2}\}\\
&+N^2\left\{
\frac{1-\sqrt{2(1-\varrho)}}{1+\sqrt{2(1-\varrho)}}
\frac{(1-\alpha)(3-2\varrho)}
{3-2\varrho-\alpha(2\varrho-1)}
\right\}^2\Bigg].
\end{aligned}
\label{eq:alpha-tuning-soft-lag-bound}
\end{align}
For fixed $\tfrac12<\varrho<1$, before the effective memory reaches
the horizon, balancing these two terms is equivalent to
\begin{equation}
\frac{(1+\log J)^2}{NJ}+\frac1{J^2}
\asymp
\left\{
\frac{1-\sqrt{2(1-\varrho)}}{1+\sqrt{2(1-\varrho)}}
\frac{3-2\varrho}{3-2\varrho-\alpha(2\varrho-1)}
\right\}^2(1-\alpha)^4.
\label{eq:alpha-tuning-varrho-balance}
\end{equation}
Suppressing these fixed model constants gives the following
rate-attaining choices for $N\geq2$:
\begin{equation}
\begin{array}{c|c|c}
\text{regime}&1-\alpha&\text{MSE upper bound}/\{pG'(\theta)^2\}\\
\hline
J(1+\log J)^2\lesssim N&\asymp J^{-1/2}&N^2/J\\[2pt]
N\lesssim J(1+\log J)^2\lesssim N^3
&\asymp\{(1+\log J)^2/(NJ)\}^{1/4}
&N^{3/2}(1+\log J)/J^{1/2}\\[2pt]
J(1+\log J)^2\gtrsim N^3&0&N^3/J.
\end{array}
\label{eq:alpha-tuning-soft-lag-regimes}
\end{equation}
The first boundary is where the two particle coefficients have the
same order. In the middle regime, compare the balanced bound with the
independent terminal bound:
\[
\min\left\{\frac{N^{3/2}(1+\log J)}{\sqrt J},\frac{N^3}{J}\right\}.
\]
They agree when $J(1+\log J)^2$ has order $N^3$.
The choice $\alpha=1$ attains the last rate by
Equation~(\ref{eq:verified-dmop-one-mse-bound}). To check minimization
of the finite bound in Equation~(\ref{eq:verified-dmop-alpha-mse-finite-sum}),
for $N\geq2$,
\begin{equation}
\left.\frac{\partial}{\partial\alpha}
\left\{
\left(\frac{N(1+\log J)^2}{J}+\frac{N^2}{J^2}\right)
\left(\sum_{r=0}^{N-1}\alpha^r\right)^2
\right\}\right|_{\alpha=1}
=\left(\frac{N(1+\log J)^2}{J}+\frac{N^2}{J^2}\right)N^2(N-1)>0,
\end{equation}
The squared population term has derivative zero there, so decreasing
$\alpha$ slightly from one reduces this finite upper bound.
This need not change the rate or reduce the actual MSE.
The first-row choice $J^{-1/2}$ is a rate for $1-\alpha$;
the corresponding MSE upper bound is $N^2/J$.

For fixed $\varrho\in(\tfrac12,1)$, its factors affect the tuning constants.
At $\varrho=\tfrac12$, the population term is zero and $\alpha=0$
minimizes the finite particle upper bound.  For $N=1$ all choices of
$\alpha$ give the same estimator.  Finally, if the signed sum in
Equation~(\ref{eq:population-bias-signed-derivative}) vanishes,
Remark~\ref{rem:population-bias-cancellation} gives, before horizon
saturation and with fixed model constants,
\begin{align}
1-\alpha&\asymp
\left(\frac{(1+\log J)^2}{NJ}+\frac1{J^2}\right)^{1/6},                 \notag\\
\E\left\|
\nabla_\theta\hat\ell^\alpha(\theta)-\nabla_\theta\ell(\theta)
\right\|_2^2
&\lesssim pG'(\theta)^2N^2
\left(\frac{(1+\log J)^2}{NJ}+\frac1{J^2}\right)^{2/3}.
\label{eq:alpha-tuning-signed-cancellation}
\end{align}
Indeed, the particle contribution is of order
$\{N(1+\log J)^2/J+N^2/J^2\}(1-\alpha)^{-2}$ and the squared population contribution
is then at most order $N^2(1-\alpha)^4$.
This cancellation is an additional condition for the last rate.
\end{remark}

\begin{remark}[Geometric discounting and lag truncation]
For $0\leq k<N$, the truncations defined in
Lemma~\ref{lem:dmop-alpha-variance} satisfy
\begin{align}
\nabla_\theta\hat s_k^1(\theta)
&=
\sum_{s=1}^{N-k-1}m(\xi_{s+k})(h_s)
+m(\xi_N)\left(\sum_{s=N-k}^{N}h_s\right),                \notag\\
\nabla_\theta\ell(\theta)-\nabla_\theta s_k^1(\theta)
&=
\sum_{s=1}^{N-k-1}\sum_{u=s+k+1}^{N}\Delta_{u,u-s}.
\end{align}
Equation~(\ref{eq:population-discount-vector-boundary-comparison}) with
$q=k$ gives, for $\tfrac12<\varrho<1$,
\begin{equation}
\left\|
\nabla_\theta\ell(\theta)-\nabla_\theta s_k^1(\theta)
\right\|_2
\leq
2(N-k-1)\sqrt p\,G'(\theta)
\frac{1-\sqrt{2(1-\varrho)}}{1+\sqrt{2(1-\varrho)}}
\left(\frac{2\varrho-1}{3-2\varrho}\right)^k.
\end{equation}
At $\varrho=\tfrac12$ the population error is zero for every $k$;
at $k=N-1$ it is zero for every $\varrho$.
The finite-window particle calculation in Equations~
(\ref{eq:fixed-lag-variance-s13}) and
(\ref{eq:fixed-lag-particle-bias-s13}) gives
\begin{align}
\Var\{\nabla_\theta\hat s_k^1(\theta)\}
&\lesssim
\frac{NpG'(\theta)^2(k+1)^2(1+\log J)^2}{J},                           \notag\\
\left\|
\E\nabla_\theta\hat s_k^1(\theta)
-\nabla_\theta s_k^1(\theta)
\right\|_2
&\lesssim
\frac{N\sqrt p\,G'(\theta)(k+1)}{J}.
\end{align}
Equation~(\ref{eq:dmop-alpha-fixed-lag-mixture-s13}) writes DMOP-$\alpha$
as their geometric average. Its effective memory is
\begin{equation}
\sum_{k=0}^{N-1}\beta_k(k+1)
=\sum_{r=0}^{N-1}\alpha^r
\end{equation}
Thus every fixed $\alpha<1$ has an $N(1+\log J)^2/J$ particle-variance upper bound,
including the DMOP-$0$ case.
\end{remark}

\section{Implementation and Figures for Bayesian Inference}
\label{appendix:bayes}

\subsection{A common sampling target}
\label{sec:bayes-common-target}

For a comparison with a common posterior, we implemented NUTS and particle
Metropolis--Hastings (PMMH) using the Dhaka model in Pypomp.
The process simulator uses Gaussian noise with 20 Euler steps per monthly
observation. We fix the initial state proportions at the library values,
with $\rho=0$, $c=1$, $\alpha=1$ and $\delta=0.02$.
Here $\alpha$ is a model parameter, distinct from the DMOP discount parameter.
The observation distribution is Gaussian, with mean equal to the simulated
monthly deaths and standard deviation equal to $\tau$ times that mean.
We omit the additive measurement-density floor used in the library's
optimization code. Paths that leave the nonnegative state space are assigned
zero likelihood for the observed series.

Let $x$ denote the 18 estimation coordinates. These comprise the logarithms
of the five positive parameters $\gamma,m,\epsilon,\sigma,\tau$, the trend
coordinate $100\beta_{\mathrm{trend}}$, and the six coefficients in each
seasonal function. The unused constant environmental-transmission parameter
in the earlier implementation is excluded.
We use independent Gaussian priors for these coordinates, common to all
chains and both methods. For each coordinate, the mean is the midpoint
of the interval in Table~\ref{tab:bayes-prior}, and the standard deviation
is half its width. These intervals come from the parameter search ranges;
they do not truncate the Gaussian distributions.
Thus, writing $a_i$ and $b_i$ for the endpoints,
\[
 z_i=\frac{x_i-(a_i+b_i)/2}{(b_i-a_i)/2},
 \qquad p(z)=\prod_{i=1}^{18}\phi(z_i;0,1).
\]
This specifies a reference prior without fitting a density to IF2 particles.
The prior density is evaluated in $z$, so the sampling calculation requires
no further Jacobian factor. On the natural parameter scale, the priors for
the five positive parameters are lognormal.

\begin{table}[htbp]
\centering
\begin{tabular}{lll}
\toprule
Estimation coordinate & $a_i$ & $b_i$ \\
\midrule
$\log\gamma$ & $\log 10$ & $\log 40$ \\
$\log m$ & $\log 0.03$ & $\log 0.6$ \\
$\log\epsilon$ & $\log 0.2$ & $\log 30$ \\
$100\beta_{\mathrm{trend}}$ & $-1$ & $0$ \\
$\log\sigma$ & $\log 1$ & $\log 5$ \\
$\log\tau$ & $\log 0.1$ & $\log 0.5$ \\
Seasonal transmission coefficients 1 and 3 & $-4$ & $4$ \\
Seasonal transmission coefficients 2, 4, 5 and 6 & $0$ & $8$ \\
Each seasonal environmental-transmission coefficient & $-10$ & $0$ \\
\bottomrule
\end{tabular}
\caption{Intervals specifying the means and standard deviations of the
common Gaussian prior on the estimation scale. A sensitivity calculation
doubles all prior standard deviations, keeping the means fixed.}
\label{tab:bayes-prior}
\end{table}

Let $u$ denote the particle-filter random numbers, with law $m$, and let
$\widehat L_J(z,u)$ be the bootstrap particle-filter likelihood estimate.
The augmented target is
\begin{equation}
 \widetilde\pi(z,u)\ \propto\ \widehat L_J(z,u)\,p(z)\,m(u).
 \label{eq:bayes-common-target}
\end{equation}
Unbiasedness of the likelihood estimate gives the posterior for the
specified Euler model as the marginal distribution of $z$ \citep{andrieu10}.
The log likelihood and log prior enter without division by the number of
observations. We use $J=10,000$ particles and double precision for this
likelihood calculation in both methods.

PMMH jointly proposes $z'=z+hB\xi$, where $\xi\sim N(0,I)$, and independent
random numbers $u'\sim m$. Its acceptance probability is
\[
 1\wedge\frac{\widehat L_J(z',u')p(z')}{\widehat L_J(z,u)p(z)}.
\]
On rejection, both $z$ and $u$ are retained, together with the cached
likelihood estimate. The matrix $BB^T$ is estimated from pilot PMMH draws
and held fixed; the scalar $h$ is adapted during 500 warmup iterations.
Preparation of the starting points and proposal scales includes numerical
optimization with derivatives, although the PMMH transition itself
requires only process simulation and measurement-density evaluation.

NUTS first proposes independent random numbers $u'\sim m$, accepting with
probability $1\wedge\widehat L_J(z,u')/\widehat L_J(z,u)$, and then updates
$z$ conditional on the retained $u$.
We use BlackJAX's NUTS trajectory construction, with maximum tree depth 5.
DMOP with discount parameter 0.95 and 1,000 particles supplies the proposal
force, evaluated in single precision with deterministic GPU operations.
Every trajectory weight uses the double-precision likelihood in
Equation~\ref{eq:bayes-common-target} and the common prior.
The force and likelihood are therefore separate numerical calculations.
Step sizes are adapted during 500 warmup iterations and then held fixed.
The implementation, prior specification and checks of the sampling kernels
are available in \texttt{code/bayesian}.

\subsection{Sampling settings and diagnostics}

\begin{table}
\centering
\resizebox{\textwidth}{!}{\begin{tabular}{lrrrrrrrr}
\toprule
& \multicolumn{4}{c}{NUTS} & \multicolumn{4}{c}{PMMH} \\
\cmidrule(lr){2-5}\cmidrule(lr){6-9}
Coordinate & $\widehat R$ & Bulk ESS & Tail ESS & MCSE & $\widehat R$ & Bulk ESS & Tail ESS & MCSE \\
\midrule
$\log\gamma$ & \textbf{1.009} & \textbf{243} & \textbf{645} & \textbf{0.00749} & 1.206 & 15 & 58 & 0.028 \\
$\log m$ & \textbf{1.019} & \textbf{257} & \textbf{439} & \textbf{0.0118} & 1.121 & 25 & 122 & 0.0358 \\
$\log\epsilon$ & \textbf{1.025} & \textbf{226} & \textbf{267} & \textbf{0.0305} & 1.032 & 105 & 259 & 0.0385 \\
$100\beta_{\mathrm{trend}}$ & \textbf{1.014} & \textbf{333} & \textbf{745} & \textbf{0.00593} & 1.025 & 73 & 165 & 0.0124 \\
$\log\sigma$ & \textbf{1.015} & \textbf{303} & \textbf{551} & \textbf{0.00287} & 1.170 & 18 & 69 & 0.0116 \\
$\log\tau$ & \textbf{1.012} & \textbf{328} & \textbf{537} & \textbf{0.00562} & 1.177 & 16 & 94 & 0.0274 \\
$\beta_1$ & \textbf{1.006} & \textbf{215} & \textbf{529} & \textbf{0.0313} & 1.300 & 11 & 26 & 0.125 \\
$\beta_2$ & \textbf{1.009} & \textbf{197} & \textbf{460} & \textbf{0.0247} & 1.266 & 12 & 43 & 0.0813 \\
$\beta_3$ & \textbf{1.011} & \textbf{196} & \textbf{284} & \textbf{0.0917} & 1.312 & 10 & 30 & 0.319 \\
$\beta_4$ & \textbf{1.008} & \textbf{282} & \textbf{549} & \textbf{0.0101} & 1.178 & 19 & 62 & 0.0368 \\
$\beta_5$ & \textbf{1.004} & \textbf{295} & \textbf{591} & \textbf{0.00877} & 1.303 & 11 & 149 & 0.0436 \\
$\beta_6$ & \textbf{1.007} & \textbf{341} & \textbf{640} & \textbf{0.00578} & 1.114 & 34 & 200 & 0.0181 \\
$\omega_1$ & \textbf{1.059} & \textbf{60} & \textbf{214} & \textbf{0.135} & 1.114 & 34 & 61 & 0.188 \\
$\omega_2$ & \textbf{1.027} & \textbf{297} & \textbf{416} & \textbf{0.0634} & 1.079 & 39 & 129 & 0.157 \\
$\omega_3$ & \textbf{1.027} & \textbf{220} & \textbf{381} & \textbf{0.0463} & 1.155 & 17 & 49 & 0.153 \\
$\omega_4$ & \textbf{1.015} & \textbf{278} & \textbf{542} & \textbf{0.0405} & 1.177 & 15 & 37 & 0.166 \\
$\omega_5$ & \textbf{1.010} & \textbf{293} & \textbf{335} & \textbf{0.118} & 1.168 & 18 & 68 & 0.495 \\
$\omega_6$ & \textbf{1.056} & \textbf{109} & \textbf{116} & \textbf{0.396} & 1.111 & 42 & 59 & 0.528 \\
\bottomrule
\end{tabular}
}
\caption{Diagnostics for four chains of each method, using 2,000 retained iterations per chain after 500 warmup iterations. Rank-normalized split $\widehat R$, bulk and tail ESS, and the Monte Carlo standard error of the mean (MCSE) are calculated across chains on the displayed estimation coordinates. All chains have the common prior and target specified above. Bold entries indicate $\widehat R$ closer to 1, larger ESS, or smaller MCSE.}
\label{tab:bayesian-diagnostics}
\end{table}

We ran four chains of each method from the same four starting points,
obtained by optimization followed by a pilot PMMH run.
Figure~\refFigNutsEb\ compares the first 2,000 retained iterations of each
chain, after 500 warmup iterations.
Figure~\ref{fig:mh} shows the marginal sample distributions of the individual chains over this window.

Both methods use fixed covariance matrices estimated from preliminary
PMMH runs. The PMMH covariance uses 1\% diagonal shrinkage and a $10^{-6}I$ ridge.
For NUTS, we combine the pooled covariance from the first 5,000 retained PMMH
iterations per chain with this matrix, with weights 0.75 and 0.25, and rescale
to the trace of the PMMH proposal covariance. These matrices affect the proposals but not the prior.
Only the scalar step size is adapted during warmup, with target acceptance
probabilities 0.65 for NUTS and 0.234 for PMMH. Initial step sizes are 0.07
and 0.5, respectively. The final NUTS step sizes range from 0.048 to 0.060.

Table~\ref{tab:bayesian-diagnostics} gives diagnostics for all 18
coordinates in the 2,000-iteration comparison. For the trend coordinate,
NUTS has $\widehat R=1.014$ and bulk ESS 333, compared with 1.025 and 73 for
PMMH. The largest NUTS $\widehat R$ is 1.059, and the smallest bulk ESS is
60, both for the first seasonal environmental-transmission coefficient.
There are two divergent NUTS transitions among 8,000 retained iterations.
Almost all trajectories reach the depth limit of 5, corresponding to
31 integration steps. Thus the traces show improved exploration without
establishing convergence of all coordinates.

The longer PMMH runs give broadly similar marginal distributions
\sout{(Figure~\ref{fig:mh})}; their largest $\widehat R$ is 1.014 and smallest bulk
ESS is 202. Doubling all prior standard deviations has little effect on
the sampled trend standard deviation (0.112 versus 0.115), but gives wider
lower tails for some seasonal environmental-transmission coefficients.
These tails also mix slowly, so the comparison does not establish prior
insensitivity. The prior overlays in Figure~\ref{fig:posteriors} make this
issue visible.

\subsection{Additional sampling figures}

\begin{figure}
\centering
\includegraphics[width=0.93\textwidth,height=0.73\textheight,keepaspectratio]{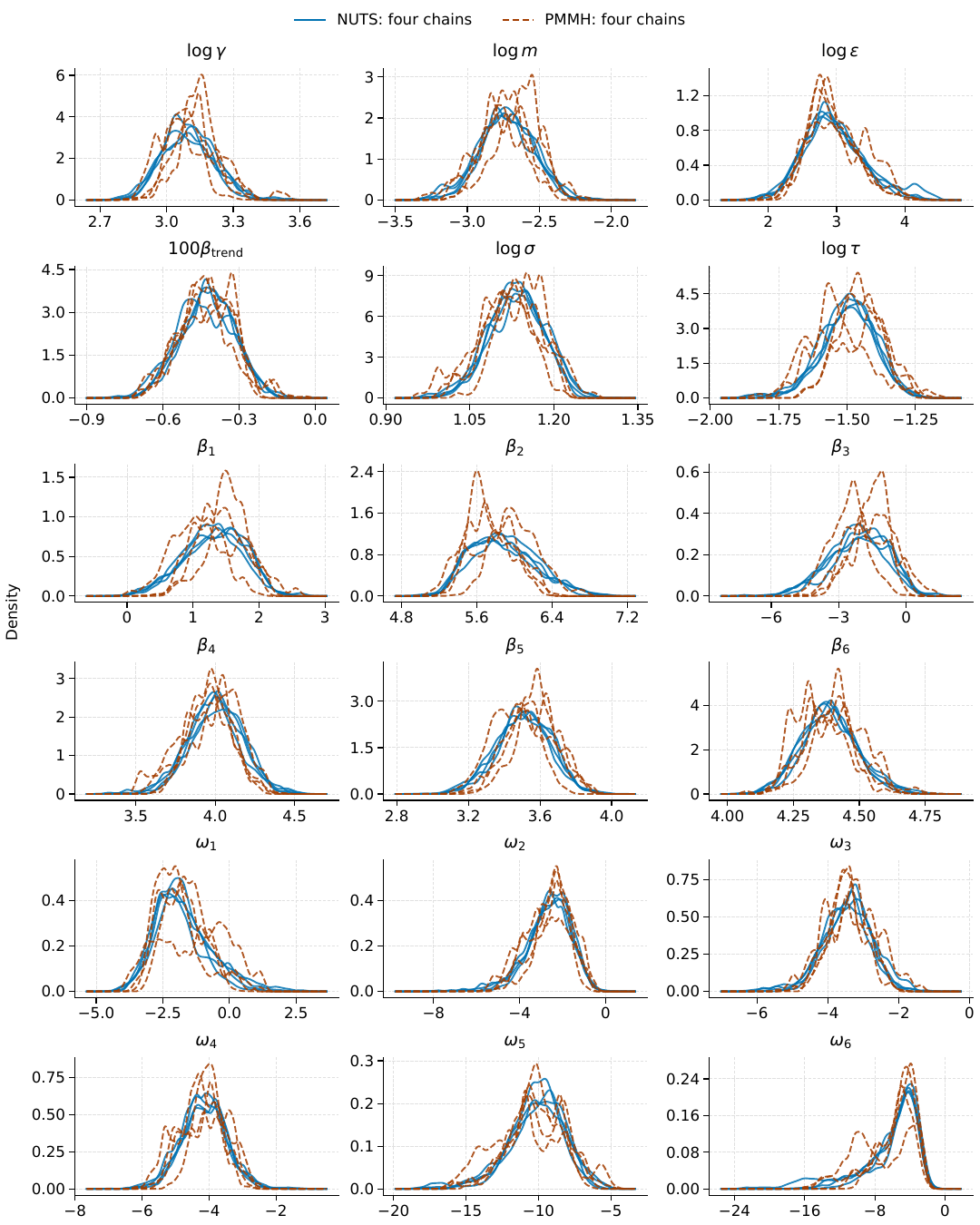}
\caption{
    Marginal sample distributions for the four NUTS chains (blue solid) and four PMMH chains (orange dashed), using all 2,000 retained iterations per chain after 500 warmup iterations. Each curve represents one chain. For each coordinate, all eight density estimates use the same Gaussian-kernel bandwidth, selected by Scott's rule from the pooled draws of both methods. The prior, target and iteration window are those of Table~\ref{tab:bayesian-diagnostics}.}
\label{fig:mh}
\end{figure}
\FloatBarrier


\begin{figure}
\centering
\includegraphics[width=0.93\textwidth,height=0.73\textheight,keepaspectratio]{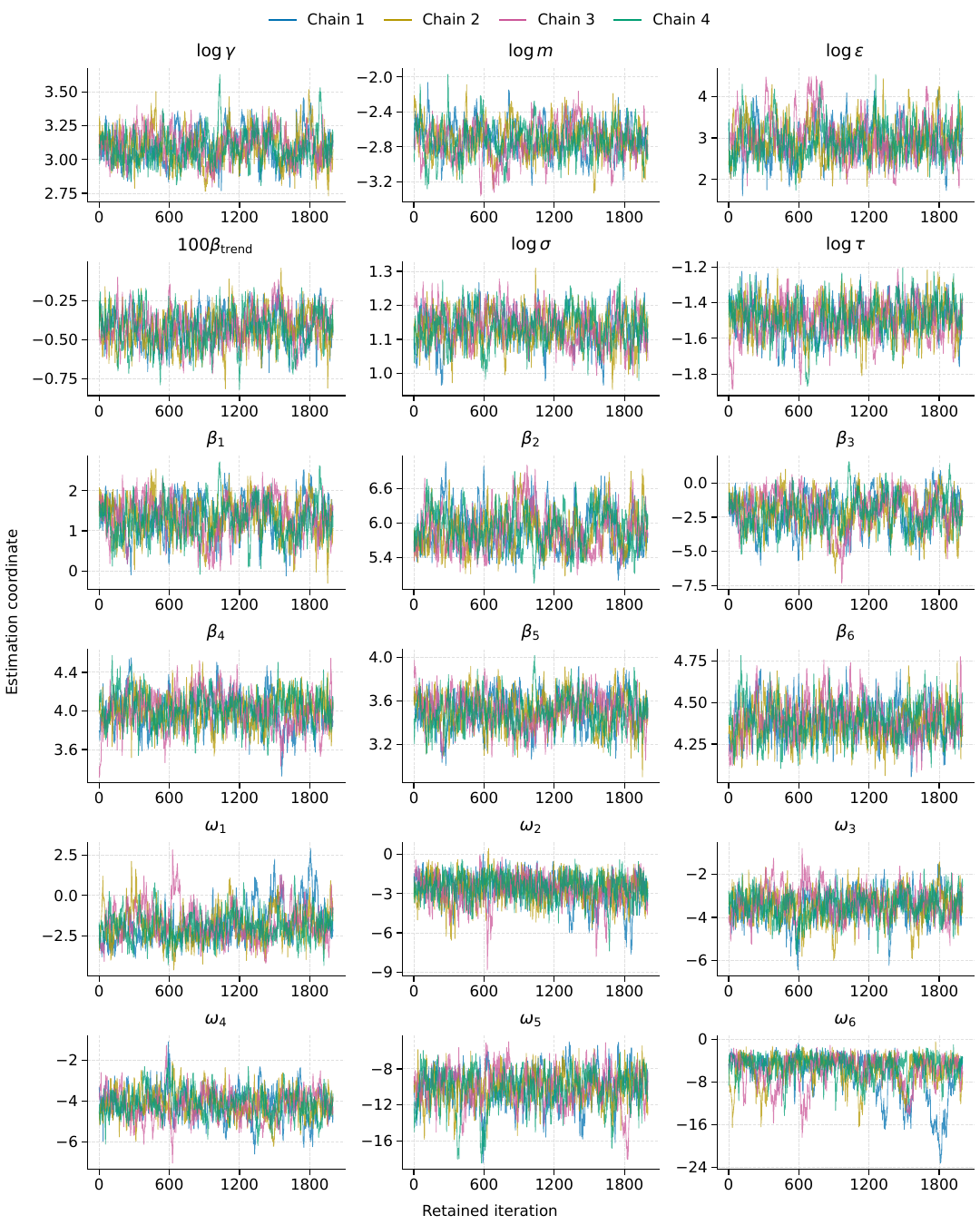}
\caption{NUTS traces for all 18 estimation coordinates. Colors identify the four chains, with 2,000 retained iterations per chain after 500 warmup iterations. The long excursions of some seasonal environmental-transmission coefficients are less apparent in the trend diagnostic alone.}
\label{fig:nuts}
\end{figure}
\FloatBarrier
\newpage

\begin{figure}
\centering
\includegraphics[width=0.93\textwidth,height=0.73\textheight,keepaspectratio]{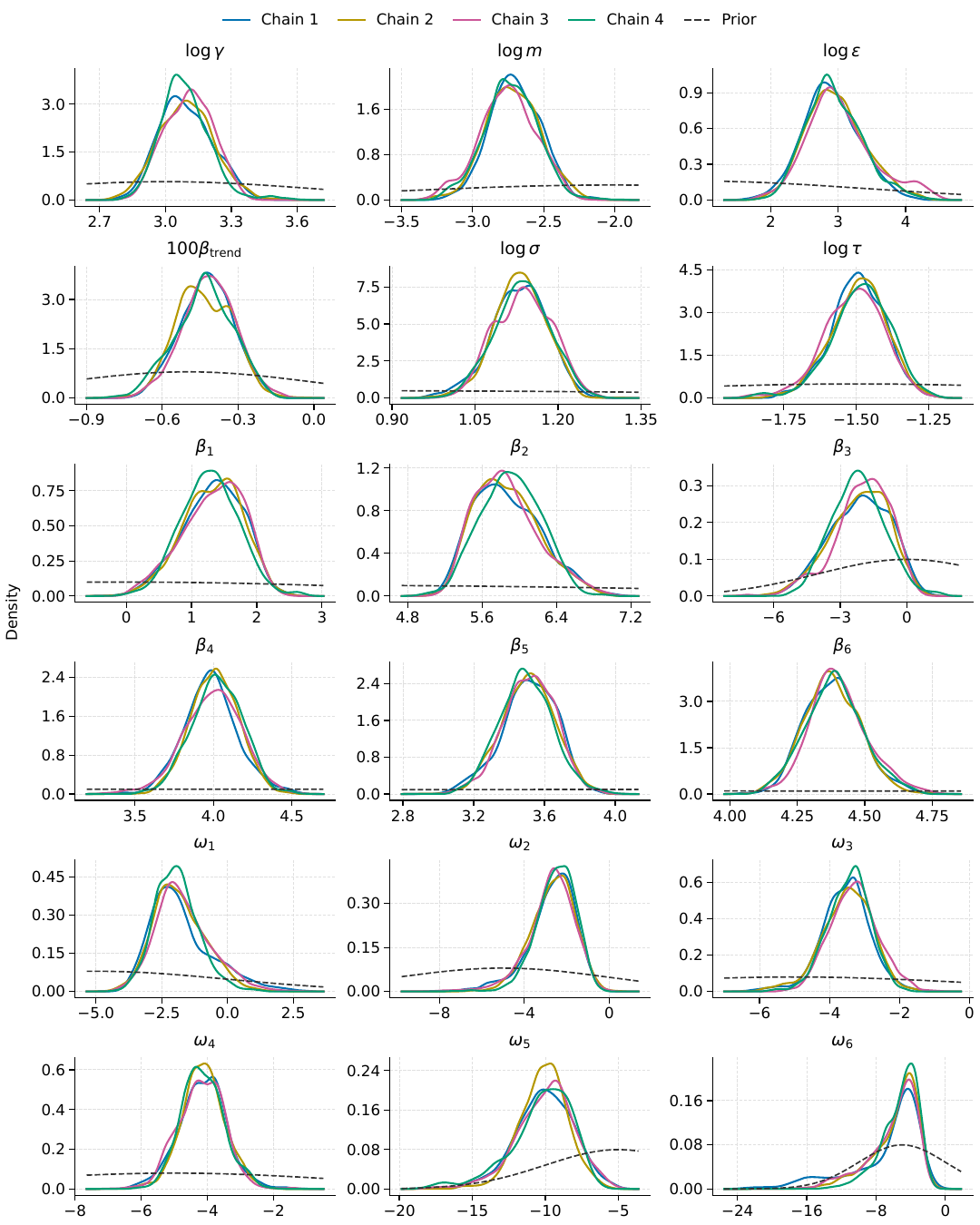}
\caption{Marginal sample distributions from the four NUTS chains for all 18 estimation coordinates, using 2,000 retained iterations per chain. Dashed curves show the common Gaussian prior. Density estimates use Gaussian kernels with Scott bandwidths. }
\label{fig:posteriors}
\end{figure}

\clearpage

\subsection{Sampling without covariance tuning}
\label{sec:bayes-identity-control}

To assess the effect of covariance tuning, we also ran NUTS with an
identity metric on the estimation coordinates $x$ defined in
Section~\ref{sec:bayes-common-target}, and PMMH with a fixed diagonal proposal.
The model, prior, likelihood particle count and four starting points are the
same as in the preceding comparison. For PMMH, the proposal standard
deviation is one hundredth of the search-interval width in each coordinate,
equivalently $z'=z+0.02\xi$ with $\xi\sim N(0,I)$. This proposal is held
fixed throughout the run. For NUTS, the metric is held fixed but the step
size is adapted during 500 warmup iterations, starting at 0.002. The final
step sizes are 0.000811, 0.001238, 0.001121 and 0.000963. The maximum tree
depth remains 5, and the DMOP force uses 1,000 particles with discount
parameter 0.95.

Each chain has 1,500 iterations. We compare iterations 501--1,500
for both methods, discarding the first 500 PMMH iterations for this
comparison although PMMH has no adaptive warmup.
Figure~\ref{fig:bayes-identity-trend} shows greater movement of the trend
coordinate under NUTS, with $\widehat R=1.380$ and bulk ESS 8.9, compared
with 2.395 and 4.9 for PMMH. However, substantial differences between chains
remain for both methods (Figures~\ref{fig:bayes-identity-traces}
and~\ref{fig:bayes-identity-marginals}). The largest $\widehat R$ is 3.491
for NUTS and 2.580 for PMMH, and the smallest bulk ESS is 4.4 and 4.8,
respectively (Table~\ref{tab:bayes-identity-diagnostics}). These runs do not
provide reliable posterior summaries.

There are no divergent NUTS transitions, but every retained
trajectory reaches the depth limit. Thus removing divergences by adapting
the step size does not ensure adequate exploration. PMMH accepts between
9.8\% and 13.9\% of proposals during the displayed iterations.

\begin{table}[H]
\centering
\resizebox{\textwidth}{!}{\begin{tabular}{lrrrrrrrr}
\toprule
& \multicolumn{4}{c}{Identity-metric NUTS} & \multicolumn{4}{c}{Diagonal PMMH} \\
\cmidrule(lr){2-5}\cmidrule(lr){6-9}
Coordinate & $\widehat R$ & Bulk ESS & Tail ESS & MCSE & $\widehat R$ & Bulk ESS & Tail ESS & MCSE \\
\midrule
$\log\gamma$ & \textbf{1.560} & \textbf{7.1} & 30.4 & \textbf{0.0204} & 2.524 & 4.8 & \textbf{31.5} & 0.0489 \\
$\log m$ & 2.475 & 4.8 & 11.2 & 0.117 & \textbf{2.246} & \textbf{5.1} & \textbf{18.2} & \textbf{0.101} \\
$\log\epsilon$ & 2.827 & 4.6 & 13.1 & 0.204 & \textbf{1.895} & \textbf{5.7} & \textbf{13.9} & \textbf{0.131} \\
$100\beta_{\mathrm{trend}}$ & \textbf{1.380} & \textbf{8.9} & \textbf{73.7} & \textbf{0.0399} & 2.395 & 4.9 & 11.6 & 0.049 \\
$\log\sigma$ & \textbf{1.239} & \textbf{12.5} & \textbf{164.3} & \textbf{0.0127} & 1.776 & 6.3 & 15.7 & 0.0183 \\
$\log\tau$ & \textbf{1.437} & \textbf{8.3} & 11.4 & 0.0416 & 2.407 & 5.0 & \textbf{14.1} & \textbf{0.0359} \\
$\beta_1$ & \textbf{1.870} & \textbf{5.8} & \textbf{33.6} & \textbf{0.119} & 2.129 & 5.3 & 22.5 & 0.192 \\
$\beta_2$ & 2.025 & 5.4 & 19.0 & \textbf{0.109} & \textbf{1.813} & \textbf{6.0} & \textbf{39.0} & 0.154 \\
$\beta_3$ & 2.741 & 4.7 & \textbf{14.2} & \textbf{0.302} & \textbf{2.572} & \textbf{4.8} & 12.4 & 0.583 \\
$\beta_4$ & 1.373 & 9.0 & 22.4 & \textbf{0.0445} & \textbf{1.314} & \textbf{10.1} & \textbf{74.5} & 0.0546 \\
$\beta_5$ & 1.688 & 6.4 & 16.7 & \textbf{0.0431} & \textbf{1.528} & \textbf{7.3} & \textbf{32.6} & 0.0508 \\
$\beta_6$ & \textbf{1.383} & 8.8 & 25.8 & 0.034 & 1.386 & \textbf{9.3} & \textbf{37.8} & \textbf{0.0333} \\
$\omega_1$ & 3.074 & 4.5 & 13.6 & 0.374 & \textbf{1.811} & \textbf{6.0} & \textbf{31.8} & \textbf{0.232} \\
$\omega_2$ & \textbf{2.179} & \textbf{5.2} & 13.3 & \textbf{0.21} & 2.580 & 4.8 & \textbf{17.3} & 0.504 \\
$\omega_3$ & 2.324 & 5.0 & \textbf{37.4} & 0.209 & \textbf{1.492} & \textbf{7.6} & 18.8 & \textbf{0.119} \\
$\omega_4$ & 2.424 & 4.9 & 14.5 & \textbf{0.145} & \textbf{2.196} & \textbf{5.2} & \textbf{24.7} & 0.222 \\
$\omega_5$ & 3.491 & 4.4 & 11.5 & 0.827 & \textbf{1.883} & \textbf{5.8} & \textbf{13.9} & \textbf{0.481} \\
$\omega_6$ & 2.605 & 4.8 & \textbf{13.5} & \textbf{0.291} & \textbf{2.471} & \textbf{4.9} & 11.2 & 0.335 \\
\bottomrule
\end{tabular}
}
\caption{Diagnostics for identity-metric NUTS and diagonal PMMH,
using four chains and iterations 501--1,500 of each chain. The columns give
rank-normalized split $\widehat R$, bulk and tail ESS, and MCSE of the mean
on the estimation scale. Bold entries indicate the more favorable value
within each pair, as in Table~\ref{tab:bayesian-diagnostics}. The large
between-chain differences limit interpretation of ESS and MCSE; bold entries
do not indicate convergence.}
\label{tab:bayes-identity-diagnostics}
\end{table}
\FloatBarrier

\begin{figure}[p]
\centering
\includegraphics[width=\textwidth]{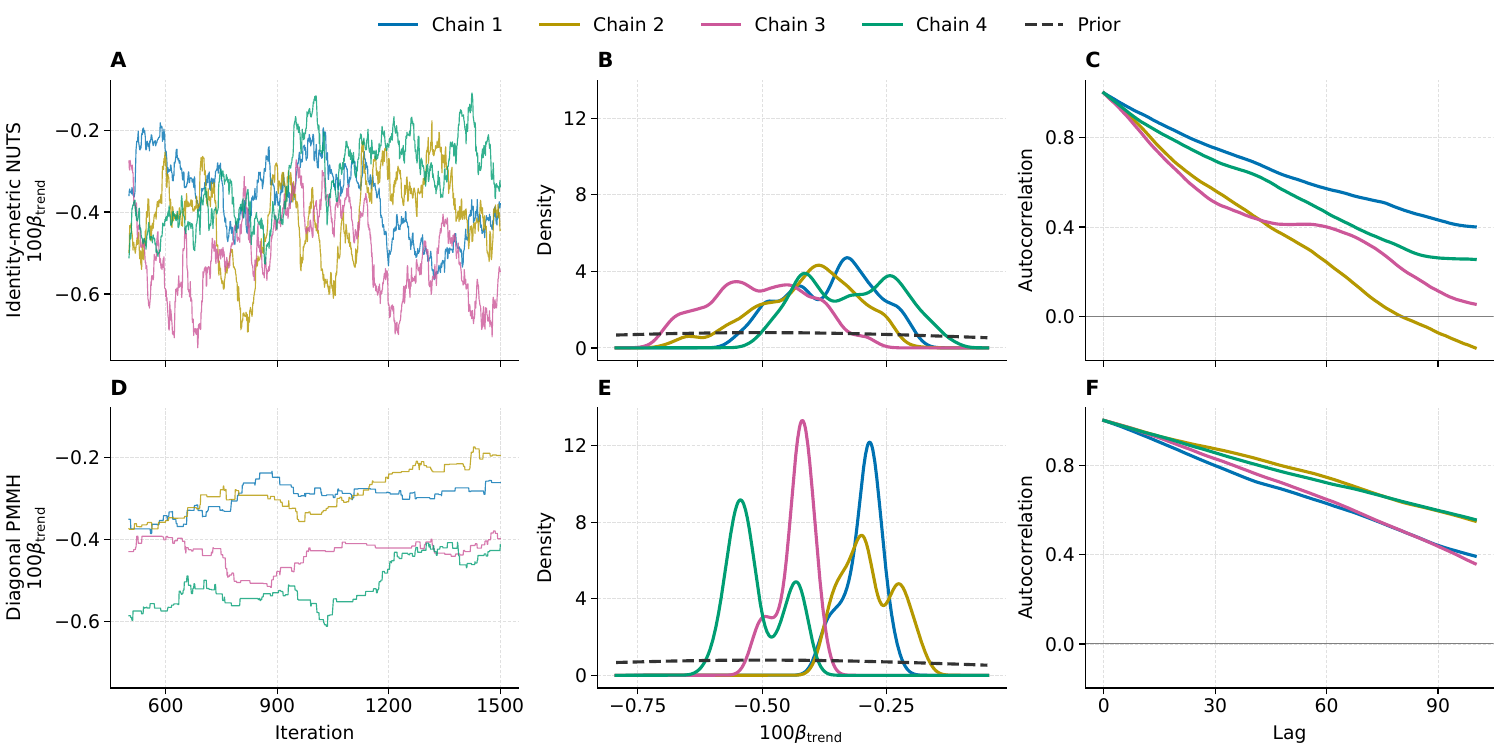}
\caption{Trend-coordinate diagnostics for identity-metric NUTS
(A--C) and diagonal PMMH (D--F). The columns show traces, marginal sample
densities and autocorrelations for iterations 501--1,500. Colors identify
the four chains; dashed curves in the density panels show the common prior.
Corresponding panels have the same scales. The density estimates use a
common Gaussian-kernel bandwidth across methods and chains.}
\label{fig:bayes-identity-trend}
\end{figure}
\FloatBarrier

\begin{figure}[p]
\centering
\includegraphics[width=0.93\textwidth,height=0.76\textheight,keepaspectratio]{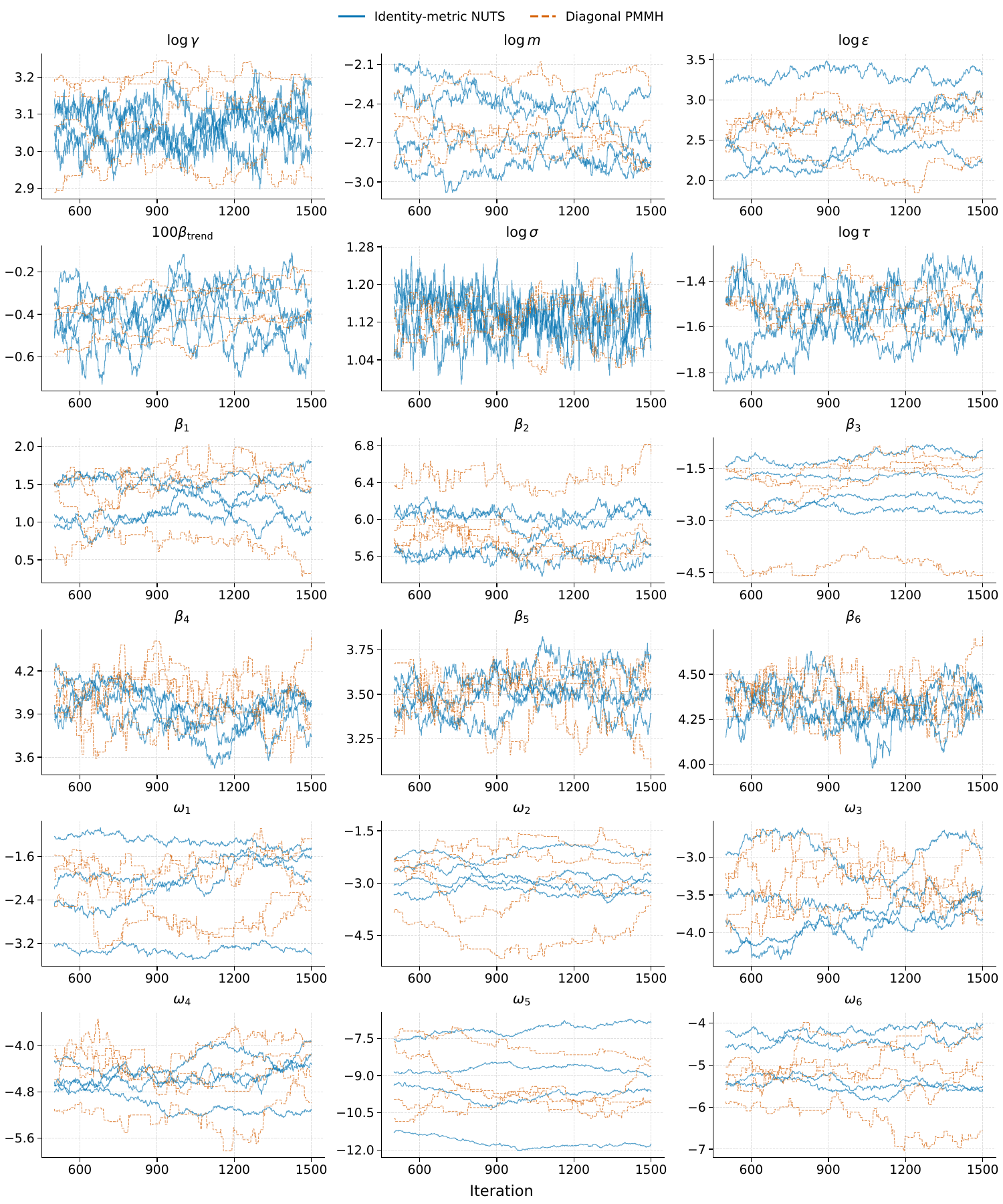}
\caption{Traces for all 18 estimation coordinates under
identity-metric NUTS (blue solid) and diagonal PMMH (orange dashed).
Each method has four chains, shown over iterations 501--1,500. Several
coordinates retain substantial differences between chains under both methods.}
\label{fig:bayes-identity-traces}
\end{figure}
\FloatBarrier

\begin{figure}[p]
\centering
\includegraphics[width=0.93\textwidth,height=0.76\textheight,keepaspectratio]{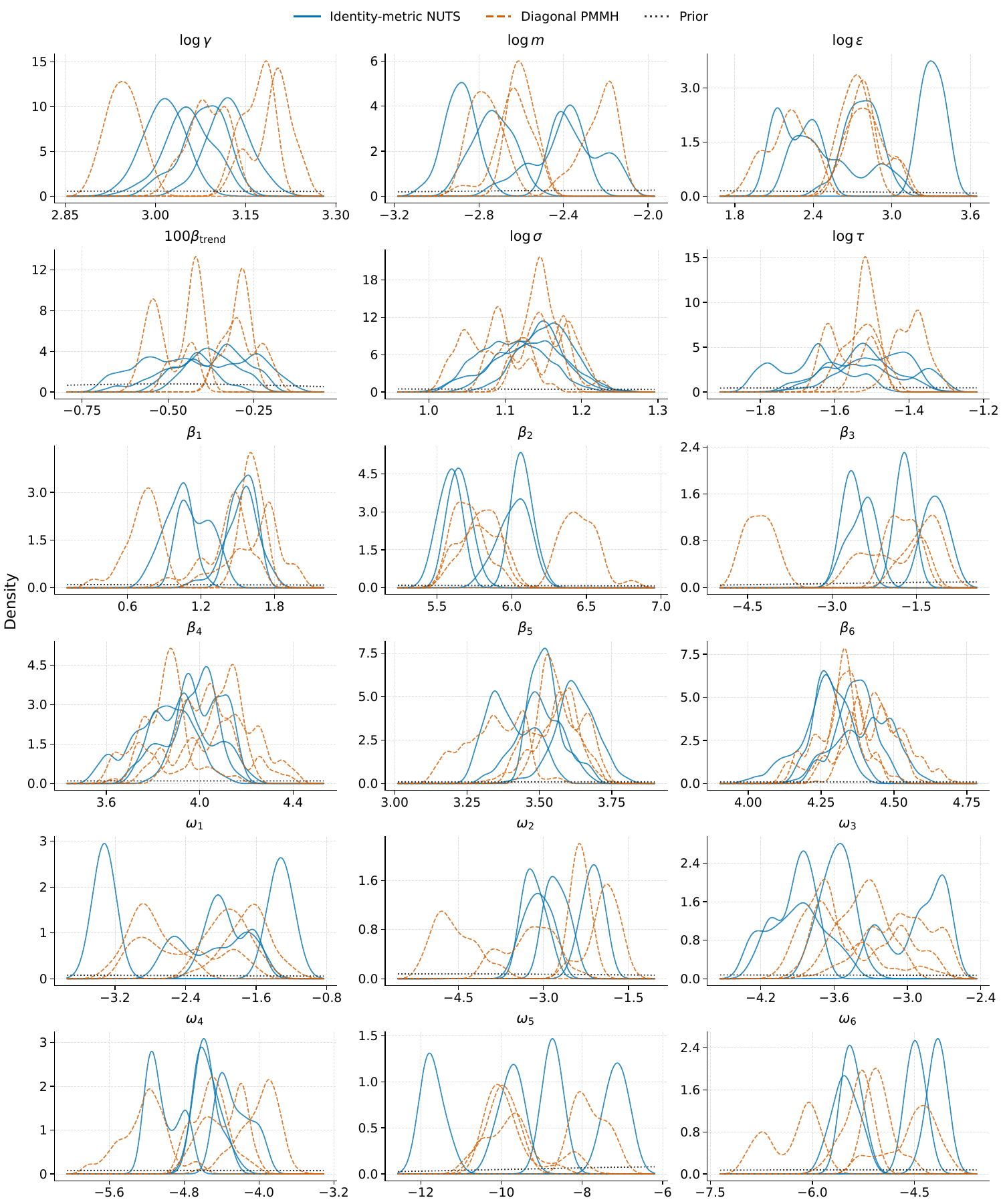}
\caption{Marginal sample distributions for the chains in
Figure~\ref{fig:bayes-identity-traces}. Blue solid curves show identity-metric
NUTS, orange dashed curves show diagonal PMMH, and dotted curves show the
common prior. Each curve uses all 1,000 retained draws from one chain. For
each coordinate, all eight estimates use the same absolute Gaussian-kernel
bandwidth, obtained by Scott's rule from the pooled draws. Differences
between curves reflect incomplete mixing and cannot be interpreted as
distinct posterior modes.}
\label{fig:bayes-identity-marginals}
\end{figure}
\FloatBarrier
\newpage

\section{Algorithmic Parameters}
\label{appendix:algparams}

We include the full set of algorithmic parameters used for the Dhaka results in Table~\ref{table:algparams}. 
The warm start iterations are IF2 iterations, which are followed by gradient descent (GD) iterations when using IFAD. 
The IF2 cooling rate is the fractional reduction in the parameter perturbations over 50 iterations, so that a cooling rate of 0.5 corresponds to approximately 1\% cooling per iteration.
The random walk standard deviation (RWSD) scales the noise added at each filtering step to each regular parameter.
Initial value parameters (IVPs), that affect the value of the latent process only at time $t_0$, are treated differently: they are perturbed only at time $t_0$, with a larger RWSD. 
Parameters are transformed to a unit scale, for example by taking logarithms or a logistic transformation, so that a single RWSD value works for all parameters.

\begin{table}[H]
\centering
\begin{tabular}{p{4.5cm}cccc}
\toprule
\textbf{Hyperparameter} & \textbf{Standard IF2} & \textbf{IFAD-0} & \textbf{IFAD-0.97} & \textbf{IFAD-1} \\
\midrule
Number of Particles ($J$) & 5,000 & 5,000 & 5,000 & 5,000 \\
Warm-start Iterations & 1,950 & 60 & 175 & 60 \\
GD Iterations & --- & 410 & 400 & 425 \\
IF2 Cooling Rate & 0.93 & 0.5 & 0.5 & 0.5 \\
Initial RWSD & 0.02 & 0.02 & 0.02 & 0.02 \\
Initial RWSD (IVPs) & 0.16 & 0.16 & 0.16 & 0.16 \\
Gradient Optimizer & --- & Adam & Adam & Adam \\
Momentum ($\beta_1$) & --- & 0.0 & 0.9 & 0.9 \\
LR Schedule & --- & Cosine Decay & Cosine Decay & Cosine Decay \\
Initial LR ($\eta_0$) & --- & 0.1 & 0.1 & 0.1 \\
Warm-up Steps & --- & 10 & 10 & 10 \\
Eval. Particles / Reps & 5,000 / 36 & 5,000 / 36 & 5,000 / 36 & 5,000 / 36 \\
\bottomrule
\end{tabular}
\caption{Algorithmic parameters for the global search benchmarking on the Dhaka cholera model.}
\label{table:algparams}
\end{table}

Following some experimentation, IFAD was implemnenteed by plugging the DMOP gradient estimate into the Adam optimizer \citep{kingma15}.
The momentum parameter in Adam was turned off for IFAD-0, since this low-variance estimate did not reqquire it.
Decreasing the learning rate (LR) was found to be helpful, and we selected the commonly used cosine decay.
Warm-up steps were taken with Adam to stabilize its initial performance and avoid noisy

\clearpage

\section{Application to a Daphnia Population Model}
\label{sec:daphnia}

\subsection{Model and data}

The SIRJPF2 model developed by \citet{yang25daphnia} describes interactions between two species of \textit{Daphnia}, their algal food, and a fungal parasite.
We use the version with unit-specific infected-adult mortality to compare likelihood maximization by IFAD and the marginalized panel iterated filter (MPIF) of \citet{wheeler25}.

The data come from the mesocosm experiment of \citet{searle16}.
We use the treatment containing the native species \textit{D.~dentifera}, the invasive species \textit{D.~lumholtzi}, and the parasite \textit{Metschnikowia bicuspidata}, with eight mesocosms, each treated as one unit in the panel.
Each mesocosm was observed at ten times, $t_n=5n+2$ days for $n=1,\ldots,10$.
At each observation time, a one-liter sample was taken and susceptible and infected adult females of both species were counted.
Juvenile counts are excluded from the likelihood.

We write $S_u^a(t)$, $I_u^a(t)$, and $J_u^a(t)$ for the densities of susceptible adult females, infected adult females, and juveniles, where $u=1,\ldots,8$ indexes the mesocosm and $a\in\{\mathrm n,\mathrm i\}$ denotes the native or invasive species.
The densities of food and parasite spores are $F_u(t)$ and $P_u(t)$.
These variables form an eight-dimensional latent state in each mesocosm.
Removal by death, consumption, or sampling is denoted by $R$ in the model name and does not require an additional state variable.
The two host species interact through their shared food and parasite populations.
The processes in different mesocosms are independent conditional on the parameters.

As in \citet{yang25daphnia}, the transition dynamics follow an It\^o stochastic differential equation (SDE):
\begin{align*}
dS_u^a
 &= \{\lambda_J J_u^a-(\theta_S^a+p^a f_S^aP_u+\delta)S_u^a\}\,dt
    +\sigma_S^a S_u^a\,dB_{S,u}^a,\\
dI_u^a
 &= \{p^a f_S^a S_u^aP_u-(\theta_{I,u}^a+\delta)I_u^a\}\,dt
    +\sigma_I^a I_u^a\,dB_{I,u}^a,\\
dJ_u^a
 &= \{r^a f_S^aF_uS_u^a-(\theta_J^a+\delta+\lambda_J)J_u^a\}\,dt
    +\sigma_J^a J_u^a\,dB_{J,u}^a,\\
dP_u
 &= \left\{\sum_a\left[\beta\theta_{I,u}^a I_u^a
          -f_S^a(S_u^a+\xi I_u^a)P_u\right]
          -(\theta_P+\delta)P_u\right\}\,dt
    +\sigma_P P_u\,dB_{P,u},\\
dF_u
 &= \left\{\mu-\delta F_u
          -\sum_a f_S^aF_u(S_u^a+\xi I_u^a+J_u^a)\right\}\,dt
    +\sigma_F F_u\,dB_{F,u},
\end{align*}
with independent standard Brownian motions $B$.
The time argument is suppressed in the equations.
Except for infected-adult mortality, estimated parameters are shared across mesocosms.
The parameters $r^a$ and $f_S^a$ are the reproductive rate and host filtration rate, $p^a$ controls infection following spore ingestion, and $\xi$ is the filtration rate of infected adults relative to susceptible adults.
The mortality or loss rates of susceptible adults, infected adults, juveniles, and spores are $\theta_S^a$, $\theta_{I,u}^a$, $\theta_J^a$, and $\theta_P$, respectively; the $\sigma$ parameters are process-noise coefficients.
The relative filtration rate of juveniles is fixed at one.

We use the units of the public Python implementation:\footnote{\href{https://github.com/pypomp/Daphnia-tutorial/blob/f3183301273a409e5b079bb0ae3a22246f69ba7a/Python-code/daphnia_tut_pypomp_advanced.qmd}{\texttt{pypomp/Daphnia-tutorial}, advanced tutorial, revision \texttt{f318330}.}}
\textit{Daphnia} densities are individuals per liter, $F$ is in $10^6$ algal cells per liter, and $P$ is in $10^3$ spores per liter, with parameters expressed in the corresponding units.
The sampling loss rate is fixed at $\delta=0.013$ per day, the food replenishment rate at $\mu=0.37$, the juvenile maturation rate at $\lambda_J=0.1$ per day, and the spore yield per infected-host death at $\beta=30$, corresponding to $3\times10^4$ spores.
The initial states are $S_u^{\mathrm n}=2.333$, $S_u^{\mathrm i}=0.667$, and $F_u=16.667$, with the remaining states set to zero.
Simulation begins at $t_0=1$ day, and parasites are introduced on day~4 at a density of $25\times10^3$ spores per liter.
We numerically solve the SDE using the Euler-Maruyama method with a time step of 0.25 days, following the initial conditions, inoculation, and numerical boundary rules of the public implementation.

The measurement model for the adult counts is
\[
Y_{S,u,n}^a\sim\operatorname{NBinomial}(S_u^a(t_n),\tau_S^a),
\qquad
Y_{I,u,n}^a\sim\operatorname{NBinomial}(I_u^a(t_n),\tau_I^a),
\]
where $\operatorname{NBinomial}(m,\tau)$ has mean $m$ and variance $m+m^2/\tau$ \citep{yang25daphnia}.
The four counts at each observation time are independent conditional on the latent state.
Since the sampled volume is one liter, the state densities and observed counts are on the same scale.
The numerical implementation sets a lower bound of $10^{-10}$ on the negative binomial mean and size parameters.
As in the public Python code, a state variable that crosses a numerical boundary is reset to zero, and that particle's log weight for the mesocosm at the next observation is set to $-150$.
The boundary thresholds and the treatment of the inoculation time are specified in the code.

We fix $\sigma_S^{\mathrm n}=\sigma_S^{\mathrm i}=0$ and estimate $\theta_{I,u}^{\mathrm n}$ and $\theta_{I,u}^{\mathrm i}$ separately for each mesocosm.
The remaining 22 estimated parameters are shared across mesocosms, giving $22+2\times8=38$ estimated parameters.
This is the model with unit-specific infected-adult mortality in Table~S3 of the archived revision \citep{yang26daphniacode}.

\subsection{Experimental settings}

We tested IFAD with discount factors $\alpha=0$, $\alpha=0.97$, and $\alpha=1$ against MPIF using Pypomp 1.0.5 \citep{abkemeier25}.
For each method, we performed 50 searches from the same starting parameter vectors, generated by independent lognormal perturbations of the initial parameter vector in the public tutorial, with a standard deviation of 0.45 on the log scale.
All positive estimated parameters were optimized on the log scale.
Each IFAD search used an MPIF warm start of 200 iterations followed by 200 iterations of the Adam optimizer \citep{kingma15}, using DMOP-$\alpha$ gradient estimates.
We ran MPIF alone for 680 iterations, chosen to approximate the computing time of IFAD; measured times are reported in Table~\ref{tab:daphnia-results}.

We used 1,000 particles per mesocosm during optimization.
The initial random walk standard deviation (RWSD) for MPIF was 0.05 on the log scale, with geometric cooling by a factor of 0.7 over 50 iterations.
Adam used $\beta_1=0.9$, $\beta_2=0.999$, and $\epsilon=10^{-8}$, with all eight mesocosms contributing to each update.
The learning-rate schedule used ten steps of linear warmup multiplied by cosine decay with a final factor of 0.1.
A summary of the algorithmic parameters is provided in Table~\ref{tab:daphnia-settings}.

The learning-rate scale was selected using four additional starting points, separate from the 50 used for the comparison.
We compared the six combinations of learning-rate scales $0.001$, $0.005$, and $0.01$ with $\alpha\in\{0.9,0.97\}$.
The combination with the highest median final log-likelihood, evaluated by independent particle filtering, was $0.01$ with $\alpha=0.97$.
We used the same learning-rate schedule for $\alpha=0$ and $\alpha=1$.

\begin{table}[htbp!]
\centering
\begin{tabular}{p{4.5cm}cccc}
\toprule
\textbf{Hyperparameter} & \textbf{MPIF} & \textbf{IFAD-0} & \textbf{IFAD-0.97} & \textbf{IFAD-1} \\
\midrule
Particles per mesocosm ($J$) & 1,000 & 1,000 & 1,000 & 1,000 \\
MPIF Iterations & 680 & 200 & 200 & 200 \\
Adam Iterations & --- & 200 & 200 & 200 \\
MPIF Cooling Factor & 0.7 & 0.7 & 0.7 & 0.7 \\
Initial RWSD (log scale) & 0.05 & 0.05 & 0.05 & 0.05 \\
Momentum ($\beta_1$) & --- & 0.9 & 0.9 & 0.9 \\
Learning-rate Scale & --- & 0.01 & 0.01 & 0.01 \\
Warm-up Steps & --- & 10 & 10 & 10 \\
Eval. Particles / Reps & 2,000 / 10 & 2,000 / 10 & 2,000 / 10 & 2,000 / 10 \\
\bottomrule
\end{tabular}
\caption{Algorithmic parameters for the Daphnia experiment. Particle numbers are per mesocosm; the cooling factor is the multiplier over 50 iterations.}
\label{tab:daphnia-settings}
\end{table}

All reported log-likelihoods were evaluated independently of optimization using ordinary particle filters, with 2,000 particles per mesocosm and ten replicates.
Likelihood estimates were averaged within each mesocosm before taking logarithms, and the eight mesocosm log-likelihoods were then summed.
Each method's best parameter vector was then held fixed and re-evaluated with ten fresh replicates at the same particle count, using the same model and Pypomp version on a CPU.
The same evaluation method was used for the final results and for intermediate saved parameter vectors.
The progress plot uses 25, 27, 26, and 22 saved parameter vectors per search for IFAD-0, IFAD-0.97, IFAD-1, and MPIF, respectively, including the final evaluations.
Time is assigned in proportion to iteration number within each stage, using the measured total optimization time for that stage.
Each set of 50 searches ran on one Tesla V100 16 GB GPU, with two searches per batch.
Optimization times are summed over the 25 batches and include initial compilation but exclude the additional particle filter evaluations.

\subsection{Results}

Our findings are summarized in Table~\ref{tab:daphnia-results} and Figures~\ref{fig:daphnia-progress} and~\ref{fig:daphnia-raincloud}.
IFAD-0.97 and IFAD-1 attained maximum log-likelihoods of $-858.48$ and $-858.27$, respectively, each with a Monte Carlo standard error (MCSE) of approximately 0.17.
The maximum for MPIF was $-860.06$, with MCSE 0.16.
IFAD-0 attained $-860.36$ and had a median log-likelihood close to that of MPIF.

\begin{table}[htbp!]
\centering
\begin{tabular}{lrrrrr}
\toprule
 & Maximum & Re-evaluated & Median & MPIF time (s) & DMOP time (s) \\
\midrule
IFAD-0.97 & $-858.48\;(0.17)$ & $-858.23\;(0.16)$ & $-861.42$ & 535 & 1,309 \\
IFAD-1 & $-858.27\;(0.17)$ & $-858.00\;(0.14)$ & $-861.46$ & 520 & 1,313 \\
IFAD-0 & $-860.36\;(0.24)$ & $-860.30\;(0.21)$ & $-864.12$ & 503 & 1,314 \\
MPIF & $-860.06\;(0.16)$ & $-860.26\;(0.14)$ & $-864.27$ & 1,591 & - \\
\bottomrule
\end{tabular}

\caption{Log-likelihoods from 50 searches per method, with MCSEs in parentheses. Re-evaluated uses fresh particle filters at each method's original maximizing parameter vector.}
\label{tab:daphnia-results}
\end{table}

Comparing searches from the same starting points, IFAD-0.97 and IFAD-1 each attained a higher final log-likelihood than MPIF in 49 of the 50 searches.
The median paired differences were 2.86 and 2.79 log-likelihood units, respectively.
The distributions of the final log-likelihoods for IFAD-0.97 and IFAD-1 were similar, which may be due to the short time series in each mesocosm.

\citet{yang26daphniacode} report a maximum log-likelihood of $-863.03$ in the revised Table~S3.
Our MPIF searches attained a higher value, and IFAD-0.97 and IFAD-1 improved on this further.

\begin{figure}[htbp!]
\centering
\includegraphics[width=14cm]{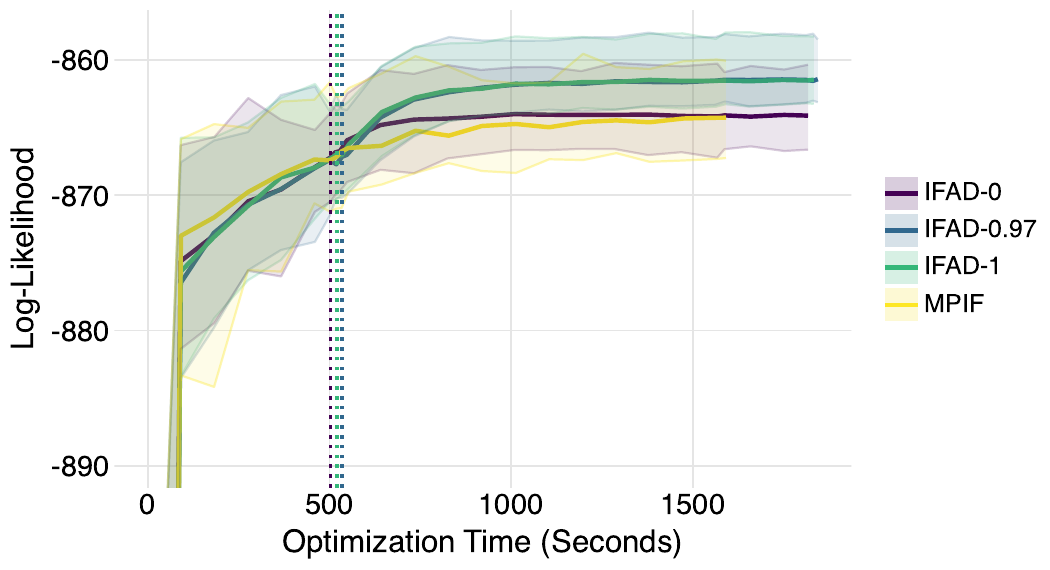}
\caption{Optimization progress of IFAD and MPIF for the Daphnia model.
The solid lines show the median log-likelihood across 50 searches, and the shaded regions cover the area between the 10th and 100th percentiles.
The vertical dotted lines indicate where the IFAD searches switch from MPIF to Adam optimization.
The plot shows log-likelihoods above $-890$; all searches contribute to the percentiles.}
\label{fig:daphnia-progress}
\end{figure}

\begin{figure}[htbp!]
\centering
\includegraphics[width=14cm]{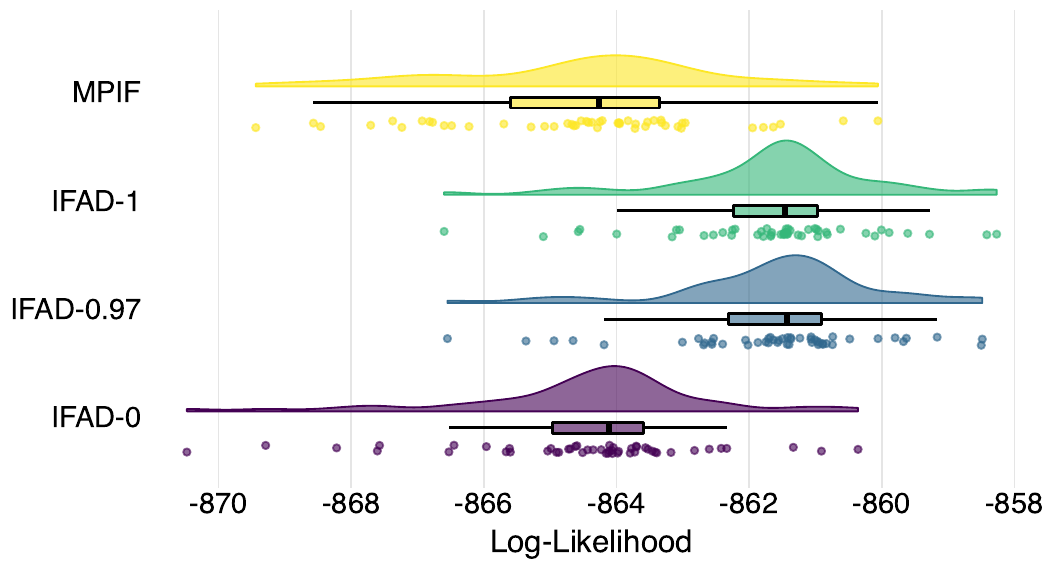}
\caption{Raincloud plot of the searches for the Daphnia model.
The half-violin plots show the distributions of all 50 final log-likelihood estimates for each method, the boxplots show the medians and quartiles, and each point represents one search.
IFAD-0.97 and IFAD-1 have similar distributions, with generally higher log-likelihoods than MPIF and IFAD-0.}
\label{fig:daphnia-raincloud}
\end{figure}

\clearpage
\section{Comparison with Other Methods for the Dhaka Model}
\label{sec:dhaka-competitors}

We compared IFAD with the transition-density method of \citet{ditlevsen19}
(DS19) and the differentiable particle filter of \citet{corenflos21}
(CTDD21) on the Dhaka
cholera model. We tested whether these methods could find high-likelihood
parameter estimates from the wide range of starting points used in the
main text, and whether they could improve estimates obtained by IF2.
We describe our implementations below.

\subsection{Methods}

\paragraph{DS19.}
The method of \citet{ditlevsen19} combines an approximate transition density
with particle smoothing and stochastic approximation expectation
maximization (SAEM). Their construction assumes that one smooth state
coordinate is observed directly and that noise acts directly on the
remaining coordinates. Applying it to Dhaka requires some modifications:
the six state coordinates $(S,I,R^1,R^2,R^3,M)$ are driven by one Brownian
motion, and the monthly death counts are observed with measurement error.

The SAEM objective includes the process log-density, which DS19 must
therefore be able to evaluate. Replacing the
SAEM maximization with a score update retains this requirement: the score
includes derivatives of the process log-density, evaluated along the
imputed path. Our backward smoother also evaluates the transition density
to sample that path. For Dhaka, a single Euler step is driven by one
Brownian increment and has no density with respect to Lebesgue measure
on the six-dimensional state space. Composing 20 Euler steps gives a
monthly transition that we can
simulate, but whose density we cannot evaluate directly after integrating
out the intermediate states. IF2 and IFAD use this simulator without
evaluating its transition density; IFAD additionally differentiates
through the simulation. They retain the Euler discretization error but
do not require a separate approximation to the monthly transition density.

We approximate the transition over each month by a Gaussian distribution.
The local strong-order-1.5 approximation has a covariance matrix of rank
at most two, so it does not have a density in the six-dimensional state
space. We instead calculate the mean and covariance over a full month by
propagating the second-order mean and covariance through 20 time steps.
We linearize the dynamics at each step and integrate out the intermediate
states. The resulting monthly covariance matrix can have full rank,
although it remains poorly conditioned. To stabilize the calculation,
we add $10^{-12}$ times its largest diagonal entry to each diagonal entry.
This gives a monthly transition density that differs from the Euler
approximation used to evaluate the final parameter estimates.

The rank deficiency also affects parameter updates on an augmented path.
At a fixed interior state, the range of the unregularized local covariance is
\[
 \operatorname{span}\left\{
 (-1,1,0,0,0,0)^{\mathsf T},\;
 (0,-\gamma-m,\gamma,0,0,m)^{\mathsf T}
 \right\}.
\]
This follows by writing the diffusion vector as $g=q(-1,1,0,0,0,0)^{\mathsf T}$:
its diffusion derivatives are parallel to $g$, while the drift derivative
$(Db)g$ supplies the second direction. The plane depends on $m/\gamma$.
Changing this ratio changes the support of the local Gaussian transition,
so augmenting the states at every internal step does not provide an
ordinary full-state density for parameter optimization. Changes in the
transition mean can also change its affine support. An augmentation in
terms of driving noise would require a different parameter-update scheme.
Our implementation adds the covariance floor only after composing the
monthly transition, not at each internal step.

There is a separate difficulty with finely augmented paths even when the
transition densities are nonsingular. As the time step decreases, transitions
become concentrated and observations at intermediate steps are uninformative,
which can make particle resampling and backward sampling ineffective
\citep{karppinen24}. This concerns mixing of the path sampler and does not
by itself establish mutual singularity of path laws at different parameter
values. Our smoother samples monthly endpoints, so it does not assess
this fine-grid behavior. We restrict the comparison to this implementable
adaptation of DS19; constructing and validating a different augmented-path
method is beyond the present comparison.

We use each month's observation to guide the particle filter toward states
that are consistent with the observed death count. For each particle, we
calculate the measurement variance using the predicted monthly death count
and hold this variance fixed when conditioning on the observation.
This gives a Gaussian distribution from which we sample the next state.
We weight each sample by the original measurement density multiplied by
the ratio of the transition density to the proposal density.

At each optimization step, we run the particle filter through the data,
then sample a trajectory of monthly states backward through the stored
particles. We differentiate the complete-data log-likelihood of this
trajectory and average the resulting scores using stochastic approximation.
For the 23-parameter Dhaka model, these searches use averaged-score ascent.
This differs from averaging and maximizing the complete-data objective in
SAEM. Section~\ref{sec:ds19-update-equations} gives both updates and describes
the separate numerical-SAEM comparison. We clip the scores and use
backtracking, with periodic checks of the log-likelihood under the monthly
Gaussian transition approximation. The 100- and 1,000-particle fits use
100 particles for these checks; the separately tuned 5,000-particle
experiment uses 5,000.
We also checked the particle filter on the harmonic oscillator example of
\citet{ditlevsen19}, for which an analytic Gaussian filter is available.
With 100 particles and 1,000 observations, the root mean squared difference
between the particle and analytic filtered means was 0.00523.

\paragraph{CTDD21.}
The method of \citet{corenflos21} uses entropy-regularized optimal transport
to make particle resampling differentiable. We implemented it using the
same Euler simulator as in the main text, with 20 steps per month.
The transport calculation uses all six state coordinates, scaled as in
the published implementation. We solve the transport problem using
Sinkhorn iterations, gradually decreasing the regularization parameter
to 0.25 with an annealing scale of 0.75. We use a convergence tolerance
of $10^{-3}$ and a maximum of 100 Sinkhorn iterations. To calculate
derivatives, we hold the converged potentials fixed and differentiate
through one further fixed-point update. As in the published implementation,
we clip each component of the derivative propagated through the transport
step to $[-1,1]$. We resample when the effective sample size falls below
half the number of particles.

We maximize the resulting log-likelihood estimate using Adam, with new
random numbers for the particle filter at each iteration and a gradient
norm limited to 100. If a search becomes invalid or repeatedly deteriorates,
we return to its best valid parameter estimate, reset the Adam moments,
and halve the learning rate. We allow at most six such restarts.
Entropic resampling changes the particle filter's likelihood approximation,
so the objective used for optimization can differ from the ordinary
particle filter estimate used to evaluate the results.

\subsection{Experimental settings}

For each method, we performed 100 global searches from starting parameter
vectors drawn uniformly from the same wide bounding box as in the main
text. We tested 100 and 1,000 particles, using learning rates of 0.02 for
CTDD21 and 0.10 for DS19. We allowed each search 800 seconds and
selected parameter estimates at approximately 700 seconds. For CTDD21,
we selected the saved estimate with the highest fitting objective before
700 seconds. If none was valid, we used another valid saved estimate or
the initial parameter vector. For DS19, we selected the first saved
estimate at or after 700 seconds, or the last estimate if the search ended
earlier. Thus a DS19 search could exceed the selection time by one
iteration.

To assess the benefit of an IF2 warm start, we also initialized both
methods at the 100 IF2 estimates used to initialize IFAD-0.97 in the main
text's comparable computational effort experiment. We used these estimates
directly, including those outside the initial bounding box. The IF2 stage
used 175 iterations and 5,000 particles, taking 92.160 seconds under the
timing procedure described in the main text. We then allowed each method
400 iterations or another 850.657 seconds, giving a total time limit of
942.817 seconds. The learning rates were 0.0002 for CTDD21 and 0.0005
for DS19, at both 100 and 1,000 particles. We selected the last valid
CTDD21 estimate before the time limit and the first DS19 estimate
at or after it, or the last estimate if the search ended earlier.

We kept the other optimization settings fixed when increasing the number
of particles from 100 to 1,000. The 1,000-particle experiments ran
sequentially on one RTX 3090 GPU, as did the 100-particle experiments
initialized by IF2. The 100-particle DS19 global searches ran four at
a time on that GPU. Consequently, both the particle count and GPU
scheduling affect the number of iterations completed within the time limit.
We also ran DS19 with 5,000 particles from the global starting points.
For that experiment, we separately tuned the learning rate to 0.20 and
began decreasing the stochastic approximation step size earlier.

We evaluated all selected parameter vectors independently using the Euler
approximation with 20 steps per month, including estimates from searches
that ended early. Each evaluation used 36 particle filter replicates with
5,000 particles. We averaged the likelihood estimates before taking the
logarithm. For the optimization progress plots, we also evaluated each
saved iterate using one replicate with 5,000 particles. These plots show
the iterates visited during optimization, which can differ from the
estimates selected for the final comparison. As in the main text, we
excluded the additional evaluations from the reported computation times.

\subsection{Results}

Our findings are summarized in Tables~\ref{tab:dhaka-competitors}
and~\ref{tab:dhaka-competitors-paired} and
Figures~\ref{fig:dhaka-competitors-final}--\ref{fig:dhaka-competitors-parameters}.
The figures show CTDD21 with 100 particles and DS19 with 1,000,
using the same particle count for each method with and without an IF2
warm start. IFAD-0.97 uses 5,000 particles. The tables include all particle
counts tested.

Both methods attained lower median log-likelihoods than IFAD-0.97 in the
global searches. At the particle counts shown in the figures, the medians
were $-4554.67$ for CTDD21 and $-3998.93$ for DS19, compared with
$-3744.40$ for IFAD-0.97. Increasing the number of particles did not
consistently improve performance within the allotted time. For CTDD21,
increasing from 100 to 1,000 particles reduced the median number of updates
from 255 to 28, and the median log-likelihood fell to $-6312.89$.
For DS19, the median improved from $-4100.34$ with 100 particles to
$-3998.93$ with 1,000 particles, while the median number of updates rose
from 55 to 147. This comparison also includes the change in GPU scheduling
described above. With 5,000 particles and separately tuned settings,
DS19 attained a median of $-4100.40$ after a median of 47 updates.

\begin{table}[htbp!]
\centering
\small
\begin{tabular}{llrrrr}
\toprule
Method & Start & $J$ & Median & Maximum & Updates \\
\midrule
CTDD21 & Box & 100 & -4554.67 & -3842.43 & 255 \\
CTDD21 & Box & 1,000 & -6312.89 & -4366.30 & 28 \\
DS19 & Box & 100 & -4100.34 & -3872.97 & 55 \\
DS19 & Box & 1,000 & -3998.93 & -3798.03 & 147 \\
DS19 & Box & 5,000 & -4100.40 & -3871.57 & 47 \\
CTDD21 & IF2 & 100 & -3765.91 & -3749.10 & 213 \\
CTDD21 & IF2 & 1,000 & -3766.61 & -3749.23 & 20 \\
DS19 & IF2 & 100 & -3768.72 & -3752.20 & 252 \\
DS19 & IF2 & 1,000 & -3769.04 & -3751.03 & 171 \\
IF2 checkpoint & Box & 5,000 & -3766.64 & -3749.30 & 175 \\
IFAD-0.97 & IF2 & 5,000 & -3744.40 & -3743.68 & 400 \\
\bottomrule
\end{tabular}

\caption{Log-likelihoods from 100 searches per method and particle count
for the Dhaka model. $J$ is the number of particles used for optimization.
Box denotes initialization from the bounding box, and IF2 denotes an IF2
warm start. Median and maximum use independent particle filter evaluations
with 20 Euler steps per month. Updates gives the median number of
optimization updates. The IF2 checkpoint row reports the 175 iterations
used for initialization; the IFAD-0.97 row reports the 400 subsequent
gradient updates. The 5,000-particle DS19 experiment used separate
tuning, as described in the text.}
\label{tab:dhaka-competitors}
\end{table}

Starting from IF2 estimates substantially improved the final log-likelihoods
of both methods, but neither consistently improved on the IF2 estimates
themselves. With 100 particles, CTDD21 improved on 45 of the 100 starting
estimates and DS19 improved on 18. The median changes were $-0.10$
and $-1.78$ log-likelihood units, respectively. With 1,000 particles, these
counts were 48 and 24, and the median changes were $-0.03$ and $-1.15$.
CTDD21 completed a median of 213 updates with 100 particles and 20 with
1,000 particles; DS19 completed 252 and 171, respectively.
IFAD-0.97 continued from the same IF2 estimates to a median log-likelihood
of $-3744.40$, compared with $-3766.64$ at the end of IF2.

\begin{table}[htbp!]
\centering
\begin{tabular}{lrrrr}
\toprule
Method & $J$ & Fitting change & Euler-20 change & Improved \\
\midrule
CTDD21 & 100 & 5.56 & -0.10 & 45/100 \\
DS19 & 100 & 17.24 & -1.78 & 18/100 \\
CTDD21 & 1,000 & 0.73 & -0.03 & 48/100 \\
DS19 & 1,000 & 15.21 & -1.15 & 24/100 \\
\bottomrule
\end{tabular}

\caption{Changes during optimization from an IF2 warm start. For each search,
we subtract the value at its IF2 starting point from the value at its
selected parameter estimate. Fitting change and Euler-20 change give the
medians of these differences for the fitting objective and the independently
evaluated log-likelihood. Improved is the number of searches whose
independently evaluated log-likelihood increased.}
\label{tab:dhaka-competitors-paired}
\end{table}

For both methods and both particle counts, the median change in the fitting
objective was positive even though the median change in the independently
evaluated log-likelihood was negative (Table~\ref{tab:dhaka-competitors-paired}).
The two evaluations use different likelihood approximations: DS19
fits the monthly Gaussian transition model, while CTDD21 uses entropic
resampling. The fitting objectives are also estimated using a single
particle filter, and the updates depend on clipping and the numerical
safeguards described above. These experiments do not determine how much
each of these factors contributes to the disagreement.

\begin{figure}[p]
\centering
\includegraphics[width=\textwidth]{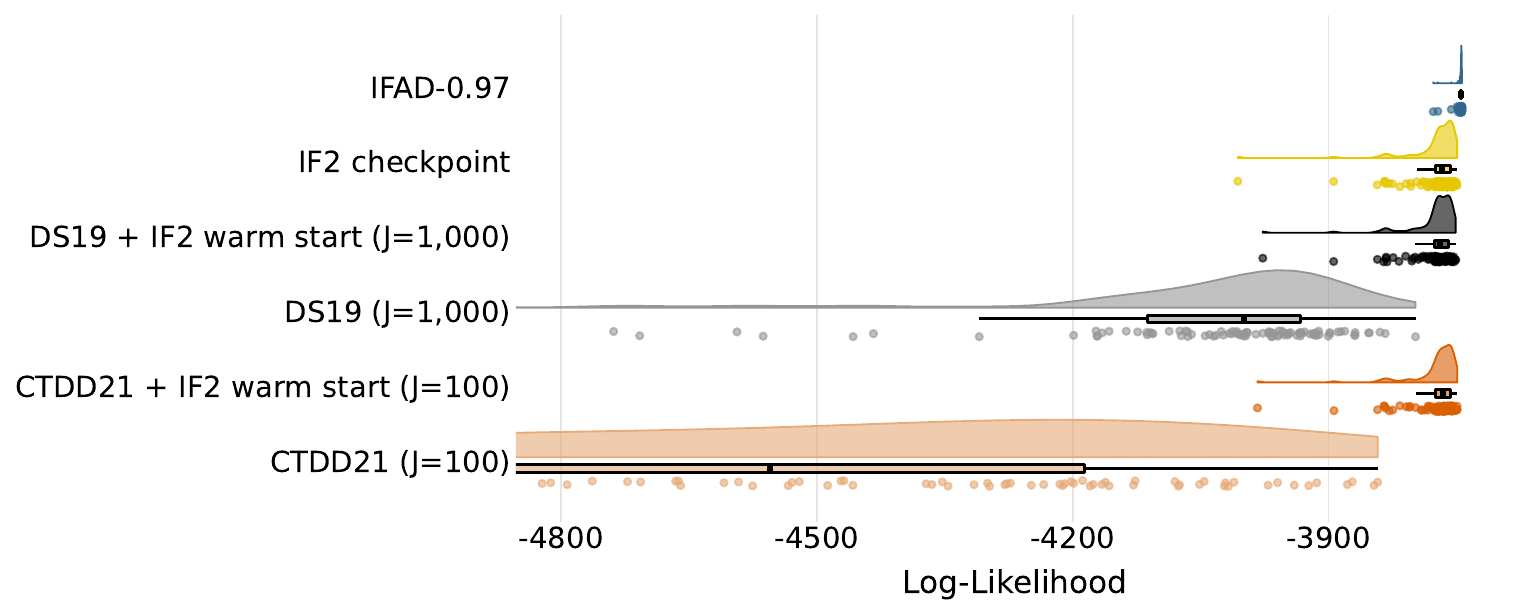}
\caption{Raincloud plots of the log-likelihoods found in 100 searches for the
Dhaka model, with and without an IF2 warm start. CTDD21 uses 100 particles,
DS19 1,000, and IFAD-0.97 5,000 during optimization. The IF2 checkpoint
shows the starting estimates for the warm starts. The plot shows
log-likelihoods between $-4800$ and $-3735$; all 100 searches contribute to
the densities, medians, and quartiles. Each half-violin is scaled to the
same maximum height.}
\label{fig:dhaka-competitors-final}
\end{figure}

\begin{figure}[p]
\centering
\includegraphics[width=\textwidth]{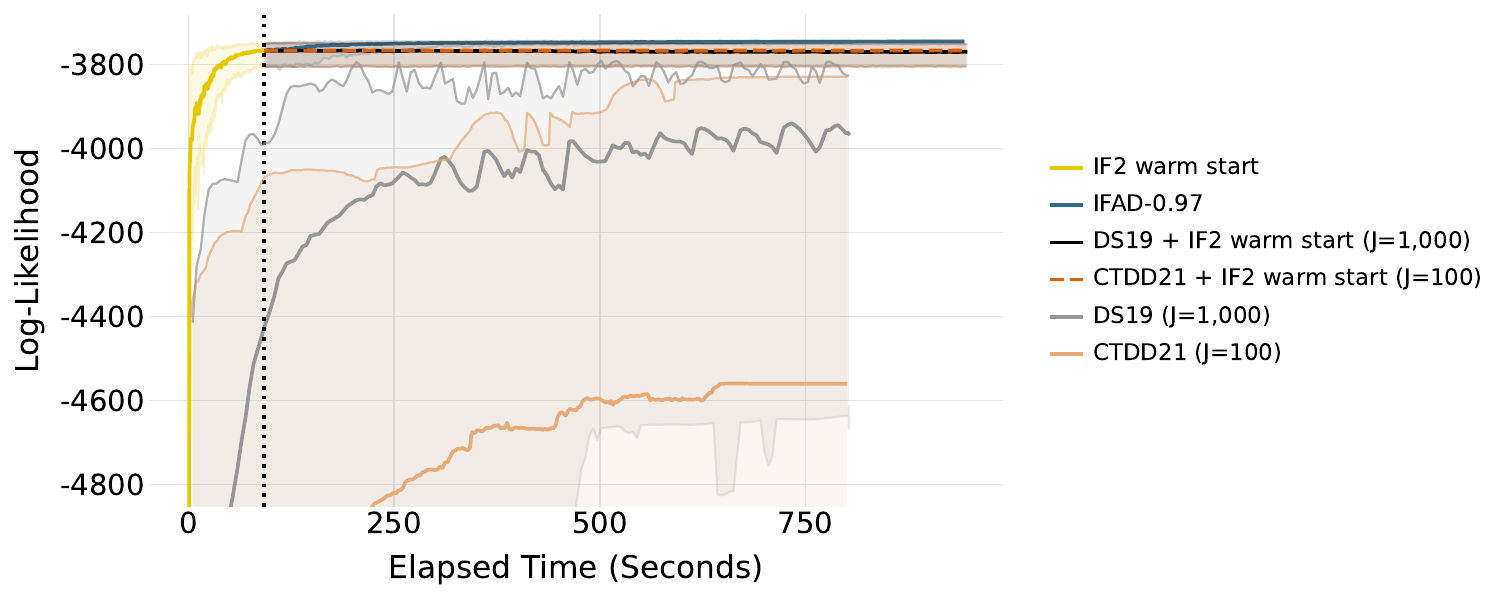}
\caption{Optimization progress for the Dhaka model, with and without an IF2
warm start. CTDD21 uses 100 particles, DS19 1,000, and IFAD-0.97
5,000 during optimization. Heavy lines show the median log-likelihood,
and shaded regions cover the area between the 10th and 100th percentiles.
The vertical dotted line indicates the end of IF2 at 92.160 seconds.
The plot shows log-likelihoods between $-4800$ and $-3735$; all 100 searches
contribute to the percentiles. Each saved iterate is evaluated with one
particle filter replicate, compared with 36 replicates for the selected
estimates in Table~\ref{tab:dhaka-competitors}, using 5,000 particles in
both cases.}
\label{fig:dhaka-competitors-progress}
\end{figure}

\begin{figure}[p]
\centering
\includegraphics[width=\textwidth]{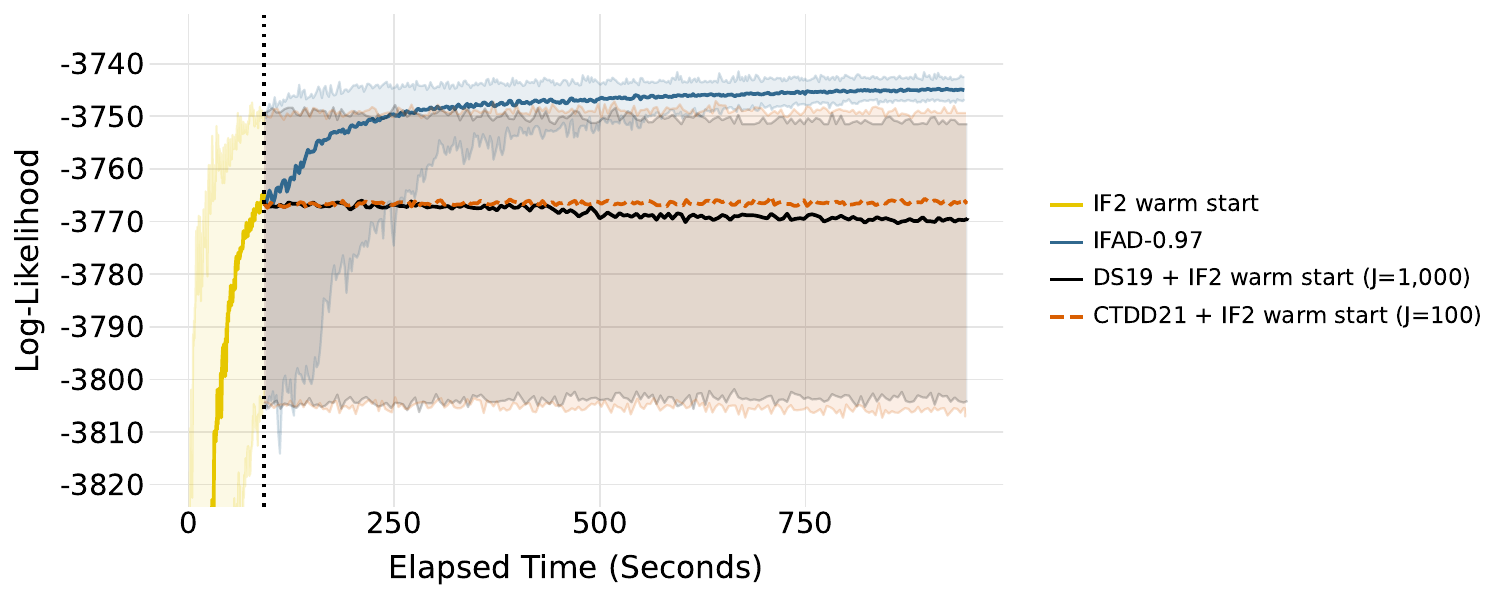}
\caption{Optimization progress with an IF2 warm start, showing log-likelihoods
between $-3820$ and $-3735$. Particle counts, lines, and shading are as in
Figure~\ref{fig:dhaka-competitors-progress}. Elapsed time includes the
92.160-second IF2 stage.}
\label{fig:dhaka-competitors-continuation}
\end{figure}

\begin{figure}[p]
\centering
\includegraphics[width=\textwidth]{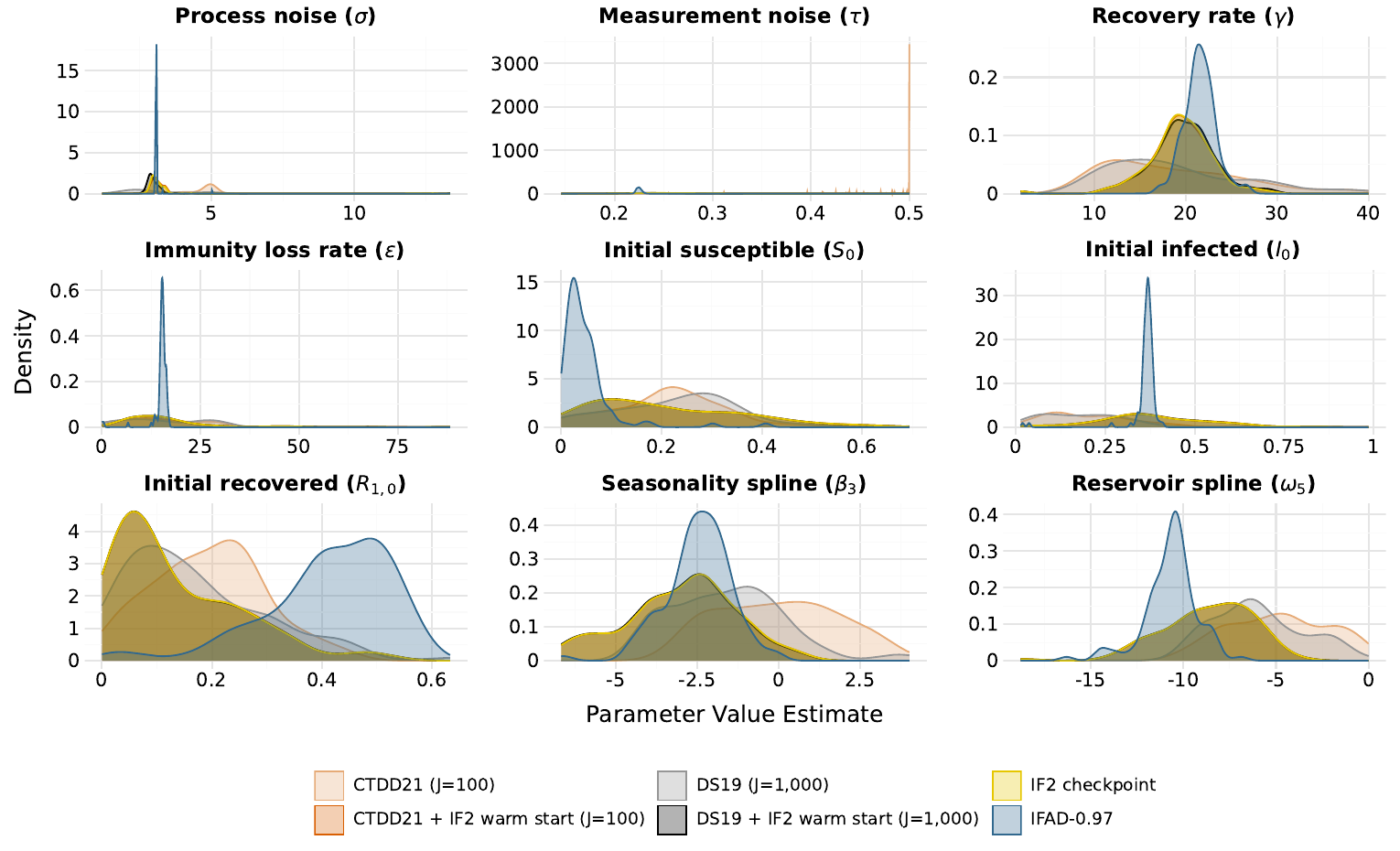}
\caption{Density plots of the final parameter estimates across 100 searches
for nine selected parameters of the Dhaka model. Each panel has its own
density scale. CTDD21 uses 100 particles and DS19 1,000, with and
without an IF2 warm start; IFAD-0.97 uses 5,000 particles. The IF2 checkpoint
shows the starting estimates for the warm starts.}
\label{fig:dhaka-competitors-parameters}
\end{figure}

\subsection{Reproducibility}

Source code and algorithmic settings are available at
\href{https://github.com/pypomp/dmop}{https://github.com/pypomp/dmop}, in the
\texttt{ditlevsen/} and \texttt{corenflos/} directories.
The file \path{imgs/competitors/README.md} gives instructions for reproducing
the tables and figures, identifies the numerical results used, and links
to plots showing the full range of log-likelihoods. The repository includes
the tables and figures needed to compile the supplement. Instructions for
retrieving archived parameter estimates and evaluations at individual
iterations are in \path{artifacts/README.md}.

\clearpage
\section{Comparisons Across Models}
\label{sec:cross-model-benchmarks}

We compared IFAD, IF2, DS19, and CTDD21 on the harmonic oscillator of
\citet{ditlevsen19}, the two-dimensional linear Gaussian model of
\citet{corenflos21}, and the Pypomp library implementation of a stochastic
volatility model at \texttt{pypomp.models.spx}, subsequently called
\texttt{spx}, in addition to the Daphnia and Dhaka models.
For each example, we used the same starting parameter vectors for all
methods and evaluated the resulting estimates under the same model.
The two Gaussian examples have analytic likelihoods, allowing us to
assess maximization error without Monte~Carlo evaluation error.

Table~\ref{tab:ds19-models} summarizes the DS19 implementation for each
panel. All five implementations evaluate process transition densities.
For the linear Gaussian model and \texttt{spx}, these
densities are available directly. For the oscillator, all methods fit
the same order-1.5 Gaussian approximation, whose likelihood we evaluate
by Kalman filtering. Thus these three comparisons use the same transition
model for fitting and evaluation.

For Daphnia and Dhaka, we also use moments from the order-1.5 expansion,
but compose them over smaller time steps and approximate the transition
between observations by a Gaussian distribution. This composition adds
a linearization of the nonlinear dynamics. It does not imply that the
resulting Gaussian approximation has strong order 1.5 over the full
observation interval. For Dhaka, a single local Gaussian approximation
has covariance rank at most two in six dimensions, as explained in
Section~\ref{sec:dhaka-competitors}. Composing the local covariances allows
noise to propagate into additional state directions. For Daphnia, we
compose moments at the 0.25-day simulation step rather than apply a
single Taylor expansion across the five-day observation interval.
We have not tested the accuracy of such a single coarse step.
Both models are evaluated using their original Euler simulators, so the
Gaussian approximation introduces a difference between the fitting and
evaluation likelihoods.

\begin{table}[htbp]
\centering
\small
\setlength{\tabcolsep}{4pt}
\renewcommand{\arraystretch}{1.15}
\begin{tabular}{@{}>{\raggedright\arraybackslash}p{0.18\textwidth}
>{\raggedright\arraybackslash}p{0.29\textwidth}
>{\raggedright\arraybackslash}p{0.21\textwidth}
>{\raggedright\arraybackslash}p{0.25\textwidth}@{}}
\toprule
Model & DS19 transition density & Parameter update & Likelihood evaluation \\
\midrule
A: Harmonic oscillator
& Order-1.5 Gaussian approximation at the observation interval
& SAEM; numerical maximization of averaged sufficient statistics
& Kalman likelihood of the same approximation \\
\addlinespace
B: Linear Gaussian
& Exact Gaussian density of the specified discrete-time model
& SAEM; explicit maximization using averaged sufficient statistics
& Kalman likelihood of the same model \\
\addlinespace
C: SPX
& \texttt{spx} transition law: Gaussian above the variance floor, with a point mass at the floor
& Averaged scores and Adam
& Particle-filter likelihood of the same \texttt{spx} model \\
\addlinespace
D: Daphnia
& Gaussian approximation across each observation interval, composed from local order-1.5 moments
& Averaged scores and Adam
& Particle-filter likelihood of the original Euler model; transition and boundary rules differ \\
\addlinespace
E: Dhaka
& Monthly Gaussian approximation composed from local order-1.5 moments
& Averaged scores and Adam; numerical SAEM in a separate comparison
& Particle-filter likelihood of the original Euler model; transition differs \\
\bottomrule
\end{tabular}
\caption{DS19 implementations for panels A--E of
Figure~\ref{fig:benchmark-final-likelihoods}.
Each uses transition-density evaluation for backward smoothing and either
maximization or differentiation of a complete-data objective.
For SPX, the density includes a discrete mass and is exact for the
implemented model, not for a continuous-time diffusion.
The Daphnia and Dhaka covariances are regularized as described below.
The separate Dhaka SAEM comparison retains the same Gaussian transition.}
\label{tab:ds19-models}
\end{table}

\subsection{SAEM and averaged-score updates}
\label{sec:ds19-update-equations}

Stochastic approximation expectation maximization (SAEM) alternates
imputation of the latent states with maximization of an averaged
complete-data log-likelihood.
Let $x$ denote a latent trajectory and write
$C(\vartheta;x)=\log\widetilde p_\vartheta(x,y)$ for the complete-data
objective of the fitted model. The tilde allows for a transition-density
approximation. Deterministic initial states are functions of the candidate
parameter; only random states are imputed. At iteration $m$, a particle
filter followed by backward
simulation draws one trajectory $X^{(m)}$ from an approximation to the
smoothing distribution at $\vartheta_{m-1}$.
SAEM updates a function of the candidate parameter $\vartheta$:
\begin{align}
 \overline Q_m(\vartheta)
 &= (1-\rho_m)\overline Q_{m-1}(\vartheta)
       +\rho_m C(\vartheta;X^{(m)}),
 \label{eq:benchmark-saem-q}\\
 \vartheta_m &\in
       \mathop{\mathrm{argmax}}_{\vartheta}\overline Q_m(\vartheta).
 \label{eq:benchmark-saem-m}
\end{align}
In the maximization step, we hold the imputed paths and their averaging
weights fixed. This maximization can be carried out numerically; a step
that increases $\overline Q_m$ without maximizing it gives a generalized
EM update.

The averaged-score implementation instead calculates
\begin{align}
 g_m &= \left.\nabla_\vartheta C(\vartheta;X^{(m)})
                         \right|_{\vartheta=\vartheta_{m-1}},
 \label{eq:benchmark-path-score}\\
 \overline g_m &=(1-\rho_m)\overline g_{m-1}+\rho_m g_m,
 \label{eq:benchmark-score-average}
\end{align}
then uses the clipped $\overline g_m$ in an Adam step.
Derivatives hold the imputed path fixed; they do not pass through the
particle filter or backward sampler.
With exact smoothing, the conditional expectation of $g_m$ is the score
of the fitted model's observed-data likelihood.
The particle smoother approximates this expectation.
Before clipping or safeguard resets, the distinction between the two
averages can be written as
\begin{align}
 \nabla\overline Q_m(\vartheta)
   &=\sum_{r=1}^m w_{mr}\nabla_\vartheta C(\vartheta;X^{(r)}),
 &\overline g_m
   &=\sum_{r=1}^m w_{mr}
       \nabla_\vartheta C(\vartheta_{r-1};X^{(r)}),
 \label{eq:benchmark-q-versus-score}\\
 w_{mr}&=\rho_r\prod_{k=r+1}^m(1-\rho_k).\nonumber
\end{align}
The SAEM gradient evaluates all path contributions at the same candidate
parameter, whereas $\overline g_m$ averages gradients evaluated at the
parameter values used in earlier iterations. Consequently,
$\overline g_m$ generally differs from the gradient of $\overline Q_m$
at the current parameter.
Our gain schedule is $\rho_m=1$ for $m\leq b$ and
$\rho_m=(m-b)^{-0.9}$ thereafter.
We use $b=30$ for the Gaussian SAEM and averaged-score runs, and $b=3$
for numerical SAEM on Dhaka.
Clipping and Adam act on the existing estimation scales.
In the Dhaka score implementation, the clipped average is also carried
into the next iteration, and failed safeguard checks can reset it.

\paragraph{Linear Gaussian model: explicit SAEM.}
The latent states satisfy
$X_t=\operatorname{diag}(a_1,a_2)X_{t-1}+\varepsilon_t$,
with $\varepsilon_t\sim\mathcal N(0,0.5I)$,
$Y_t\mid X_t\sim\mathcal N(X_t,0.1I)$, and
$X_0\sim\mathcal N(0,I)$.
For a fixed path, the terms depending on $a=(a_1,a_2)$ are
\[
 C(a;x)=-\frac{1}{2(0.5)}
     \sum_{t=1}^T\sum_{d=1}^2(x_{t,d}-a_dx_{t-1,d})^2
     +\text{constant}.
\]
For each coordinate, the sufficient statistics are
\[
 S_{00,d}(x)=\sum_t x_{t-1,d}^2,
 \qquad S_{10,d}(x)=\sum_t x_{t,d}x_{t-1,d}.
\]
We update their averages by
$\overline S_m=(1-\rho_m)\overline S_{m-1}+\rho_m S(X^{(m)})$.
The M-step is
\begin{equation}
 a_{d,m}=\operatorname{clip}_{[-0.98,0.98]}
       \left(\frac{\overline S_{10,d,m}}{\overline S_{00,d,m}}\right).
 \label{eq:benchmark-linear-mstep}
\end{equation}
Each iteration therefore requires a particle filter and backward
simulation of one path, after which the maximization reduces to two
scalar ratios.

\paragraph{Harmonic oscillator: SAEM with a numerical M-step.}
Write $x_t=(v_t,u_t)^\mathsf T$, where the positions $v_t$ are observed
without measurement error and the velocities $u_t$ are imputed.
For $\theta=(D,\gamma,\sigma)$ and observation interval $\Delta=0.02$,
the fitted Gaussian transition has
\begin{align*}
 A_\theta&=\begin{pmatrix}0&1\\-D&-\gamma\end{pmatrix},
 &F_\theta&=I+\Delta A_\theta+\tfrac12\Delta^2 A_\theta^2,\\
 Q_\theta&=\sigma^2
 \begin{pmatrix}
 \Delta^3/3&\Delta^2/2-\gamma\Delta^3/3\\
 \Delta^2/2-\gamma\Delta^3/3&
       \Delta-\gamma\Delta^2+\gamma^2\Delta^3/3
 \end{pmatrix}.
\end{align*}
Conditional on the first position, $U_0\sim\mathcal N(0,v_\theta)$,
where $v_\theta=\sigma^2/(2\gamma)$.
Let $S_{ij}=\sum_{t=1}^T x_{t-1+i}x_{t-1+j}^\mathsf T$ for
$(i,j)=(0,0),(1,0),(1,1)$, and let $h=u_0^2$.
We average these quantities using the recursion above, giving the
following objective up to constants independent of $\theta$:
\begin{equation}
 \overline Q_m(\theta)=-\tfrac12\left\{
  \log v_\theta+\frac{\overline h_m}{v_\theta}
  +T\log\det Q_\theta
  +\operatorname{tr}\!\left(Q_\theta^{-1}\overline R_m(\theta)\right)
 \right\},
 \label{eq:benchmark-oscillator-mstep}
\end{equation}
where
\[
 \overline R_m(\theta)=\overline S_{11,m}
     -F_\theta\overline S_{10,m}^{\mathsf T}
     -\overline S_{10,m}F_\theta^{\mathsf T}
     +F_\theta\overline S_{00,m}F_\theta^{\mathsf T}.
\]
We use L-BFGS-B on $(\log D,\log\gamma,\log\sigma)$, retaining a finite
step only if this objective does not decrease, up to numerical tolerance.
To reduce numerical cancellation, we evaluate this quadratic form using
accumulated state increments. The M-step thus uses averaged sufficient
statistics of the imputed paths.

\paragraph{SPX: averaged complete-data scores.}
Let $r_t$ be the observed return and $v_t$ the latent variance.
The \texttt{spx} conditional variance transition has Gaussian mean
and standard deviation
\[
 m_t(v)=v+\kappa(\bar v-v)+\xi\rho(r_{t-1}-\mu+v/2),
 \qquad s(v)=\xi\sqrt{v(1-\rho^2)},
\]
followed by a floor $\epsilon=10^{-32}$.
Writing $\phi(z;m,s^2)$ for a Gaussian density and $\Phi$ for the
standard normal distribution function, its transition is
\[
 f_{\theta,t}(v'\mid v)=
 \begin{cases}
  \phi(v';m_t(v),s(v)^2),&v'>\epsilon,\\
  \Phi((\epsilon-m_t(v))/s(v)),&v'=\epsilon.
 \end{cases}
\]
The second case is a probability mass; the reference measure is Lebesgue
measure plus an atom at $\epsilon$.
The complete-path objective is
\begin{equation}
 C(\theta;v)=\sum_{t=1}^T\left\{
    \log f_{\theta,t}(v_t\mid v_{t-1})
    +\log\phi(r_t;\mu-v_t/2,v_t)\right\}.
 \label{eq:benchmark-spx-score-objective}
\end{equation}
The initial variance $v_0=V_0$ is a parameter, so its current candidate
value enters the first transition when differentiating.
For $v>0$, $\xi>0$, and $|\rho|<1$, the conditional transition variance
$\xi^2v(1-\rho^2)$ is positive. This scalar latent process therefore has
no missing covariance direction in the interior; its boundary atom is
included explicitly in the density above.
For this example, we update the parameters with Adam using the averaged
scores in Equations~(\ref{eq:benchmark-path-score})--(\ref{eq:benchmark-score-average}),
without a SAEM maximization step.
We obtain these complete-data derivatives by AD with the imputed path fixed,
without differentiating through the filtering and resampling calculations.
The runtime difference from
IFAD also reflects the particle counts and optimization settings below.

\paragraph{Daphnia: averaged scores for the Gaussian block approximation.}
For mesocosm $u$, let $m_{\theta,u,t}$ and $\Sigma_{\theta,u,t}$ be the
mean and covariance obtained by composing the local Gaussian transitions
over an observation interval, including the inoculation in the first interval.
Let $g_{\theta,u}$ be the product of the four negative-binomial measurement
probabilities in Section~S10.
For the eight imputed mesocosm paths, the complete-path objective is
\begin{equation}
 C(\theta;x)=\sum_{u=1}^{8}\sum_{t=1}^{10}
 \left\{\log\phi_8\!\left(x_{u,t};
 m_{\theta,u,t}(x_{u,t-1}),\Sigma_{\theta,u,t}(x_{u,t-1})\right)
       +\log g_{\theta,u}(y_{u,t}\mid x_{u,t})\right\}.
 \label{eq:benchmark-daphnia-score-objective}
\end{equation}
Initial states are fixed. The covariance uses a relative diagonal floor
of $10^{-8}$ on scaled states. The particle filter uses the observations
to guide proposals, with the original negative-binomial measurement
probabilities in the importance weights.
We approximate each observation interval by one eight-dimensional
Gaussian transition, composing local moments over 24 steps for the first
interval and 20 thereafter. The step size is 0.25 days, so these
transitions cover six days initially and five days between observations.
This approximates the distribution produced by the nonlinear dynamics
over the whole interval. Boundary handling also differs: we assign zero
weight to Gaussian proposals outside the allowed state range, whereas
the Euler simulator resets an offending state to zero and assigns a
log weight of $-150$ to that particle at the next observation.
Both differences can change the likelihood being maximized.
The local covariance need not be rank deficient in this model.
At a strictly positive state, the six independent noises span the six
directly forced coordinates. The integrated juvenile noises contribute
two further directions through maturation into susceptible adults, with
nonzero coefficients $\lambda_J\sigma_J^{\mathrm n}J^{\mathrm n}$ and
$\lambda_J\sigma_J^{\mathrm i}J^{\mathrm i}$.
The implemented order-1.5 covariance is therefore full rank in this
interior when the six fitted noise coefficients are positive, and its
range is independent of the parameter values there. At the specified
initial state several populations are zero, and the local covariance has
rank three. Daphnia thus differs from Dhaka's rank-two interior covariance.
The algebra and numerical checks, including parameter perturbations at
fixed states, are provided in \path{code/benchmarks/covariance_support.py}
and its accompanying note.
After backward simulation of one path per mesocosm, we differentiate this
sum with respect to the 38 estimated parameters and update them with
Adam using the averaged scores. As for SPX, we do not carry out a SAEM
maximization step. We evaluate the resulting parameter estimates using
the original Euler model from Section~S10.

\paragraph{Dhaka: score updates and numerical generalized SAEM.}
We compare score updates with numerical SAEM using the monthly Gaussian transition described in
Section~\ref{sec:dhaka-competitors}, with 20 internal time steps and
a relative covariance floor of $10^{-12}$.
Let $\mathcal R x$ reset the accumulated monthly deaths to zero, and
write $d_t=P_tM_t$ for the predicted deaths at a sampled endpoint.
Here $P_t$ is the population and $M_t$ is the accumulated death count
divided by population. Both methods use the complete-data objective
\begin{align}
 C(\theta;x)=\sum_{t=1}^T\big\{&
   \log\phi_6\!\left(x_t;m_{\theta,t}(\mathcal R x_{t-1}),
                     \Sigma_{\theta,t}(\mathcal R x_{t-1})\right)
   \nonumber\\[-2pt]
   &+\log[\phi(y_t;d_t,(\tau d_t+\delta)^2)+\delta]\big\},
 \qquad \delta=10^{-18}.
 \label{eq:benchmark-dhaka-objective}
\end{align}
The initial state $x_0=x_0(\theta)$ is recalculated at each candidate
parameter, while the sampled monthly endpoints are held fixed.
For the score updates, we use the gradient average in
Equation~(\ref{eq:benchmark-score-average}), with Adam and the safeguards
described above.
For numerical SAEM, we retain the paths and weights to calculate
$\overline Q_m(\theta)=\sum_{r=1}^m w_{mr}C(\theta;X^{(r)})$.
Each L-BFGS-B evaluation recalculates every retained path's objective
and gradient at the same candidate $\theta$.
Within each M-step, we divide the objective and its gradient by the fixed
scale $\max\{1,\|\nabla\overline Q_m(\theta_{m-1})\|_\infty\}$.
This reduces numerical line-search failures without changing the maximizer.
We allow at most 25 L-BFGS-B iterations per M-step and retain a candidate
only if it is finite and increases or leaves unchanged this fixed
$\overline Q_m$. An M-step interrupted by the time limit is discarded.
This gives a generalized-SAEM update when an explicit maximizer in terms
of sufficient statistics is unavailable.
We evaluate both methods with the Euler likelihood, so the evaluation
model differs from the monthly Gaussian approximation used for fitting.

\paragraph{Computation times for the Gaussian examples.}
We used 500 particles for DS19, IFAD, and IF2 in both Gaussian examples.
DS19 ran for 80 SAEM iterations, IFAD for 100 IF2 iterations followed by
300 gradient iterations, and IF2 for 600 iterations.
CTDD21 used 25 particles and 300 gradient iterations.
The common particle count for DS19, IFAD, and IF2 therefore does not
imply equal computational effort.

DS19 runs a particle filter and samples one trajectory by backward
simulation at every iteration, and both are included in its computation
time. Its maximization step is inexpensive for these models: the linear
Gaussian model has an explicit update, while the oscillator requires
optimization over three parameters using averaged sufficient statistics.
Neither maximization requires differentiation through the particle filter.
The analytic likelihoods are used only to evaluate the fitted parameters.
For all methods, we wait for the numerical calculations to finish before
recording the time, excluding compilation and likelihood evaluation.
IFAD, IF2, and DS19 used four efficiency cores of an Intel Core i9-13900K;
CTDD21 used two performance cores with four hardware threads.
Concurrent batches used separate core groups. These hardware differences
also affect the comparison of computation times.

\subsection{Search and evaluation settings}
\label{sec:benchmark-settings}

We ran 20 searches per method for the oscillator, linear Gaussian, SPX,
and Daphnia models, using the same starting parameter vectors for each
method. For Dhaka, we used the 100 searches from
Section~\ref{sec:dhaka-competitors}.

For SPX, we used 500 particles for IFAD and IF2, 100 for DS19, and 25
for CTDD21. IFAD used $\alpha=0.97$ and ran for 100 IF2 iterations followed by 300 gradient
iterations, IF2 for 600 iterations, and DS19 and CTDD21 for 300 updates
each. The Adam learning rates started at $0.001$ for IFAD and $0.01$
for DS19 and CTDD21, decreasing to one tenth of these values by cosine decay.

For Daphnia, we used 1,000 particles per mesocosm for MPIF and IFAD,
and 100 for DS19 and CTDD21. We ran MPIF for 680 iterations.
IFAD used a separate 200-iteration MPIF warm start followed by 200
gradient steps. To compare refinement of the same starting estimate,
we initialized DS19 and CTDD21 at this MPIF warm start and allowed each
the time taken by the IFAD gradient stage, with a maximum of 2,000 updates.
Their learning rates decreased from $0.001$ to $0.0001$ by cosine decay
over the allotted time. We discarded updates completed after the time
limit and included the common MPIF warm start in the computation times
for all three methods.

We evaluated the final SPX estimates using 24 independent particle-filter
replicates with 5,000 particles. For Daphnia, we used 10 replicates with
2,000 particles per mesocosm under the original Euler model.
We averaged the likelihood estimates for each mesocosm before taking
logarithms, then summed the log-likelihoods over mesocosms, combining
their independent Monte~Carlo errors.
For the optimization trajectories, we used four replicates with 1,000
particles for SPX and two replicates with 500 particles per mesocosm
for Daphnia. All these evaluations were independent of fitting;
we did not use them to select the reported parameter estimates.

SPX fits used two performance cores with four hardware threads, except
for one recovered starting point for which all four methods used four
efficiency cores. Daphnia fits used four efficiency cores.
These runs used the same processor as the Gaussian comparisons.
The archived Dhaka comparison used GPU computations as described in
Section~\ref{sec:dhaka-competitors}.

\subsection{Repeated-search results}

IFAD and DS19 both approached the maximum likelihood for the two
Gaussian models (Table~\ref{tab:benchmark-results} and
Figure~\ref{fig:benchmark-final-likelihoods}). DS19 did so more quickly,
benefiting from the inexpensive maximization steps described above.
IF2 gave more variable results, while CTDD21 attained lower likelihoods,
particularly for the oscillator. DS19 also had the highest median
log-likelihood and shortest median computation time for SPX.
We therefore found no consistent improvement from IFAD on these examples.

For Daphnia, IFAD attained the highest maximum and median log-likelihoods.
The median was $-861.30$, compared with $-864.98$ for MPIF, $-867.02$
for DS19, and $-865.03$ for CTDD21.
Figure~\ref{fig:benchmark-optimization-detail} shows IFAD continuing to
improve after the common MPIF warm start, while DS19 and CTDD21 made
less progress with the same time allowance.
The full trajectories are shown in Figure~\ref{fig:benchmark-optimization-full};
the shaded regions in both figures cover the 10th through 100th
percentiles, as in Section~\ref{sec:dhaka-competitors}.
We retained the last finite estimate from one DS19 search that stopped
prematurely. For Dhaka, the differences were larger: IFAD attained a
median log-likelihood of $-3744.40$, compared with $-3755.41$ for IF2,
$-3769.04$ for DS19, and $-3765.91$ for CTDD21.

The Daphnia comparison does not determine how much of the DS19 shortfall
comes from the Gaussian transition approximation. Smoothing error and
optimization may also contribute. DS19 used 100 particles per mesocosm,
compared with 1,000 for IFAD, and estimated each gradient from one
backward-sampled path per mesocosm for the 38-parameter optimization.
Moreover, the Daphnia implementation uses averaged-score Adam updates,
whereas the Gaussian examples use SAEM maximization with sufficient
statistics. We have not compared these updates with numerical SAEM on
Daphnia. Evaluating both the Gaussian fitting objective and the Euler
likelihood along the same searches would help assess the effect of the
transition approximation: improvement in the former accompanied by
deterioration in the latter would indicate that improving the approximate
objective need not improve the likelihood of interest. We have not
carried out this diagnostic.

The Gaussian examples provide useful numerical checks because their
likelihoods are available analytically. Once methods repeatedly approach
these maxima, the examples offer limited evidence about further
improvements in likelihood maximization. The SPX comparison still shows
differences between methods and is not saturated in this sense.
Daphnia and Dhaka retain difficulties encountered in scientific
applications, including many estimated parameters, nonlinear dynamics,
partially observed states, and, for Daphnia, multiple experimental units.
Assessing methods for such applications requires benchmarks that retain
these difficulties alongside tractable examples used to verify implementations.

\subsection{Score updates and numerical SAEM on Dhaka}

To assess whether the score updates accounted for the poor performance
of DS19 on Dhaka, we repeated the comparison using numerical SAEM with
the same monthly Gaussian transition approximation and particle smoother.
We used 1,000 particles and allowed 850 seconds per search, starting
from 20 global parameter vectors and 20 IF2 estimates.
Numerical SAEM gave higher Euler log-likelihoods than the score updates
for only two of the 20 global starts
(Figure~\ref{fig:dhaka-saem-ab} and Table~\ref{tab:dhaka-saem-ab}).
The median paired difference was $-608.83$ log-likelihood units, taking
numerical SAEM minus score updates. The numerical M-steps were expensive:
SAEM completed a median of 15.5 outer iterations in the allotted time,
compared with 86 for the score updates.

Starting from the IF2 estimates reduced this difference substantially.
The median paired difference was $-0.62$ log-likelihood units, although
one SAEM search finished $39.59$ units below its score counterpart.
SAEM gave higher likelihoods in four of the 20 pairs.
Neither method improved on the IF2 estimates at the median, with changes
of $-0.27$ for score updates and $-1.37$ for numerical SAEM.
Thus numerical maximization of the SAEM objective did not improve the
Dhaka results.
Both methods retained the same monthly Gaussian transition and particle
smoother, so this comparison does not measure the effect of the transition
approximation. Increasing the number of particles can improve the
particle approximation to this Gaussian model, but does not remove its
difference from the Euler model. The results therefore leave open the
relative contributions of transition approximation, smoothing error,
numerical conditioning, and optimization cost to the Dhaka shortfall.

\clearpage
\begin{figure}[p]
\centering
\includegraphics[width=\textwidth]{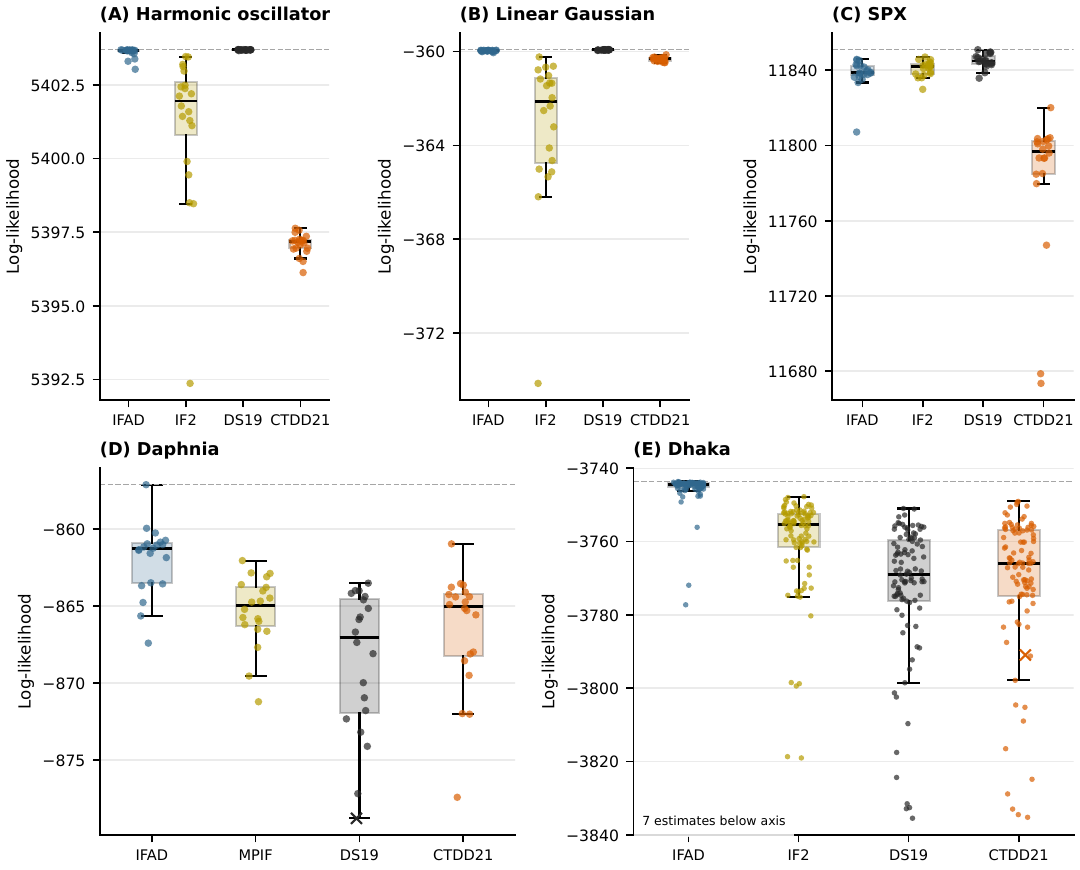}
\caption{Final log-likelihoods from repeated parameter searches.
Each point represents one of 20 searches per method in panels A--D
or 100 searches in panel E.
Boxes show the median and interquartile range; whiskers extend to the most
extreme values within 1.5 interquartile ranges of the box.
All estimates are shown in panels A--D. Panel E uses the range
$[-3840,-3740]$; the number of estimates below the axis is indicated.
The uncropped figure is available with the source code, and all estimates
contribute to the boxes and numerical summaries.
Crosses mark searches that stopped prematurely, for which we evaluated
the last estimate obtained before failure or the time limit.
Dashed lines mark the analytic maximum in panels A and B
and the best displayed estimate in panels C--E.
Likelihoods are evaluated independently of fitting, using the same model
within each panel. DS19 and CTDD21 use the common MPIF warm start in D
and the common IF2 warm start in E. Panel E uses the same estimates as
the warm-start comparison in Figure~\ref{fig:dhaka-competitors-final}.
All axes are linear.}
\label{fig:benchmark-final-likelihoods}
\end{figure}

\begin{figure}[p]
\centering
\includegraphics[width=\textwidth]{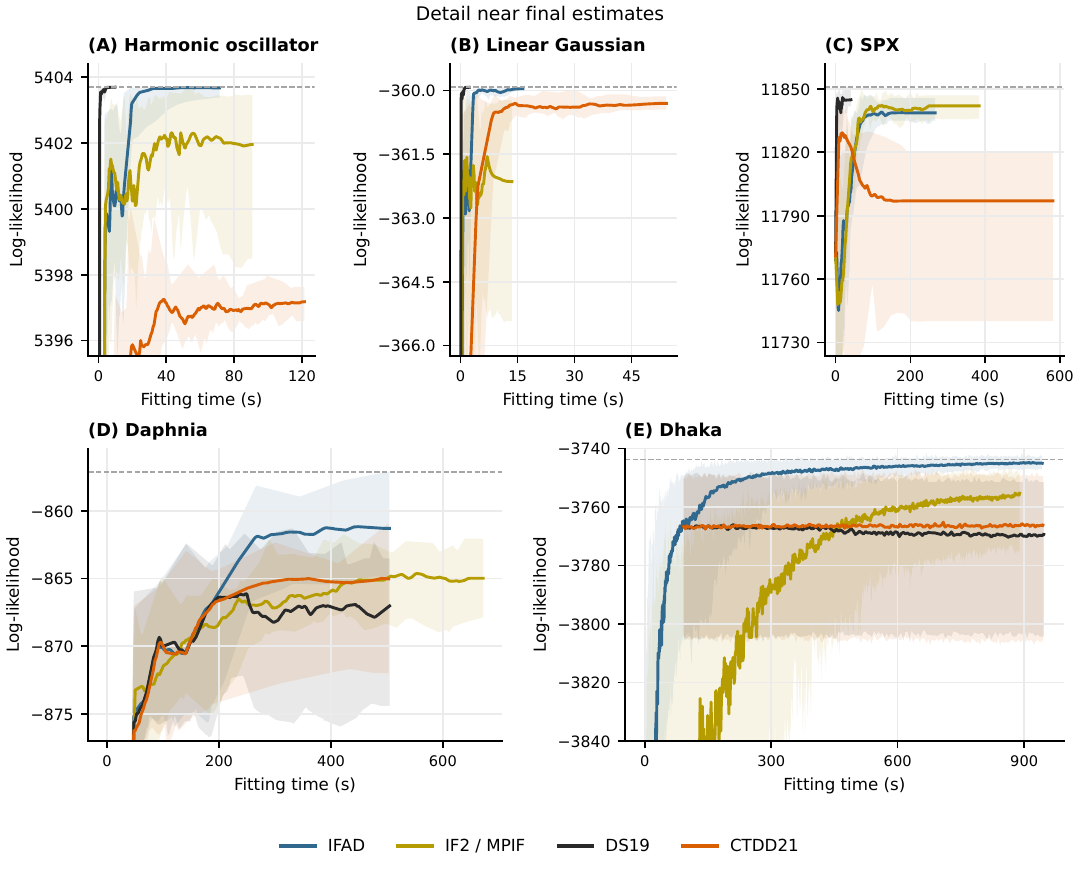}
\caption{Optimization progress near the final estimates.
Solid lines show the median log-likelihoods, and the shaded regions cover
the area between the 10th and 100th percentiles across searches.
We evaluated likelihoods at saved iterations and interpolated linearly
between them, retaining each search's final value after it stopped.
These evaluations were not used to select the fitted parameters.
The vertical axes are linear and restricted to the region near the
final estimates; Figure~\ref{fig:benchmark-optimization-full} shows
the full range, with the same shading opacity.
Panel E uses the range $[-3840,-3740]$.
The starting points, reference lines and final evaluations are those of
Figure~\ref{fig:benchmark-final-likelihoods}.}
\label{fig:benchmark-optimization-detail}
\end{figure}

\begin{figure}[p]
\centering
\includegraphics[width=\textwidth]{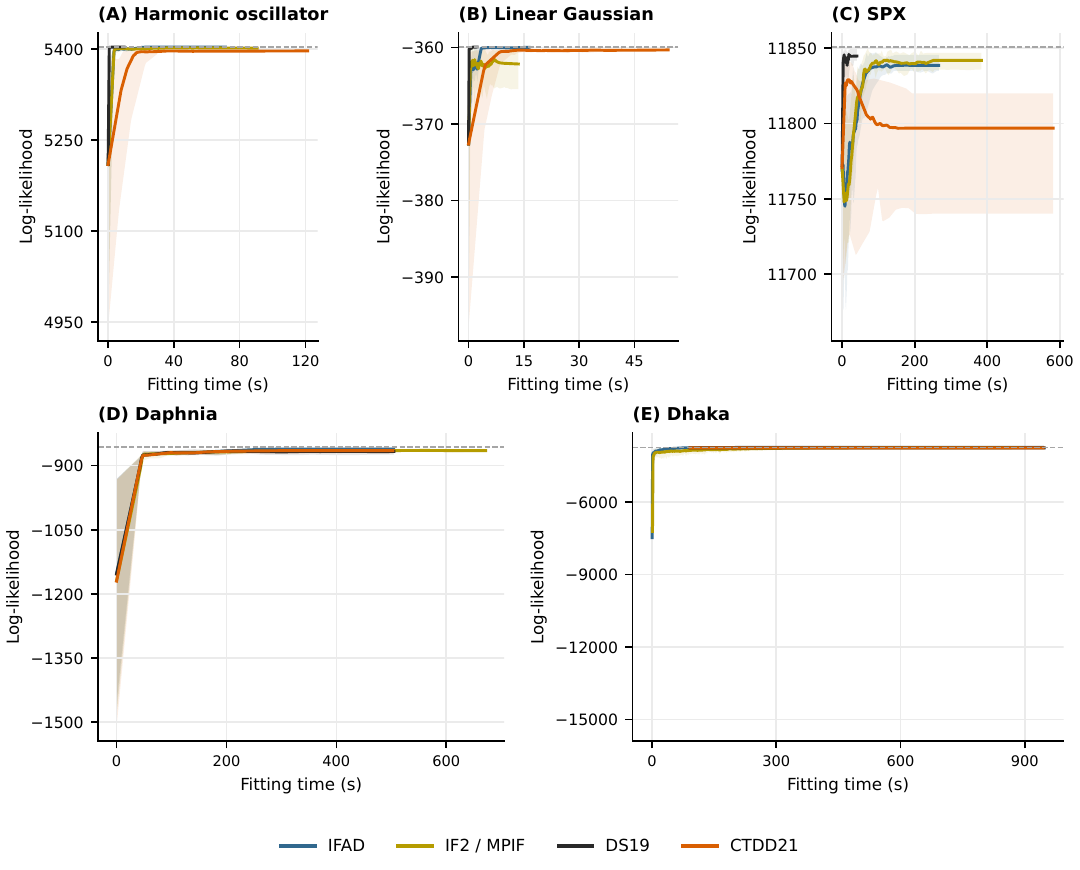}
\caption{Optimization progress over the full range of the median
trajectories and percentile bands. Lines, shading and reference lines are
defined in Figure~\ref{fig:benchmark-optimization-detail}.
All axes are linear, with separate ranges for each model and data set.
Computation times exclude compilation and independent likelihood
evaluation. The hardware and algorithmic settings are given in
Section~\ref{sec:benchmark-settings} and the preceding paragraph.}
\label{fig:benchmark-optimization-full}
\end{figure}

\begin{table}[p]
\centering
\small
\begin{tabular}{lrrrrrr}
\toprule
Method & Median & Maximum & IQR & MCSE & Time (s) & Failed \\
\midrule
\multicolumn{7}{l}{\textit{Harmonic oscillator (20 starts)}} \\
IFAD & 5403.67 & 5403.69 & 0.04 & 0.00 & 70.0 & 0 \\
IF2 & 5401.95 & 5403.46 & 1.80 & 0.00 & 89.4 & 0 \\
DS19 & 5403.69 & 5403.69 & 0.01 & 0.00 & 9.4 & 0 \\
CTDD21 & 5397.18 & 5397.63 & 0.28 & 0.00 & 117.3 & 0 \\
\addlinespace
\multicolumn{7}{l}{\textit{Linear Gaussian (20 starts)}} \\
IFAD & -359.96 & -359.94 & 0.03 & 0.00 & 16.1 & 0 \\
IF2 & -362.14 & -360.23 & 3.60 & 0.00 & 13.3 & 0 \\
DS19 & -359.93 & -359.92 & 0.01 & 0.00 & 2.1 & 0 \\
CTDD21 & -360.31 & -360.14 & 0.13 & 0.00 & 51.6 & 0 \\
\addlinespace
\multicolumn{7}{l}{\textit{SPX (20 starts)}} \\
IFAD & 11838.63 & 11845.69 & 4.95 & 0.31 & 144.5 & 0 \\
IF2 & 11841.96 & 11846.95 & 5.06 & 0.31 & 249.8 & 0 \\
DS19 & 11844.83 & 11850.85 & 4.11 & 0.47 & 21.5 & 0 \\
CTDD21 & 11796.97 & 11820.04 & 17.37 & 0.19 & 165.9 & 0 \\
\addlinespace
\multicolumn{7}{l}{\textit{Daphnia (20 starts)}} \\
IFAD & -861.30 & -857.13 & 2.57 & 0.19 & 487.8 & 0 \\
MPIF & -864.98 & -862.05 & 2.53 & 0.19 & 647.3 & 0 \\
DS19 & -867.02 & -863.52 & 7.37 & 0.20 & 488.1 & 1 \\
CTDD21 & -865.03 & -860.97 & 4.00 & 0.17 & 488.7 & 0 \\
\addlinespace
\multicolumn{7}{l}{\textit{Dhaka (100 starts)}} \\
IFAD & -3744.40 & -3743.68 & 0.95 & 0.14 & 942.8 & 0 \\
IF2 & -3755.41 & -3747.82 & 9.21 & 0.15 & 889.3 & 0 \\
DS19 & -3769.04 & -3751.03 & 16.53 & 0.17 & 942.8 & 0 \\
CTDD21 & -3765.91 & -3749.10 & 17.78 & 0.15 & 942.8 & 2 \\
\bottomrule
\end{tabular}

\caption{Repeated-search results for the five models.
Median, maximum and interquartile range (IQR) summarize final log-likelihoods
across searches. MCSE is the median Monte Carlo standard error of an
individual final log-likelihood estimate; it is zero for the analytic
Gaussian evaluations. Time is the median fitting time, including any IF2
or MPIF warm start and excluding compilation and independent evaluation.
Failed is the number of searches that stopped prematurely. We retained
their last estimates before failure or the time limit in all summaries.
Hardware and fitting settings differ
between experiments, as described in the text.}
\label{tab:benchmark-results}
\end{table}

\begin{figure}[p]
\centering
\includegraphics[width=\textwidth]{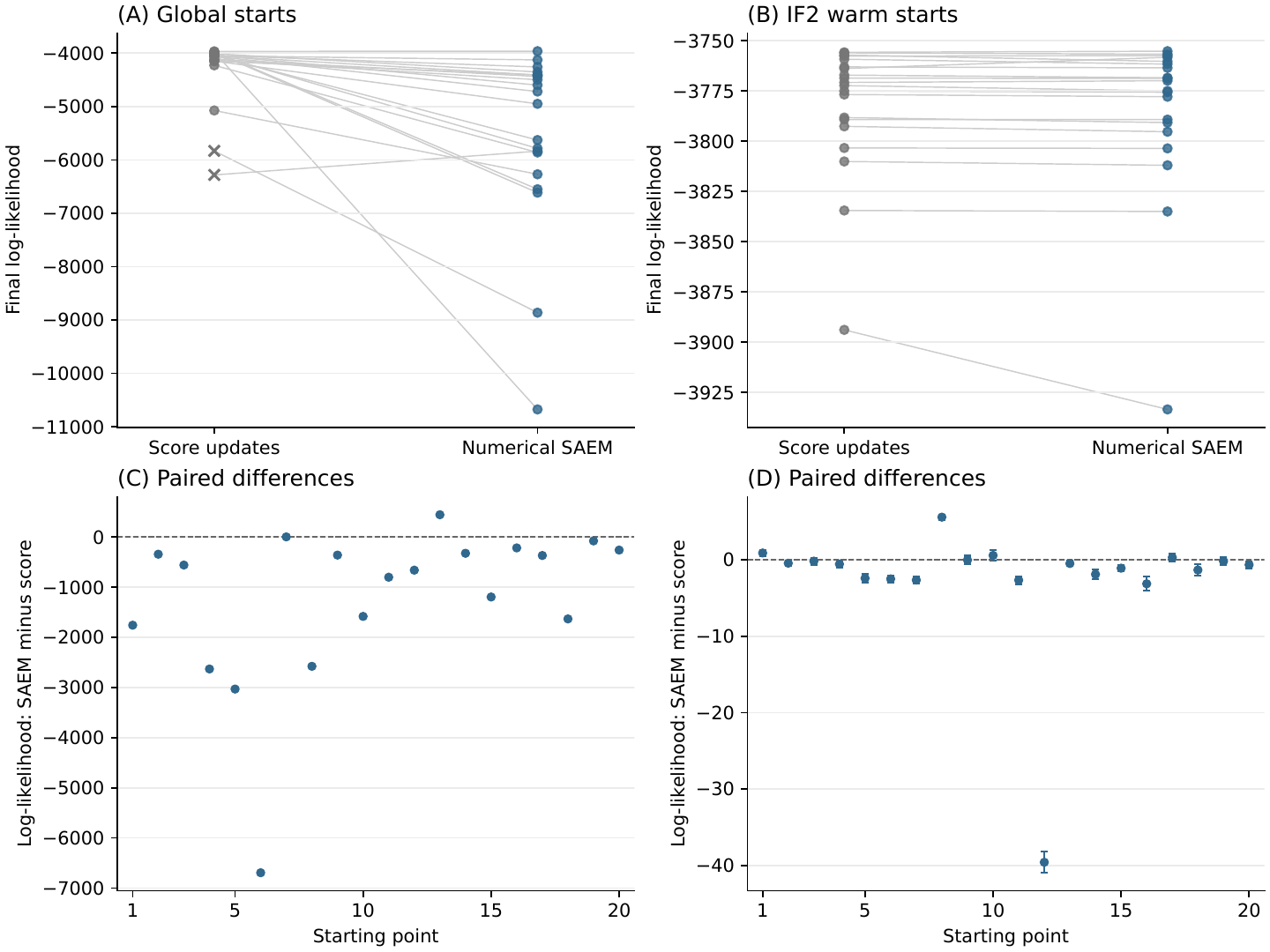}
\caption{Averaged-score updates and numerical SAEM on Dhaka, using the
same monthly Gaussian transition approximation and particle smoother.
Both methods use 1,000 particles and are allowed 850 seconds per search, from
20 global starts (A, C) and 20 saved IF2 estimates (B, D).
Panels A and B join final estimates from the same starting point.
Crosses mark searches that stopped prematurely, for which we retained
the last estimate before failure or the time limit.
Panels C and D show paired log-likelihood differences, numerical SAEM minus
score updates; bars show twice the paired Monte Carlo standard error.
Final likelihoods use the original Euler model, with 5,000 particles and
36 evaluation replicates. Evaluation random numbers are shared within a
pair and independent of fitting. All axes are linear.}
\label{fig:dhaka-saem-ab}
\end{figure}

\begin{table}[p]
\centering
{\small
\setlength{\tabcolsep}{4pt}
\begin{tabular}{llrrrrr}
\toprule
Initialization & Optimizer & Median $\ell$ & Median change & Iterations & Time (s) & Failed \\
\midrule
Global & Score updates & -4050.49 & 3827.36 & 86 & 854.5 & 2 \\
Global & Numerical SAEM & -4834.43 & 3109.60 & 15.5 & 851.3 & 0 \\
IF2 & Score updates & -3771.54 & -0.27 & 89 & 855.7 & 0 \\
IF2 & Numerical SAEM & -3772.44 & -1.37 & 13 & 851.1 & 0 \\
\bottomrule
\end{tabular}

}
\caption{Dhaka score and numerical-SAEM comparison in
Figure~\ref{fig:dhaka-saem-ab}.
Median $\ell$ is the median final Euler log-likelihood.
Median change is the median difference between each search's final and
initial log-likelihoods. Iterations is the median number of outer iterations
completed at the reported estimate, including iterations without a parameter change.
Time is the median computation time, including computation past the time
limit; we discarded updates that did not finish within the allotted time.
Compilation, independent evaluation and the cost of producing the saved
IF2 starting estimates are excluded.
For failed searches, we retained the last estimate before failure or the
time limit in all summaries.}
\label{tab:dhaka-saem-ab}
\end{table}

\clearpage